\documentclass[letterpaper,11pt]{article}
\usepackage[sort,compress]{cite}
\usepackage[top=1.1in, bottom=1.13in, left=1.1in, right=1.1in]{geometry}
\usepackage[utf8]{inputenc}
\usepackage{amsmath}
\usepackage{amssymb}
\usepackage{amsfonts}
\usepackage{amsthm}
\usepackage{graphicx}
\usepackage{tikz}
\usepackage{xspace}
\usepackage{subcaption}
\usepackage{color}
\usepackage{comment}
\usepackage{titling}
\usepackage{enumitem}
\usepackage[font={small},skip=0pt]{caption}
\usepackage{changepage}
\usepackage[export]{adjustbox}
\usepackage{mathtools}
\usepackage{rotating}
\usepackage{float}
\usepackage[normalem]{ulem}
\usepackage[textsize=tiny]{todonotes}
\usepackage{changepage} \usepackage{environ}
\usepackage{hyperref}

\newcommand{\rhsr}{\ensuremath{\mathcal R}}
\newcommand{\srhsr}{\ensuremath{\mathcal R\setminus\rho(\mathbb{R})}}
\newcommand{\hp}[3]{\ensuremath{\textrm{mid}{(#1_{#2} #1_{#3})}}} 

\usepackage[normalem]{ulem}

\newcommand{\Pik}{\ensuremath{P_{i+1,i+2,\ldots,k}}\xspace}
\newcommand{\Pjk}{\ensuremath{P_{j+1,j+2,\ldots,k}+\vect{P_j P_i}}\xspace}

\makeatletter
\newcommand\footnoteref[1]{\protected@xdef\@thefnmark{\ref{#1}}\@footnotemark}
\makeatother

\usetikzlibrary{decorations.pathmorphing}
\usetikzlibrary{shapes.misc} \usetikzlibrary{backgrounds}

\newcommand{\vect}{\protect\overrightarrow}

\newcommand{\resp}{respectively\xspace}
\newcommand{\termasm}[1]{\mathcal{A}_{\Box}[{#1}]}
\newcommand{\prodasm}[1]{\mathcal{A}[{#1}]}
\newcommand{\dom}[1]{{\rm dom}(#1)}
\newcommand{\Z}{\mathbb{Z}}
\newcommand{\N}{\mathbb{N}}
\newcommand{\R}{\ensuremath{\mathbb{R}}}
\newcommand{\calT}{\mathcal{T}}

\newcommand{\prodT}{\prodasm{\mathcal{T}}}
\newcommand{\prodpaths}[1]{{\bf{P}}[{#1}]}

\newcommand{\prodpathsT}{\prodpaths{\mathcal{T}}}

\newcommand{\rev}[1]{\ensuremath{{#1}^\leftarrow}}

\newcommand{\pathassembly}[1]{\mathrm{asm}{\left(#1\right)}}
\newcommand{\asm}[1]{\pathassembly{#1}}
\newcommand{\glue}[3]{\mathrm{glue}(#1_{#2} #1_{#3})}
\newcommand{\glu}[2]{\glue{#1}{#2}{#2+1}}
\newcommand{\glueP}[2]{\mathrm{glue}(P_{#1} P_{#2})}

\newcommand{\pos}[1]{\mathrm{pos}(#1)}
\newcommand\type[1]{\mathrm{type}(#1)}
\newcommand\defeq{\mathrel{\overset{\makebox[0pt]{\mbox{\normalfont\tiny\sffamily def}}}{=}}}
\newcommand{\gs}[2]{\ensuremath{\Big[#1,#2\Big]}}

\newcommand{\olq}{{\ensuremath{\overline{q}}}}

\newcommand\torture[4]{{#1}_{{#3}+1 + (({#2}-{#3}-1)\mod({#4}-{#3}))} + \left\lfloor \frac{{#2}-{#3}-1}{{#4}-{#3}}\right\rfloor\vect{{#1}_{#3}{#1}_{#4}}}

\newcommand\embed[1]{{\ensuremath{\frak{E}[#1] }}}

\newcommand{\reverse}[1]{\ensuremath{{#1}^\leftarrow}}
\newcommand{\xcoord}[1]{\mathrm{x}_{#1} }
\newcommand{\ycoord}[1]{\mathrm{y}_{#1} }

\newcommand{\vectwo}[2]{\tiny\begin{pmatrix}#1\\#2\end{pmatrix}}

\newcommand\vpij{\vect{P_iP_j}}
\newcommand\vpji{\vect{P_jP_i}}

\newcommand{\concat}[1]{\mathrm{concat}\!\left(#1\right)}
\newcommand\bound{\ensuremath{(8|T|)^{4|T|+1}(5|\sigma|+6)}\xspace}
\newcommand\boundHam{\ensuremath{(88|T|)^{4|T|+1}}\xspace}

\newtheorem{theorem}{Theorem}[section]

\newtheorem{lemma}[theorem]{Lemma}
\newcommand\sublname{claim\xspace}
\newcommand\sublName{Claim\xspace}
\newcommand\sublemmanames{\sublName{s}\xspace}
\newcommand\subl[1]{\sublName~\ref{#1}}
\newtheorem{sublemma}[theorem]{\sublName}
\newtheorem{corollary}[theorem]{Corollary}

\newtheorem{observation}[theorem]{Observation}
\newtheorem{fact}{Fact}  \numberwithin{fact}{theorem}
\theoremstyle{definition}
\newtheorem{definition}{Definition}  
\newcommand\argument[1]{{\vspace*{\baselineskip}\noindent\bf {#1}.\hspace*{1em}}}

\newcommand\range[3]{{#1},{#2},\ldots,{#3}}
\newcommand\rng[2]{\range {#1} {#1+1} {#2}}

\makeatletter
\newtheorem*{rep@theorem}{\rep@title}
\newcommand{\newreptheorem}[2]{\newenvironment{rep#1}[1]{ \def\rep@title{#2 \ref{##1}} \begin{rep@theorem}} {\end{rep@theorem}}}
\makeatother
\newreptheorem{theorem}{Theorem}
\newreptheorem{lemma}{Lemma}
\newreptheorem{corollary}{Corollary}

\newcommand\scale{0.4}

\setlist[itemize]{leftmargin=1.6em,topsep=0.3em,itemsep=0.1em} \setlist[enumerate]{leftmargin=1.6em,noitemsep,topsep=0.3em,itemsep=0.1em} 

\title{The program-size complexity of self-assembled paths\thanks{Supported by European Research Council (ERC) award number 772766 and Science foundation Ireland (SFI) grant 18/ERCS/5746 (this manuscript reflects only the authors' view and the ERC is not responsible for any use that may be made of the information it contains). Some of this work was supported by, and carried out at, Inria, Paris, France.}}

\author{Pierre-\'Etienne Meunier\thanks{Hamilton Institute, Maynooth University, Co.\ Kildare, Ireland. \href{mailto:pierre-etienne.meunier@tutanota.com}{pierre-etienne.meunier@tutanota.com}}\\
  \and
Damien Regnault\thanks{IBISC, Université Évry, Université Paris-Saclay, 91025, Evry, France. \href{mailto:damien.regnault@univ-evry.fr}{damien.regnault@univ-evry.fr}}\\
\and
Damien Woods\thanks{Hamilton Institute and Department of Computer Science, Maynooth University, Co.\ Kildare, Ireland. \href{mailto:damien.woods@mu.ie}{damien.woods@mu.ie}}\\
}
\date{}
\usepackage{chngcntr}

\begin{document}
\numberwithin{figure}{section} \maketitle
\begin{abstract}
We prove a Pumping Lemma for the \emph{noncooperative} abstract Tile Assembly Model, a model central to the theory of algorithmic self-assembly since the beginning of the field.  This theory suggests, and our result proves, that small differences in the nature of adhesive bindings between abstract square molecules gives rise to vastly different 
expressive capabilities.  

In the \emph{cooperative} abstract Tile Assembly Model, square tiles attach to each other using multi-sided cooperation of one, two or more sides.  This precise control of tile binding is directly exploited for algorithmic tasks including growth of specified shapes using very few tile types, as well as simulation of Turing machines and even self-simulation of self-assembly systems.  But are cooperative bindings required for these computational tasks? The definitionally simpler \emph{noncooperative (or Temperature 1)} model has poor control over local binding events: tiles stick if they bind on at least one side.  This has led to the conjecture that it is impossible for it to exhibit precisely controlled growth of computationally-defined shapes.

Here, we prove such an impossibility result.  We show that any planar noncooperative system that attempts to grow large algorithmically-controlled tile-efficient assemblies must also grow infinite non-algorithmic (pumped) structures with a simple closed-form description, or else suffer blocking of intended algorithmic structures.  Our result holds for both directed and nondirected systems, and gives an explicit upper bound of \bound, where $|T|$ is the size of the tileset and $|\sigma|$ is the size of the seed assembly, beyond which any path of tiles is pumpable or blockable.

\end{abstract}

\clearpage \section{Introduction}

The main challenge of molecular programming is to understand, build and control matter at the molecular level.
The dynamics of molecules can embed algorithms~\cite{Winf98} and the theory of algorithmic self-assembly~\cite{DotCACM,  PatitzSurvey,woods2015ntrinsic} is one formal way to think about the computational capabilities of autonomic self-assembling molecular systems. 
That theory, and more broadly the theory of models of computation, guides advances in experimental work to this day: 
the self-assembling binary counter of Winfree and Rothemund~\cite{RotWin00} was later implemented using tiles made of DNA~\cite{evans2014crystals},
as were  bit-copying systems~\cite{barish2005two,barish2009information,schulman2012robust} and
discrete self-similar fractals~\cite{RoPaWi04,FujHarParWinMur07}.
More recently, twenty-one self-assembly algorithms were implemented using DNA single-stranded tiles~\cite{WoodsDotyWinfree}, including 
a tile-based implementation~\cite{chalk2015flipping} of von Neuman's fair-bit-from-an-unfair-coin and a 3-bit  instance of the computationally universal cellular automata Rule 110~\cite{cook2004universality,nearyWoods2006rule110}. 
Besides guiding experiment, the theory itself has undergone significant developments, with the long-term vision of understanding the kinds of structure-building capabilities and computational mechanisms that are implementable and permitted by molecular processes.

Perhaps the most studied model of algorithmic self-assembly is the abstract Tile Assembly Model (aTAM), introduced by Winfree~\cite{Winf98} 
as a computational model of DNA tile-based self-assembly.
The model is an algorithmic version of Wang tilings~\cite{Wang61}, can be thought of as an asynchronous cellular automaton~\cite{bersini1994asynchrony} and has features seen in other distributed computing models.
In each instance of the model, we have a finite set of unit square tile \emph{types}, with colours on their four sides, and an infinite supply of each type. Starting from a small connected arrangement of tiles on the $\mathbb{Z}^2$ plane, called the \emph{seed assembly}, we attach tiles to that assembly asynchronously and nondeterministically based on a local rule depending only on the colour of the sides of the newly placed tile and of the sides of the assembly that are adjacent to the attachment position.

This model can simulate Turing machines~\cite{Winf98}, 
self-assemble squares with few tile types~\cite{RotWin00,Roth01}, 
assemble any finite spatially-scaled shape using a small, Kolmogorov-efficient, tile set~\cite{SolWin07}, and there is an intrinsically universal tile set capable of simulating the behaviour (produced shapes and growth dynamics) of any other tile set, up to spatial rescaling~\cite{IUSA}.

However, these results have all been proven using the so-called {\em cooperative} tile assembly model. 
In the cooperative, or temperature 2, model there are two kinds of bonds:  strong and weak. A tile can attach to an assembly by one side if that side forms a strong  (``strength 2'') bond with the assembly, or it can attach if two of its sides each match with a weak (``strength 1'') bond.
Intuitively, the cooperative model exploits weak bonds to create a form of synchronisation. The attachment of a tile by two weak glues sticking to two neighbour tiles can only occur {\em after} both neighbour tiles are present, allowing the system to {\em wait} rather than (say) proceeding with potentially inaccurate or incomplete information.

But what happens if we allow only one kind of bond? We get a simple model called the  \emph{non-cooperative}, or \emph{temperature 1}, model.
In the non-cooperative model, tiles may attach to the assembly whenever at least one of their sides' colour matches the colour of a side of the assembly adjacent to the position where they attach. Intuitively, they need not wait for more bonds to appear adjacent to their attachment position.
In this model it seems non-obvious how to implement synchronisation;  we have no {\em obvious} programmable feature that enables one  growth process to wait until another is complete. Our attempts to build such things typically lead to rampant uncontrolled growth.

The question of whether the non-cooperative (or ``non-waiting'') model has {\em any} non-trivial computational abilities has been open since the beginning of the field~\cite{RotWin00}.
Perhaps a reason for this is that actually proving that one can not synchronise growth is tricky; maybe noncooperative self-assembly can somehow simulate synchronisation using some complicated form of in-plane geometric blocking?
Restrictions of the model have been shown to be extremely weak~\cite{RotWin00, ManuchTemp1},
generalisations shown to be extremely powerful~\cite{Versus,2HAMIU,OneTile,geotiles,Cook-2011,Roth01,Patitz-2011,BMS-DNA2012a,Signals,Jonoska2014,Fekete2014,Hendricks-2014,gilbert2015continuous,SFTSAFT},
and, to further deepen the mystery,
the model has been shown capable of some (albeit limited) efficient tile reuse~\cite{meunier2015,Meunier2019}.

\subsection{Main result}
\newcommand{\introtheorem}{Let $\mathcal T=(T,\sigma,1)$ be any 
tile assembly system in the abstract Tile Assembly Model (aTAM), and let $P$ be a path producible by  $\mathcal T$. 
If $P$ has vertical height or horizontal width at least \bound, then $P$ is pumpable or fragile.}

\newcommand{\introtheoremTwoHAM}{Let $\mathcal H=(T,1)$ be any 
tile assembly system in the Two-Handed  Assembly Model (2HAM), and let $P$ be a path producible by  $\mathcal H$. 
If $P$ has vertical height or horizontal width at least \boundHam, then $P$ is pumpable or fragile.}

Our main result is stated in Theorem~\ref{thm:intro main thm}, although a number of notions have yet to be formally defined (see Section~\ref{sec:defs} for definitions).
Intuitively if a noncooperative tile assembly system produces a large assembly, it is capable of also producing any path of tiles in that assembly.
Our statement says that if the tile assembly system can produce a long enough path $P$, then it must also produce  assemblies where either
an infinite ultimately periodic path appears ($P$ is ``pumpable''), or else
the path cannot appear in all terminal assemblies (because some other tiles can be placed to block the growth of $P$), in which case we say that $P$ is \emph{fragile}.
Let $T$ be a set of tile types, and let $|\sigma|$ denote the number of tiles in the (seed) assembly $\sigma$.
\begin{theorem}\label{thm:intro main thm}\introtheorem
\end{theorem}

Our result rules out the kind of Turing machine simulations, and other kinds of computations, that have appeared in the literature to date and execute precisely controlled growth patterns~\cite{Versus,2HAMIU,OneTile,geotiles,Cook-2011,Roth01,Patitz-2011,BMS-DNA2012a,Signals,Jonoska2014,Fekete2014,Hendricks-2014,gilbert2015continuous,ManuchTemp1,SFTSAFT}.
We do so by showing that any attempt to carry out long computations such as these, which provably require large assemblies and thus long paths, will result in unbounded pumped growth or the blocking of such paths.

The essence of algorithmic self-assembly is tile reuse: growing structures that are larger than the number of available tile types~\cite{RotWin00,DotCACM,PatitzSurvey,woods2015ntrinsic}. 
 Meunier and Regnault~\cite{Meunier2019} show that some noncooperative systems are capable of tile reuse in the following sense:
there is a tile assembly system with multiple terminal assemblies, all of finite size, such that each of them contains the same long path $P$, where $P$ is of width $O(|T| \log |T|)$ (i.e.\ larger than $|T|$).
In that construction, $P$ is neither pumpable (all assemblies are finite) nor fragile ($P$ appears in all terminal assemblies, hence no assembly or path can block it).
Their  result should be contrasted with ours since here we show that any attempt to generalise such a construction beyond our exponential-in-$|T|$ bound will fail -- thus we give a limitation of the amount of tile reuse possible in such constructions.
Analogous tile reuse limitations do not appear in cooperative systems~\cite{RotWin00},
due to their ability to run arbitrary algorithms.

Our theorem statement is quite similar to the pumpability conjecture of Doty, Patitz and Summers~\cite{Doty-2011} (Conjecture 6.1). In that work~\cite{Doty-2011} under the assumption that the conjecture is true, they achieve a complete characterisation of the producible assemblies.
Our result is slightly different from that conjecture, being stronger in two ways, and weaker in one:
\begin{itemize}
\item First, their conjecture was stated for {\em directed} systems (that produce a single terminal assembly), but here we prove the result for both directed and undirected systems (systems that produce many terminal assemblies).\footnote{Directed systems do not have fragile paths, hence for directed systems the conclusion of the theorem statement  systems simply reads ``... then $P$ is pumpable''.} In the directed case, our result shows that any attempt to simulate computations by growing longer and longer paths is doomed to also produce a terminal assembly littered with more and more infinite pumped paths.   
\item Second, we give an explicit bound (exponential in $|T|$) on the length a path can reach before it is pumpable or fragile.
\item Our result is weaker in one way: indeed, while our result only applies to paths grown all the way from the seed, the conjecture is that arbitrary paths are pumpable, regardless of their position relative to the seed.

  However, we conjecture that our result is sufficient to achieve the same characterisation of producible assemblies, using the same techniques and arguments as~\cite{Doty-2011}.
\end{itemize}

Our result can be applied to other models. After the aTAM, another well-studied model in the theory of algorithmic self-assembly is the hierarchical, or two-handed, model (2HAM)~\cite{AGKS05g,Versus,SFTSAFT,2HAMIU}.
There is no seed assembly in the 2HAM:  in the noncooperative (temperature 1) 2HAM  tiles stick together if glues match on one tile side, forming a collection of assemblies, and those assemblies can in turn stick to each other if they can be translated to touch without overlapping and with adjacent matching glues between them. As an almost direct corollary of Theorem~\ref{thm:intro main thm} we get:
\begin{corollary}
  \label{cor:2ham}
  \introtheoremTwoHAM
\end{corollary}

Intuitively, the corollary comes from the fact that aTAM-like growth is permitted in the 2HAM. In fact, if we fix a tile set $T$ and a temperature of 1, the set of producible assemblies in a 2HAM system over $T$  is a superset (sometimes a strict superset~\cite{Versus}) of those in the aTAM over $T$. A brief proof sketch is given in Appendix~\ref{sec:2HAM}.

\subsection{Relationship with other prior work}\label{sec:related work}

Perhaps due to the difficulty of analysing the standard  noncooperative model  (2D, square tiles, tiles attach one at a time), researchers have looked at different variants of that model.

The first restriction studied was where we permit only terminal assemblies that are ``fully connected'' meaning all tiles are fully bound to all of their neighbour tiles):  
Rothemund and Winfree~\cite{RotWin00}  showed that for each large enough $n\in\mathbb{N}$ there does not exist a noncooperative system that builds a fully-connected $n\times n $ square in a tile-efficient way (using $< n^2$ tile types).
Since embedding algorithms in tiles are essentially our best (and perhaps only) way to exhibit efficient tile reuse, our result shuts the door on a wide class of algorithmic approaches.  Other restrictions proven weak are where we disallow adjacent mismatching colours~\cite{ManuchTemp1}, or even force any pair of adjacent tiles to bind to each other~\cite{RotWin00}.

Another, quite productive, approach has been to study generalisations of the noncooperative model (e.g.\ 3D, non-square tiles, multi-tile assembly steps, more complicated `active' tiles, etc.): it turns out that these  generalisations  and others  are powerful enough to simulate Turing machines~\cite{Versus,2HAMIU,OneTile,geotiles,Cook-2011,Roth01,Patitz-2011,BMS-DNA2012a,Signals,Jonoska2014,Fekete2014,Hendricks-2014,gilbert2015continuous,ManuchTemp1,SFTSAFT,furcy2018optimal,furcy2019newBounds}.
Each such result points to a specific feature or set of features in a generalised model that provably needs to be exploited in order to avoid our negative result. 

Since Cook, Fu, and Schweller~\cite{Cook-2011} have shown that for any Turing machine there is an undirected tile assembly system, whose seed encodes an input, and where the largest producible terminal assembly (which is possibly infinite) simulates the Turing machine computation  on that input. 
However, in that construction, ``blocking errors'' can occur where growth is prematurely blocked and is stopped before the simulation, or computation, is completed (hence their result is stated in a probabilistic setting).
Indeed
their tile assembly systems that
simulate Turing machines will produce many such erroneous assemblies. Our result shows that this kind of blocking is unavoidable.

Assemblies that cannot be ``blocked'' have the opposite issue, where it seems there is always a part of the assembly that can be repeated forever, which led Doty, Patitz and Summers~\cite{Doty-2011} to their pumpability conjecture. They go on to show  that, assuming the pumpability conjecture holds, projections to the vertical/horizontal axes of assemblies produced by directed noncooperative systems have a straightforward closed-form description (as the union of semi-linear sets). 

In the direction of negative results on Temperature 1, reference~\cite{STOC2017}
showed that the noncooperative planar aTAM is not capable of simulating -- in the sense used in intrinsic universality~\cite{woods2015ntrinsic} -- other noncooperative aTAM systems. In other words intrinsic universality is not possible for the planar noncooperative model. 
The Temperature 2 (or, cooperative) model is capable of such simulations~\cite{IUSA}, hence the main result of~\cite{STOC2017} shows a difference in the self-simulation capabilities of the two models. 
Prior to that work, another result showed that Temperature 1 cannot simulate Temperature 2 intrinsically~\cite{SODA2014}, and hence the former is strictly weaker than the later in this setting (where we ignore spatial scaling).
However, to obtain both of those results, the use of simulation and intrinsic universality permitted a technique where we choose a particular class of shapes we want to simulate, and restrict the analysis of produced assemblies to (scaled versions of) these shapes.
The proofs~\cite{STOC2017,SODA2014} then involved forcing certain paths to grow outside of these pre-determined shapes.
This paper makes use of a number of tools from~\cite{STOC2017}. Here, however, the setting is significantly more challenging as we have no geometric hypotheses whatsoever on producible assemblies and therefore can not directly leverage~\cite{STOC2017},
although we do make use of the tools of visibility and the notion of right/left priority exploited in that prior work. As already noted, \emph{positive} constructions have been found that some limited form of efficient reuse of tile types is possible in the standard 2D noncooperative model~\cite{meunier2015,Meunier2019}.
This paper also builds on extensive previous work by two of the authors~\cite{pumpabilityLargeBound}.

\subsection{New tools and future work}

We develop a collection of new tools to reason about paths in  $\mathbb{Z}^2$.
In order to carry out computation in tile-assembly, a key idea is to reuse tile types (analogous to how a Boolean circuit reuses gates of a given type, or how a computer program reuses instructions via loops).
Our main technical lemma, which we term the ``shield'' lemma, shows that if a path $P$ has a certain form that reaches so far to the east that it reuses some tile types, then we can use $P$ to construct a curve in $\R^2$ that is an almost-vertical cut~$c$ of the plane.
We use this cut of the plane to show one of two things must happen: either that iterations of a pumping of $P$ are separated from one another, and hence that the pumping is simple, which means that $P$ is pumpable, or else that a path can be grown that blocks the growth of $P$. 

This cut, and our subsequent argument can be thought of as a kind of ``Window Movie Lemma''~\cite{SODA2014}, or pumping tool, but targeted specifically at noncooperative self-assembly.

The shield lemma can only be applied when $P$ is of a certain form. 
Our second tool is a combinatorial argument on the height and width of a path. We begin by applying some straight-froward transformations to put any wide enough or tall enough path (i.e., that meets the hypothesis of Theorem~\ref{thm:intro main thm}) into a form that its last tile is also its eastern-most tile. 
Then, for any such path $P$, our combinatorial argument shows that we can always find a cut that satisfies the hypothesis of the shield lemma. 
We hope that the techniques developed in this paper can be applied to other self-assembly models -- as an example we apply them to the 2HAM (Corollary~\ref{cor:2ham}). 

Although we answer one of the main unresolved questions on noncooperative self-assembly, our result does not close all questions on the model.
First, there is still a big gap between the best known lower bounds on the size of assemblies (which is $O(|T|\log |T|)$), and the exponential bound \bound we prove here. It remains as future work to reduce that gap.
We also conjecture that our result can be used to characterise the assemblies producible by directed systems, using the same argument as~\cite{Doty-2011}.
A number of decidability questions are also still open such as: can we decide whether a planar nondirected tile assembly system is directed (i.e. produces exactly one terminal assembly)?
An important part of self-assembly is related to building shapes (as opposed to decidability questions). 
In this direction, can we build $n\times n$ squares any more efficiently than the best known result of $2n-1$ tile types for noncooperative systems? For this common benchmark of shape-building, cooperative systems achieve a tileset size as low as $\Theta{(\log n / \log \log n)}$~\cite{RotWin00,AdChGoHu01,furcy2019newBounds}. 

\subsection{Structure of the paper}
Section~\ref{sec:defs} contains key  definitions and notions required for the paper. 
Section~\ref{sec:shield main section} proves our main technical lemma (Lemma~\ref{lem:shield}), which we call the Shield lemma.
The section begins with the definition of a shield of a path (a triple of indices of tiles along the path) and then
gives some high-level intuition for the proof of Lemma~\ref{lem:shield}.
The section proceeds to prove quite a number of claims (Claims~\ref{lem:c-cuts} to \ref{lem:u0v0}) that provide an arsenal of technical tools to reason about paths that have a shield.
The section ends with an inductive argument that proves the Shield lemma, making use of the previous claims.
Finally, Section~\ref{sec:main thm proof} proves our main result, Theorem~\ref{thm:intro main thm}, using a combinatorial argument to show that every wide enough path has a shield. It begins with Subsection~\ref{sec:main thm intuition} which contains the Intuition behind the proof of Theorem~\ref{thm:intro main thm} ---  some may find it helpful to begin reading there. 

Appendix~\ref{sec:2HAM} proves the 2HAM result (Corollary~\ref{cor:2ham}) and 
Appendix~\ref{sec:app:wee lemmas}  states a particular version of the Jordan curve theorem for infinite polygonal curves and states related definitions (one curve turning from another, left-/right-hand side of the plane).

\section{Definitions and preliminaries}\label{sec:defs}
\label{definitions}

As usual, let $\mathbb{R}$ be the set of real numbers, let $\mathbb{Z}$ be the integers, and let $\mathbb{N}$ be the natural numbers including 0.  
The domain of a function $f$ is denoted $\dom{f}$, and its range (or image) is denoted $f(\dom{f})$.

\subsection{Abstract tile assembly model}\label{sec:atam}

The abstract tile assembly model was introduced by Winfree~\cite{Winf98}. In this paper we study a restriction of the abstract tile assembly model called the temperature~1 abstract tile assembly model, or noncooperative abstract tile assembly model. For definitions of the full model, as well as intuitive explanations, see for example~\cite{RotWin00,Roth01}.

A \emph{tile type} is a unit square with four sides,
each consisting of a glue \emph{type} and a nonnegative integer \emph{strength}. Let  $T$  be a a finite set of tile types.
In any set of tile types used in this paper, we assume the existence of a well-defined total ordering which we call the {\em canonical ordering} of the tile set. The sides of a tile type are respectively called  north, east, south, and west, as shown  in the following picture:
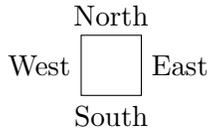
\begin{center}
\vspace{-1ex}
\begin{tikzpicture}[scale=0.8]
\draw(0,0)rectangle(1,1);
\draw(0,0.5)node[anchor=east]{West};
\draw(1,0.5)node[anchor=west]{East};
\draw(0.5,0)node[anchor=north]{South};
\draw(0.5,1)node[anchor=south]{North};
\end{tikzpicture}
\vspace{-1ex}\end{center}

An \emph{assembly} is a partial function $\alpha:\mathbb{Z}^2\dashrightarrow T$ where $T$ is a set of tile types and the domain of $\alpha$ (denoted $\dom{\alpha}$) is connected.\footnote{Intuitively, an assembly is a positioning of unit-sized tiles, each from some set of tile types $T$, so that their centers are placed on (some of) the elements of the discrete plane $\mathbb{Z}^2$ and such that those elements of $\mathbb{Z}^2$ form a connected set of points.} 
We let $\mathcal{A}^T$ denote the set of all assemblies over the set of tile types $T$. 
In this paper, two tile types in an assembly are said to  {\em bind} (or \emph{interact}, or are
\emph{stably attached}), if the glue types on their abutting sides are
equal, and have strength $\geq 1$.  An assembly $\alpha$ induces an undirected
weighted \emph{binding graph} $G_\alpha=(V,E)$, where $V=\dom{\alpha}$, and
there is an edge $\{ a,b \} \in E$ if and only if the tiles at positions $a$ and $b$ interact, and
this edge is weighted by the glue strength of that interaction.  The
assembly is said to be $\tau$-stable if every cut of $G_{\alpha}$ has weight at
least $\tau$.

A \emph{tile assembly system} is a triple $\mathcal{T}=(T,\sigma,\tau)$,
where $T$ is a finite set of tile types, $\sigma$ is a $\tau$-stable assembly called the \emph{seed}, and
$\tau \in \mathbb{N}$ is the \emph{temperature}.
Throughout this paper,  $\tau=1$.

Given two $\tau$-stable assemblies $\alpha$ and $\beta$, we say that $\alpha$ is a
\emph{subassembly} of $\beta$, and write $\alpha\sqsubseteq\beta$, if
$\dom{\alpha}\subseteq \dom{\beta}$ and for all $p\in \dom{\alpha}$,
$\alpha(p)=\beta(p)$.
We also write
$\alpha\rightarrow_1^{\mathcal{T}}\beta$ if we can obtain $\beta$ from
$\alpha$ by the binding of a single tile type, that is:  $\alpha\sqsubseteq \beta$, $|\dom{\beta}\setminus\dom{\alpha}|=1$ and the tile type at the position $\dom{\beta}\setminus\dom{\alpha}$ stably binds to $\alpha$ at that position.  We say that~$\gamma$ is
\emph{producible} from $\alpha$, and write
$\alpha\rightarrow^{\mathcal{T}}\gamma$ if there is a (possibly empty)
sequence $\alpha_1,\alpha_2,\ldots,\alpha_n$ where $n \in \N \cup \{ \infty \} $, $\alpha= \alpha_1$ and $\alpha_n =\gamma$, such that
$\alpha_1\rightarrow_1^{\mathcal{T}}\alpha_2\rightarrow_1^{\mathcal{T}}\ldots\rightarrow_1^{\mathcal{T}}\alpha_n$. 
A sequence of $n\in\mathbb{Z}^+ \cup \{\infty\}$ assemblies
$\alpha_0,\alpha_1,\ldots$ over $\mathcal{A}^T$ is a
\emph{$\mathcal{T}$-assembly sequence} if, for all $1 \leq i < n$,
$\alpha_{i-1} \to_1^\mathcal{T} \alpha_{i}$.

The set of \emph{productions}, or \emph{producible assemblies}, of a tile assembly system $\mathcal{T}=(T,\sigma,\tau)$ is the set of all assemblies producible
from the seed assembly $\sigma$ and is written~$\prodasm{\mathcal{T}}$. An assembly $\alpha$ is called \emph{terminal} if there is no $\beta$ such that $\alpha\rightarrow_1^{\mathcal{T}}\beta$. The set of all terminal assemblies of $\mathcal{T}$ is denoted~$\termasm{\mathcal{T}}$.  

\subsection{Paths and non-cooperative self-assembly}\label{sec:defs-paths}

This section introduces quite a number of key definitions and concepts that will be used extensively throughout the paper. 

Let $T$ be a set of tile types. 
A {\em tile} is a pair $((x,y),t)$ where $(x,y) \in \mathbb{Z}^2$ is a position and $t\in T$ is a tile type. 
Intuitively, a path is a finite or one-way-infinite simple (non-self-intersecting) sequence of tiles placed on points of $\mathbb{Z}^2$ so that each tile in the sequence interacts with the previous one, or more precisely: 

\begin{definition}[Path]\label{def:path}
  A {\em path} is a (finite or infinite) sequence  $P = P_0 P_1 P_2 \ldots$  of tiles   $P_i = ((x_i,y_i),t_i) \in \mathbb{Z}^2 \times T$, such that:
\begin{itemize}
\item for all $P_j$ and $P_{j+1}$ defined on $P$, 
their positions $(x_j,y_j)$ and $(x_{j+1}, y_{j+1})$ are adjacent nodes in the grid graph of $\mathbb{Z}^2$, moreover~$t_{j}$ and~$t_{j+1}$ interact (have matching glues on their abutting sides), and
\item for all $P_j,P_k$ such that $j\neq k$ it is the case that $ (x_j,y_j) \neq (x_k,y_k)$.
\end{itemize}
\end{definition}

By definition, paths are simple (or self-avoiding), and this fact will be repeatedly used throughout the paper.
For a tile $P_i$ on some path $P$, its x-coordinate is denoted~$\xcoord{P_i}$ and its y-coordinate is denoted~$\ycoord{P_i}$. 
The \emph{concatenation} of two paths $P$ and $Q$ is the concatenation $PQ$ of these two paths as sequences, and $PQ$ is a path if and only if (1) the last tile of $P$ interacts with the first tile of $Q$ and (2)  $P$ and $Q$ do not intersect each other.

For a path $P = P_0 \ldots  P_i P_{i+1} \ldots P_j  \ldots $, we define the notation $P_{i,i+1,\ldots,j} = P_i P_{i+1} \ldots P_j$, i.e.\  ``the subpath of $P$ between indices $i$ and $j$, inclusive''.
Whenever $P$ is finite, i.e. $P = P_0P_1P_2\ldots P_{n-1}$ for some $n\in\mathbb{N}$, $n$ is termed the {\em length} of $P$ and denoted by $|P|$. In the special case of a subpath where $i=0$, we say that $P_{0,1,\ldots,j}$ is a prefix of $P$ and that $P$ is \emph{an extension} of $P_{0,1,\ldots,j}$. The prefix or extension are said to be \emph{strict} if $j < |P|-1$. Else, when $j=|P|-1$, we say that $P_{i,\ldots, |P|-1}$ is a suffix of $P$, and is a \emph{strict suffix} of $P$ if $i > 0$.

For any path $P = P_0 P_1 P_2, \ldots$ and integer $i\geq 0$, we write $\pos{P_i} \in \mathbb{Z}^2$, or  $(\xcoord{P_i},\ycoord{P_i}) \in \mathbb{Z}^2$, for the position of $P_i$ and $\type{P_i}$ for the tile type of $P_i$. Hence if  $P_i = ((x_i,y_i),t_i) $ then $\pos{P_i} =  (\xcoord{P_i},\ycoord{P_i}) = (x_i,y_i) $ and $\type{P_i} = t_i$.
A ``\emph{position of}  $P$'' is an element of $\mathbb{Z}^2$ that appears in $P$ (and therefore appears exactly once), and an \emph{index} $i$ of a path $P$ of length $n\in \mathbb{N}$ is a natural number  $i \in \{0,1,\ldots,n-1\}$.
For a path $P=P_0P_1P_2\ldots$ we write $\pos{P}$ to mean ``the sequence of positions of tiles along $P$'', which is $\pos{P}=\pos{P_0}\pos{P_1}\pos{P_2}\ldots\;$.

Although a path is not an assembly, we know that each adjacent pair of tiles in the path sequence interact implying that the set of path positions forms a connected set in $\Z^2$ and hence every path uniquely represents an assembly containing exactly the tiles of the path, more formally:
for a path $P = P_0 P_1 P_2 \ldots$ we define the set of tiles  $\pathassembly{P} = \{ P_0, P_1, P_2, \ldots\}$ which we observe is  an assembly\footnote{I.e.  $\pathassembly{P}$ is  a partial function from $\Z^2$ to tile types, and is defined on a connected set.} and  we call $\pathassembly{P}$ a {\em path assembly}.
Given a tile assembly system $\calT = (T,\sigma,1)$ the path $P$ is a {\em  producible path of $\calT$} if 
$\asm{P}$ does not intersect\footnote{Formally, non-intersection of a path $P = P_0 P_1, \ldots $ and a seed assembly $\sigma$ is defined as: $\forall t$ such that  $t \in \sigma$, $\nexists i $ such that $\pos{P_i} = \pos{t}$.} the seed $\sigma$
and the assembly $(\asm{P} \cup \sigma ) $ is producible by $\calT$, i.e.\ $(\asm{P} \cup \sigma ) \in \prodT$, and $P_0$ interacts with a tile of $\sigma$.
As a convenient abuse of notation we sometimes write $\sigma \cup P$ as a shorthand for $\sigma \cup \asm{P}$.
Given  a tile assembly system $\calT = (T,\sigma,1)$
we define  the set of producible paths of $\calT$ to be:\footnote{Intuitively, although producible paths are not assemblies, any  producible path $P$ has the nice property that it encodes an unambiguous description of how to grow $\asm{P}$ from the seed $\sigma$, in path  ($P$) order, to produce  the assembly $  \asm{P}\cup \sigma$.}
$$\prodpathsT = \{ P  \mid P \textrm{ is a path that does not intersect } \sigma \textrm{ and } (\asm{P}\cup\sigma) \in \prodT \} $$

So far, we have defined paths of tiles (Definition~\ref{def:path}). In our proofs, we will also reason about (untiled) {\em binding paths} in the binding graph of an assembly.

\begin{definition}[Binding path]\label{def:binding graph}
Let $G=(V,E)$ be a binding graph. 
A {\em binding path} $q$ in $G$ is a sequence $q_{\range{0}{1}{|q|-1}}$ of vertices from $V$ such that 
for all $i \in \{\range{0}{1}{|q|-2} \}$,
$\{ q_i, q_{i+1} \} \in E$ ($q$~is connected) and no vertex appears twice in $q$ ($q$ is simple).
\end{definition}

\begin{observation}
Let $\calT = (T,\sigma,1)$ be a tile assembly system and let $\alpha \in \prodT$.
For any tile $((x,y),t) \in \alpha$ either $((x,y),t)$ is a tile of $\sigma$ or else there is a producible path $P \in \prodpathsT$  that for some $j \in \N$ contains $P_j = ((x,y),t)$.
\end{observation}
\begin{proof}
If $((x,y),t)$ is a tile of $\sigma$ we are done. Assume otherwise for the rest of the proof. 

  Since $\dom\alpha$ is a connected subset of $\Z^2$  there is an integer $n\geq 0$ and a binding path $p_{0,1,\ldots ,n}$ in the binding graph of $\alpha$ where $p_0\in\dom\sigma$ and $p_n = (x,y)$.
  Let $i_0$ be the largest integer such that $p_{i_0}\in\dom\sigma$.
  We can then define $P $ as the path:
  $$P = P_{0,1,\ldots,n - (i_0+1)} = (p_{i_0+1},\alpha(p_{i_0+1}))(p_{i_0+2},\alpha(p_{i_0+2}))\ldots(p_{n},\alpha(p_{n}))$$
  By definition of binding graph, for all $i \in \{i_0+1,i_0+2,\ldots,n-1 \}$, the tiles 
  $(p_{i},\alpha(p_{i})) = P_{i-(i_0+1)}$ and $(p_{i+1},\alpha(p_{i+1})) = P_{i+1 - (i_0+1) }$ on $P$ are adjacent in $\mathbb{Z}^2$ and interact on their abutting sides. 
  Moreover, by the definition of $i_0$, $P$ does not intersect $\sigma$ and the tile $(p_{i_0+1},\alpha(p_{i_0+1})) =  P_0$ (the first tile of $P$) interacts with a tile of $\sigma$, meaning that $P\in\prodpathsT$, thus proving the statement. 
  \end{proof}

For $A,B\in\mathbb{Z}^2$, we define $\vect{AB} = B - A$ to be the vector from $A$ to $B$, and  for two \emph{tiles} $P_i = ((x_i,y_i),t_i)$ and $P_j = ((x_j,y_j),t_j)$ we define $\vect{P_i P_j} = \pos{P_j} - \pos{P_i}$ to mean the vector from $\pos{P_i}=(x_i,y_i)$ to $\pos{P_j}=(x_j,y_j)$.
The translation of a path $P$ by a vector $\vect{v} \in \mathbb{Z}^2$, written $P+\vect{v}$, is  the path $Q$ such that $|P|=|Q|$
and for all indices $i \in \{0,1,\ldots,|P|-1 \}$,
$\pos{Q_i}=\pos{P_i}+\vect{v}$ 
and 
$\type{Q_i}=\type{P_i}$. 
We always use parentheses for scoping of translations when necessary, i.e.\ $P(Q+\vect{v})$ is the sequence containing the path $P$ followed by the translation of the entire path $Q$ by vector $\vect{v}$. 
The translation of an assembly $\alpha$ by a vector $\vect{v}$, written $\alpha+\vect{v}$, is the assembly $\beta$ defined for all $(x,y)\in(\dom\alpha+\vect{v})$ as $\beta(x,y)=\alpha((x,y)-\vect{v})$.

Let $P$ be a path, let $i\in\{1,2,\ldots,|P|-2\}$, and let $A\neq P_{i+1}$ be a tile such that $P_{0,1,\ldots,i}A$ is a path.
Let also $\rho$ be the clockwise rotation matrix defined as $\rho= \bigl(\begin{smallmatrix}0&1 \\ -1&0\end{smallmatrix} \bigr)$, and let
$\Delta = (\rho\vect{P_iP_{i-1}}, \rho^2\vect{P_iP_{i-1}}, \rho^3\vect{P_iP_{i-1}})$ (intuitively, $\Delta$ is the vector of possible steps after $P_i$, ordered clockwise).
We say that $P_{0,1,\ldots,i}A$ \emph{turns right} (\resp \emph{turns left}) from $P_{0,1,\ldots,i+1}$ if $\vect{P_{i}A}$ appears after (\resp before) $\vect{P_{i} P_{i+1}}$ in $\Delta$.

\begin{definition}[The right priority path of a set of paths or binding paths]\label{def:rp}
Let $P$ and $Q$ be two paths, where $P \neq Q$ and moreover neither is a prefix of the other, and
with $\pos{P_0} = \pos{Q_0} $ and $\pos{P_1} = \pos{Q_1}$.
Let $i$ be the smallest index such that $i \geq 0$ and  $P_i\neq Q_i$.
We say that $P$ is the {\em right priority path} of $P$ and $Q$  if either (a) $P_{0,1,\ldots, i}$ is a right turn from $Q$ or (b) $\pos{P_i}=\pos{Q_i}$ and the type of $P_i$ is smaller than the type of $Q_i$ in the canonical ordering of tile types.

Similarly,  let $p$ and $q$ be two binding paths, where $p \neq q$ and moreover neither is a prefix of the other, and
with $p_0 = q_0$ and $p_1 = q_1$.
Let $i$ be the smallest index such that $i \geq 0$ and  $p_i\neq q_i$.
We say that $p$ is the {\em right priority path} of $p$ and $q$  if  $p_{0,1,\ldots, i}$ is a right turn from $q$.

For any finite
set $S$ of paths, or of binding paths, we extend this definition as follows: let $p_{0} \in \mathbb{Z}^{2},  p_{1} \in \mathbb{Z}^{2}$ be two adjacent positions. If for all $s\in S$, we have $s_0=p_0$ and $s_1=p_1$, we call the \emph{right-priority path of $S$} the path that is the right-priority path of all other paths~in~$S$.
\end{definition}

For all $i\in\{0,1,\ldots,|P|-2\}$, we define $\glu P i = (g, i)$, where $g$ is the shared glue type between consecutive tiles $P_i$ and $P_{i+1}$ on the path $P$. Similarly, when we say ``glue'' in the context of a path $P$, we mean a pair of the form (glue type, path index).
We define $\type{\glue{P}{i}{i+1}} = g$ to denote the glue type of $\glue{P}{i}{i+1}$.
The {\em position} of the glue $\glue{P}{i}{i+1}$ is the midpoint of the unit-length line segment $\gs{\pos{P_i}}{\pos{P_{i+1}}}$ and is written as $\pos{\glue{P}{i}{i+1}}$.
We say that $\glue{P}{i}{i+1}$ is
{\em pointing to the north} (or {\em points to the north}, for short) if $\vect{P_i P_{i+1}} =  \vectwo{0}{1}$,
{\em pointing to the west} if $\vect{P_i P_{i+1}} =  \vectwo{-1}{0}$,
{\em pointing to the south} if $\vect{P_i P_{i+1}} =  \vectwo{0}{-1}$, and
{\em pointing to the east} if $\vect{P_i P_{i+1}} =  \vectwo{1}{0}$.

\subsection{Fragile paths}

If two paths, or two assemblies, or a path and an assembly, share a common position we say that they {\em intersect} at that position. Furthermore, we say that two paths, or two assemblies, or a path and an assembly,  {\em agree} on a position if they both place the same tile type at that position and {\em conflict} if they place a different tile type at that position.
We say that a path $P$ is {\em fragile} to mean that there is a producible assembly $\alpha$ that conflicts with $P$ (intuitively, if we grow $\alpha$ first, then there is at least one tile that $P$ cannot place), or more formally: 

\begin{definition}[Fragile]\label{def:fragile}
  Let $\calT = (T,\sigma,1)$ be a tile assembly system and $P\in\prodpathsT$. We say that $P$ is fragile if there exists a producible assembly $\alpha \in \prodasm{\calT}$ and a position $(x,y)\in(\dom\alpha\cap\dom{\asm P})$
  such that $\alpha((x,y))\neq\asm{P}((x,y))$.\footnote{
    Here, it might be the case that $\alpha$ and $P$ conflict at only one position by placing two different tile types $t$ and $t'$, but that $t$ and $t'$ may place the same glues along $P$. In this case $P$ is not producible when starting from the assembly  $\alpha$ because one of the tiles along the positions of $P$ is of the wrong type.
  }
\end{definition}

\subsection{Pumping a path}\label{sec:Pumping a path}

Next, for a path $P$ and two indices $i,j$ on $P$, we define a sequence of points and tile types (not necessarily a path) called the \emph{pumping of $P$ between $i$ and $j$}:
\begin{definition}[Pumping of $P$ between $i$ and $j$]
  \label{def:pumpingPbetweeniandj}
Let $\calT = (T,\sigma,1)$ be a tile assembly system and $P\in\prodpathsT$.
We say that the \emph{``pumping of $P$ between $i$ and $j$''} is the sequence $\olq$ of elements from $\Z^2\times T$ defined by:

\begin{equation*}
\olq_k =
\begin{cases}
 P_k &\qquad \textrm{for } 0\leq k \leq i \\
 \torture P k i j &\qquad \textrm{for }  i < k  ,
\end{cases}
\end{equation*}
\end{definition}

Intuitively, $\olq$ is the concatenation of a finite path $P_{0,1,\ldots,i}$ and an infinite periodic sequence of tile types and positions (possibly intersecting $\sigma\cup P_{0,1,\ldots,i}$, and possibly intersecting itself). We formalise this intuition in Lemma~\ref{lem:torture}. 

The following definition gives the notion of pumpable path used in our proofs. It is followed by a less formal but more intuitive description.

\begin{definition}[Pumpable path]\label{def:pumpable path}
  Let $\calT = (T,\sigma,1)$ be a tile assembly system.  We say that a producible path $P \in\prodpathsT$, is {\em infinitely pumpable}, or simply {\em pumpable}, if there are two integers $i<j$ such that the pumping of $P$ between $i$ and $j$ is  an infinite  producible path, i.e.\ formally: $ \olq \in\prodpaths{\mathcal{T}}$.

  In this case, we say that the \emph{pumping vector} of $\olq$ is $\vect{P_iP_j}$, and that $P$ is \emph{pumpable with pumping vector $\vect{P_iP_j}$}.
\end{definition}

For a path $P$ to be pumpable between $i$ and $j$ implies that $P_{i+1}+\vect{P_iP_j}$ interacts with $P_j$. It also implies that $\olq$ is self-avoiding and that in particular, for any positive integers $s\neq t $, the path $P_{i+1,\ldots, j}+s\vect{P_iP_j}$ does not intersect with the path $P_{i+1,\ldots, j}+t\vect{P_iP_j}$.
Lemma~\ref{lem:precious} shows that a sufficient condition for this is that $P_{i+1,\ldots, j}$ does not intersect $P_{i+1,\ldots, j}+\vect{P_iP_j}$.

\subsection{2D plane}

\subsubsection{Column, glue column, row, glue row}
When referring to sets of positions, we use the term ``\emph{the column $x$}'' for some fixed $x\in\Z$ to mean the set $\{(x, y)\mid y\in\Z\}$, and the term ``\emph{the row $y$}'' for some fixed $y\in\Z$ to mean the set $\{(x, y)\mid  x\in\Z\}$. 
The \emph{glue column $x$}, for some fixed $x\in\Z$, is the set of 2D half-integer positions $\{ (x+0.5,y) \mid y \in \Z \}$.
The \emph{glue row $y$},       for some fixed $y\in\Z$, is the set of 2D half-integer positions $\{ (x,y+0.5) \mid x \in \Z \}$.

Using the canonical embedding of $\Z^2$, the definition of a \emph{glue column $x$} can also be defined as the set of edges of the grid graph of $\Z^2$ between column $x$ and column $x+1$, and the \emph{glue row $y$} is the set of edges of the grid graph of $\Z^2$ between row $x$ and row $x+1$. Which definition we use will always be clear from the context.

\subsubsection{Curves}
\label{subsubsec:defs:curves}
A curve $c : I \rightarrow \mathbb{R}^2$ is a function from an interval $I \subset \mathbb{R}$ to $\mathbb{R}^2$, where $I$ is one of a closed, open, or half-open. 
All the curves in this paper are polygonal, i.e. unions of line segments and rays.

For a finite path $P$, we call the \emph{embedding $\embed{P}$ of $P$} the curve defined for all $t\in[0,|P|-1]\subset\R$ by:
$$\embed{P}(t) = \pos{P_{\lfloor t\rfloor}} + (t-\lfloor t\rfloor)\vect{P_{\lfloor t\rfloor}P_{\lfloor t\rfloor+1}}$$
Similarly, for a finite binding path $p$, the \emph{embedding $\embed p$ of $p$} is the curve defined for all
$t\in[0,|p|-1]\subset\R$ by:
$$\embed{p}(t) = p_{\lfloor t\rfloor} + (t-\lfloor t\rfloor)\vect{p_{\lfloor t\rfloor}p_{\lfloor t\rfloor+1}}$$

The ray of vector $\vect v$ {\em from} (or, {\em that starts at}) point $A\in\mathbb{R}$ is defined as the curve $r:[0, +\infty[\rightarrow \R^2$ such that $r(t) = A + t\vect v$.
The \emph{vertical ray from a point $A$ to the south (\resp to the north)} is the ray of vector $(0,-1)$ (\resp $(0,1)$) from $A$, and the \emph{horizontal ray from a point $A$ to the west (\resp to the east)} is the ray of vector $(-1, 0)$ (\resp $(1, 0)$) from~$A$.

If $C$ is a curve defined on some real interval of the form $[a, b]$ or $]a, b]$, and $D$ is a curve defined on some real interval of the form $[c, d]$ or $[c, d[$, and moreover $C(b) = D(c)$, then the \emph{concatenation $\concat{ C,  D}$ of $ C$ and $D$} is the curve defined on $\dom{ C}\cup(\dom{ D} - (c-b))$~by:\footnote{$\dom{ D} - (c-b)$ means $[b, d-(c-b)]$ if $\dom{ D}=[c, d]$, and $[b, d-(c-b)[$ if $\dom{ D}=[c, d[$}
$$\concat{ C,  D}\!(t) = \begin{cases} C(t)\text{ if }t\leq b\\ D(t+(c-b))\text{ otherwise}\end{cases}$$

A curve $c$ is said to be \emph{simple} or \emph{self-avoiding} if all its points are distinct, i.e. if for all $x, y\in\dom c$, $c(x) = c(y)\Rightarrow x=y$.

For $a,b \in\mathbb R$ with $a\leq b$, the notation 
$[a, b]$ denotes a closed real interval, $]a, b[$ an open real interval, and $[a, b[$ and $]a, b]$ are open on one end and closed on the other. 
The reverse $\reverse c$ of a curve $c$ defined on some interval $[a, b]$ (\resp $[a, b[$, $]a, b]$, $]a, b[$\;) is the curve defined on $[-b, -a]$ (\resp $]\!-b,-a]$, $[-b, -a[$, $]\!-b,-a[$\;) as $\reverse c(t) = c(-t)$.

If $A=(x_a,y_a)\in\R^2$ and $B=(x_b,y_b)\in\R^2$, the \emph{segment $\gs A B$} is defined to be the curve $s:[0,1]\rightarrow\R^2$ such that for all $t\in[0,1]$,  $s(t) = ((1-t)x_a+t x_b, (1-t)y_a+t y_b)$. We sometimes abuse notation and write $\gs{A}{B}$ even if $A$ or $B$ (or both) is a \emph{tile}, in which case we mean the position of that tile instead of the tile itself.

For a curve  $c:\mathbb{R} \rightarrow \mathbb{R}^2$ we write $c(\mathbb{R})$ to denote the range of $c$ (whenever we use this notation the curve $c$ has all of $\mathbb{R}$ as its domain). When it is clear from the context, we sometimes write $c$ to mean $c(\mathbb{R})$, for example 

For a path or binding path, $p_{0,1,\ldots,k}$ of length $\geq 1$, for $0 \leq i < k$ the notation $\hp{p}{i}{i+1}$ denotes the midpoint of the unit-length line segment $\embed{p_{i,i+1}}$. For a path $P$, we have $\hp p i {i+1} = \pos{\glu P i}$, hence this notation is especially useful for binding paths, since they do not have glues.

\subsubsection{Cutting the plane with curves; left and right turns}
In this paper we use finite and infinite polygonal curves to cut the $\mathbb{R}^2$ plane into two pieces. 
The finite polygonal curves we use consist of a finite number of concatenations of vertical and horizontal segments of length 1 or 0.5. 
If the curve is simple and closed we may apply the Jordan Curve Theorem
to cut the plane into connected components. 
\begin{theorem}[Jordan Curve Theorem]
  \label{thm:jordan}
  Let $c$ be a simple closed curve, then $c$ cuts $\R^2$ into two connected components.\end{theorem}
Here, we have stated the theorem in its general form, although for our results the (easier to prove) polygonal version suffices. 

The second kind of curve we use is composed of one or two infinite rays, along with a finite number of length 0.5 or length 1 segments. For such infinite polygonal curves we also state and prove a slightly different version of the polygonal Jordan Curve Theorem, as Theorem~\ref{thm:infinite-jordan}.

In Section~\ref{app:sec:inf JCT and turns} we define what it means for one curve to turn left or right from another, as well as left hand side and right hand side of a cut of the real plane.

\subsubsection{Visibility}\label{ref:subsubsec: visibility}
Let $P$ be a path producible by some tile assembly system $\mathcal T = (T, \sigma, 1)$, and let $i\in\{0,1,\ldots,|P|-2\}$ be such that $\glu P i$ points east or west. We say that $\glu P i$ is \emph{visible from the south} if and only if the ray $\ell^i$ of vector $(0,-1)$ starting at $\ell^i(0) = \left(\frac{\xcoord{P_i}+\xcoord{P_{i+1}}}{2},\ycoord{P_{i}}\right)$ does not intersect $\embed{P}$ nor $\sigma$.\footnote{For a glue ray (whose range has half-integer x-coordinate), to not intersect  $\sigma$ (whose tiles are on integer points), we mean that the glue ray does not intersect any segment between two adjacent tiles of $\sigma$ (even if these adjacent tiles do not interact).}

We define the terms \emph{visible from the east}, \emph{visible from the west} and \emph{visible from the north} similarly.

In many of our proofs, we will use curves, in particular curves that include visibility ray(s), to define connected components. For example, consider the curve $e$, where we have a path $P_{\range{i+1}{i+2}{j,j+1}}$ 
with $i<j$ and two glues $\glu P i$ and $\glu P j$ which are visible from the south with respective visibility rays $l^i$ and $l^j$:
$$e = \concat{\reverse{l^i}, \gs{l^i(0)}{ \pos{P_{i+1}}},\embed{P_{i+1,\ldots,j}},\gs{\pos{P_j}}{l^j(0)},l^j}$$

The curve $e$ is defined on $]\!-\!\infty,+\infty[ \,= \R$.

We will use the following lemma about visibility, which was stated in~\cite{STOC2017} (it is the fusion of Lemmas 5.2 and 6.3 in that paper). For the sake of completeness, we prove this result again, with slightly different notation.

\begin{lemma}
  \label{lem:glue:east}
  Let $P$ be a path producible by some tile assembly system $\mathcal T=(T,\sigma,1)$ such that the last glue of $P$ is visible from the north.
  Let $i, j\in\{ \range 0 1 {|P|-2} \}$ be two integers.
  If both $\glu P i$ and $\glu P j$ are visible from the south and $\glu P i$ points to the east (\resp to the west), and $\xcoord{P_i} < \xcoord{P_j}$ (\resp $\xcoord{P_i}>\xcoord{P_j}$), then $i<j$ and $\glu P j$ points to the east (\resp to the west).
\end{lemma}
\begin{proof}
  We first assume that $\glu P i$ points to the east and $\xcoord{P_i}<\xcoord{P_j}$: indeed by taking the vertical symmetric of each tile in $T$, of $\sigma$, and of $P$, we get the ``respective'' version of this statement, where $\glu P i$ points to the west and $\xcoord{P_i}>\xcoord{P_j}$. In particular, this flip about a vertical line does not change the visibility hypotheses.

  Let $l^i$ and $l^j$ be the respective visibility rays of $\glu P i$ and $\glu P j$. Additionally, let $m_i$ and $m_j$ be two real numbers such that $l^i(m_i)\in\R^2$ and $l^j(m_j)\in\R^2$ have the same y-coordinate, and are below all the points of $\sigma\cup\asm{P}$.
  Moreover, let $t_i = i + 0.5\in\R$,  $t_j = j+0.5\in\R$ and $t_k = |P|-1.5\in\R$, and note that $\embed{P}(t_i) = \pos{\glu P i}$, $\embed{P}(t_j) = \pos{\glu P j}$ and $\embed{P}(t_k) = \pos{\glue P {|P|-2}{|P|-1}}$.

  Now, assume for the sake of contradiction that $i > j$, which means that $t_i > t_j$. We define a curve $c$ as the concatenation of $l^i([0,m_i])$, $\gs{l^i(m_i)}{l^j(m_j)}$, $\reverse{l^j([0,m_j])}$ and the restriction to the interval $[t_j, t_i]$ of $\embed{P}$.

  By the Jordan Curve Theorem (Theorem~\ref{thm:jordan}), $c$ cuts the plane $\R^2$ into two connected components. Let $\mathcal C$ be the connected component inside $c$. We claim that $\embed{P}([t_i,t_k])$ has at least one point inside $\mathcal C$: indeed, the translation by a distance $0.1$ to the east of $l^i$, i.e. $l^i+(0.1, 0)$, which starts at $\embed{P}(t_i+0.1)$, does not intersect $\embed {P_{\range j {j+1} i}}$ (by visibility of $l^i$), does not intersect $l^i$ nor $l^j$, and intersects $\gs{l^i(m_i)}{l^j(m_j)}$ exactly once (since $\xcoord{P_i}<\xcoord{P_j}$). Therefore, $\embed{P}(t_i+0.1)$ is inside $\mathcal C$.

  However, $\glue P {|P|-2}{|P|-1}$, the last glue of $P$, is visible from the north, hence cannot be in $\mathcal C$, since there is at least one glue of $P_{\rng j i}$ to the north of any glue inside $\mathcal C$.
  Therefore, $\embed{P}([t_i,t_k])$ starts in $\mathcal C$ and has at least one point outside $\mathcal C$. But $\embed{P_{\range{i+1}{i+2}{|P|-1}}}$ cannot cross $l^i$ nor $l^j$ (because of visibility), nor $\embed{P}([t_j,t_i])$ (because $P$ is simple), nor $\gs{l^i(m_i)}{l^j(m_j)}$ (because $\gs{l^i(m_i)}{l^j(m_j)}$ is strictly below all points of $P$).
  This is a contradiction, and hence $t_i<t_j$ and $i<j$.

  Finally, we claim that $\glu P j$ points to the east.
  We first redefine $c$ to be the concatenation of $l^j([0,m_j])$, $\gs{l^j(m_j)}{l^i(m_i)}$, $\reverse{l^i([0,m_i])}$, and the restriction to interval $[t_i, t_j]$ of $\embed{P}$\footnote{this new definition of $c$ follows the same points as the previous one, but some parts are reversed because now we know that $i<j$.}.
  Assume, for the sake of contradiction that $\glu P j$ points to the west. This means that $\embed{P}(t_j+0.1)$ is inside $\mathcal C$, since by the same argument as above, $l^j-(0.1, 0)$ starts at $\embed{P}(t_j+0.1)$ and intersects $c$ exactly once.
  But $\embed{P}([t_j,t_k])$ has at least one point inside $\mathcal C$ and ends outside $\mathcal C$, yet cannot cross $c$ (for the same reasons as before). This is a contradiction, and hence $\glu P j$ points to the east.
\end{proof}

We will sometimes use this lemma when the last tile of $P$ is the unique easternmost tile of $\sigma\cup\asm{P}$, in this case the last glue of $P$ is visible from the south and from the north.

\newcommand\cp{\mathcal C^+}
\newcommand\cm{\mathcal C^-}

\section{Shield lemma}\label{sec:shield main section}

The goal of this section is to prove Lemma~\ref{lem:shield},
which is the main technical tool we prove in this paper. The following definition is crucial to the lemma statement and defines  notation used throughout this section. It is illustrated in Figure~\ref{fig:shield-setup}.

\begin{definition}[A shield $(i,j,k)$ for $P$]\label{def:shield}
  Let $P$ be a path producible by some tile assembly system $\mathcal T = (T, \sigma, 1)$.
  We say that the triple $(i, j, k)$ of integers is a \emph{shield for $P$} if $0 \leq i<j\leq k < |P|-1$,
  and the following three conditions hold:
  \begin{enumerate}
  \item\label{lem:hp:shield ij} $\glueP{i}{i+1}$ and $\glueP{j}{j+1}$ are both of the same type,
    visible from the south relative to $P$
    and pointing east; and
  \item\label{lem:hp:shield k} $\glu P k$ is visible from the north relative to $P$, for notation let $l^k$ be the visibility ray to the north of $\glu P k$; and
  \item\label{lem:hp:shield backup}  $\embed{P_{i,i+1,\ldots,k}} \cap ( l^k+\vect{P_jP_i}) \subseteq\{l^k(0) + \vpji\}$.
    In other words, if $\embed{P_{i,i+1,\ldots,k}}$ intersects $l^k+\vpji$ (which may not be the case), there is exactly one intersection, which is at the start-point of the ray $l^k+\vpji$.
  \end{enumerate}
\end{definition}
It should be noted that $x_{\vect{P_iP_j}}>0$, i.e.\ the x-component of the vector $\vect{P_iP_j}$ is strictly positive, which follows from  $i < j$ together with the contrapositive of  Lemma~\ref{lem:glue:east}: indeed, the contrapositive of Lemma~\ref{lem:glue:east}, with $i$ and $j$ swapped, is $i \leq j \Rightarrow x_{P_i}\leq x_{P_j}$, and since here $P$ is simple, we also have $i<j\Rightarrow x_{P_i}<x_{P_j}$).
Throughout the paper, $l^i$ and $l^j$  denote the vertical (visibility) rays to the south that start at $\pos{\glu P i}$ and $\pos{\glu P j}$, \resp.

\begin{definition}[The cut $c$ and workspace $\mathcal C$ of shield $(i,j,k)$]\label{def:workspace}\label{def:c}\label{def:C}
  Let $P$ be a path producible by some tile assembly system $\mathcal T = (T, \sigma, 1)$, and let $(i, j, k)$ be a shield for $P$.
  We say that the \emph{cut of the shield $(i, j, k)$} is the curve $c$ defined by:
  $$c = \concat{\reverse{l^i}, \gs{l^i(0)}{P_{i+1}}, \embed{P_{\range{i+1}{i+2}{k}}}, \gs{P_k}{l^k(0)}, l^k}$$
  By \subl{lem:c-cuts} below, $c$ cuts the plane into the two connected components defined in the conclusion of Theorem~\ref{thm:infinite-jordan}: the left-hand side and right-hand side of $c$.
  The right-hand side of $c$ is the connected component $\mathcal C$ that is intuitively to the east of $l^i$ and $l^k$, and includes $c$ itself. We say that $\mathcal C$ is the  \emph{workspace of shield $(i, j, k)$} (see Figure~\ref{fig:shield-setup-c} for an example).
\end{definition}

\vspace*{\baselineskip}
\newcounter{shield}
\newcommand\shieldStatement{
  Let $P$ be a path producible by some tile assembly system $\mathcal T = (T, \sigma, 1)$,
  such that $(i, j, k)$ is a shield for $P$ (see Definition~\ref{def:shield}).
  Then $P$ is pumpable with pumping vector $\vect{P_iP_j}$, or $P$ is fragile.

  Moreover, if $P$ is fragile, there is a path $Q$, entirely contained in the workspace of shield $(i, j, k)$  (see Definition~\ref{def:workspace}), such that $P_{\range 0 1 i}Q$ is a producible path and conflicts with~$P_{\range{i+1}{i+2}{k}}$.
}
\begin{lemma}[Shield lemma]\label{lem:shield}
  \shieldStatement
\end{lemma}

\begin{figure}[ht]
  \begin{center}
    \begin{tikzpicture}[scale=\scale]\draw[draw={rgb,255:red,200; green,200; blue,200}](0.5,-0.2) rectangle (26.5, -15.2);
\draw[draw={rgb,255:red,200; green,200; blue,200}](0.5,-15.2)--(0.5,-0.2);
\draw[draw={rgb,255:red,200; green,200; blue,200}](1.5,-15.2)--(1.5,-0.2);
\draw[draw={rgb,255:red,200; green,200; blue,200}](2.5,-15.2)--(2.5,-0.2);
\draw[draw={rgb,255:red,200; green,200; blue,200}](3.5,-15.2)--(3.5,-0.2);
\draw[draw={rgb,255:red,200; green,200; blue,200}](4.5,-15.2)--(4.5,-0.2);
\draw[draw={rgb,255:red,200; green,200; blue,200}](5.5,-15.2)--(5.5,-0.2);
\draw[draw={rgb,255:red,200; green,200; blue,200}](6.5,-15.2)--(6.5,-0.2);
\draw[draw={rgb,255:red,200; green,200; blue,200}](7.5,-15.2)--(7.5,-0.2);
\draw[draw={rgb,255:red,200; green,200; blue,200}](8.5,-15.2)--(8.5,-0.2);
\draw[draw={rgb,255:red,200; green,200; blue,200}](9.5,-15.2)--(9.5,-0.2);
\draw[draw={rgb,255:red,200; green,200; blue,200}](10.5,-15.2)--(10.5,-0.2);
\draw[draw={rgb,255:red,200; green,200; blue,200}](11.5,-15.2)--(11.5,-0.2);
\draw[draw={rgb,255:red,200; green,200; blue,200}](12.5,-15.2)--(12.5,-0.2);
\draw[draw={rgb,255:red,200; green,200; blue,200}](13.5,-15.2)--(13.5,-0.2);
\draw[draw={rgb,255:red,200; green,200; blue,200}](14.5,-15.2)--(14.5,-0.2);
\draw[draw={rgb,255:red,200; green,200; blue,200}](15.5,-15.2)--(15.5,-0.2);
\draw[draw={rgb,255:red,200; green,200; blue,200}](16.5,-15.2)--(16.5,-0.2);
\draw[draw={rgb,255:red,200; green,200; blue,200}](17.5,-15.2)--(17.5,-0.2);
\draw[draw={rgb,255:red,200; green,200; blue,200}](18.5,-15.2)--(18.5,-0.2);
\draw[draw={rgb,255:red,200; green,200; blue,200}](19.5,-15.2)--(19.5,-0.2);
\draw[draw={rgb,255:red,200; green,200; blue,200}](20.5,-15.2)--(20.5,-0.2);
\draw[draw={rgb,255:red,200; green,200; blue,200}](21.5,-15.2)--(21.5,-0.2);
\draw[draw={rgb,255:red,200; green,200; blue,200}](22.5,-15.2)--(22.5,-0.2);
\draw[draw={rgb,255:red,200; green,200; blue,200}](23.5,-15.2)--(23.5,-0.2);
\draw[draw={rgb,255:red,200; green,200; blue,200}](24.5,-15.2)--(24.5,-0.2);
\draw[draw={rgb,255:red,200; green,200; blue,200}](25.5,-15.2)--(25.5,-0.2);
\draw[draw={rgb,255:red,200; green,200; blue,200}](0.5,-0.2)--(26.5,-0.2);
\draw[draw={rgb,255:red,200; green,200; blue,200}](0.5,-1.2)--(26.5,-1.2);
\draw[draw={rgb,255:red,200; green,200; blue,200}](0.5,-2.2)--(26.5,-2.2);
\draw[draw={rgb,255:red,200; green,200; blue,200}](0.5,-3.2)--(26.5,-3.2);
\draw[draw={rgb,255:red,200; green,200; blue,200}](0.5,-4.2)--(26.5,-4.2);
\draw[draw={rgb,255:red,200; green,200; blue,200}](0.5,-5.2)--(26.5,-5.2);
\draw[draw={rgb,255:red,200; green,200; blue,200}](0.5,-6.2)--(26.5,-6.2);
\draw[draw={rgb,255:red,200; green,200; blue,200}](0.5,-7.2)--(26.5,-7.2);
\draw[draw={rgb,255:red,200; green,200; blue,200}](0.5,-8.2)--(26.5,-8.2);
\draw[draw={rgb,255:red,200; green,200; blue,200}](0.5,-9.2)--(26.5,-9.2);
\draw[draw={rgb,255:red,200; green,200; blue,200}](0.5,-10.2)--(26.5,-10.2);
\draw[draw={rgb,255:red,200; green,200; blue,200}](0.5,-11.2)--(26.5,-11.2);
\draw[draw={rgb,255:red,200; green,200; blue,200}](0.5,-12.2)--(26.5,-12.2);
\draw[draw={rgb,255:red,200; green,200; blue,200}](0.5,-13.2)--(26.5,-13.2);
\draw[draw={rgb,255:red,200; green,200; blue,200}](0.5,-14.2)--(26.5,-14.2);
\draw[draw={rgb,255:red,0; green,0; blue,0},fill={rgb,255:red,200; green,113; blue,55},opacity=0.29899997,fill opacity=0.29899997](1.65,-8.35) rectangle (2.35, -9.05);
\draw[draw={rgb,255:red,0; green,0; blue,0},fill={rgb,255:red,200; green,113; blue,55},opacity=0.29899997,fill opacity=0.29899997](2.65,-8.35) rectangle (3.35, -9.05);
\draw[draw={rgb,255:red,0; green,0; blue,0},fill={rgb,255:red,200; green,113; blue,55},opacity=0.29899997,fill opacity=0.29899997](3.65,-8.35) rectangle (4.35, -9.05);
\draw[draw={rgb,255:red,0; green,0; blue,0},fill={rgb,255:red,200; green,113; blue,55},opacity=0.29899997,fill opacity=0.29899997](4.65,-8.35) rectangle (5.35, -9.05);
\draw[draw={rgb,255:red,0; green,0; blue,0},fill={rgb,255:red,200; green,113; blue,55},opacity=0.29899997,fill opacity=0.29899997](5.65,-8.35) rectangle (6.35, -9.05);
\draw[draw={rgb,255:red,0; green,0; blue,0},fill={rgb,255:red,200; green,113; blue,55},opacity=0.29899997,fill opacity=0.29899997](6.65,-8.35) rectangle (7.35, -9.05);
\draw[draw={rgb,255:red,0; green,0; blue,0},fill={rgb,255:red,200; green,113; blue,55},opacity=0.29899997,fill opacity=0.29899997](7.65,-8.35) rectangle (8.35, -9.05);
\draw[draw={rgb,255:red,0; green,0; blue,0},fill={rgb,255:red,200; green,113; blue,55},opacity=0.29899997,fill opacity=0.29899997](7.65,-9.35) rectangle (8.35, -10.05);
\draw[draw={rgb,255:red,0; green,0; blue,0},fill={rgb,255:red,200; green,113; blue,55},opacity=0.29899997,fill opacity=0.29899997](7.65,-10.35) rectangle (8.35, -11.05);
\draw[draw={rgb,255:red,0; green,0; blue,0},fill={rgb,255:red,200; green,113; blue,55},opacity=0.29899997,fill opacity=0.29899997](7.65,-11.35) rectangle (8.35, -12.05);
\draw[draw={rgb,255:red,0; green,0; blue,0},fill={rgb,255:red,200; green,113; blue,55},opacity=0.29899997,fill opacity=0.29899997](6.65,-11.35) rectangle (7.35, -12.05);
\draw[draw={rgb,255:red,0; green,0; blue,0},fill={rgb,255:red,200; green,113; blue,55},opacity=0.29899997,fill opacity=0.29899997](5.65,-11.35) rectangle (6.35, -12.05);
\draw[draw={rgb,255:red,0; green,0; blue,0},fill={rgb,255:red,200; green,113; blue,55},opacity=0.29899997,fill opacity=0.29899997](4.65,-11.35) rectangle (5.35, -12.05);
\draw[draw={rgb,255:red,0; green,0; blue,0},fill={rgb,255:red,200; green,113; blue,55},opacity=0.29899997,fill opacity=0.29899997](3.65,-11.35) rectangle (4.35, -12.05);
\draw[draw={rgb,255:red,0; green,0; blue,0},fill={rgb,255:red,200; green,113; blue,55},opacity=0.29899997,fill opacity=0.29899997](3.65,-12.35) rectangle (4.35, -13.05);
\draw[draw={rgb,255:red,0; green,0; blue,0},fill={rgb,255:red,200; green,113; blue,55},opacity=0.29899997,fill opacity=0.29899997](3.65,-13.35) rectangle (4.35, -14.05);
\draw[draw={rgb,255:red,0; green,0; blue,0},fill={rgb,255:red,200; green,113; blue,55},opacity=0.29899997,fill opacity=0.29899997](4.65,-13.35) rectangle (5.35, -14.05);
\draw[draw={rgb,255:red,200; green,113; blue,55},opacity=0.29899997,thick](2,-8.7)--(8,-8.7)--(8,-11.7)--(4,-11.7)--(4,-13.7)--(5.5,-13.7);
\draw[draw={rgb,255:red,0; green,0; blue,0},fill={rgb,255:red,200; green,113; blue,55},opacity=0.5,fill opacity=0.5](5.65,-13.35) rectangle (6.35, -14.05);
\draw[draw={rgb,255:red,0; green,0; blue,0},fill={rgb,255:red,200; green,113; blue,55},opacity=0.5,fill opacity=0.5](6.65,-13.35) rectangle (7.35, -14.05);
\draw[draw={rgb,255:red,0; green,0; blue,0},fill={rgb,255:red,200; green,113; blue,55},opacity=0.5,fill opacity=0.5](7.65,-13.35) rectangle (8.35, -14.05);
\draw[draw={rgb,255:red,0; green,0; blue,0},fill={rgb,255:red,200; green,113; blue,55},opacity=0.5,fill opacity=0.5](7.65,-12.35) rectangle (8.35, -13.05);
\draw[draw={rgb,255:red,0; green,0; blue,0},fill={rgb,255:red,200; green,113; blue,55},opacity=0.5,fill opacity=0.5](8.65,-12.35) rectangle (9.35, -13.05);
\draw[draw={rgb,255:red,0; green,0; blue,0},fill={rgb,255:red,200; green,113; blue,55},opacity=0.5,fill opacity=0.5](9.65,-12.35) rectangle (10.35, -13.05);
\draw[draw={rgb,255:red,0; green,0; blue,0},fill={rgb,255:red,200; green,113; blue,55},opacity=0.5,fill opacity=0.5](9.65,-11.35) rectangle (10.35, -12.05);
\draw[draw={rgb,255:red,0; green,0; blue,0},fill={rgb,255:red,200; green,113; blue,55},opacity=0.5,fill opacity=0.5](9.65,-10.35) rectangle (10.35, -11.05);
\draw[draw={rgb,255:red,0; green,0; blue,0},fill={rgb,255:red,200; green,113; blue,55},opacity=0.5,fill opacity=0.5](9.65,-9.35) rectangle (10.35, -10.05);
\draw[draw={rgb,255:red,0; green,0; blue,0},fill={rgb,255:red,200; green,113; blue,55},opacity=0.5,fill opacity=0.5](10.65,-9.35) rectangle (11.35, -10.05);
\draw[draw={rgb,255:red,0; green,0; blue,0},fill={rgb,255:red,200; green,113; blue,55},opacity=0.5,fill opacity=0.5](11.65,-9.35) rectangle (12.35, -10.05);
\draw[draw={rgb,255:red,0; green,0; blue,0},fill={rgb,255:red,200; green,113; blue,55},opacity=0.5,fill opacity=0.5](12.65,-9.35) rectangle (13.35, -10.05);
\draw[draw={rgb,255:red,0; green,0; blue,0},fill={rgb,255:red,200; green,113; blue,55},opacity=0.5,fill opacity=0.5](13.65,-9.35) rectangle (14.35, -10.05);
\draw[draw={rgb,255:red,0; green,0; blue,0},fill={rgb,255:red,200; green,113; blue,55},opacity=0.5,fill opacity=0.5](14.65,-9.35) rectangle (15.35, -10.05);
\draw[draw={rgb,255:red,0; green,0; blue,0},fill={rgb,255:red,200; green,113; blue,55},opacity=0.5,fill opacity=0.5](15.65,-9.35) rectangle (16.35, -10.05);
\draw[draw={rgb,255:red,0; green,0; blue,0},fill={rgb,255:red,200; green,113; blue,55},opacity=0.5,fill opacity=0.5](16.65,-9.35) rectangle (17.35, -10.05);
\draw[draw={rgb,255:red,0; green,0; blue,0},fill={rgb,255:red,200; green,113; blue,55},opacity=0.5,fill opacity=0.5](16.65,-10.35) rectangle (17.35, -11.05);
\draw[draw={rgb,255:red,0; green,0; blue,0},fill={rgb,255:red,200; green,113; blue,55},opacity=0.5,fill opacity=0.5](16.65,-11.35) rectangle (17.35, -12.05);
\draw[draw={rgb,255:red,0; green,0; blue,0},fill={rgb,255:red,200; green,113; blue,55},opacity=0.5,fill opacity=0.5](16.65,-12.35) rectangle (17.35, -13.05);
\draw[draw={rgb,255:red,0; green,0; blue,0},fill={rgb,255:red,200; green,113; blue,55},opacity=0.5,fill opacity=0.5](17.65,-12.35) rectangle (18.35, -13.05);
\draw[draw={rgb,255:red,0; green,0; blue,0},fill={rgb,255:red,200; green,113; blue,55},opacity=0.5,fill opacity=0.5](18.65,-12.35) rectangle (19.35, -13.05);
\draw[draw={rgb,255:red,0; green,0; blue,0},fill={rgb,255:red,200; green,113; blue,55},opacity=0.5,fill opacity=0.5](18.65,-11.35) rectangle (19.35, -12.05);
\draw[draw={rgb,255:red,0; green,0; blue,0},fill={rgb,255:red,200; green,113; blue,55},opacity=0.5,fill opacity=0.5](18.65,-10.35) rectangle (19.35, -11.05);
\draw[draw={rgb,255:red,0; green,0; blue,0},fill={rgb,255:red,200; green,113; blue,55},opacity=0.5,fill opacity=0.5](18.65,-9.35) rectangle (19.35, -10.05);
\draw[draw={rgb,255:red,0; green,0; blue,0},fill={rgb,255:red,200; green,113; blue,55},opacity=0.5,fill opacity=0.5](19.65,-9.35) rectangle (20.35, -10.05);
\draw[draw={rgb,255:red,0; green,0; blue,0},fill={rgb,255:red,200; green,113; blue,55},opacity=0.5,fill opacity=0.5](19.65,-8.35) rectangle (20.35, -9.05);
\draw[draw={rgb,255:red,0; green,0; blue,0},fill={rgb,255:red,200; green,113; blue,55},opacity=0.5,fill opacity=0.5](19.65,-7.35) rectangle (20.35, -8.05);
\draw[draw={rgb,255:red,0; green,0; blue,0},fill={rgb,255:red,200; green,113; blue,55},opacity=0.5,fill opacity=0.5](20.65,-7.35) rectangle (21.35, -8.05);
\draw[draw={rgb,255:red,0; green,0; blue,0},fill={rgb,255:red,200; green,113; blue,55},opacity=0.5,fill opacity=0.5](21.65,-7.35) rectangle (22.35, -8.05);
\draw[draw={rgb,255:red,0; green,0; blue,0},fill={rgb,255:red,200; green,113; blue,55},opacity=0.5,fill opacity=0.5](22.65,-7.35) rectangle (23.35, -8.05);
\draw[draw={rgb,255:red,0; green,0; blue,0},fill={rgb,255:red,200; green,113; blue,55},opacity=0.5,fill opacity=0.5](22.65,-8.35) rectangle (23.35, -9.05);
\draw[draw={rgb,255:red,0; green,0; blue,0},fill={rgb,255:red,200; green,113; blue,55},opacity=0.5,fill opacity=0.5](22.65,-9.35) rectangle (23.35, -10.05);
\draw[draw={rgb,255:red,0; green,0; blue,0},fill={rgb,255:red,200; green,113; blue,55},opacity=0.5,fill opacity=0.5](22.65,-10.35) rectangle (23.35, -11.05);
\draw[draw={rgb,255:red,0; green,0; blue,0},fill={rgb,255:red,200; green,113; blue,55},opacity=0.5,fill opacity=0.5](23.65,-10.35) rectangle (24.35, -11.05);
\draw[draw={rgb,255:red,0; green,0; blue,0},fill={rgb,255:red,200; green,113; blue,55},opacity=0.5,fill opacity=0.5](24.65,-10.35) rectangle (25.35, -11.05);
\draw[draw={rgb,255:red,0; green,0; blue,0},fill={rgb,255:red,200; green,113; blue,55},opacity=0.5,fill opacity=0.5](24.65,-9.35) rectangle (25.35, -10.05);
\draw[draw={rgb,255:red,0; green,0; blue,0},fill={rgb,255:red,200; green,113; blue,55},opacity=0.5,fill opacity=0.5](24.65,-8.35) rectangle (25.35, -9.05);
\draw[draw={rgb,255:red,0; green,0; blue,0},fill={rgb,255:red,200; green,113; blue,55},opacity=0.5,fill opacity=0.5](24.65,-7.35) rectangle (25.35, -8.05);
\draw[draw={rgb,255:red,0; green,0; blue,0},fill={rgb,255:red,200; green,113; blue,55},opacity=0.5,fill opacity=0.5](24.65,-6.35) rectangle (25.35, -7.05);
\draw[draw={rgb,255:red,0; green,0; blue,0},fill={rgb,255:red,200; green,113; blue,55},opacity=0.5,fill opacity=0.5](24.65,-5.35) rectangle (25.35, -6.05);
\draw[draw={rgb,255:red,0; green,0; blue,0},fill={rgb,255:red,200; green,113; blue,55},opacity=0.5,fill opacity=0.5](23.65,-5.35) rectangle (24.35, -6.05);
\draw[draw={rgb,255:red,0; green,0; blue,0},fill={rgb,255:red,200; green,113; blue,55},opacity=0.5,fill opacity=0.5](22.65,-5.35) rectangle (23.35, -6.05);
\draw[draw={rgb,255:red,0; green,0; blue,0},fill={rgb,255:red,200; green,113; blue,55},opacity=0.5,fill opacity=0.5](21.65,-5.35) rectangle (22.35, -6.05);
\draw[draw={rgb,255:red,0; green,0; blue,0},fill={rgb,255:red,200; green,113; blue,55},opacity=0.5,fill opacity=0.5](20.65,-5.35) rectangle (21.35, -6.05);
\draw[draw={rgb,255:red,0; green,0; blue,0},fill={rgb,255:red,200; green,113; blue,55},opacity=0.5,fill opacity=0.5](19.65,-5.35) rectangle (20.35, -6.05);
\draw[draw={rgb,255:red,0; green,0; blue,0},fill={rgb,255:red,200; green,113; blue,55},opacity=0.5,fill opacity=0.5](19.65,-4.35) rectangle (20.35, -5.05);
\draw[draw={rgb,255:red,0; green,0; blue,0},fill={rgb,255:red,200; green,113; blue,55},opacity=0.5,fill opacity=0.5](18.65,-4.35) rectangle (19.35, -5.05);
\draw[draw={rgb,255:red,0; green,0; blue,0},fill={rgb,255:red,200; green,113; blue,55},opacity=0.5,fill opacity=0.5](17.65,-4.35) rectangle (18.35, -5.05);
\draw[draw={rgb,255:red,0; green,0; blue,0},fill={rgb,255:red,200; green,113; blue,55},opacity=0.5,fill opacity=0.5](16.65,-4.35) rectangle (17.35, -5.05);
\draw[draw={rgb,255:red,0; green,0; blue,0},fill={rgb,255:red,200; green,113; blue,55},opacity=0.5,fill opacity=0.5](16.65,-3.35) rectangle (17.35, -4.05);
\draw[draw={rgb,255:red,0; green,0; blue,0},fill={rgb,255:red,200; green,113; blue,55},opacity=0.5,fill opacity=0.5](17.65,-3.35) rectangle (18.35, -4.05);
\draw[draw={rgb,255:red,0; green,0; blue,0},fill={rgb,255:red,200; green,113; blue,55},opacity=0.5,fill opacity=0.5](18.65,-3.35) rectangle (19.35, -4.05);
\draw[draw={rgb,255:red,0; green,0; blue,0},fill={rgb,255:red,200; green,113; blue,55},opacity=0.5,fill opacity=0.5](18.65,-2.35) rectangle (19.35, -3.05);
\draw[draw={rgb,255:red,0; green,0; blue,0},fill={rgb,255:red,200; green,113; blue,55},opacity=0.5,fill opacity=0.5](18.65,-1.35) rectangle (19.35, -2.05);
\draw[draw={rgb,255:red,0; green,0; blue,0},fill={rgb,255:red,200; green,113; blue,55},opacity=0.5,fill opacity=0.5](17.65,-1.35) rectangle (18.35, -2.05);
\draw[draw={rgb,255:red,0; green,0; blue,0},fill={rgb,255:red,200; green,113; blue,55},opacity=0.5,fill opacity=0.5](16.65,-1.35) rectangle (17.35, -2.05);
\draw[draw={rgb,255:red,0; green,0; blue,0},fill={rgb,255:red,200; green,113; blue,55},opacity=0.5,fill opacity=0.5](15.65,-1.35) rectangle (16.35, -2.05);
\draw[draw={rgb,255:red,0; green,0; blue,0},fill={rgb,255:red,200; green,113; blue,55},opacity=0.5,fill opacity=0.5](15.65,-2.35) rectangle (16.35, -3.05);
\draw[draw={rgb,255:red,0; green,0; blue,0},fill={rgb,255:red,200; green,113; blue,55},opacity=0.5,fill opacity=0.5](15.65,-3.35) rectangle (16.35, -4.05);
\draw[draw={rgb,255:red,0; green,0; blue,0},fill={rgb,255:red,200; green,113; blue,55},opacity=0.5,fill opacity=0.5](15.65,-4.35) rectangle (16.35, -5.05);
\draw[draw={rgb,255:red,0; green,0; blue,0},fill={rgb,255:red,200; green,113; blue,55},opacity=0.5,fill opacity=0.5](15.65,-5.35) rectangle (16.35, -6.05);
\draw[draw={rgb,255:red,0; green,0; blue,0},fill={rgb,255:red,200; green,113; blue,55},opacity=0.5,fill opacity=0.5](14.65,-5.35) rectangle (15.35, -6.05);
\draw[draw={rgb,255:red,0; green,0; blue,0},fill={rgb,255:red,200; green,113; blue,55},opacity=0.5,fill opacity=0.5](13.65,-5.35) rectangle (14.35, -6.05);
\draw[draw={rgb,255:red,0; green,0; blue,0},fill={rgb,255:red,200; green,113; blue,55},opacity=0.5,fill opacity=0.5](12.65,-5.35) rectangle (13.35, -6.05);
\draw[draw={rgb,255:red,200; green,113; blue,55},opacity=0.5,thick](6,-13.7)--(8,-13.7)--(8,-12.7)--(10,-12.7)--(10,-9.7)--(17,-9.7)--(17,-12.7)--(19,-12.7)--(19,-9.7)--(20,-9.7)--(20,-7.7)--(23,-7.7)--(23,-10.7)--(25,-10.7)--(25,-5.7)--(20,-5.7)--(20,-4.7)--(17,-4.7)--(17,-3.7)--(19,-3.7)--(19,-1.7)--(16,-1.7)--(16,-5.7)--(13,-5.7);
\draw[draw=none,fill={rgb,255:red,28; green,36; blue,31},thin](9, -12.7) ellipse (0.1cm and 0.1cm);\draw[draw=none,fill={rgb,255:red,28; green,36; blue,31},thin](5, -13.71) ellipse (0.1cm and 0.1cm);\draw(8.13, -12.59) node[anchor=south west] {$P_j$};
\draw(4.11, -13.6) node[anchor=south west] {$P_i$};
\draw[draw=none,fill={rgb,255:red,28; green,36; blue,31},thin](14, -5.7) ellipse (0.1cm and 0.1cm);\draw(13.4, -7.39) node[anchor=south west] {$P_k$};
\draw[draw={rgb,255:red,0; green,0; blue,0}](9.5,-12.69)--(9.5,-16.7);
\draw(8.23, -16.93) node[anchor=south west] {$l_j$};
\draw[draw=none,fill={rgb,255:red,28; green,36; blue,31},thin](9.5, -12.69) ellipse (0.1cm and 0.1cm);\draw[draw={rgb,255:red,0; green,0; blue,0}](5.5,-13.7)--(5.5,-16.7);
\draw(4.32, -16.93) node[anchor=south west] {$l_i$};
\draw[draw={rgb,255:red,0; green,0; blue,0}](13.5,-5.7)--(13.5,0.3);
\draw(12.34, -0.81) node[anchor=south west] {$l^k$};
\draw[draw=none,fill={rgb,255:red,28; green,36; blue,31},thin](13.5, -5.7) ellipse (0.1cm and 0.1cm);\draw[draw=none,fill={rgb,255:red,28; green,36; blue,31},thin](5.5, -13.7) ellipse (0.1cm and 0.1cm);\draw[draw={rgb,255:red,0; green,0; blue,0}](9.5,-6.71)--(9.5,0.3);
\draw(5, -0.96) node[anchor=south west] {$l^k+\vect{P_jP_i}$};
\draw[draw=none,fill={rgb,255:red,28; green,36; blue,31},thin](9.5, -6.71) ellipse (0.1cm and 0.1cm);
\end{tikzpicture}
  \end{center}
  \caption{A suffix $P_{i,i+1,\ldots,k+1}$  of a path $P$. 
    Tiles $P_i$, $P_j$ and $P_k$ are shown along with the four rays 
    and the three glues (at ray starting points) of Hypotheses~\ref{lem:hp:shield ij}--\ref{lem:hp:shield backup} of Definition~\ref{def:shield},
   thus $(i, j, k)$ is a shield for $P$.
   The goal of this section is to prove Lemma~\ref{lem:shield} showing that such a path is pumpable or fragile.}
  \label{fig:shield-setup}
\end{figure}
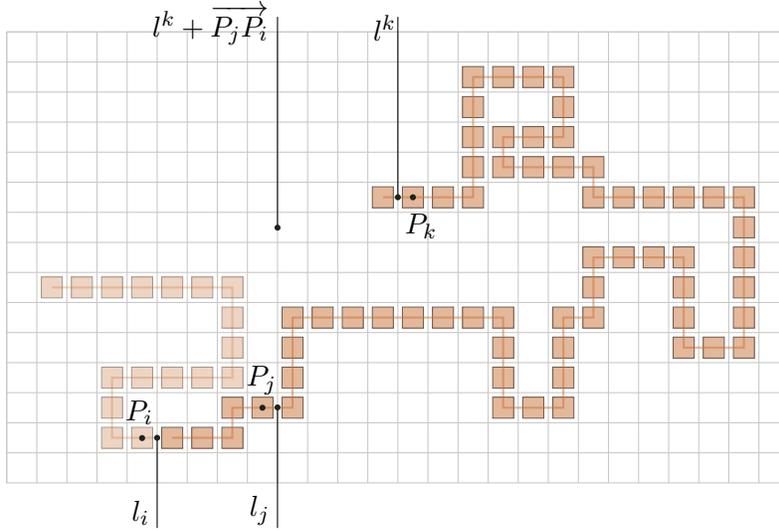

We prove Lemma~\ref{lem:shield} in Section~\ref{subsec:dominant}.
First, we state a simplifying assumption on $P$ (that we show is without loss of generality), then give an intuition for the overall proof strategy. The section then proceeds with a number of technical claims followed by the actual proof of Lemma~\ref{lem:shield}.
In this section, for the sake of brevity, claims are stated without repeating the hypotheses on $P$ and related notation. 

\argument{Assumption on $P$ used throughout this section} 
Observe that if we prove that any prefix of $P$ is pumpable or fragile, this means that $P$ is also pumpable or fragile.
Therefore, without loss of generality, for the remainder of this section (including the proof of Lemma~\ref{lem:shield}) we suppose that the last tile of $P$ is $P_{k+1}$, i.e. that $P=P_{0,1,\ldots,k+1}$.

\argument{Intuition for the proof of Lemma~\ref{lem:shield}}
Starting from a producible path $P$ with the  properties described in  Definition~\ref{def:shield},
we define three indices on $P$ (called a ``shield'') which in turn are used to define an infinite curve $c$ that partitions $\mathbb{R}^2$.
Then to build a path $R$ that stays on the right hand side of $c$, and that will ultimately allow us to reason about $P$ and show that $P$ is pumpable or fragile.
Since the proof is rather involved, we split it up into parts, each containing one or more \sublemmanames:

\begin{itemize}
\item Using Definition~\ref{def:shield}, in Subsection~\ref{subsec:c} we define a bi-infinite curve $c$ using the ray $l^i$, the path $P_{i+1,i+2,\ldots,k}$ and the ray $l^k$.
By Definition~\ref{def:simple infinite almost-vertical polygonal curve}, and Theorem~\ref{thm:infinite-jordan},
$c$  cuts $\mathbb{R}^2$ into two pieces: the left-hand and right-hand side of $c$.
 In the rest of the proof, we will use the right-hand side $\mathcal C \subset \mathbb{R}^2$ of $c$ as a ``workspace'' where we can edit paths freely (hence the name of that connected component in Definition~\ref{def:workspace}). The intuition is that $c$ ``shields our edits'' from $\sigma\cup\asm {P_{0,1,\ldots,i}}$, which is entirely in the left-hand side of $c$, and thus prevents $\sigma\cup \asm{P_{0,1,\ldots,i}}$ from blocking these paths in the workspace $\mathcal C$.

\item We then reason by induction on the length of $P$.
  The initial setup for the inductive argument  is given in Subsection~\ref{subsec:initial}, and goes as follows (in a number of places we may reach the early conclusion that $P$ is fragile, in which case we are done with the entire proof of Lemma~\ref{lem:shield}):
  \begin{itemize}
  \item In Subsection~\ref{subsec:first}, we define a tile $P_{m_0}$ of $P$ called a \emph{dominant tile}, which means that $P_{m_0}$ is such that for all integers $n\geq 0$, $P_{m_0}+n\vect{P_iP_j}$ is in $\mathcal C$.
      \item Then, in Subsection~\ref{subsec:r}, we define a binding path $r$ in $\mathbb{Z}^2$ and prove a number of key properties about about it and it's translation $r+\vpji$. We then use $r$ as a sequence of locations along which we we can either tile a producible path $R$, or else show that $P$ is fragile. 
  The path $R$ is built in such a way that both $R$ and $R+\vpji$ are producible and~in~$\mathcal C$.
  \item To complete the setup for the  inductive argument, in Subsection~\ref{subsec:u0v0}  we use $R$ and $m_0$ to define the initial inductive indices $u_0$ and $v_0$.
    To define these indices, we use a ray $L^{m_0}$ that starts from the tile $P_{m_0}$ and splits the component $\mathcal C$ into two parts (called $\cp$ and $\cm$), which guarantees that $u_0\leq m_0\leq v_0$, which in turn means that the pumping of $P$ between $u_0$ and $v_0$ is well-defined (i.e. $u_0<v_0$) and has pumping vector $\vpij$.
    (note that the pumping is not a simple path until the last step of the induction, where we eventually find a simple pumping of $P$).
  \end{itemize}

  The inductive step is then defined in Subsection~\ref{subsec:dominant}, where we show that either $P$ is pumpable or fragile, or else we can use $R$ again, along with inductive indices $u_n$, $m_n$, $v_n$, to find new indices $u_{n+1}$, $m_{n+1}$ and $v_{n+1}$, but with $m_{n+1}>m_n$.
  Since $P$ is of finite length, we will eventually run out of new indices (values for $m_{n+1}$, in particular), leading to the conclusion that $P$ is either pumpable or fragile.
\end{itemize}

\subsection{Reasoning about the curve $c$ and the workspace $\mathcal C$}\label{subsec:c}
\begin{sublemma}\label{lem:c-cuts}
  Let $(i, j, k)$ be a shield for $P$. The cut of shield $(i, j, k)$, called curve $c$ Definition~\ref{def:c}, is simple and cuts the plane $\R^2$ into two connected components. Moreover, these two components are defined in the conclusion of Theorem~\ref{thm:infinite-jordan}.
\end{sublemma}
\begin{proof}
  We first claim that $c$ is a simple curve. Indeed, $\embed{P_{i+1,i+2,\ldots,k}}$ intersects neither $\reverse{l^i}$ nor $l^k$ (by visibility of $\glu P i$ and $\glu P k$, respectively).
 The half-line $\reverse{l^i}$ (from the south) and the ray $l^k$ (to the north) do not intersect by Definition~\ref{def:shield} and in particular by the visibility of $\glu P i$ and $\glu P k$ and since $i \neq k$.
  Moreover, the length-$1/2$ segments
  $\gs{l^i(0)}{\pos{P_{i+1}}}$
  (that joins the ray $\reverse{l^i}$ to $\embed{P_{i+1,i+2,\ldots,k}}$), and 
   $\gs{\pos{P_k}}{l^k(0)}$ (that joins 
  $\embed{P_{i+1,i+2,\ldots,k}}$ to $l^k$) are horizontal, and only intersect $\embed{P_{i+1,i+2,\ldots,k}}$ and those rays at their respective endpoints.
  Hence, $c$ is a simple curve.
  Moreover, $c$ is connected since the endpoints of the five connected curves that define it are equal in the order given. 
  
 Hence $c$ satisfies Definition~\ref{def:simple infinite almost-vertical polygonal curve}. 
  Then, by Theorem~\ref{thm:infinite-jordan}, $c$ cuts the plane $\R^2$ into the two connected components defined in its conclusion.
\end{proof}

\begin{figure}[ht]
  \begin{center}
    \begin{tikzpicture}[scale=\scale]\draw[draw={rgb,255:red,200; green,200; blue,200}](0.5,-0.2) rectangle (26.5, -15.2);
\draw[draw={rgb,255:red,200; green,200; blue,200}](0.5,-15.2)--(0.5,-0.2);
\draw[draw={rgb,255:red,200; green,200; blue,200}](1.5,-15.2)--(1.5,-0.2);
\draw[draw={rgb,255:red,200; green,200; blue,200}](2.5,-15.2)--(2.5,-0.2);
\draw[draw={rgb,255:red,200; green,200; blue,200}](3.5,-15.2)--(3.5,-0.2);
\draw[draw={rgb,255:red,200; green,200; blue,200}](4.5,-15.2)--(4.5,-0.2);
\draw[draw={rgb,255:red,200; green,200; blue,200}](5.5,-15.2)--(5.5,-0.2);
\draw[draw={rgb,255:red,200; green,200; blue,200}](6.5,-15.2)--(6.5,-0.2);
\draw[draw={rgb,255:red,200; green,200; blue,200}](7.5,-15.2)--(7.5,-0.2);
\draw[draw={rgb,255:red,200; green,200; blue,200}](8.5,-15.2)--(8.5,-0.2);
\draw[draw={rgb,255:red,200; green,200; blue,200}](9.5,-15.2)--(9.5,-0.2);
\draw[draw={rgb,255:red,200; green,200; blue,200}](10.5,-15.2)--(10.5,-0.2);
\draw[draw={rgb,255:red,200; green,200; blue,200}](11.5,-15.2)--(11.5,-0.2);
\draw[draw={rgb,255:red,200; green,200; blue,200}](12.5,-15.2)--(12.5,-0.2);
\draw[draw={rgb,255:red,200; green,200; blue,200}](13.5,-15.2)--(13.5,-0.2);
\draw[draw={rgb,255:red,200; green,200; blue,200}](14.5,-15.2)--(14.5,-0.2);
\draw[draw={rgb,255:red,200; green,200; blue,200}](15.5,-15.2)--(15.5,-0.2);
\draw[draw={rgb,255:red,200; green,200; blue,200}](16.5,-15.2)--(16.5,-0.2);
\draw[draw={rgb,255:red,200; green,200; blue,200}](17.5,-15.2)--(17.5,-0.2);
\draw[draw={rgb,255:red,200; green,200; blue,200}](18.5,-15.2)--(18.5,-0.2);
\draw[draw={rgb,255:red,200; green,200; blue,200}](19.5,-15.2)--(19.5,-0.2);
\draw[draw={rgb,255:red,200; green,200; blue,200}](20.5,-15.2)--(20.5,-0.2);
\draw[draw={rgb,255:red,200; green,200; blue,200}](21.5,-15.2)--(21.5,-0.2);
\draw[draw={rgb,255:red,200; green,200; blue,200}](22.5,-15.2)--(22.5,-0.2);
\draw[draw={rgb,255:red,200; green,200; blue,200}](23.5,-15.2)--(23.5,-0.2);
\draw[draw={rgb,255:red,200; green,200; blue,200}](24.5,-15.2)--(24.5,-0.2);
\draw[draw={rgb,255:red,200; green,200; blue,200}](25.5,-15.2)--(25.5,-0.2);
\draw[draw={rgb,255:red,200; green,200; blue,200}](0.5,-0.2)--(26.5,-0.2);
\draw[draw={rgb,255:red,200; green,200; blue,200}](0.5,-1.2)--(26.5,-1.2);
\draw[draw={rgb,255:red,200; green,200; blue,200}](0.5,-2.2)--(26.5,-2.2);
\draw[draw={rgb,255:red,200; green,200; blue,200}](0.5,-3.2)--(26.5,-3.2);
\draw[draw={rgb,255:red,200; green,200; blue,200}](0.5,-4.2)--(26.5,-4.2);
\draw[draw={rgb,255:red,200; green,200; blue,200}](0.5,-5.2)--(26.5,-5.2);
\draw[draw={rgb,255:red,200; green,200; blue,200}](0.5,-6.2)--(26.5,-6.2);
\draw[draw={rgb,255:red,200; green,200; blue,200}](0.5,-7.2)--(26.5,-7.2);
\draw[draw={rgb,255:red,200; green,200; blue,200}](0.5,-8.2)--(26.5,-8.2);
\draw[draw={rgb,255:red,200; green,200; blue,200}](0.5,-9.2)--(26.5,-9.2);
\draw[draw={rgb,255:red,200; green,200; blue,200}](0.5,-10.2)--(26.5,-10.2);
\draw[draw={rgb,255:red,200; green,200; blue,200}](0.5,-11.2)--(26.5,-11.2);
\draw[draw={rgb,255:red,200; green,200; blue,200}](0.5,-12.2)--(26.5,-12.2);
\draw[draw={rgb,255:red,200; green,200; blue,200}](0.5,-13.2)--(26.5,-13.2);
\draw[draw={rgb,255:red,200; green,200; blue,200}](0.5,-14.2)--(26.5,-14.2);
\draw[draw=none,fill={rgb,255:red,179; green,179; blue,179},opacity=0.25](13.5,-5.7)--(13.5,-0.2)--(26.5,-0.2)--(26.5,-15.2)--(5.5,-15.2)--(5.5,-13.7)(5.5,-13.7)--(8,-13.7)--(8,-12.7)--(10,-12.7)--(10,-9.7)--(17,-9.7)--(17,-12.7)--(19,-12.7)--(19,-9.7)--(20,-9.7)--(20,-7.7)--(23,-7.7)--(23,-10.7)--(25,-10.7)--(25,-5.7)--(20,-5.7)--(20,-4.7)--(17,-4.7)--(17,-3.7)--(19,-3.7)--(19,-1.7)--(16,-1.7)--(16,-5.7)--(13.5,-5.7);
\draw(22.5, -3.7) node[anchor=south west] {$\mathcal C$};
\draw[draw={rgb,255:red,0; green,0; blue,0},fill={rgb,255:red,200; green,113; blue,55},opacity=0.29899997,fill opacity=0.29899997](1.65,-8.35) rectangle (2.35, -9.05);
\draw[draw={rgb,255:red,0; green,0; blue,0},fill={rgb,255:red,200; green,113; blue,55},opacity=0.29899997,fill opacity=0.29899997](2.65,-8.35) rectangle (3.35, -9.05);
\draw[draw={rgb,255:red,0; green,0; blue,0},fill={rgb,255:red,200; green,113; blue,55},opacity=0.29899997,fill opacity=0.29899997](3.65,-8.35) rectangle (4.35, -9.05);
\draw[draw={rgb,255:red,0; green,0; blue,0},fill={rgb,255:red,200; green,113; blue,55},opacity=0.29899997,fill opacity=0.29899997](4.65,-8.35) rectangle (5.35, -9.05);
\draw[draw={rgb,255:red,0; green,0; blue,0},fill={rgb,255:red,200; green,113; blue,55},opacity=0.29899997,fill opacity=0.29899997](5.65,-8.35) rectangle (6.35, -9.05);
\draw[draw={rgb,255:red,0; green,0; blue,0},fill={rgb,255:red,200; green,113; blue,55},opacity=0.29899997,fill opacity=0.29899997](6.65,-8.35) rectangle (7.35, -9.05);
\draw[draw={rgb,255:red,0; green,0; blue,0},fill={rgb,255:red,200; green,113; blue,55},opacity=0.29899997,fill opacity=0.29899997](7.65,-8.35) rectangle (8.35, -9.05);
\draw[draw={rgb,255:red,0; green,0; blue,0},fill={rgb,255:red,200; green,113; blue,55},opacity=0.29899997,fill opacity=0.29899997](7.65,-9.35) rectangle (8.35, -10.05);
\draw[draw={rgb,255:red,0; green,0; blue,0},fill={rgb,255:red,200; green,113; blue,55},opacity=0.29899997,fill opacity=0.29899997](7.65,-10.35) rectangle (8.35, -11.05);
\draw[draw={rgb,255:red,0; green,0; blue,0},fill={rgb,255:red,200; green,113; blue,55},opacity=0.29899997,fill opacity=0.29899997](7.65,-11.35) rectangle (8.35, -12.05);
\draw[draw={rgb,255:red,0; green,0; blue,0},fill={rgb,255:red,200; green,113; blue,55},opacity=0.29899997,fill opacity=0.29899997](6.65,-11.35) rectangle (7.35, -12.05);
\draw[draw={rgb,255:red,0; green,0; blue,0},fill={rgb,255:red,200; green,113; blue,55},opacity=0.29899997,fill opacity=0.29899997](5.65,-11.35) rectangle (6.35, -12.05);
\draw[draw={rgb,255:red,0; green,0; blue,0},fill={rgb,255:red,200; green,113; blue,55},opacity=0.29899997,fill opacity=0.29899997](4.65,-11.35) rectangle (5.35, -12.05);
\draw[draw={rgb,255:red,0; green,0; blue,0},fill={rgb,255:red,200; green,113; blue,55},opacity=0.29899997,fill opacity=0.29899997](3.65,-11.35) rectangle (4.35, -12.05);
\draw[draw={rgb,255:red,0; green,0; blue,0},fill={rgb,255:red,200; green,113; blue,55},opacity=0.29899997,fill opacity=0.29899997](3.65,-12.35) rectangle (4.35, -13.05);
\draw[draw={rgb,255:red,0; green,0; blue,0},fill={rgb,255:red,200; green,113; blue,55},opacity=0.29899997,fill opacity=0.29899997](3.65,-13.35) rectangle (4.35, -14.05);
\draw[draw={rgb,255:red,0; green,0; blue,0},fill={rgb,255:red,200; green,113; blue,55},opacity=0.29899997,fill opacity=0.29899997](4.65,-13.35) rectangle (5.35, -14.05);
\draw[draw={rgb,255:red,200; green,113; blue,55},opacity=0.29899997,thick](2,-8.7)--(8,-8.7)--(8,-11.7)--(4,-11.7)--(4,-13.7)--(5.5,-13.7);
\draw[draw={rgb,255:red,0; green,0; blue,0},fill={rgb,255:red,200; green,113; blue,55},opacity=0.5,fill opacity=0.5](5.65,-13.35) rectangle (6.35, -14.05);
\draw[draw={rgb,255:red,0; green,0; blue,0},fill={rgb,255:red,200; green,113; blue,55},opacity=0.5,fill opacity=0.5](6.65,-13.35) rectangle (7.35, -14.05);
\draw[draw={rgb,255:red,0; green,0; blue,0},fill={rgb,255:red,200; green,113; blue,55},opacity=0.5,fill opacity=0.5](7.65,-13.35) rectangle (8.35, -14.05);
\draw[draw={rgb,255:red,0; green,0; blue,0},fill={rgb,255:red,200; green,113; blue,55},opacity=0.5,fill opacity=0.5](7.65,-12.35) rectangle (8.35, -13.05);
\draw[draw={rgb,255:red,0; green,0; blue,0},fill={rgb,255:red,200; green,113; blue,55},opacity=0.5,fill opacity=0.5](8.65,-12.35) rectangle (9.35, -13.05);
\draw[draw={rgb,255:red,0; green,0; blue,0},fill={rgb,255:red,200; green,113; blue,55},opacity=0.5,fill opacity=0.5](9.65,-12.35) rectangle (10.35, -13.05);
\draw[draw={rgb,255:red,0; green,0; blue,0},fill={rgb,255:red,200; green,113; blue,55},opacity=0.5,fill opacity=0.5](9.65,-11.35) rectangle (10.35, -12.05);
\draw[draw={rgb,255:red,0; green,0; blue,0},fill={rgb,255:red,200; green,113; blue,55},opacity=0.5,fill opacity=0.5](9.65,-10.35) rectangle (10.35, -11.05);
\draw[draw={rgb,255:red,0; green,0; blue,0},fill={rgb,255:red,200; green,113; blue,55},opacity=0.5,fill opacity=0.5](9.65,-9.35) rectangle (10.35, -10.05);
\draw[draw={rgb,255:red,0; green,0; blue,0},fill={rgb,255:red,200; green,113; blue,55},opacity=0.5,fill opacity=0.5](10.65,-9.35) rectangle (11.35, -10.05);
\draw[draw={rgb,255:red,0; green,0; blue,0},fill={rgb,255:red,200; green,113; blue,55},opacity=0.5,fill opacity=0.5](11.65,-9.35) rectangle (12.35, -10.05);
\draw[draw={rgb,255:red,0; green,0; blue,0},fill={rgb,255:red,200; green,113; blue,55},opacity=0.5,fill opacity=0.5](12.65,-9.35) rectangle (13.35, -10.05);
\draw[draw={rgb,255:red,0; green,0; blue,0},fill={rgb,255:red,200; green,113; blue,55},opacity=0.5,fill opacity=0.5](13.65,-9.35) rectangle (14.35, -10.05);
\draw[draw={rgb,255:red,0; green,0; blue,0},fill={rgb,255:red,200; green,113; blue,55},opacity=0.5,fill opacity=0.5](14.65,-9.35) rectangle (15.35, -10.05);
\draw[draw={rgb,255:red,0; green,0; blue,0},fill={rgb,255:red,200; green,113; blue,55},opacity=0.5,fill opacity=0.5](15.65,-9.35) rectangle (16.35, -10.05);
\draw[draw={rgb,255:red,0; green,0; blue,0},fill={rgb,255:red,200; green,113; blue,55},opacity=0.5,fill opacity=0.5](16.65,-9.35) rectangle (17.35, -10.05);
\draw[draw={rgb,255:red,0; green,0; blue,0},fill={rgb,255:red,200; green,113; blue,55},opacity=0.5,fill opacity=0.5](16.65,-10.35) rectangle (17.35, -11.05);
\draw[draw={rgb,255:red,0; green,0; blue,0},fill={rgb,255:red,200; green,113; blue,55},opacity=0.5,fill opacity=0.5](16.65,-11.35) rectangle (17.35, -12.05);
\draw[draw={rgb,255:red,0; green,0; blue,0},fill={rgb,255:red,200; green,113; blue,55},opacity=0.5,fill opacity=0.5](16.65,-12.35) rectangle (17.35, -13.05);
\draw[draw={rgb,255:red,0; green,0; blue,0},fill={rgb,255:red,200; green,113; blue,55},opacity=0.5,fill opacity=0.5](17.65,-12.35) rectangle (18.35, -13.05);
\draw[draw={rgb,255:red,0; green,0; blue,0},fill={rgb,255:red,200; green,113; blue,55},opacity=0.5,fill opacity=0.5](18.65,-12.35) rectangle (19.35, -13.05);
\draw[draw={rgb,255:red,0; green,0; blue,0},fill={rgb,255:red,200; green,113; blue,55},opacity=0.5,fill opacity=0.5](18.65,-11.35) rectangle (19.35, -12.05);
\draw[draw={rgb,255:red,0; green,0; blue,0},fill={rgb,255:red,200; green,113; blue,55},opacity=0.5,fill opacity=0.5](18.65,-10.35) rectangle (19.35, -11.05);
\draw[draw={rgb,255:red,0; green,0; blue,0},fill={rgb,255:red,200; green,113; blue,55},opacity=0.5,fill opacity=0.5](18.65,-9.35) rectangle (19.35, -10.05);
\draw[draw={rgb,255:red,0; green,0; blue,0},fill={rgb,255:red,200; green,113; blue,55},opacity=0.5,fill opacity=0.5](19.65,-9.35) rectangle (20.35, -10.05);
\draw[draw={rgb,255:red,0; green,0; blue,0},fill={rgb,255:red,200; green,113; blue,55},opacity=0.5,fill opacity=0.5](19.65,-8.35) rectangle (20.35, -9.05);
\draw[draw={rgb,255:red,0; green,0; blue,0},fill={rgb,255:red,200; green,113; blue,55},opacity=0.5,fill opacity=0.5](19.65,-7.35) rectangle (20.35, -8.05);
\draw[draw={rgb,255:red,0; green,0; blue,0},fill={rgb,255:red,200; green,113; blue,55},opacity=0.5,fill opacity=0.5](20.65,-7.35) rectangle (21.35, -8.05);
\draw[draw={rgb,255:red,0; green,0; blue,0},fill={rgb,255:red,200; green,113; blue,55},opacity=0.5,fill opacity=0.5](21.65,-7.35) rectangle (22.35, -8.05);
\draw[draw={rgb,255:red,0; green,0; blue,0},fill={rgb,255:red,200; green,113; blue,55},opacity=0.5,fill opacity=0.5](22.65,-7.35) rectangle (23.35, -8.05);
\draw[draw={rgb,255:red,0; green,0; blue,0},fill={rgb,255:red,200; green,113; blue,55},opacity=0.5,fill opacity=0.5](22.65,-8.35) rectangle (23.35, -9.05);
\draw[draw={rgb,255:red,0; green,0; blue,0},fill={rgb,255:red,200; green,113; blue,55},opacity=0.5,fill opacity=0.5](22.65,-9.35) rectangle (23.35, -10.05);
\draw[draw={rgb,255:red,0; green,0; blue,0},fill={rgb,255:red,200; green,113; blue,55},opacity=0.5,fill opacity=0.5](22.65,-10.35) rectangle (23.35, -11.05);
\draw[draw={rgb,255:red,0; green,0; blue,0},fill={rgb,255:red,200; green,113; blue,55},opacity=0.5,fill opacity=0.5](23.65,-10.35) rectangle (24.35, -11.05);
\draw[draw={rgb,255:red,0; green,0; blue,0},fill={rgb,255:red,200; green,113; blue,55},opacity=0.5,fill opacity=0.5](24.65,-10.35) rectangle (25.35, -11.05);
\draw[draw={rgb,255:red,0; green,0; blue,0},fill={rgb,255:red,200; green,113; blue,55},opacity=0.5,fill opacity=0.5](24.65,-9.35) rectangle (25.35, -10.05);
\draw[draw={rgb,255:red,0; green,0; blue,0},fill={rgb,255:red,200; green,113; blue,55},opacity=0.5,fill opacity=0.5](24.65,-8.35) rectangle (25.35, -9.05);
\draw[draw={rgb,255:red,0; green,0; blue,0},fill={rgb,255:red,200; green,113; blue,55},opacity=0.5,fill opacity=0.5](24.65,-7.35) rectangle (25.35, -8.05);
\draw[draw={rgb,255:red,0; green,0; blue,0},fill={rgb,255:red,200; green,113; blue,55},opacity=0.5,fill opacity=0.5](24.65,-6.35) rectangle (25.35, -7.05);
\draw[draw={rgb,255:red,0; green,0; blue,0},fill={rgb,255:red,200; green,113; blue,55},opacity=0.5,fill opacity=0.5](24.65,-5.35) rectangle (25.35, -6.05);
\draw[draw={rgb,255:red,0; green,0; blue,0},fill={rgb,255:red,200; green,113; blue,55},opacity=0.5,fill opacity=0.5](23.65,-5.35) rectangle (24.35, -6.05);
\draw[draw={rgb,255:red,0; green,0; blue,0},fill={rgb,255:red,200; green,113; blue,55},opacity=0.5,fill opacity=0.5](22.65,-5.35) rectangle (23.35, -6.05);
\draw[draw={rgb,255:red,0; green,0; blue,0},fill={rgb,255:red,200; green,113; blue,55},opacity=0.5,fill opacity=0.5](21.65,-5.35) rectangle (22.35, -6.05);
\draw[draw={rgb,255:red,0; green,0; blue,0},fill={rgb,255:red,200; green,113; blue,55},opacity=0.5,fill opacity=0.5](20.65,-5.35) rectangle (21.35, -6.05);
\draw[draw={rgb,255:red,0; green,0; blue,0},fill={rgb,255:red,200; green,113; blue,55},opacity=0.5,fill opacity=0.5](19.65,-5.35) rectangle (20.35, -6.05);
\draw[draw={rgb,255:red,0; green,0; blue,0},fill={rgb,255:red,200; green,113; blue,55},opacity=0.5,fill opacity=0.5](19.65,-4.35) rectangle (20.35, -5.05);
\draw[draw={rgb,255:red,0; green,0; blue,0},fill={rgb,255:red,200; green,113; blue,55},opacity=0.5,fill opacity=0.5](18.65,-4.35) rectangle (19.35, -5.05);
\draw[draw={rgb,255:red,0; green,0; blue,0},fill={rgb,255:red,200; green,113; blue,55},opacity=0.5,fill opacity=0.5](17.65,-4.35) rectangle (18.35, -5.05);
\draw[draw={rgb,255:red,0; green,0; blue,0},fill={rgb,255:red,200; green,113; blue,55},opacity=0.5,fill opacity=0.5](16.65,-4.35) rectangle (17.35, -5.05);
\draw[draw={rgb,255:red,0; green,0; blue,0},fill={rgb,255:red,200; green,113; blue,55},opacity=0.5,fill opacity=0.5](16.65,-3.35) rectangle (17.35, -4.05);
\draw[draw={rgb,255:red,0; green,0; blue,0},fill={rgb,255:red,200; green,113; blue,55},opacity=0.5,fill opacity=0.5](17.65,-3.35) rectangle (18.35, -4.05);
\draw[draw={rgb,255:red,0; green,0; blue,0},fill={rgb,255:red,200; green,113; blue,55},opacity=0.5,fill opacity=0.5](18.65,-3.35) rectangle (19.35, -4.05);
\draw[draw={rgb,255:red,0; green,0; blue,0},fill={rgb,255:red,200; green,113; blue,55},opacity=0.5,fill opacity=0.5](18.65,-2.35) rectangle (19.35, -3.05);
\draw[draw={rgb,255:red,0; green,0; blue,0},fill={rgb,255:red,200; green,113; blue,55},opacity=0.5,fill opacity=0.5](18.65,-1.35) rectangle (19.35, -2.05);
\draw[draw={rgb,255:red,0; green,0; blue,0},fill={rgb,255:red,200; green,113; blue,55},opacity=0.5,fill opacity=0.5](17.65,-1.35) rectangle (18.35, -2.05);
\draw[draw={rgb,255:red,0; green,0; blue,0},fill={rgb,255:red,200; green,113; blue,55},opacity=0.5,fill opacity=0.5](16.65,-1.35) rectangle (17.35, -2.05);
\draw[draw={rgb,255:red,0; green,0; blue,0},fill={rgb,255:red,200; green,113; blue,55},opacity=0.5,fill opacity=0.5](15.65,-1.35) rectangle (16.35, -2.05);
\draw[draw={rgb,255:red,0; green,0; blue,0},fill={rgb,255:red,200; green,113; blue,55},opacity=0.5,fill opacity=0.5](15.65,-2.35) rectangle (16.35, -3.05);
\draw[draw={rgb,255:red,0; green,0; blue,0},fill={rgb,255:red,200; green,113; blue,55},opacity=0.5,fill opacity=0.5](15.65,-3.35) rectangle (16.35, -4.05);
\draw[draw={rgb,255:red,0; green,0; blue,0},fill={rgb,255:red,200; green,113; blue,55},opacity=0.5,fill opacity=0.5](15.65,-4.35) rectangle (16.35, -5.05);
\draw[draw={rgb,255:red,0; green,0; blue,0},fill={rgb,255:red,200; green,113; blue,55},opacity=0.5,fill opacity=0.5](15.65,-5.35) rectangle (16.35, -6.05);
\draw[draw={rgb,255:red,0; green,0; blue,0},fill={rgb,255:red,200; green,113; blue,55},opacity=0.5,fill opacity=0.5](14.65,-5.35) rectangle (15.35, -6.05);
\draw[draw={rgb,255:red,0; green,0; blue,0},fill={rgb,255:red,200; green,113; blue,55},opacity=0.5,fill opacity=0.5](13.65,-5.35) rectangle (14.35, -6.05);
\draw[draw={rgb,255:red,0; green,0; blue,0},fill={rgb,255:red,200; green,113; blue,55},opacity=0.5,fill opacity=0.5](12.65,-5.35) rectangle (13.35, -6.05);
\draw[draw={rgb,255:red,200; green,113; blue,55},opacity=0.5,thick](6,-13.7)--(8,-13.7)--(8,-12.7)--(10,-12.7)--(10,-9.7)--(17,-9.7)--(17,-12.7)--(19,-12.7)--(19,-9.7)--(20,-9.7)--(20,-7.7)--(23,-7.7)--(23,-10.7)--(25,-10.7)--(25,-5.7)--(20,-5.7)--(20,-4.7)--(17,-4.7)--(17,-3.7)--(19,-3.7)--(19,-1.7)--(16,-1.7)--(16,-5.7)--(13,-5.7);
\draw[draw={rgb,255:red,0; green,0; blue,0}](9.5,-12.69)--(9.5,-16.7);
\draw(8.23, -16.93) node[anchor=south west] {$l_j$};
\draw[draw=none,fill={rgb,255:red,28; green,36; blue,31},thin](9.5, -12.69) ellipse (0.1cm and 0.1cm);\draw[draw={rgb,255:red,0; green,0; blue,0}](5.5,-13.7)--(5.5,-16.7);
\draw(4.32, -16.93) node[anchor=south west] {$l_i$};
\draw[draw={rgb,255:red,0; green,0; blue,0}](13.5,-5.7)--(13.5,0.3);
\draw(12.34, -0.81) node[anchor=south west] {$l^k$};
\draw[draw=none,fill={rgb,255:red,28; green,36; blue,31},thin](13.5, -5.7) ellipse (0.1cm and 0.1cm);\draw[draw=none,fill={rgb,255:red,28; green,36; blue,31},thin](5.5, -13.7) ellipse (0.1cm and 0.1cm);\draw[draw={rgb,255:red,0; green,0; blue,0}](9.5,-6.71)--(9.5,0.3);
\draw(5, -0.96) node[anchor=south west] {$l^k+\vect{P_jP_i}$};
\draw[draw=none,fill={rgb,255:red,28; green,36; blue,31},thin](9.5, -6.71) ellipse (0.1cm and 0.1cm);
\end{tikzpicture}
  \end{center}
  \caption{The path and shield triple $(i,j,k)$ from Figure~\ref{fig:shield-setup}, annotated with curve $c$ and component~$\mathcal C$.
    The border of the shaded region is the curve $c$  (from Definition~\ref{def:c}), 
    and the shaded region itself is the component~$\mathcal C$ (Definition~\ref{def:C}).}
  \label{fig:shield-setup-c}
\end{figure}
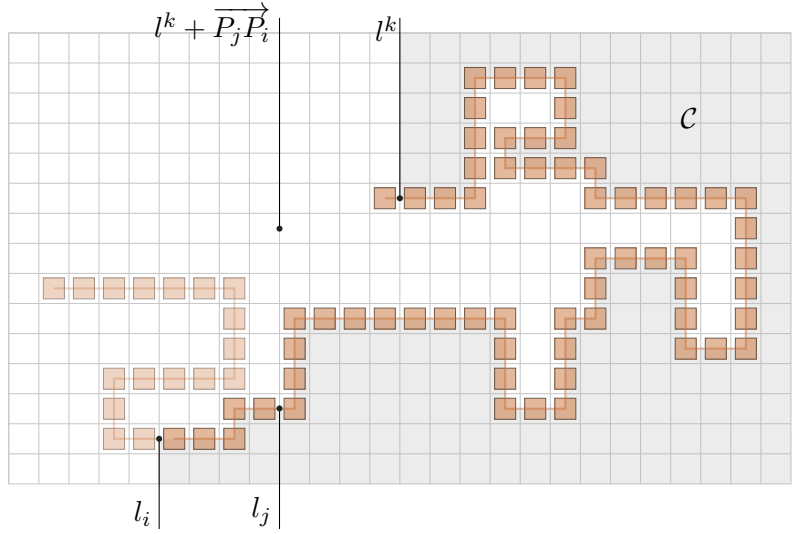

By definition, $\dom{\sigma \cup \asm{P_{0,1,\ldots, i}}}$ is a subset of $\Z^2$. The following claim captures the intuition that allows us to think of $\mathcal C$ as a ``workspace'' that is ``shielded'' from $\sigma \cup \asm{P_{0,1,\ldots, i}}$:

\begin{sublemma}\label{lem:c}
  Let $(i, j, k)$ be a shield for $P$ and let $\mathcal C$ be the workspace of shield $(i, j, k)$.
  Then $\dom{\sigma \cup \asm{P_{0,1,\ldots, i}}}$ is a subset of $\R^2 \setminus \mathcal{C}$.
\end{sublemma}
\begin{proof}
  By the definition of $c$, since $\glu P i$ points east, $\pos{P_i}$ is on the left-hand side of~$c$, hence $\pos{P_i}$ is not in $\mathcal{C}$.
  Moreover, since $P$ is a producible path, $\sigma \cup \asm{P_{0,1,\ldots,i}}$ is a connected assembly.
  We claim that $\sigma\cup\asm{P_{0,1,\ldots,i}}$ has no tile in $\mathcal C$, since otherwise one or both of (a)
  $\embed{P_{0,1,\ldots,i}}$ or 
(b) some of the tile positions of $\sigma$, or glue positions of abutting tiles of $\sigma$,
      would have to intersect $c$ to reach that tile. If that were the case, the intersection would be on one of the five curves used to define $c$:
  \begin{itemize}
  \item If $\embed{P_{0,1,\ldots,i}}$ intersects any of $\embed{P_{i+1,i+2,\ldots,k}}$, $\gs{l^i(0)}{\pos{P_{i+1}}}$, or $\gs{\pos{P_k}}{l^k(0)}$, this contradicts the fact that $P$ is simple.
  \item If $\dom{\sigma}$ intersects $\dom{\asm{P_{i+1,i+2,\ldots,k}}}$, this contradicts the fact that $P$ is a producible path.
  \item If $\dom{\sigma\cup\asm{P_{0,1,\ldots,i}}}$  has a glue positioned on $l^i$ or $l^k$, this contradicts the fact that $\glu P i$ is visible from the south, and $\glu P k$ is visible from the north, relative to $P$.
  \end{itemize}
  In all cases, we get a contradiction, and hence $\dom{\sigma\cup\asm{P_{0,1,\ldots,i}}}$ is disjoint from $\mathcal C$, and thus contained in $\R^2 \setminus \mathcal{C}$.
\end{proof}

The following claim restates Hypothesis~\ref{lem:hp:shield backup} of Definition~\ref{def:shield} to be in a form more suited to our proofs.  Specifically, \subl{lem:c-lk} states that $l^k(0)+\vpji$ is the only position of the ray $l^k+\vpij$ that may be in $\mathcal{C}$ and if this is the case then  $l^k(0)+\vpji$ is also a point of $c$.

\begin{sublemma}\label{lem:c-lk}
 $
 (( l^k+\vpji)  \cap \mathcal{C} ) 
  = 
 ( (l^k+\vpji) \cap c )
 \subseteq 
 \{l^k(0)+\vpji \}
 $. 
\end{sublemma}

\begin{proof}
First, we prove that if there is an intersection between $l^k+\vpji$ and the curve $c$ then 
this happens only at the point $l^k(0)+\vpji$,
 by analysing the five curves used to define $c$ (Definition~\ref{def:workspace}): 
  \begin{itemize}
  \item since $x_{\vpji}$ is a non-zero integer then $l^k+\vpji$ does not intersect $l^k$ nor $\gs{\pos{P_{k}}}{l^k(0)}$;
  \item by Hypothesis~\ref{lem:hp:shield backup} of Definition~\ref{def:shield}, the point $l^k(0)+\vpji$ is the only intersection permitted  between the ray $l^k + \vpji$ and either of the curves $\gs{l^i(0)}{\pos{P_{i+1}}^{}}$ or $\embed{ P_{i+1,i+2,\ldots,k}}$;
  \item since $l^i$ is a ray to the south, if $l^k+\vpji$ intersects $l^i$ then $l^k+\vpji$ also intersects $\gs{l^i(0)}{\pos{P_{i+1}}^{}}$ which means, by the previous case, that $l^k(0)+\vpji=l^i(0)$ is the only permitted intersection between these two rays.
  \end{itemize}
Thus $l^k(0)+\vpji$ is the only possible intersection of $l^k+\vpji$ with $c$, which shows the second part of our claim, i.e.  $( (l^k+\vpji) \cap c ) \subseteq \{l^k(0)+\vpji \}$.

We now show the first part of our claim, i.e.
$(( l^k+\vpji)  \cap \mathcal{C} ) = ( (l^k+\vpji) \cap c )$. First, since $l^k$ is a ray to the north, there exists $z\in\mathbb{R}, z\geq0$ such that $l^k(z)$ is strictly to the north of all positions of $\concat{\reverse{l^i}, \,\gs{l^i(0)}{\pos{P_{i+1}}^{}}, \,\embed{ P_{i+1,i+2,\ldots,k}}, \,\gs{\pos{P_{k}}}{l^k(0)}}$. Also, since $x_{\vpji}<0$, the point $l^k(z)+\vpji$ is to the west of all positions of $l^k$. Thus, $l^k(z)+\vpji$ is on the strict left-hand side of $c$, that is, in $\mathbb{R}^2\setminus \mathcal C$.

Therefore, if there were a $z'>0$ such that $l^k(z')+\vpji\in\mathcal C$ (the right hand side of $c$), and since $\{l^k(z)+\vpji \}$ is not in the right hand side of $c$, 
there would be at least one real number $z''\in[z', z[$ such that $l^k(z'')+\vpji$ is on $c$, contradicting the first part of this proof. Therefore, 
$(( l^k+\vpji)  \cap \mathcal{C} ) = ( (l^k+\vpji) \cap c )$.
\end{proof}

\subsection{The first dominant tile $P_{m_0}$ and the path $R$}
\label{subsec:initial}
In this subsection, we define a tile $P_{m_0}$ of $P$, and a path $R$, that will be used in the inductive argument in Section~\ref{subsec:dominant}.

\subsubsection{The ray $\rho$ used to define the first dominant tile $P_{m_0}$}

\label{subsec:first}

\newcommand\lmz{L^{m_0}}
Let $\rho$ be the lowest (southernmost) ray of vector $\vect{P_i P_j}$ that starts on $l^i$ and intersects at least one tile of $P_{i+1,i+2,\ldots,k}$.
Such a ray exists because in particular, the ray of vector $\vpij$ starting on $l^i$ and going through $\pos{P_{i+1}}$ intersects $\pos{P_{j+1}}$ (these two positions are positions of tiles of $P_{i+1,i+2,\ldots,k}$ because $i < j \leq k$), and since $P_{i+1,i+2,\ldots,k}$ is of finite length there is one ray that is the southernmost ray (meaning that out of all such rays, $\rho$ is the ray whose start position $\rho(0)$ is the southernmost on $l^i$).

Then, let $m_0 \in \{i+1,i+2,\ldots,k \} $ be the index such that $P_{m_0}$ is positioned on $\rho$, and is the easternmost tile of $P_{i+1,i+2,\ldots,k}$ that is positioned on $\rho$.
See Figure~\ref{fig:Pm0} for an example.
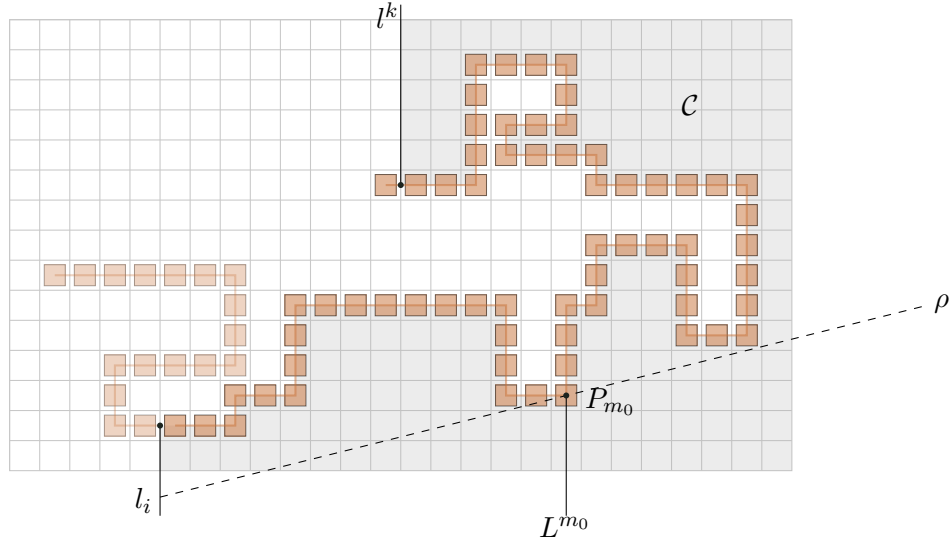
\begin{figure}[ht]
  \centering
  \begin{tikzpicture}[scale=\scale]\draw[draw={rgb,255:red,200; green,200; blue,200}](0.5,-0.2) rectangle (26.5, -15.2);
\draw[draw={rgb,255:red,200; green,200; blue,200}](0.5,-15.2)--(0.5,-0.2);
\draw[draw={rgb,255:red,200; green,200; blue,200}](1.5,-15.2)--(1.5,-0.2);
\draw[draw={rgb,255:red,200; green,200; blue,200}](2.5,-15.2)--(2.5,-0.2);
\draw[draw={rgb,255:red,200; green,200; blue,200}](3.5,-15.2)--(3.5,-0.2);
\draw[draw={rgb,255:red,200; green,200; blue,200}](4.5,-15.2)--(4.5,-0.2);
\draw[draw={rgb,255:red,200; green,200; blue,200}](5.5,-15.2)--(5.5,-0.2);
\draw[draw={rgb,255:red,200; green,200; blue,200}](6.5,-15.2)--(6.5,-0.2);
\draw[draw={rgb,255:red,200; green,200; blue,200}](7.5,-15.2)--(7.5,-0.2);
\draw[draw={rgb,255:red,200; green,200; blue,200}](8.5,-15.2)--(8.5,-0.2);
\draw[draw={rgb,255:red,200; green,200; blue,200}](9.5,-15.2)--(9.5,-0.2);
\draw[draw={rgb,255:red,200; green,200; blue,200}](10.5,-15.2)--(10.5,-0.2);
\draw[draw={rgb,255:red,200; green,200; blue,200}](11.5,-15.2)--(11.5,-0.2);
\draw[draw={rgb,255:red,200; green,200; blue,200}](12.5,-15.2)--(12.5,-0.2);
\draw[draw={rgb,255:red,200; green,200; blue,200}](13.5,-15.2)--(13.5,-0.2);
\draw[draw={rgb,255:red,200; green,200; blue,200}](14.5,-15.2)--(14.5,-0.2);
\draw[draw={rgb,255:red,200; green,200; blue,200}](15.5,-15.2)--(15.5,-0.2);
\draw[draw={rgb,255:red,200; green,200; blue,200}](16.5,-15.2)--(16.5,-0.2);
\draw[draw={rgb,255:red,200; green,200; blue,200}](17.5,-15.2)--(17.5,-0.2);
\draw[draw={rgb,255:red,200; green,200; blue,200}](18.5,-15.2)--(18.5,-0.2);
\draw[draw={rgb,255:red,200; green,200; blue,200}](19.5,-15.2)--(19.5,-0.2);
\draw[draw={rgb,255:red,200; green,200; blue,200}](20.5,-15.2)--(20.5,-0.2);
\draw[draw={rgb,255:red,200; green,200; blue,200}](21.5,-15.2)--(21.5,-0.2);
\draw[draw={rgb,255:red,200; green,200; blue,200}](22.5,-15.2)--(22.5,-0.2);
\draw[draw={rgb,255:red,200; green,200; blue,200}](23.5,-15.2)--(23.5,-0.2);
\draw[draw={rgb,255:red,200; green,200; blue,200}](24.5,-15.2)--(24.5,-0.2);
\draw[draw={rgb,255:red,200; green,200; blue,200}](25.5,-15.2)--(25.5,-0.2);
\draw[draw={rgb,255:red,200; green,200; blue,200}](0.5,-0.2)--(26.5,-0.2);
\draw[draw={rgb,255:red,200; green,200; blue,200}](0.5,-1.2)--(26.5,-1.2);
\draw[draw={rgb,255:red,200; green,200; blue,200}](0.5,-2.2)--(26.5,-2.2);
\draw[draw={rgb,255:red,200; green,200; blue,200}](0.5,-3.2)--(26.5,-3.2);
\draw[draw={rgb,255:red,200; green,200; blue,200}](0.5,-4.2)--(26.5,-4.2);
\draw[draw={rgb,255:red,200; green,200; blue,200}](0.5,-5.2)--(26.5,-5.2);
\draw[draw={rgb,255:red,200; green,200; blue,200}](0.5,-6.2)--(26.5,-6.2);
\draw[draw={rgb,255:red,200; green,200; blue,200}](0.5,-7.2)--(26.5,-7.2);
\draw[draw={rgb,255:red,200; green,200; blue,200}](0.5,-8.2)--(26.5,-8.2);
\draw[draw={rgb,255:red,200; green,200; blue,200}](0.5,-9.2)--(26.5,-9.2);
\draw[draw={rgb,255:red,200; green,200; blue,200}](0.5,-10.2)--(26.5,-10.2);
\draw[draw={rgb,255:red,200; green,200; blue,200}](0.5,-11.2)--(26.5,-11.2);
\draw[draw={rgb,255:red,200; green,200; blue,200}](0.5,-12.2)--(26.5,-12.2);
\draw[draw={rgb,255:red,200; green,200; blue,200}](0.5,-13.2)--(26.5,-13.2);
\draw[draw={rgb,255:red,200; green,200; blue,200}](0.5,-14.2)--(26.5,-14.2);
\draw[draw=none,fill={rgb,255:red,179; green,179; blue,179},opacity=0.25](13.5,-5.7)--(13.5,-0.2)--(26.5,-0.2)--(26.5,-15.2)--(5.5,-15.2)--(5.5,-13.7)(5.5,-13.7)--(8,-13.7)--(8,-12.7)--(10,-12.7)--(10,-9.7)--(17,-9.7)--(17,-12.7)--(19,-12.7)--(19,-9.7)--(20,-9.7)--(20,-7.7)--(23,-7.7)--(23,-10.7)--(25,-10.7)--(25,-5.7)--(20,-5.7)--(20,-4.7)--(17,-4.7)--(17,-3.7)--(19,-3.7)--(19,-1.7)--(16,-1.7)--(16,-5.7)--(13.5,-5.7);
\draw(22.5, -3.7) node[anchor=south west] {$\mathcal C$};
\draw[draw={rgb,255:red,0; green,0; blue,0},fill={rgb,255:red,200; green,113; blue,55},opacity=0.29899997,fill opacity=0.29899997](1.65,-8.35) rectangle (2.35, -9.05);
\draw[draw={rgb,255:red,0; green,0; blue,0},fill={rgb,255:red,200; green,113; blue,55},opacity=0.29899997,fill opacity=0.29899997](2.65,-8.35) rectangle (3.35, -9.05);
\draw[draw={rgb,255:red,0; green,0; blue,0},fill={rgb,255:red,200; green,113; blue,55},opacity=0.29899997,fill opacity=0.29899997](3.65,-8.35) rectangle (4.35, -9.05);
\draw[draw={rgb,255:red,0; green,0; blue,0},fill={rgb,255:red,200; green,113; blue,55},opacity=0.29899997,fill opacity=0.29899997](4.65,-8.35) rectangle (5.35, -9.05);
\draw[draw={rgb,255:red,0; green,0; blue,0},fill={rgb,255:red,200; green,113; blue,55},opacity=0.29899997,fill opacity=0.29899997](5.65,-8.35) rectangle (6.35, -9.05);
\draw[draw={rgb,255:red,0; green,0; blue,0},fill={rgb,255:red,200; green,113; blue,55},opacity=0.29899997,fill opacity=0.29899997](6.65,-8.35) rectangle (7.35, -9.05);
\draw[draw={rgb,255:red,0; green,0; blue,0},fill={rgb,255:red,200; green,113; blue,55},opacity=0.29899997,fill opacity=0.29899997](7.65,-8.35) rectangle (8.35, -9.05);
\draw[draw={rgb,255:red,0; green,0; blue,0},fill={rgb,255:red,200; green,113; blue,55},opacity=0.29899997,fill opacity=0.29899997](7.65,-9.35) rectangle (8.35, -10.05);
\draw[draw={rgb,255:red,0; green,0; blue,0},fill={rgb,255:red,200; green,113; blue,55},opacity=0.29899997,fill opacity=0.29899997](7.65,-10.35) rectangle (8.35, -11.05);
\draw[draw={rgb,255:red,0; green,0; blue,0},fill={rgb,255:red,200; green,113; blue,55},opacity=0.29899997,fill opacity=0.29899997](7.65,-11.35) rectangle (8.35, -12.05);
\draw[draw={rgb,255:red,0; green,0; blue,0},fill={rgb,255:red,200; green,113; blue,55},opacity=0.29899997,fill opacity=0.29899997](6.65,-11.35) rectangle (7.35, -12.05);
\draw[draw={rgb,255:red,0; green,0; blue,0},fill={rgb,255:red,200; green,113; blue,55},opacity=0.29899997,fill opacity=0.29899997](5.65,-11.35) rectangle (6.35, -12.05);
\draw[draw={rgb,255:red,0; green,0; blue,0},fill={rgb,255:red,200; green,113; blue,55},opacity=0.29899997,fill opacity=0.29899997](4.65,-11.35) rectangle (5.35, -12.05);
\draw[draw={rgb,255:red,0; green,0; blue,0},fill={rgb,255:red,200; green,113; blue,55},opacity=0.29899997,fill opacity=0.29899997](3.65,-11.35) rectangle (4.35, -12.05);
\draw[draw={rgb,255:red,0; green,0; blue,0},fill={rgb,255:red,200; green,113; blue,55},opacity=0.29899997,fill opacity=0.29899997](3.65,-12.35) rectangle (4.35, -13.05);
\draw[draw={rgb,255:red,0; green,0; blue,0},fill={rgb,255:red,200; green,113; blue,55},opacity=0.29899997,fill opacity=0.29899997](3.65,-13.35) rectangle (4.35, -14.05);
\draw[draw={rgb,255:red,0; green,0; blue,0},fill={rgb,255:red,200; green,113; blue,55},opacity=0.29899997,fill opacity=0.29899997](4.65,-13.35) rectangle (5.35, -14.05);
\draw[draw={rgb,255:red,200; green,113; blue,55},opacity=0.29899997,thick](2,-8.7)--(8,-8.7)--(8,-11.7)--(4,-11.7)--(4,-13.7)--(5.5,-13.7);
\draw[draw={rgb,255:red,0; green,0; blue,0},fill={rgb,255:red,200; green,113; blue,55},opacity=0.5,fill opacity=0.5](5.65,-13.35) rectangle (6.35, -14.05);
\draw[draw={rgb,255:red,0; green,0; blue,0},fill={rgb,255:red,200; green,113; blue,55},opacity=0.5,fill opacity=0.5](6.65,-13.35) rectangle (7.35, -14.05);
\draw[draw={rgb,255:red,0; green,0; blue,0},fill={rgb,255:red,200; green,113; blue,55},opacity=0.5,fill opacity=0.5](7.65,-13.35) rectangle (8.35, -14.05);
\draw[draw={rgb,255:red,0; green,0; blue,0},fill={rgb,255:red,200; green,113; blue,55},opacity=0.5,fill opacity=0.5](7.65,-12.35) rectangle (8.35, -13.05);
\draw[draw={rgb,255:red,0; green,0; blue,0},fill={rgb,255:red,200; green,113; blue,55},opacity=0.5,fill opacity=0.5](8.65,-12.35) rectangle (9.35, -13.05);
\draw[draw={rgb,255:red,0; green,0; blue,0},fill={rgb,255:red,200; green,113; blue,55},opacity=0.5,fill opacity=0.5](9.65,-12.35) rectangle (10.35, -13.05);
\draw[draw={rgb,255:red,0; green,0; blue,0},fill={rgb,255:red,200; green,113; blue,55},opacity=0.5,fill opacity=0.5](9.65,-11.35) rectangle (10.35, -12.05);
\draw[draw={rgb,255:red,0; green,0; blue,0},fill={rgb,255:red,200; green,113; blue,55},opacity=0.5,fill opacity=0.5](9.65,-10.35) rectangle (10.35, -11.05);
\draw[draw={rgb,255:red,0; green,0; blue,0},fill={rgb,255:red,200; green,113; blue,55},opacity=0.5,fill opacity=0.5](9.65,-9.35) rectangle (10.35, -10.05);
\draw[draw={rgb,255:red,0; green,0; blue,0},fill={rgb,255:red,200; green,113; blue,55},opacity=0.5,fill opacity=0.5](10.65,-9.35) rectangle (11.35, -10.05);
\draw[draw={rgb,255:red,0; green,0; blue,0},fill={rgb,255:red,200; green,113; blue,55},opacity=0.5,fill opacity=0.5](11.65,-9.35) rectangle (12.35, -10.05);
\draw[draw={rgb,255:red,0; green,0; blue,0},fill={rgb,255:red,200; green,113; blue,55},opacity=0.5,fill opacity=0.5](12.65,-9.35) rectangle (13.35, -10.05);
\draw[draw={rgb,255:red,0; green,0; blue,0},fill={rgb,255:red,200; green,113; blue,55},opacity=0.5,fill opacity=0.5](13.65,-9.35) rectangle (14.35, -10.05);
\draw[draw={rgb,255:red,0; green,0; blue,0},fill={rgb,255:red,200; green,113; blue,55},opacity=0.5,fill opacity=0.5](14.65,-9.35) rectangle (15.35, -10.05);
\draw[draw={rgb,255:red,0; green,0; blue,0},fill={rgb,255:red,200; green,113; blue,55},opacity=0.5,fill opacity=0.5](15.65,-9.35) rectangle (16.35, -10.05);
\draw[draw={rgb,255:red,0; green,0; blue,0},fill={rgb,255:red,200; green,113; blue,55},opacity=0.5,fill opacity=0.5](16.65,-9.35) rectangle (17.35, -10.05);
\draw[draw={rgb,255:red,0; green,0; blue,0},fill={rgb,255:red,200; green,113; blue,55},opacity=0.5,fill opacity=0.5](16.65,-10.35) rectangle (17.35, -11.05);
\draw[draw={rgb,255:red,0; green,0; blue,0},fill={rgb,255:red,200; green,113; blue,55},opacity=0.5,fill opacity=0.5](16.65,-11.35) rectangle (17.35, -12.05);
\draw[draw={rgb,255:red,0; green,0; blue,0},fill={rgb,255:red,200; green,113; blue,55},opacity=0.5,fill opacity=0.5](16.65,-12.35) rectangle (17.35, -13.05);
\draw[draw={rgb,255:red,0; green,0; blue,0},fill={rgb,255:red,200; green,113; blue,55},opacity=0.5,fill opacity=0.5](17.65,-12.35) rectangle (18.35, -13.05);
\draw[draw={rgb,255:red,0; green,0; blue,0},fill={rgb,255:red,200; green,113; blue,55},opacity=0.5,fill opacity=0.5](18.65,-12.35) rectangle (19.35, -13.05);
\draw[draw={rgb,255:red,0; green,0; blue,0},fill={rgb,255:red,200; green,113; blue,55},opacity=0.5,fill opacity=0.5](18.65,-11.35) rectangle (19.35, -12.05);
\draw[draw={rgb,255:red,0; green,0; blue,0},fill={rgb,255:red,200; green,113; blue,55},opacity=0.5,fill opacity=0.5](18.65,-10.35) rectangle (19.35, -11.05);
\draw[draw={rgb,255:red,0; green,0; blue,0},fill={rgb,255:red,200; green,113; blue,55},opacity=0.5,fill opacity=0.5](18.65,-9.35) rectangle (19.35, -10.05);
\draw[draw={rgb,255:red,0; green,0; blue,0},fill={rgb,255:red,200; green,113; blue,55},opacity=0.5,fill opacity=0.5](19.65,-9.35) rectangle (20.35, -10.05);
\draw[draw={rgb,255:red,0; green,0; blue,0},fill={rgb,255:red,200; green,113; blue,55},opacity=0.5,fill opacity=0.5](19.65,-8.35) rectangle (20.35, -9.05);
\draw[draw={rgb,255:red,0; green,0; blue,0},fill={rgb,255:red,200; green,113; blue,55},opacity=0.5,fill opacity=0.5](19.65,-7.35) rectangle (20.35, -8.05);
\draw[draw={rgb,255:red,0; green,0; blue,0},fill={rgb,255:red,200; green,113; blue,55},opacity=0.5,fill opacity=0.5](20.65,-7.35) rectangle (21.35, -8.05);
\draw[draw={rgb,255:red,0; green,0; blue,0},fill={rgb,255:red,200; green,113; blue,55},opacity=0.5,fill opacity=0.5](21.65,-7.35) rectangle (22.35, -8.05);
\draw[draw={rgb,255:red,0; green,0; blue,0},fill={rgb,255:red,200; green,113; blue,55},opacity=0.5,fill opacity=0.5](22.65,-7.35) rectangle (23.35, -8.05);
\draw[draw={rgb,255:red,0; green,0; blue,0},fill={rgb,255:red,200; green,113; blue,55},opacity=0.5,fill opacity=0.5](22.65,-8.35) rectangle (23.35, -9.05);
\draw[draw={rgb,255:red,0; green,0; blue,0},fill={rgb,255:red,200; green,113; blue,55},opacity=0.5,fill opacity=0.5](22.65,-9.35) rectangle (23.35, -10.05);
\draw[draw={rgb,255:red,0; green,0; blue,0},fill={rgb,255:red,200; green,113; blue,55},opacity=0.5,fill opacity=0.5](22.65,-10.35) rectangle (23.35, -11.05);
\draw[draw={rgb,255:red,0; green,0; blue,0},fill={rgb,255:red,200; green,113; blue,55},opacity=0.5,fill opacity=0.5](23.65,-10.35) rectangle (24.35, -11.05);
\draw[draw={rgb,255:red,0; green,0; blue,0},fill={rgb,255:red,200; green,113; blue,55},opacity=0.5,fill opacity=0.5](24.65,-10.35) rectangle (25.35, -11.05);
\draw[draw={rgb,255:red,0; green,0; blue,0},fill={rgb,255:red,200; green,113; blue,55},opacity=0.5,fill opacity=0.5](24.65,-9.35) rectangle (25.35, -10.05);
\draw[draw={rgb,255:red,0; green,0; blue,0},fill={rgb,255:red,200; green,113; blue,55},opacity=0.5,fill opacity=0.5](24.65,-8.35) rectangle (25.35, -9.05);
\draw[draw={rgb,255:red,0; green,0; blue,0},fill={rgb,255:red,200; green,113; blue,55},opacity=0.5,fill opacity=0.5](24.65,-7.35) rectangle (25.35, -8.05);
\draw[draw={rgb,255:red,0; green,0; blue,0},fill={rgb,255:red,200; green,113; blue,55},opacity=0.5,fill opacity=0.5](24.65,-6.35) rectangle (25.35, -7.05);
\draw[draw={rgb,255:red,0; green,0; blue,0},fill={rgb,255:red,200; green,113; blue,55},opacity=0.5,fill opacity=0.5](24.65,-5.35) rectangle (25.35, -6.05);
\draw[draw={rgb,255:red,0; green,0; blue,0},fill={rgb,255:red,200; green,113; blue,55},opacity=0.5,fill opacity=0.5](23.65,-5.35) rectangle (24.35, -6.05);
\draw[draw={rgb,255:red,0; green,0; blue,0},fill={rgb,255:red,200; green,113; blue,55},opacity=0.5,fill opacity=0.5](22.65,-5.35) rectangle (23.35, -6.05);
\draw[draw={rgb,255:red,0; green,0; blue,0},fill={rgb,255:red,200; green,113; blue,55},opacity=0.5,fill opacity=0.5](21.65,-5.35) rectangle (22.35, -6.05);
\draw[draw={rgb,255:red,0; green,0; blue,0},fill={rgb,255:red,200; green,113; blue,55},opacity=0.5,fill opacity=0.5](20.65,-5.35) rectangle (21.35, -6.05);
\draw[draw={rgb,255:red,0; green,0; blue,0},fill={rgb,255:red,200; green,113; blue,55},opacity=0.5,fill opacity=0.5](19.65,-5.35) rectangle (20.35, -6.05);
\draw[draw={rgb,255:red,0; green,0; blue,0},fill={rgb,255:red,200; green,113; blue,55},opacity=0.5,fill opacity=0.5](19.65,-4.35) rectangle (20.35, -5.05);
\draw[draw={rgb,255:red,0; green,0; blue,0},fill={rgb,255:red,200; green,113; blue,55},opacity=0.5,fill opacity=0.5](18.65,-4.35) rectangle (19.35, -5.05);
\draw[draw={rgb,255:red,0; green,0; blue,0},fill={rgb,255:red,200; green,113; blue,55},opacity=0.5,fill opacity=0.5](17.65,-4.35) rectangle (18.35, -5.05);
\draw[draw={rgb,255:red,0; green,0; blue,0},fill={rgb,255:red,200; green,113; blue,55},opacity=0.5,fill opacity=0.5](16.65,-4.35) rectangle (17.35, -5.05);
\draw[draw={rgb,255:red,0; green,0; blue,0},fill={rgb,255:red,200; green,113; blue,55},opacity=0.5,fill opacity=0.5](16.65,-3.35) rectangle (17.35, -4.05);
\draw[draw={rgb,255:red,0; green,0; blue,0},fill={rgb,255:red,200; green,113; blue,55},opacity=0.5,fill opacity=0.5](17.65,-3.35) rectangle (18.35, -4.05);
\draw[draw={rgb,255:red,0; green,0; blue,0},fill={rgb,255:red,200; green,113; blue,55},opacity=0.5,fill opacity=0.5](18.65,-3.35) rectangle (19.35, -4.05);
\draw[draw={rgb,255:red,0; green,0; blue,0},fill={rgb,255:red,200; green,113; blue,55},opacity=0.5,fill opacity=0.5](18.65,-2.35) rectangle (19.35, -3.05);
\draw[draw={rgb,255:red,0; green,0; blue,0},fill={rgb,255:red,200; green,113; blue,55},opacity=0.5,fill opacity=0.5](18.65,-1.35) rectangle (19.35, -2.05);
\draw[draw={rgb,255:red,0; green,0; blue,0},fill={rgb,255:red,200; green,113; blue,55},opacity=0.5,fill opacity=0.5](17.65,-1.35) rectangle (18.35, -2.05);
\draw[draw={rgb,255:red,0; green,0; blue,0},fill={rgb,255:red,200; green,113; blue,55},opacity=0.5,fill opacity=0.5](16.65,-1.35) rectangle (17.35, -2.05);
\draw[draw={rgb,255:red,0; green,0; blue,0},fill={rgb,255:red,200; green,113; blue,55},opacity=0.5,fill opacity=0.5](15.65,-1.35) rectangle (16.35, -2.05);
\draw[draw={rgb,255:red,0; green,0; blue,0},fill={rgb,255:red,200; green,113; blue,55},opacity=0.5,fill opacity=0.5](15.65,-2.35) rectangle (16.35, -3.05);
\draw[draw={rgb,255:red,0; green,0; blue,0},fill={rgb,255:red,200; green,113; blue,55},opacity=0.5,fill opacity=0.5](15.65,-3.35) rectangle (16.35, -4.05);
\draw[draw={rgb,255:red,0; green,0; blue,0},fill={rgb,255:red,200; green,113; blue,55},opacity=0.5,fill opacity=0.5](15.65,-4.35) rectangle (16.35, -5.05);
\draw[draw={rgb,255:red,0; green,0; blue,0},fill={rgb,255:red,200; green,113; blue,55},opacity=0.5,fill opacity=0.5](15.65,-5.35) rectangle (16.35, -6.05);
\draw[draw={rgb,255:red,0; green,0; blue,0},fill={rgb,255:red,200; green,113; blue,55},opacity=0.5,fill opacity=0.5](14.65,-5.35) rectangle (15.35, -6.05);
\draw[draw={rgb,255:red,0; green,0; blue,0},fill={rgb,255:red,200; green,113; blue,55},opacity=0.5,fill opacity=0.5](13.65,-5.35) rectangle (14.35, -6.05);
\draw[draw={rgb,255:red,0; green,0; blue,0},fill={rgb,255:red,200; green,113; blue,55},opacity=0.5,fill opacity=0.5](12.65,-5.35) rectangle (13.35, -6.05);
\draw[draw={rgb,255:red,200; green,113; blue,55},opacity=0.5,thick](6,-13.7)--(8,-13.7)--(8,-12.7)--(10,-12.7)--(10,-9.7)--(17,-9.7)--(17,-12.7)--(19,-12.7)--(19,-9.7)--(20,-9.7)--(20,-7.7)--(23,-7.7)--(23,-10.7)--(25,-10.7)--(25,-5.7)--(20,-5.7)--(20,-4.7)--(17,-4.7)--(17,-3.7)--(19,-3.7)--(19,-1.7)--(16,-1.7)--(16,-5.7)--(13,-5.7);
\draw[draw={rgb,255:red,0; green,0; blue,0}](5.5,-13.7)--(5.5,-16.7);
\draw(4.32, -16.93) node[anchor=south west] {$l_i$};
\draw[draw={rgb,255:red,0; green,0; blue,0}](13.5,-5.7)--(13.5,0.3);
\draw(12.34, -0.81) node[anchor=south west] {$l^k$};
\draw[draw=none,fill={rgb,255:red,28; green,36; blue,31},thin](13.5, -5.7) ellipse (0.1cm and 0.1cm);\draw[draw=none,fill={rgb,255:red,28; green,36; blue,31},thin](5.5, -13.7) ellipse (0.1cm and 0.1cm);\draw[draw={rgb,255:red,0; green,0; blue,0},dashed,thin](5.5,-16.08)--(31,-9.7);
\draw(30.89, -10.28) node[anchor=south west] {$\rho$};
\draw(19.29, -13.68) node[anchor=south west] {$P_{m_0}$};
\draw[draw={rgb,255:red,0; green,0; blue,0}](19,-12.7)--(19,-16.7);
\draw(17.76, -17.73) node[anchor=south west] {$\lmz$};
\draw[draw=none,fill={rgb,255:red,28; green,36; blue,31},thin](19, -12.7) ellipse (0.1cm and 0.1cm);
\end{tikzpicture}
  \caption{The ray $\rho$ and tile $P_{m_0}$.
    We define $\rho$ as the southernmost ray of vector $\vect{P_i P_j}$ that starts on $l^i$ and intersects the position of at least one tile of $P_{i+1,i+2,\ldots,k}$. The easternmost such intersection is then defined to be   $\pos{P_{m_0}}$, and $P_{m_0}$ is called the dominant tile.}\label{fig:Pm0}
\end{figure}
\enlargethispage{-2\baselineskip}
\begin{sublemma}
  \label{lem:m0i1}
  $m_0>i+1$.
\end{sublemma}
\begin{proof}
 By its definition,  we know that $m_0\geq i+1$. If we had $m_0 = i+1$, then $P_{j+1}$ would also be on $\rho$, and hence
  since $\xcoord{\vpij}>0$ (i.e., $\vpij$ has positive x-component, a consequence of  Definition~\ref{def:shield} and Lemma~\ref{lem:glue:east}), we would have $m_0 \geq j+1$, contradicting the fact that $i<j$.
\end{proof}

\begin{sublemma}
  \label{lem:cmz}
  Let $\lmz$ be a vertical ray from $\pos{P_{m_0}}$ to the south.
  Then $\lmz$ only intersects $\embed{P_{i+1,i+2,\ldots,k}}$ at $\pos{P_{m_0}}$, and for all integers $n>0$,    $\lmz+n\vpij$ does not intersect $\embed{P_{i+1,i+2,\ldots,k}}$.
\end{sublemma}
\begin{proof}
Remember that $m_0 \in \{i+1,i+2,\ldots,k \}$.
Thus, by definition of $\lmz$, $\lmz$ intersects $\embed{P_{i+1,i+2,\ldots,k}}$ at 
 $\pos{P_{m_0}}$. If there were another intersection point then that intersection would be strictly lower (strictly to the south) on $\lmz$ and that intersection would define a ray of vector $\vect{P_i P_j}$ that intersects $P_{i,i+1,\ldots,k}$ strictly lower than $\rho$  thus contradicting the definition of $\rho$ as the lowest such ray. Thus $\lmz$ only intersects $\embed{P_{i+1,i+2,\ldots,k}}$ at $\pos{P_{m_0}}$.

For the second conclusion of this lemma,
suppose, for the sake of contradiction, that there is an integer $n > 0$ such that $\embed{P_{i+1,i+2,\ldots,k}}$ and $\lmz+n\vpij$ intersect.
  Then that intersection is either:
  \begin{itemize}
  \item At $\lmz(0)+n\vpij$.  
      Since $n>0$, that intersection is on $\rho$ (since $\rho$ is of vector $\vpij$ and goes through $\pos{P_{m_0}} = \lmz(0)$ for $n=0$) and strictly to the east of $P_{m_0}$ (because $\xcoord \vpij > 0$), contradicting the definition of $P_{m_0}$ as the easternmost tile of $P$ whose position is on $\rho$.

  \item Or else strictly lower (more southern) on $\lmz+n\vpij$ than $\lmz(0)+n\vpij$. 
    Since $\lmz+n\vpij$ is on a column (vertical line of integer positions), that intersection happens between $\lmz+n\vpij$ and a tile of $P_{i+1,i+2,\ldots,k}$, which contradicts the definition of $\rho$ as the lowest ray of vector $\vpij$ through the position of a tile of $P_{i+1,i+2,\ldots,k}$.
  \end{itemize}
\end{proof}

\begin{sublemma}
  \label{lem:lmzInC}
  $\lmz$ is entirely in $\mathcal C$.
\end{sublemma}
\begin{proof}
  By \subl{lem:cmz}, $\lmz$ only intersects $P_{i+1,i+2,\ldots,k}$ at $\lmz(0) = \pos{P_{m_0}}$,
  which is on the border $c$ of $\mathcal C$ since $m_0\in\{\range{i+1}{i+2}{k}\}$.
Also, $\lmz$ does not otherwise intersect $c$ since $l^i$ and $l^k$ 
are both on glue columns (half-integer x-coordinate) and $\lmz$ is on a column  (integer x-coordinate). 
Finally, $\lmz$ is strictly to the east of $l^i$ (by definition of $\rho$), and hence by Definition~\ref{def:workspace}, $\lmz$ is entirely in $\mathcal C$.
\end{proof}

\begin{sublemma}
  \label{lem:lmzEast}
  Recall that $\lmz$ and $l^j$ are vertical rays to the south starting at $\pos{P_{m_0}}$ and $\pos{\glu P j}$, \resp.
  If $m_0>j$, the vertical ray to the south $\lmz$ is strictly to the east of the ray $l^j$, and $\lmz+\vpji\cap\embed{P_{i+1,i+2,\ldots,k}}\subseteq\{\lmz(0)+\vpji\}$.
\end{sublemma}
\begin{proof}
  The proof has four cases, depending on the glues around $P_{m_0}$.
  We show that cases 1 and 2 below cannot occur, and that cases 3 and 4 yield the conclusion of this claim:
  \begin{enumerate}
  \item If $\glue P {m_0}{m_0+1}$ is pointing south,
    the ray of vector $\vpij$ through $\pos{P_{m_0 +1}}$ is strictly lower than $\rho$, contradicting the definition of $\rho$.
  \item If $\glue P {m_0} {m_0+1}$ is pointing west, then we claim that $\glue P {m_0} {m_0+1}$ would be visible from the south relative to $P$. Indeed, by \subl{lem:lmzInC}, $L^{m_0}\in\mathcal C$, and since $L^{m_0}$ is on a tile column, $L^{m_0}-(0.5,0)$ is on a glue column, and at x-coordinate at least $\xcoord{\pos{\glu P i}}$.
  That, together with the fact that $\lmz$ intersects $P_{\range{i+1}{i+2}{k}}$ only at $\pos{P_{m_0}}$, implies that no glue of $P$ can be positioned on the line defined as  $L^{m_0}(z) - (0.5,0)$ for $z>0,z\in \mathbb{R}$ (i.e. on line $L^{m_0}-(0.5,0)$ and strictly below $\glu P {m_0}$).
  Thus $\glue P {m_0} {m_0+1}$ is visible from the south.
    But, since $m_0>j$ and $\glu P j$ points east, $\glue P {m_0} {m_0+1}$ pointing west and being visible from the south contradicts Lemma~\ref{lem:glue:east}.

  \item\label{case:lmzInter} If $\glue P {m_0}{m_0+1}$ is pointing east, then $\glue P {m_0}{m_0+1}$ is visible from the south relative to $P$ 
    (by the same argument as in the previous case, with $L^{m_0}+(0.5,0)$ instead of $L^{m_0}-(0.5,0)$).
  Therefore by \sublName~\ref{lem:m0i1} and Lemma~\ref{lem:glue:east}, $\lmz$ is strictly to the east of $l^i$, and if $m_0>j$, then $\lmz$ is strictly to the east of $l^j$ too.

  Furthermore, if $m_0>j$ and if $\lmz + \vpji$ intersected $\embed{P_{i+1,i+2,\ldots,k}}$ other than at $\lmz(0) + \vpji$, we could find a ray of vector $\vpij$ through a tile of $P_{i+1,i+2,\ldots,k}$ strictly lower (strictly to the south) than $\rho$, which is a contradiction.
  Hence, if $m_0>j$, then $\lmz+\vpji$ does not intersect $\embed{P_{\range {i+1}{i+2}k}}$ except possibly at $\lmz(0)+\vpji$.
\item If $\glue P {m_0}{m_0+1}$ is pointing north, we consider all possibilities for $\glue P {m_0-1}{m_0}$:
    \begin{enumerate}
    \item If $\glue P {m_0-1}{m_0}$ points north, this contradicts the definition of $\rho$, since the ray of vector $\vpij$ going through $P_{m_0-1}$ is strictly lower than $\rho$.
    \item If $\glue P {m_0-1}{m_0}$ points west, then $\glue P {m_0-1}{m_0}$ is visible relative to $P$ (by \sublName~\ref{lem:cmz}), and this contradicts Lemma~\ref{lem:glue:east} since $m_0>i$.
    \item $\glue P {m_0-1}{m_0}$ pointing south  contradicts $P$ being simple: indeed, this would mean that $\pos{P_{m_0-1}} = \pos{P_{m_0+1}}$.
    \item Therefore, $\glue P {m_0-1}{m_0}$ points east, and is therefore visible from the south relative to $P_{i+1,i+2,\ldots,k}$ (by \sublName~\ref{lem:cmz}). By Lemma~\ref{lem:glue:east}, since $m_0>i$, the visibility ray of $\glue P {m_0-1}{m_0}$ is strictly to the east of $l^i$, and hence $\lmz$ is strictly to the east of $l^i$. If $m_0>j$, then $\lmz$ is also strictly to the east of $l^j$, hence $\lmz+\vpji$ does not intersect $P_{i+1,i+2,\ldots,k}$, except possibly at $\lmz(0) + \vpji$, by the same argument as in Case~\ref{case:lmzInter} above.
    \end{enumerate}
  \end{enumerate}
\end{proof}

\subsubsection{Splitting $\mathcal C$ into two components, $\cm$ and $\cp$, using $P_{m_0}$}\label{shield:split}
By \subl{lem:lmzInC}, $\lmz$ is entirely in $\mathcal{C}$. Moreover, by \subl{lem:cmz}, $\lmz$ intersects $c$ only at $\pos{P_{m_0}}$. Therefore, we use $\lmz$ to split $\mathcal{C}$ into two components:
let $c^{m_0}$ be the curve defined as 
$$c^{m_0} = \concat{\reverse\lmz, \embed{P_{\rng{m_0}k}},\gs{\pos{P_k}}{l^k(0)},l^k}$$ 
Since $c^{m_0}$ starts and ends with vertical rays, and is otherwise made of a finite concatenation of horizontal and vertical segments, $c^{m_0}$ splits $\R^2$ into two connected components by Theorem~\ref{thm:infinite-jordan}. Moreover, since $c^{m_0}$ is in $\mathcal C$ (because the four curves that define $c^{m_0}$ are in $\mathcal C$), let $\cp\subseteq\mathcal C$ be the right-hand side of $c^{m_0}$ (including $c^{m_0}$), and let $\cm=(\mathcal C\setminus\cp)\cup \lmz(\R)$\footnote{Recall that $\lmz(\R)$ is the set of all points in the range of of $\lmz$}. See Figure~\ref{fig:shield-cpm} for an illustration.

Moreover, $P_{i+1,i+2,\ldots,m_0}$ is entirely in $\cm$ (and on the border of $\cm$, by definition of $\mathcal C$) and $P_{m_0,m_0+1,\ldots,k}$ is entirely in $\cp$ (and on the border of $\cp$, by definition of $\mathcal C$).
\begin{figure}[ht]
  \centering
  \begin{tikzpicture}[scale=\scale]\draw[draw={rgb,255:red,200; green,200; blue,200}](0.5,-0.2) rectangle (26.5, -15.2);
\draw[draw={rgb,255:red,200; green,200; blue,200}](0.5,-15.2)--(0.5,-0.2);
\draw[draw={rgb,255:red,200; green,200; blue,200}](1.5,-15.2)--(1.5,-0.2);
\draw[draw={rgb,255:red,200; green,200; blue,200}](2.5,-15.2)--(2.5,-0.2);
\draw[draw={rgb,255:red,200; green,200; blue,200}](3.5,-15.2)--(3.5,-0.2);
\draw[draw={rgb,255:red,200; green,200; blue,200}](4.5,-15.2)--(4.5,-0.2);
\draw[draw={rgb,255:red,200; green,200; blue,200}](5.5,-15.2)--(5.5,-0.2);
\draw[draw={rgb,255:red,200; green,200; blue,200}](6.5,-15.2)--(6.5,-0.2);
\draw[draw={rgb,255:red,200; green,200; blue,200}](7.5,-15.2)--(7.5,-0.2);
\draw[draw={rgb,255:red,200; green,200; blue,200}](8.5,-15.2)--(8.5,-0.2);
\draw[draw={rgb,255:red,200; green,200; blue,200}](9.5,-15.2)--(9.5,-0.2);
\draw[draw={rgb,255:red,200; green,200; blue,200}](10.5,-15.2)--(10.5,-0.2);
\draw[draw={rgb,255:red,200; green,200; blue,200}](11.5,-15.2)--(11.5,-0.2);
\draw[draw={rgb,255:red,200; green,200; blue,200}](12.5,-15.2)--(12.5,-0.2);
\draw[draw={rgb,255:red,200; green,200; blue,200}](13.5,-15.2)--(13.5,-0.2);
\draw[draw={rgb,255:red,200; green,200; blue,200}](14.5,-15.2)--(14.5,-0.2);
\draw[draw={rgb,255:red,200; green,200; blue,200}](15.5,-15.2)--(15.5,-0.2);
\draw[draw={rgb,255:red,200; green,200; blue,200}](16.5,-15.2)--(16.5,-0.2);
\draw[draw={rgb,255:red,200; green,200; blue,200}](17.5,-15.2)--(17.5,-0.2);
\draw[draw={rgb,255:red,200; green,200; blue,200}](18.5,-15.2)--(18.5,-0.2);
\draw[draw={rgb,255:red,200; green,200; blue,200}](19.5,-15.2)--(19.5,-0.2);
\draw[draw={rgb,255:red,200; green,200; blue,200}](20.5,-15.2)--(20.5,-0.2);
\draw[draw={rgb,255:red,200; green,200; blue,200}](21.5,-15.2)--(21.5,-0.2);
\draw[draw={rgb,255:red,200; green,200; blue,200}](22.5,-15.2)--(22.5,-0.2);
\draw[draw={rgb,255:red,200; green,200; blue,200}](23.5,-15.2)--(23.5,-0.2);
\draw[draw={rgb,255:red,200; green,200; blue,200}](24.5,-15.2)--(24.5,-0.2);
\draw[draw={rgb,255:red,200; green,200; blue,200}](25.5,-15.2)--(25.5,-0.2);
\draw[draw={rgb,255:red,200; green,200; blue,200}](0.5,-0.2)--(26.5,-0.2);
\draw[draw={rgb,255:red,200; green,200; blue,200}](0.5,-1.2)--(26.5,-1.2);
\draw[draw={rgb,255:red,200; green,200; blue,200}](0.5,-2.2)--(26.5,-2.2);
\draw[draw={rgb,255:red,200; green,200; blue,200}](0.5,-3.2)--(26.5,-3.2);
\draw[draw={rgb,255:red,200; green,200; blue,200}](0.5,-4.2)--(26.5,-4.2);
\draw[draw={rgb,255:red,200; green,200; blue,200}](0.5,-5.2)--(26.5,-5.2);
\draw[draw={rgb,255:red,200; green,200; blue,200}](0.5,-6.2)--(26.5,-6.2);
\draw[draw={rgb,255:red,200; green,200; blue,200}](0.5,-7.2)--(26.5,-7.2);
\draw[draw={rgb,255:red,200; green,200; blue,200}](0.5,-8.2)--(26.5,-8.2);
\draw[draw={rgb,255:red,200; green,200; blue,200}](0.5,-9.2)--(26.5,-9.2);
\draw[draw={rgb,255:red,200; green,200; blue,200}](0.5,-10.2)--(26.5,-10.2);
\draw[draw={rgb,255:red,200; green,200; blue,200}](0.5,-11.2)--(26.5,-11.2);
\draw[draw={rgb,255:red,200; green,200; blue,200}](0.5,-12.2)--(26.5,-12.2);
\draw[draw={rgb,255:red,200; green,200; blue,200}](0.5,-13.2)--(26.5,-13.2);
\draw[draw={rgb,255:red,200; green,200; blue,200}](0.5,-14.2)--(26.5,-14.2);
\draw[draw=none,fill={rgb,255:red,255; green,0; blue,0},opacity=0.1](5.45,-13.7)--(7.95,-13.7)--(7.95,-12.7)--(9.95,-12.7)--(9.95,-9.7)--(16.95,-9.7)--(16.95,-12.7)--(18.95,-12.7)--(18.95,-15.2)--(5.45,-15.2)--(5.45,-13.7);
\draw(11.81, -12.96) node[anchor=south west] {$\cm$};
\draw[draw=none,fill={rgb,255:red,0; green,102; blue,255},opacity=0.1](13.45,-5.7)--(13.5,-0.2)--(26.5,-0.2)--(26.45,-15.2)--(18.95,-15.2)--(18.95,-9.7)--(19.95,-9.7)--(19.95,-7.7)--(22.95,-7.7)--(22.95,-10.7)--(24.95,-10.7)--(24.95,-5.7)--(19.95,-5.7)--(19.95,-4.7)--(17,-4.7)--(17,-3.7)--(19,-3.7)--(19,-1.7)--(16,-1.7)--(15.95,-3.7)--(15.95,-5.7)--(13.45,-5.7);
\draw(22.87, -4.18) node[anchor=south west] {$\cp$};
\draw[draw={rgb,255:red,0; green,0; blue,0},fill={rgb,255:red,200; green,113; blue,55},opacity=0.29899997,fill opacity=0.29899997](1.65,-8.35) rectangle (2.35, -9.05);
\draw[draw={rgb,255:red,0; green,0; blue,0},fill={rgb,255:red,200; green,113; blue,55},opacity=0.29899997,fill opacity=0.29899997](2.65,-8.35) rectangle (3.35, -9.05);
\draw[draw={rgb,255:red,0; green,0; blue,0},fill={rgb,255:red,200; green,113; blue,55},opacity=0.29899997,fill opacity=0.29899997](3.65,-8.35) rectangle (4.35, -9.05);
\draw[draw={rgb,255:red,0; green,0; blue,0},fill={rgb,255:red,200; green,113; blue,55},opacity=0.29899997,fill opacity=0.29899997](4.65,-8.35) rectangle (5.35, -9.05);
\draw[draw={rgb,255:red,0; green,0; blue,0},fill={rgb,255:red,200; green,113; blue,55},opacity=0.29899997,fill opacity=0.29899997](5.65,-8.35) rectangle (6.35, -9.05);
\draw[draw={rgb,255:red,0; green,0; blue,0},fill={rgb,255:red,200; green,113; blue,55},opacity=0.29899997,fill opacity=0.29899997](6.65,-8.35) rectangle (7.35, -9.05);
\draw[draw={rgb,255:red,0; green,0; blue,0},fill={rgb,255:red,200; green,113; blue,55},opacity=0.29899997,fill opacity=0.29899997](7.65,-8.35) rectangle (8.35, -9.05);
\draw[draw={rgb,255:red,0; green,0; blue,0},fill={rgb,255:red,200; green,113; blue,55},opacity=0.29899997,fill opacity=0.29899997](7.65,-9.35) rectangle (8.35, -10.05);
\draw[draw={rgb,255:red,0; green,0; blue,0},fill={rgb,255:red,200; green,113; blue,55},opacity=0.29899997,fill opacity=0.29899997](7.65,-10.35) rectangle (8.35, -11.05);
\draw[draw={rgb,255:red,0; green,0; blue,0},fill={rgb,255:red,200; green,113; blue,55},opacity=0.29899997,fill opacity=0.29899997](7.65,-11.35) rectangle (8.35, -12.05);
\draw[draw={rgb,255:red,0; green,0; blue,0},fill={rgb,255:red,200; green,113; blue,55},opacity=0.29899997,fill opacity=0.29899997](6.65,-11.35) rectangle (7.35, -12.05);
\draw[draw={rgb,255:red,0; green,0; blue,0},fill={rgb,255:red,200; green,113; blue,55},opacity=0.29899997,fill opacity=0.29899997](5.65,-11.35) rectangle (6.35, -12.05);
\draw[draw={rgb,255:red,0; green,0; blue,0},fill={rgb,255:red,200; green,113; blue,55},opacity=0.29899997,fill opacity=0.29899997](4.65,-11.35) rectangle (5.35, -12.05);
\draw[draw={rgb,255:red,0; green,0; blue,0},fill={rgb,255:red,200; green,113; blue,55},opacity=0.29899997,fill opacity=0.29899997](3.65,-11.35) rectangle (4.35, -12.05);
\draw[draw={rgb,255:red,0; green,0; blue,0},fill={rgb,255:red,200; green,113; blue,55},opacity=0.29899997,fill opacity=0.29899997](3.65,-12.35) rectangle (4.35, -13.05);
\draw[draw={rgb,255:red,0; green,0; blue,0},fill={rgb,255:red,200; green,113; blue,55},opacity=0.29899997,fill opacity=0.29899997](3.65,-13.35) rectangle (4.35, -14.05);
\draw[draw={rgb,255:red,0; green,0; blue,0},fill={rgb,255:red,200; green,113; blue,55},opacity=0.29899997,fill opacity=0.29899997](4.65,-13.35) rectangle (5.35, -14.05);
\draw[draw={rgb,255:red,200; green,113; blue,55},opacity=0.29899997,thick](2,-8.7)--(8,-8.7)--(8,-11.7)--(4,-11.7)--(4,-13.7)--(5.5,-13.7);
\draw[draw={rgb,255:red,0; green,0; blue,0},fill={rgb,255:red,200; green,113; blue,55},opacity=0.5,fill opacity=0.5](5.65,-13.35) rectangle (6.35, -14.05);
\draw[draw={rgb,255:red,0; green,0; blue,0},fill={rgb,255:red,200; green,113; blue,55},opacity=0.5,fill opacity=0.5](6.65,-13.35) rectangle (7.35, -14.05);
\draw[draw={rgb,255:red,0; green,0; blue,0},fill={rgb,255:red,200; green,113; blue,55},opacity=0.5,fill opacity=0.5](7.65,-13.35) rectangle (8.35, -14.05);
\draw[draw={rgb,255:red,0; green,0; blue,0},fill={rgb,255:red,200; green,113; blue,55},opacity=0.5,fill opacity=0.5](7.65,-12.35) rectangle (8.35, -13.05);
\draw[draw={rgb,255:red,0; green,0; blue,0},fill={rgb,255:red,200; green,113; blue,55},opacity=0.5,fill opacity=0.5](8.65,-12.35) rectangle (9.35, -13.05);
\draw[draw={rgb,255:red,0; green,0; blue,0},fill={rgb,255:red,200; green,113; blue,55},opacity=0.5,fill opacity=0.5](9.65,-12.35) rectangle (10.35, -13.05);
\draw[draw={rgb,255:red,0; green,0; blue,0},fill={rgb,255:red,200; green,113; blue,55},opacity=0.5,fill opacity=0.5](9.65,-11.35) rectangle (10.35, -12.05);
\draw[draw={rgb,255:red,0; green,0; blue,0},fill={rgb,255:red,200; green,113; blue,55},opacity=0.5,fill opacity=0.5](9.65,-10.35) rectangle (10.35, -11.05);
\draw[draw={rgb,255:red,0; green,0; blue,0},fill={rgb,255:red,200; green,113; blue,55},opacity=0.5,fill opacity=0.5](9.65,-9.35) rectangle (10.35, -10.05);
\draw[draw={rgb,255:red,0; green,0; blue,0},fill={rgb,255:red,200; green,113; blue,55},opacity=0.5,fill opacity=0.5](10.65,-9.35) rectangle (11.35, -10.05);
\draw[draw={rgb,255:red,0; green,0; blue,0},fill={rgb,255:red,200; green,113; blue,55},opacity=0.5,fill opacity=0.5](11.65,-9.35) rectangle (12.35, -10.05);
\draw[draw={rgb,255:red,0; green,0; blue,0},fill={rgb,255:red,200; green,113; blue,55},opacity=0.5,fill opacity=0.5](12.65,-9.35) rectangle (13.35, -10.05);
\draw[draw={rgb,255:red,0; green,0; blue,0},fill={rgb,255:red,200; green,113; blue,55},opacity=0.5,fill opacity=0.5](13.65,-9.35) rectangle (14.35, -10.05);
\draw[draw={rgb,255:red,0; green,0; blue,0},fill={rgb,255:red,200; green,113; blue,55},opacity=0.5,fill opacity=0.5](14.65,-9.35) rectangle (15.35, -10.05);
\draw[draw={rgb,255:red,0; green,0; blue,0},fill={rgb,255:red,200; green,113; blue,55},opacity=0.5,fill opacity=0.5](15.65,-9.35) rectangle (16.35, -10.05);
\draw[draw={rgb,255:red,0; green,0; blue,0},fill={rgb,255:red,200; green,113; blue,55},opacity=0.5,fill opacity=0.5](16.65,-9.35) rectangle (17.35, -10.05);
\draw[draw={rgb,255:red,0; green,0; blue,0},fill={rgb,255:red,200; green,113; blue,55},opacity=0.5,fill opacity=0.5](16.65,-10.35) rectangle (17.35, -11.05);
\draw[draw={rgb,255:red,0; green,0; blue,0},fill={rgb,255:red,200; green,113; blue,55},opacity=0.5,fill opacity=0.5](16.65,-11.35) rectangle (17.35, -12.05);
\draw[draw={rgb,255:red,0; green,0; blue,0},fill={rgb,255:red,200; green,113; blue,55},opacity=0.5,fill opacity=0.5](16.65,-12.35) rectangle (17.35, -13.05);
\draw[draw={rgb,255:red,0; green,0; blue,0},fill={rgb,255:red,200; green,113; blue,55},opacity=0.5,fill opacity=0.5](17.65,-12.35) rectangle (18.35, -13.05);
\draw[draw={rgb,255:red,0; green,0; blue,0},fill={rgb,255:red,200; green,113; blue,55},opacity=0.5,fill opacity=0.5](18.65,-12.35) rectangle (19.35, -13.05);
\draw[draw={rgb,255:red,0; green,0; blue,0},fill={rgb,255:red,200; green,113; blue,55},opacity=0.5,fill opacity=0.5](18.65,-11.35) rectangle (19.35, -12.05);
\draw[draw={rgb,255:red,0; green,0; blue,0},fill={rgb,255:red,200; green,113; blue,55},opacity=0.5,fill opacity=0.5](18.65,-10.35) rectangle (19.35, -11.05);
\draw[draw={rgb,255:red,0; green,0; blue,0},fill={rgb,255:red,200; green,113; blue,55},opacity=0.5,fill opacity=0.5](18.65,-9.35) rectangle (19.35, -10.05);
\draw[draw={rgb,255:red,0; green,0; blue,0},fill={rgb,255:red,200; green,113; blue,55},opacity=0.5,fill opacity=0.5](19.65,-9.35) rectangle (20.35, -10.05);
\draw[draw={rgb,255:red,0; green,0; blue,0},fill={rgb,255:red,200; green,113; blue,55},opacity=0.5,fill opacity=0.5](19.65,-8.35) rectangle (20.35, -9.05);
\draw[draw={rgb,255:red,0; green,0; blue,0},fill={rgb,255:red,200; green,113; blue,55},opacity=0.5,fill opacity=0.5](19.65,-7.35) rectangle (20.35, -8.05);
\draw[draw={rgb,255:red,0; green,0; blue,0},fill={rgb,255:red,200; green,113; blue,55},opacity=0.5,fill opacity=0.5](20.65,-7.35) rectangle (21.35, -8.05);
\draw[draw={rgb,255:red,0; green,0; blue,0},fill={rgb,255:red,200; green,113; blue,55},opacity=0.5,fill opacity=0.5](21.65,-7.35) rectangle (22.35, -8.05);
\draw[draw={rgb,255:red,0; green,0; blue,0},fill={rgb,255:red,200; green,113; blue,55},opacity=0.5,fill opacity=0.5](22.65,-7.35) rectangle (23.35, -8.05);
\draw[draw={rgb,255:red,0; green,0; blue,0},fill={rgb,255:red,200; green,113; blue,55},opacity=0.5,fill opacity=0.5](22.65,-8.35) rectangle (23.35, -9.05);
\draw[draw={rgb,255:red,0; green,0; blue,0},fill={rgb,255:red,200; green,113; blue,55},opacity=0.5,fill opacity=0.5](22.65,-9.35) rectangle (23.35, -10.05);
\draw[draw={rgb,255:red,0; green,0; blue,0},fill={rgb,255:red,200; green,113; blue,55},opacity=0.5,fill opacity=0.5](22.65,-10.35) rectangle (23.35, -11.05);
\draw[draw={rgb,255:red,0; green,0; blue,0},fill={rgb,255:red,200; green,113; blue,55},opacity=0.5,fill opacity=0.5](23.65,-10.35) rectangle (24.35, -11.05);
\draw[draw={rgb,255:red,0; green,0; blue,0},fill={rgb,255:red,200; green,113; blue,55},opacity=0.5,fill opacity=0.5](24.65,-10.35) rectangle (25.35, -11.05);
\draw[draw={rgb,255:red,0; green,0; blue,0},fill={rgb,255:red,200; green,113; blue,55},opacity=0.5,fill opacity=0.5](24.65,-9.35) rectangle (25.35, -10.05);
\draw[draw={rgb,255:red,0; green,0; blue,0},fill={rgb,255:red,200; green,113; blue,55},opacity=0.5,fill opacity=0.5](24.65,-8.35) rectangle (25.35, -9.05);
\draw[draw={rgb,255:red,0; green,0; blue,0},fill={rgb,255:red,200; green,113; blue,55},opacity=0.5,fill opacity=0.5](24.65,-7.35) rectangle (25.35, -8.05);
\draw[draw={rgb,255:red,0; green,0; blue,0},fill={rgb,255:red,200; green,113; blue,55},opacity=0.5,fill opacity=0.5](24.65,-6.35) rectangle (25.35, -7.05);
\draw[draw={rgb,255:red,0; green,0; blue,0},fill={rgb,255:red,200; green,113; blue,55},opacity=0.5,fill opacity=0.5](24.65,-5.35) rectangle (25.35, -6.05);
\draw[draw={rgb,255:red,0; green,0; blue,0},fill={rgb,255:red,200; green,113; blue,55},opacity=0.5,fill opacity=0.5](23.65,-5.35) rectangle (24.35, -6.05);
\draw[draw={rgb,255:red,0; green,0; blue,0},fill={rgb,255:red,200; green,113; blue,55},opacity=0.5,fill opacity=0.5](22.65,-5.35) rectangle (23.35, -6.05);
\draw[draw={rgb,255:red,0; green,0; blue,0},fill={rgb,255:red,200; green,113; blue,55},opacity=0.5,fill opacity=0.5](21.65,-5.35) rectangle (22.35, -6.05);
\draw[draw={rgb,255:red,0; green,0; blue,0},fill={rgb,255:red,200; green,113; blue,55},opacity=0.5,fill opacity=0.5](20.65,-5.35) rectangle (21.35, -6.05);
\draw[draw={rgb,255:red,0; green,0; blue,0},fill={rgb,255:red,200; green,113; blue,55},opacity=0.5,fill opacity=0.5](19.65,-5.35) rectangle (20.35, -6.05);
\draw[draw={rgb,255:red,0; green,0; blue,0},fill={rgb,255:red,200; green,113; blue,55},opacity=0.5,fill opacity=0.5](19.65,-4.35) rectangle (20.35, -5.05);
\draw[draw={rgb,255:red,0; green,0; blue,0},fill={rgb,255:red,200; green,113; blue,55},opacity=0.5,fill opacity=0.5](18.65,-4.35) rectangle (19.35, -5.05);
\draw[draw={rgb,255:red,0; green,0; blue,0},fill={rgb,255:red,200; green,113; blue,55},opacity=0.5,fill opacity=0.5](17.65,-4.35) rectangle (18.35, -5.05);
\draw[draw={rgb,255:red,0; green,0; blue,0},fill={rgb,255:red,200; green,113; blue,55},opacity=0.5,fill opacity=0.5](16.65,-4.35) rectangle (17.35, -5.05);
\draw[draw={rgb,255:red,0; green,0; blue,0},fill={rgb,255:red,200; green,113; blue,55},opacity=0.5,fill opacity=0.5](16.65,-3.35) rectangle (17.35, -4.05);
\draw[draw={rgb,255:red,0; green,0; blue,0},fill={rgb,255:red,200; green,113; blue,55},opacity=0.5,fill opacity=0.5](17.65,-3.35) rectangle (18.35, -4.05);
\draw[draw={rgb,255:red,0; green,0; blue,0},fill={rgb,255:red,200; green,113; blue,55},opacity=0.5,fill opacity=0.5](18.65,-3.35) rectangle (19.35, -4.05);
\draw[draw={rgb,255:red,0; green,0; blue,0},fill={rgb,255:red,200; green,113; blue,55},opacity=0.5,fill opacity=0.5](18.65,-2.35) rectangle (19.35, -3.05);
\draw[draw={rgb,255:red,0; green,0; blue,0},fill={rgb,255:red,200; green,113; blue,55},opacity=0.5,fill opacity=0.5](18.65,-1.35) rectangle (19.35, -2.05);
\draw[draw={rgb,255:red,0; green,0; blue,0},fill={rgb,255:red,200; green,113; blue,55},opacity=0.5,fill opacity=0.5](17.65,-1.35) rectangle (18.35, -2.05);
\draw[draw={rgb,255:red,0; green,0; blue,0},fill={rgb,255:red,200; green,113; blue,55},opacity=0.5,fill opacity=0.5](16.65,-1.35) rectangle (17.35, -2.05);
\draw[draw={rgb,255:red,0; green,0; blue,0},fill={rgb,255:red,200; green,113; blue,55},opacity=0.5,fill opacity=0.5](15.65,-1.35) rectangle (16.35, -2.05);
\draw[draw={rgb,255:red,0; green,0; blue,0},fill={rgb,255:red,200; green,113; blue,55},opacity=0.5,fill opacity=0.5](15.65,-2.35) rectangle (16.35, -3.05);
\draw[draw={rgb,255:red,0; green,0; blue,0},fill={rgb,255:red,200; green,113; blue,55},opacity=0.5,fill opacity=0.5](15.65,-3.35) rectangle (16.35, -4.05);
\draw[draw={rgb,255:red,0; green,0; blue,0},fill={rgb,255:red,200; green,113; blue,55},opacity=0.5,fill opacity=0.5](15.65,-4.35) rectangle (16.35, -5.05);
\draw[draw={rgb,255:red,0; green,0; blue,0},fill={rgb,255:red,200; green,113; blue,55},opacity=0.5,fill opacity=0.5](15.65,-5.35) rectangle (16.35, -6.05);
\draw[draw={rgb,255:red,0; green,0; blue,0},fill={rgb,255:red,200; green,113; blue,55},opacity=0.5,fill opacity=0.5](14.65,-5.35) rectangle (15.35, -6.05);
\draw[draw={rgb,255:red,0; green,0; blue,0},fill={rgb,255:red,200; green,113; blue,55},opacity=0.5,fill opacity=0.5](13.65,-5.35) rectangle (14.35, -6.05);
\draw[draw={rgb,255:red,0; green,0; blue,0},fill={rgb,255:red,200; green,113; blue,55},opacity=0.5,fill opacity=0.5](12.65,-5.35) rectangle (13.35, -6.05);
\draw[draw={rgb,255:red,200; green,113; blue,55},opacity=0.5,thick](6,-13.7)--(8,-13.7)--(8,-12.7)--(10,-12.7)--(10,-9.7)--(17,-9.7)--(17,-12.7)--(19,-12.7)--(19,-9.7)--(20,-9.7)--(20,-7.7)--(23,-7.7)--(23,-10.7)--(25,-10.7)--(25,-5.7)--(20,-5.7)--(20,-4.7)--(17,-4.7)--(17,-3.7)--(19,-3.7)--(19,-1.7)--(16,-1.7)--(16,-5.7)--(13,-5.7);
\draw[draw={rgb,255:red,0; green,0; blue,0}](9.5,-12.69)--(9.5,-16.7);
\draw(8.23, -16.93) node[anchor=south west] {$l_j$};
\draw[draw=none,fill={rgb,255:red,28; green,36; blue,31},thin](9.5, -12.69) ellipse (0.1cm and 0.1cm);\draw[draw={rgb,255:red,0; green,0; blue,0}](5.5,-13.7)--(5.5,-16.7);
\draw(4.32, -16.93) node[anchor=south west] {$l_i$};
\draw[draw={rgb,255:red,0; green,0; blue,0}](13.5,-5.7)--(13.5,0.3);
\draw(12.34, -0.81) node[anchor=south west] {$l^k$};
\draw[draw=none,fill={rgb,255:red,28; green,36; blue,31},thin](13.5, -5.7) ellipse (0.1cm and 0.1cm);\draw[draw=none,fill={rgb,255:red,28; green,36; blue,31},thin](5.5, -13.7) ellipse (0.1cm and 0.1cm);\draw[draw={rgb,255:red,0; green,0; blue,0},dashed,thin](5.5,-16.08)--(31,-9.7);
\draw(30.89, -10.28) node[anchor=south west] {$\rho$};
\draw(19.29, -13.68) node[anchor=south west] {$P_{m_0}$};
\draw[draw={rgb,255:red,0; green,0; blue,0}](19,-12.7)--(19,-16.7);
\draw(17.76, -17.73) node[anchor=south west] {$\lmz$};
\draw[draw=none,fill={rgb,255:red,28; green,36; blue,31},thin](19, -12.7) ellipse (0.1cm and 0.1cm);
\end{tikzpicture}
  \caption{The components $\cp$ and $\cm$.  In Subsection~\ref{subsec:first},  $\cp$ and $\cm$ were defined so that $\cp \cup \cm = \mathcal{C}$ and  $\cp \cap \cm = L^{m_0}$.}\label{fig:shield-cpm}
\end{figure}

\enlargethispage{-2\baselineskip}
\begin{sublemma}
  \label{lem:lmzInCp}
  For any integer $n > 0$, $\lmz+n\vpij$ is in $\cp$.
\end{sublemma}
\begin{proof}
  Let integer $n>0$.
  By \subl{lem:cmz}, $\lmz+n\vpij$ does not intersect $\embed{P_{i+1,i+2,\ldots,k}}$.
  Moreover, since $\xcoord\vpij>0$ we know that $\lmz+n\vpij$ is strictly to the east of $\lmz$, and $\lmz+n\vpij$ does not intersect $l^k$ (because $\lmz$ is on a tile column, and $l^k$ on a glue column) nor $\gs{P_k}{l^k(0)}$ (because that would mean that $\lmz+n\vpij$ intersects $P_k$).

  Therefore, $\lmz+n\vpij$ does not intersect the border of $\cp$, and starts to the east of $\lmz$ (which is on the border of $\cp$). Therefore, $\lmz+n\vpij$ is entirely in $\cp$.
\end{proof}

\subsubsection{A path $R$ that can grow in two different translations}
\label{subsec:r}
The goal of this subsection is to define a path $R$ that can grow in two possible translations (respectively, $R$ that starts at $\pos{P_{i+1}}$ and $R+\vpij$ that starts at $\pos{P_{j+1}}$). We proceed in two steps: the first step is to specify a \emph{binding path} called $r$ below (from Definition~\ref{def:binding graph}, a binding path is a simple sequence of adjacent positions in $\mathbb{Z}^2$, i.e. a path but without tiles), and the second step is to ``tile'' that binding path to get the path $R$.

\paragraph{The binding path $r$.}

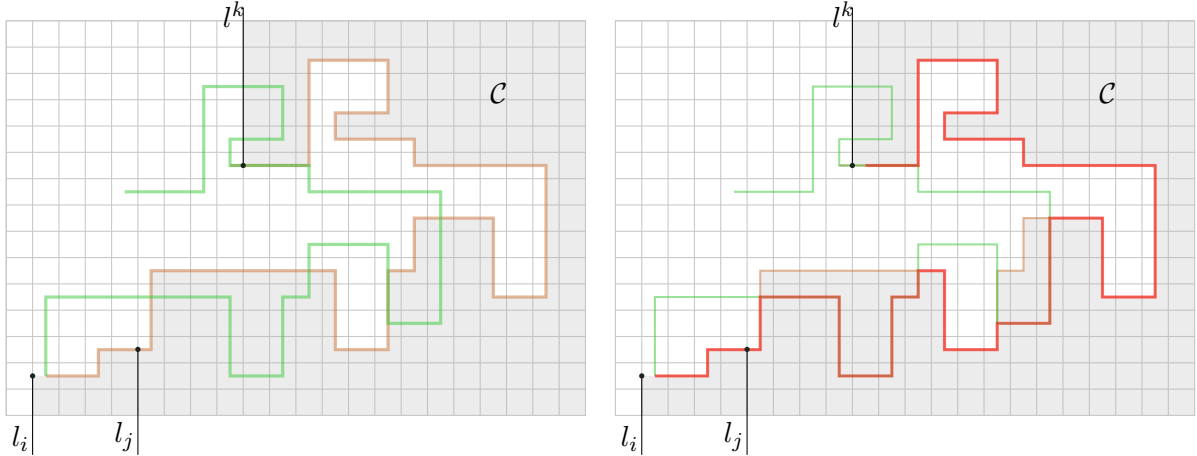
\begin{figure}[h]
  \centering
  \renewcommand\scale{0.35}
  \begin{tikzpicture}[scale=\scale]\draw[draw={rgb,255:red,200; green,200; blue,200}](4.5,-0.2) rectangle (26.5, -15.2);
\draw[draw={rgb,255:red,200; green,200; blue,200}](4.5,-15.2)--(4.5,-0.2);
\draw[draw={rgb,255:red,200; green,200; blue,200}](5.5,-15.2)--(5.5,-0.2);
\draw[draw={rgb,255:red,200; green,200; blue,200}](6.5,-15.2)--(6.5,-0.2);
\draw[draw={rgb,255:red,200; green,200; blue,200}](7.5,-15.2)--(7.5,-0.2);
\draw[draw={rgb,255:red,200; green,200; blue,200}](8.5,-15.2)--(8.5,-0.2);
\draw[draw={rgb,255:red,200; green,200; blue,200}](9.5,-15.2)--(9.5,-0.2);
\draw[draw={rgb,255:red,200; green,200; blue,200}](10.5,-15.2)--(10.5,-0.2);
\draw[draw={rgb,255:red,200; green,200; blue,200}](11.5,-15.2)--(11.5,-0.2);
\draw[draw={rgb,255:red,200; green,200; blue,200}](12.5,-15.2)--(12.5,-0.2);
\draw[draw={rgb,255:red,200; green,200; blue,200}](13.5,-15.2)--(13.5,-0.2);
\draw[draw={rgb,255:red,200; green,200; blue,200}](14.5,-15.2)--(14.5,-0.2);
\draw[draw={rgb,255:red,200; green,200; blue,200}](15.5,-15.2)--(15.5,-0.2);
\draw[draw={rgb,255:red,200; green,200; blue,200}](16.5,-15.2)--(16.5,-0.2);
\draw[draw={rgb,255:red,200; green,200; blue,200}](17.5,-15.2)--(17.5,-0.2);
\draw[draw={rgb,255:red,200; green,200; blue,200}](18.5,-15.2)--(18.5,-0.2);
\draw[draw={rgb,255:red,200; green,200; blue,200}](19.5,-15.2)--(19.5,-0.2);
\draw[draw={rgb,255:red,200; green,200; blue,200}](20.5,-15.2)--(20.5,-0.2);
\draw[draw={rgb,255:red,200; green,200; blue,200}](21.5,-15.2)--(21.5,-0.2);
\draw[draw={rgb,255:red,200; green,200; blue,200}](22.5,-15.2)--(22.5,-0.2);
\draw[draw={rgb,255:red,200; green,200; blue,200}](23.5,-15.2)--(23.5,-0.2);
\draw[draw={rgb,255:red,200; green,200; blue,200}](24.5,-15.2)--(24.5,-0.2);
\draw[draw={rgb,255:red,200; green,200; blue,200}](25.5,-15.2)--(25.5,-0.2);
\draw[draw={rgb,255:red,200; green,200; blue,200}](4.5,-0.2)--(26.5,-0.2);
\draw[draw={rgb,255:red,200; green,200; blue,200}](4.5,-1.2)--(26.5,-1.2);
\draw[draw={rgb,255:red,200; green,200; blue,200}](4.5,-2.2)--(26.5,-2.2);
\draw[draw={rgb,255:red,200; green,200; blue,200}](4.5,-3.2)--(26.5,-3.2);
\draw[draw={rgb,255:red,200; green,200; blue,200}](4.5,-4.2)--(26.5,-4.2);
\draw[draw={rgb,255:red,200; green,200; blue,200}](4.5,-5.2)--(26.5,-5.2);
\draw[draw={rgb,255:red,200; green,200; blue,200}](4.5,-6.2)--(26.5,-6.2);
\draw[draw={rgb,255:red,200; green,200; blue,200}](4.5,-7.2)--(26.5,-7.2);
\draw[draw={rgb,255:red,200; green,200; blue,200}](4.5,-8.2)--(26.5,-8.2);
\draw[draw={rgb,255:red,200; green,200; blue,200}](4.5,-9.2)--(26.5,-9.2);
\draw[draw={rgb,255:red,200; green,200; blue,200}](4.5,-10.2)--(26.5,-10.2);
\draw[draw={rgb,255:red,200; green,200; blue,200}](4.5,-11.2)--(26.5,-11.2);
\draw[draw={rgb,255:red,200; green,200; blue,200}](4.5,-12.2)--(26.5,-12.2);
\draw[draw={rgb,255:red,200; green,200; blue,200}](4.5,-13.2)--(26.5,-13.2);
\draw[draw={rgb,255:red,200; green,200; blue,200}](4.5,-14.2)--(26.5,-14.2);
\draw[draw=none,fill={rgb,255:red,179; green,179; blue,179},opacity=0.25](13.5,-5.7)--(13.5,-0.2)--(26.5,-0.2)--(26.5,-15.2)--(5.5,-15.2)--(5.5,-13.7)(5.5,-13.7)--(8,-13.7)--(8,-12.7)--(10,-12.7)--(10,-9.7)--(17,-9.7)--(17,-12.7)--(19,-12.7)--(19,-9.7)--(20,-9.7)--(20,-7.7)--(23,-7.7)--(23,-10.7)--(25,-10.7)--(25,-5.7)--(20,-5.7)--(20,-4.7)--(17,-4.7)--(17,-3.7)--(19,-3.7)--(19,-1.7)--(16,-1.7)--(16,-5.7)--(13.5,-5.7);
\draw(22.5, -3.7) node[anchor=south west] {$\mathcal C$};
\draw[draw={rgb,255:red,200; green,113; blue,55},opacity=0.5,very thick](6,-13.7)--(8,-13.7)--(8,-12.7)--(10,-12.7)--(10,-9.7)--(17,-9.7)--(17,-12.7)--(19,-12.7)--(19,-9.7)--(20,-9.7)--(20,-7.7)--(23,-7.7)--(23,-10.7)--(25,-10.7)--(25,-5.7)--(20,-5.7)--(20,-4.7)--(17,-4.7)--(17,-3.7)--(19,-3.7)--(19,-1.7)--(16,-1.7)--(16,-5.7)--(13,-5.7);
\draw[draw={rgb,255:red,55; green,200; blue,55},opacity=0.5,very thick](6,-13.7)--(6,-10.7)--(13,-10.7)--(13,-13.7)--(15,-13.7)--(15,-10.7)--(16,-10.7)--(16,-8.7)--(19,-8.7)--(19,-11.7)--(21,-11.7)--(21,-6.7)--(16,-6.7)--(16,-5.7)--(13,-5.7)--(13,-4.7)--(15,-4.7)--(15,-2.7)--(12,-2.7)--(12,-6.7)--(9,-6.7);
\draw[draw={rgb,255:red,0; green,0; blue,0}](9.5,-12.69)--(9.5,-16.7);
\draw(8.23, -16.93) node[anchor=south west] {$l_j$};
\draw[draw=none,fill={rgb,255:red,28; green,36; blue,31},thin](9.5, -12.69) ellipse (0.1cm and 0.1cm);\draw[draw={rgb,255:red,0; green,0; blue,0}](5.5,-13.7)--(5.5,-16.7);
\draw(4.32, -16.93) node[anchor=south west] {$l_i$};
\draw[draw={rgb,255:red,0; green,0; blue,0}](13.5,-5.7)--(13.5,0.3);
\draw(12.34, -0.81) node[anchor=south west] {$l^k$};
\draw[draw=none,fill={rgb,255:red,28; green,36; blue,31},thin](13.5, -5.7) ellipse (0.1cm and 0.1cm);\draw[draw=none,fill={rgb,255:red,28; green,36; blue,31},thin](5.5, -13.7) ellipse (0.1cm and 0.1cm);
\end{tikzpicture}\hfill
  \begin{tikzpicture}[scale=\scale]\draw[draw={rgb,255:red,200; green,200; blue,200}](4.5,-0.2) rectangle (26.5, -15.2);
\draw[draw={rgb,255:red,200; green,200; blue,200}](4.5,-15.2)--(4.5,-0.2);
\draw[draw={rgb,255:red,200; green,200; blue,200}](5.5,-15.2)--(5.5,-0.2);
\draw[draw={rgb,255:red,200; green,200; blue,200}](6.5,-15.2)--(6.5,-0.2);
\draw[draw={rgb,255:red,200; green,200; blue,200}](7.5,-15.2)--(7.5,-0.2);
\draw[draw={rgb,255:red,200; green,200; blue,200}](8.5,-15.2)--(8.5,-0.2);
\draw[draw={rgb,255:red,200; green,200; blue,200}](9.5,-15.2)--(9.5,-0.2);
\draw[draw={rgb,255:red,200; green,200; blue,200}](10.5,-15.2)--(10.5,-0.2);
\draw[draw={rgb,255:red,200; green,200; blue,200}](11.5,-15.2)--(11.5,-0.2);
\draw[draw={rgb,255:red,200; green,200; blue,200}](12.5,-15.2)--(12.5,-0.2);
\draw[draw={rgb,255:red,200; green,200; blue,200}](13.5,-15.2)--(13.5,-0.2);
\draw[draw={rgb,255:red,200; green,200; blue,200}](14.5,-15.2)--(14.5,-0.2);
\draw[draw={rgb,255:red,200; green,200; blue,200}](15.5,-15.2)--(15.5,-0.2);
\draw[draw={rgb,255:red,200; green,200; blue,200}](16.5,-15.2)--(16.5,-0.2);
\draw[draw={rgb,255:red,200; green,200; blue,200}](17.5,-15.2)--(17.5,-0.2);
\draw[draw={rgb,255:red,200; green,200; blue,200}](18.5,-15.2)--(18.5,-0.2);
\draw[draw={rgb,255:red,200; green,200; blue,200}](19.5,-15.2)--(19.5,-0.2);
\draw[draw={rgb,255:red,200; green,200; blue,200}](20.5,-15.2)--(20.5,-0.2);
\draw[draw={rgb,255:red,200; green,200; blue,200}](21.5,-15.2)--(21.5,-0.2);
\draw[draw={rgb,255:red,200; green,200; blue,200}](22.5,-15.2)--(22.5,-0.2);
\draw[draw={rgb,255:red,200; green,200; blue,200}](23.5,-15.2)--(23.5,-0.2);
\draw[draw={rgb,255:red,200; green,200; blue,200}](24.5,-15.2)--(24.5,-0.2);
\draw[draw={rgb,255:red,200; green,200; blue,200}](25.5,-15.2)--(25.5,-0.2);
\draw[draw={rgb,255:red,200; green,200; blue,200}](4.5,-0.2)--(26.5,-0.2);
\draw[draw={rgb,255:red,200; green,200; blue,200}](4.5,-1.2)--(26.5,-1.2);
\draw[draw={rgb,255:red,200; green,200; blue,200}](4.5,-2.2)--(26.5,-2.2);
\draw[draw={rgb,255:red,200; green,200; blue,200}](4.5,-3.2)--(26.5,-3.2);
\draw[draw={rgb,255:red,200; green,200; blue,200}](4.5,-4.2)--(26.5,-4.2);
\draw[draw={rgb,255:red,200; green,200; blue,200}](4.5,-5.2)--(26.5,-5.2);
\draw[draw={rgb,255:red,200; green,200; blue,200}](4.5,-6.2)--(26.5,-6.2);
\draw[draw={rgb,255:red,200; green,200; blue,200}](4.5,-7.2)--(26.5,-7.2);
\draw[draw={rgb,255:red,200; green,200; blue,200}](4.5,-8.2)--(26.5,-8.2);
\draw[draw={rgb,255:red,200; green,200; blue,200}](4.5,-9.2)--(26.5,-9.2);
\draw[draw={rgb,255:red,200; green,200; blue,200}](4.5,-10.2)--(26.5,-10.2);
\draw[draw={rgb,255:red,200; green,200; blue,200}](4.5,-11.2)--(26.5,-11.2);
\draw[draw={rgb,255:red,200; green,200; blue,200}](4.5,-12.2)--(26.5,-12.2);
\draw[draw={rgb,255:red,200; green,200; blue,200}](4.5,-13.2)--(26.5,-13.2);
\draw[draw={rgb,255:red,200; green,200; blue,200}](4.5,-14.2)--(26.5,-14.2);
\draw[draw=none,fill={rgb,255:red,179; green,179; blue,179},opacity=0.25](13.5,-5.7)--(13.5,-0.2)--(26.5,-0.2)--(26.5,-15.2)--(5.5,-15.2)--(5.5,-13.7)(5.5,-13.7)--(8,-13.7)--(8,-12.7)--(10,-12.7)--(10,-9.7)--(17,-9.7)--(17,-12.7)--(19,-12.7)--(19,-9.7)--(20,-9.7)--(20,-7.7)--(23,-7.7)--(23,-10.7)--(25,-10.7)--(25,-5.7)--(20,-5.7)--(20,-4.7)--(17,-4.7)--(17,-3.7)--(19,-3.7)--(19,-1.7)--(16,-1.7)--(16,-5.7)--(13.5,-5.7);
\draw(22.5, -3.7) node[anchor=south west] {$\mathcal C$};
\draw[draw={rgb,255:red,200; green,113; blue,55},opacity=0.5,thick](6,-13.7)--(8,-13.7)--(8,-12.7)--(10,-12.7)--(10,-9.7)--(17,-9.7)--(17,-12.7)--(19,-12.7)--(19,-9.7)--(20,-9.7)--(20,-7.7)--(23,-7.7)--(23,-10.7)--(25,-10.7)--(25,-5.7)--(20,-5.7)--(20,-4.7)--(17,-4.7)--(17,-3.7)--(19,-3.7)--(19,-1.7)--(16,-1.7)--(16,-5.7)--(13,-5.7);
\draw[draw={rgb,255:red,55; green,200; blue,55},opacity=0.5,thick](6,-13.7)--(6,-10.7)--(13,-10.7)--(13,-13.7)--(15,-13.7)--(15,-10.7)--(16,-10.7)--(16,-8.7)--(19,-8.7)--(19,-11.7)--(21,-11.7)--(21,-6.7)--(16,-6.7)--(16,-5.7)--(13,-5.7)--(13,-4.7)--(15,-4.7)--(15,-2.7)--(12,-2.7)--(12,-6.7)--(9,-6.7);
\draw[draw={rgb,255:red,255; green,0; blue,0},opacity=0.5,very thick](6,-13.7)--(8,-13.7)--(8,-12.7)--(10,-12.7)--(10,-10.7)--(13,-10.7)--(13,-13.7)--(15,-13.7)--(15,-10.7)--(16,-10.7)--(16,-9.7)--(17,-9.7)--(17,-12.7)--(19,-12.7)--(19,-11.7)--(21,-11.7)--(21,-7.7)--(23,-7.7)--(23,-10.7)--(25,-10.7)--(25,-5.7)--(20,-5.7)--(20,-4.7)--(17,-4.7)--(17,-3.7)--(19,-3.7)--(19,-1.7)--(16,-1.7) .. controls (16,-3.03) and (16,-4.37) .. (16,-5.7)--(14,-5.7);
\draw[draw={rgb,255:red,0; green,0; blue,0}](9.5,-12.69)--(9.5,-16.7);
\draw(8.23, -16.93) node[anchor=south west] {$l_j$};
\draw[draw=none,fill={rgb,255:red,28; green,36; blue,31},thin](9.5, -12.69) ellipse (0.1cm and 0.1cm);\draw[draw={rgb,255:red,0; green,0; blue,0}](5.5,-13.7)--(5.5,-16.7);
\draw(4.32, -16.93) node[anchor=south west] {$l_i$};
\draw[draw={rgb,255:red,0; green,0; blue,0}](13.5,-5.7)--(13.5,0.3);
\draw(12.34, -0.81) node[anchor=south west] {$l^k$};
\draw[draw=none,fill={rgb,255:red,28; green,36; blue,31},thin](13.5, -5.7) ellipse (0.1cm and 0.1cm);\draw[draw=none,fill={rgb,255:red,28; green,36; blue,31},thin](5.5, -13.7) ellipse (0.1cm and 0.1cm);
\end{tikzpicture}
  \renewcommand\scale{0.4}
  \caption{Left: Following from Figure~\ref{fig:shield-setup-c}, $\embed{P_{i+1,i+2,\ldots,k+1}}$ is shown in brown and $\embed{P_{j+1,j+2,\ldots,k+1}}+ \vect{P_j P_i}$ is shown in green.  Right: The route traced by (the embedding of) the binding path $r$ defined in Subsection~\ref{subsec:r}: red indicates  when $r$ takes positions from $P_{i+1,i+2,\ldots,k}$ only, and  brown indicates when $r$ takes positions from $P_{j+1,j+2,\ldots,k}+ \vect{P_j P_i}$ and/or $P_{i+1,i+2,\ldots,k}$.}\label{fig:shield-r}
\end{figure}

The definition of $r$ is illustrated in Figure~\ref{fig:shield-r}.
Let $G$ be the binding graph $G = (V, E)$ where:
\begin{eqnarray*}
V& = & \{\pos{P_n} \mid n\in\{\range {i+1}{i+2} k\}\}\cup\{\pos{P_{n}}+\vpji \mid n\in\{\range {j+1}{j+2}k\}\} \\
E& = & \{ \{ \pos{P_{n}},\pos{P_{n+1}} \} \mid n\in\{\range {i+1}{i+2}{k-1}\}\}\\
  &  & \cup \;\; \{ \{ \pos{P_{n}}+\vpji,\pos{P_{n+1}}+\vpji \} \mid n\in\{\range {j+1}{j+2}{k-1}\}\}
\end{eqnarray*}

We define the set $\overline{\mathbb{S}}$ of all binding paths
in the graph $G' = (V \cup \{ \pos{P_i} \} , E\cup \{ \pos{P_i} ,\pos{P_{i+1}} \})$ such that $q\in\overline{\mathbb{S}}$ if and only if:
\begin{itemize}
\item $q$ starts with $q_0 q_1 = \pos{P_i} \pos{P_{i+1}}$ and has its final vertex $q_{|q|-1}$ being a vertex of $G$ that is ``adjacent'' to $l^k$; more precisely there exists $z\in \mathbb{R}, z \geq 0$ such that  $l^k(z) \in \{ q_{|q|-1} - (0.5,0), \, q_{|q|-1} + (0.5,0) \}$; and
\item the curve $\concat{ \gs{ l^i(0) }{\pos{P_{i+1}}}, \, \embed{q_{1,2,\ldots,q_{|q|-1}}}  ,\, \gs{q_{|q|-1}}{l^k(z)}}$ is entirely in $\mathcal C$.
\end{itemize}

We now define the set $\mathbb S\subseteq\overline{\mathbb S}$ of all paths $q\in\overline{\mathbb S}$ such that for all $q'\in\overline{\mathbb S}$ that is a strict prefix or a strict extension of $q$, the last tile of $q'$ is strictly lower than the last tile of $q$, i.e.\ $y_{q'_{|q'|-1}} < y_{q_{|q|-1}}$.

Let $q'$ be the most right priority binding path of $\mathbb{S}$  (see Definition~\ref{def:rp}). 
We claim that $q'$ is well-defined.
First, we claim that $\mathbb{S}$ is nonempty: indeed, $\overline{\mathbb{S}}$ is nonempty as there is at least one binding path $q$ that starts with  
$q_0 q_1 = \pos{P_i} \pos{P_{i+1}}$, ends with a vertex adjacent to $l^k$, and stays in $\mathcal C$, since $\pos{ P_{\rng i k} }$ satisfies the first and second conditions. 
Second, either $P$ or one of its strict prefixes is in~$\mathbb{S}$.  
Third, since all binding paths $q \in \mathbb{S}$ have their first two vertices $q_0 q_1 = \pos{P_i} \pos{P_{i+1}}$ in common, 
the right priority path of $\mathbb{S}$ is well-defined. Thus $q'$ is well-defined. 
Finally, let $r = r_{\range{0}{1}{|r|-1}} = q'_{\range 1 2 {|q'|-1}}$, i.e. $r$ is the same binding path as $q'$ but without its the first vertex $q'_0$. 

Thus, by definition, $r$ is entirely in $\mathcal C$. We now go on to prove some properties about $r$. The proof of \subl{lem:order} is illustrated in Figures~\ref{fig:shield-order} and~\ref{fig:shield-order1}.

\begin{sublemma}
  \label{lem:order}
  Let $a$ and $b$ be two integers such that $0\leq a<b \leq |r|-1$.
  If there are indices $d, e \in \{\range{i+1}{i+2}{k} \}$ such that $\pos{P_d} = r_a$ and $\pos{P_e} = r_b$, then $d<e$.
\end{sublemma}
\begin{proof}
  Assume, for the sake of contradiction, that there are two indices $d$ and $e$ such that $\pos{P_d}=r_a$, $\pos{P_e}=r_b$ and $d\geq e$.
  Because $r$ and $P$ are simple and $a\neq b$, we also have $d > e$. Note that $e>i+1$ (since $r_0=\pos{P_{i+1}}$, $b>a\geq 0$ and $r$ is simple).
  We assume without loss of generality that $e$ is the smallest integer satisfying the hypotheses of this \sublname. (See Figure~\ref{fig:shield-order} for an example.)
    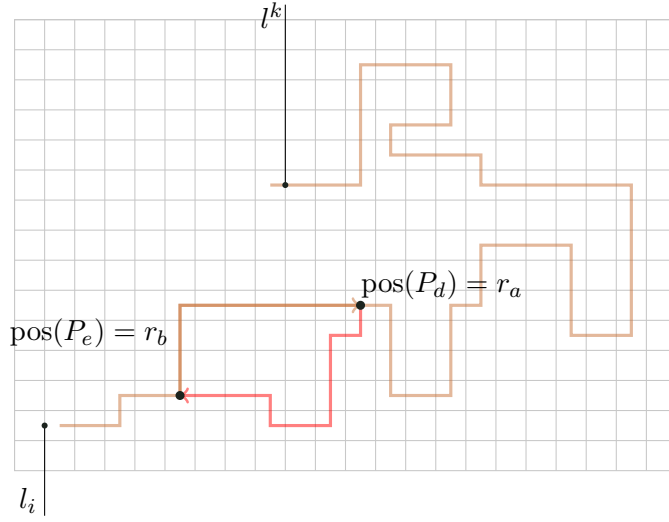
\begin{figure}[ht]
    \centering
    \begin{tikzpicture}[scale=\scale]\draw[draw={rgb,255:red,200; green,200; blue,200}](4.5,-0.2) rectangle (26.5, -15.2);
\draw[draw={rgb,255:red,200; green,200; blue,200}](4.5,-15.2)--(4.5,-0.2);
\draw[draw={rgb,255:red,200; green,200; blue,200}](5.5,-15.2)--(5.5,-0.2);
\draw[draw={rgb,255:red,200; green,200; blue,200}](6.5,-15.2)--(6.5,-0.2);
\draw[draw={rgb,255:red,200; green,200; blue,200}](7.5,-15.2)--(7.5,-0.2);
\draw[draw={rgb,255:red,200; green,200; blue,200}](8.5,-15.2)--(8.5,-0.2);
\draw[draw={rgb,255:red,200; green,200; blue,200}](9.5,-15.2)--(9.5,-0.2);
\draw[draw={rgb,255:red,200; green,200; blue,200}](10.5,-15.2)--(10.5,-0.2);
\draw[draw={rgb,255:red,200; green,200; blue,200}](11.5,-15.2)--(11.5,-0.2);
\draw[draw={rgb,255:red,200; green,200; blue,200}](12.5,-15.2)--(12.5,-0.2);
\draw[draw={rgb,255:red,200; green,200; blue,200}](13.5,-15.2)--(13.5,-0.2);
\draw[draw={rgb,255:red,200; green,200; blue,200}](14.5,-15.2)--(14.5,-0.2);
\draw[draw={rgb,255:red,200; green,200; blue,200}](15.5,-15.2)--(15.5,-0.2);
\draw[draw={rgb,255:red,200; green,200; blue,200}](16.5,-15.2)--(16.5,-0.2);
\draw[draw={rgb,255:red,200; green,200; blue,200}](17.5,-15.2)--(17.5,-0.2);
\draw[draw={rgb,255:red,200; green,200; blue,200}](18.5,-15.2)--(18.5,-0.2);
\draw[draw={rgb,255:red,200; green,200; blue,200}](19.5,-15.2)--(19.5,-0.2);
\draw[draw={rgb,255:red,200; green,200; blue,200}](20.5,-15.2)--(20.5,-0.2);
\draw[draw={rgb,255:red,200; green,200; blue,200}](21.5,-15.2)--(21.5,-0.2);
\draw[draw={rgb,255:red,200; green,200; blue,200}](22.5,-15.2)--(22.5,-0.2);
\draw[draw={rgb,255:red,200; green,200; blue,200}](23.5,-15.2)--(23.5,-0.2);
\draw[draw={rgb,255:red,200; green,200; blue,200}](24.5,-15.2)--(24.5,-0.2);
\draw[draw={rgb,255:red,200; green,200; blue,200}](25.5,-15.2)--(25.5,-0.2);
\draw[draw={rgb,255:red,200; green,200; blue,200}](4.5,-0.2)--(26.5,-0.2);
\draw[draw={rgb,255:red,200; green,200; blue,200}](4.5,-1.2)--(26.5,-1.2);
\draw[draw={rgb,255:red,200; green,200; blue,200}](4.5,-2.2)--(26.5,-2.2);
\draw[draw={rgb,255:red,200; green,200; blue,200}](4.5,-3.2)--(26.5,-3.2);
\draw[draw={rgb,255:red,200; green,200; blue,200}](4.5,-4.2)--(26.5,-4.2);
\draw[draw={rgb,255:red,200; green,200; blue,200}](4.5,-5.2)--(26.5,-5.2);
\draw[draw={rgb,255:red,200; green,200; blue,200}](4.5,-6.2)--(26.5,-6.2);
\draw[draw={rgb,255:red,200; green,200; blue,200}](4.5,-7.2)--(26.5,-7.2);
\draw[draw={rgb,255:red,200; green,200; blue,200}](4.5,-8.2)--(26.5,-8.2);
\draw[draw={rgb,255:red,200; green,200; blue,200}](4.5,-9.2)--(26.5,-9.2);
\draw[draw={rgb,255:red,200; green,200; blue,200}](4.5,-10.2)--(26.5,-10.2);
\draw[draw={rgb,255:red,200; green,200; blue,200}](4.5,-11.2)--(26.5,-11.2);
\draw[draw={rgb,255:red,200; green,200; blue,200}](4.5,-12.2)--(26.5,-12.2);
\draw[draw={rgb,255:red,200; green,200; blue,200}](4.5,-13.2)--(26.5,-13.2);
\draw[draw={rgb,255:red,200; green,200; blue,200}](4.5,-14.2)--(26.5,-14.2);
\draw[draw={rgb,255:red,200; green,113; blue,55},opacity=0.5,very thick](6,-13.7)--(8,-13.7)--(8,-12.7)--(10,-12.7)--(10,-9.7)--(17,-9.7)--(17,-12.7)--(19,-12.7)--(19,-9.7)--(20,-9.7)--(20,-7.7)--(23,-7.7)--(23,-10.7)--(25,-10.7)--(25,-5.7)--(20,-5.7)--(20,-4.7)--(17,-4.7)--(17,-3.7)--(19,-3.7)--(19,-1.7)--(16,-1.7)--(16,-5.7)--(13,-5.7);
\draw[draw={rgb,255:red,255; green,0; blue,0},opacity=0.5,<-,very thick](10,-12.7)--(13,-12.7)--(13,-13.7)--(15,-13.7)--(15,-10.7)--(16,-10.7)--(16,-9.7);
\draw[draw={rgb,255:red,200; green,113; blue,55},opacity=0.5,->,very thick](10,-12.7)--(10,-9.7)--(16,-9.7);
\draw(4, -11.44) node[anchor=south west] {$\pos{P_e} = r_b$};
\draw(15.68, -9.84) node[anchor=south west] {$\pos{P_d} = r_a$};
\draw[draw=none,fill={rgb,255:red,28; green,36; blue,31},thin](16, -9.7) ellipse (0.15cm and 0.15cm);\draw[draw=none,fill={rgb,255:red,28; green,36; blue,31},thin](10, -12.7) ellipse (0.15cm and 0.15cm);\draw[draw={rgb,255:red,0; green,0; blue,0}](5.5,-13.7)--(5.5,-16.7);
\draw(4.32, -16.93) node[anchor=south west] {$l_i$};
\draw[draw={rgb,255:red,0; green,0; blue,0}](13.5,-5.7)--(13.5,0.3);
\draw(12.34, -0.81) node[anchor=south west] {$l^k$};
\draw[draw=none,fill={rgb,255:red,28; green,36; blue,31},thin](13.5, -5.7) ellipse (0.1cm and 0.1cm);\draw[draw=none,fill={rgb,255:red,28; green,36; blue,31},thin](5.5, -13.7) ellipse (0.1cm and 0.1cm);
\end{tikzpicture}
    \caption{Initial setup for the proof \subl{lem:order}. We assume, for the sake of contradiction, that there are two indices $a<b$ such that $r_a = \pos{P_d}$, $r_b = \pos{P_e}$, for some $d$ and $e$ such that $d>e$. Here the brown path is $\pos{P_{\range{i+1}{i+2}{k+1}}}$, and the red path is $r$.}\label{fig:shield-order}
  \end{figure}
  \begin{figure}[ht]
    \centering
    \begin{tikzpicture}[scale=\scale]\draw[draw={rgb,255:red,200; green,200; blue,200}](4.5,0.8) rectangle (29.5, -17.2);
\draw[draw={rgb,255:red,200; green,200; blue,200}](4.5,-17.2)--(4.5,0.8);
\draw[draw={rgb,255:red,200; green,200; blue,200}](5.5,-17.2)--(5.5,0.8);
\draw[draw={rgb,255:red,200; green,200; blue,200}](6.5,-17.2)--(6.5,0.8);
\draw[draw={rgb,255:red,200; green,200; blue,200}](7.5,-17.2)--(7.5,0.8);
\draw[draw={rgb,255:red,200; green,200; blue,200}](8.5,-17.2)--(8.5,0.8);
\draw[draw={rgb,255:red,200; green,200; blue,200}](9.5,-17.2)--(9.5,0.8);
\draw[draw={rgb,255:red,200; green,200; blue,200}](10.5,-17.2)--(10.5,0.8);
\draw[draw={rgb,255:red,200; green,200; blue,200}](11.5,-17.2)--(11.5,0.8);
\draw[draw={rgb,255:red,200; green,200; blue,200}](12.5,-17.2)--(12.5,0.8);
\draw[draw={rgb,255:red,200; green,200; blue,200}](13.5,-17.2)--(13.5,0.8);
\draw[draw={rgb,255:red,200; green,200; blue,200}](14.5,-17.2)--(14.5,0.8);
\draw[draw={rgb,255:red,200; green,200; blue,200}](15.5,-17.2)--(15.5,0.8);
\draw[draw={rgb,255:red,200; green,200; blue,200}](16.5,-17.2)--(16.5,0.8);
\draw[draw={rgb,255:red,200; green,200; blue,200}](17.5,-17.2)--(17.5,0.8);
\draw[draw={rgb,255:red,200; green,200; blue,200}](18.5,-17.2)--(18.5,0.8);
\draw[draw={rgb,255:red,200; green,200; blue,200}](19.5,-17.2)--(19.5,0.8);
\draw[draw={rgb,255:red,200; green,200; blue,200}](20.5,-17.2)--(20.5,0.8);
\draw[draw={rgb,255:red,200; green,200; blue,200}](21.5,-17.2)--(21.5,0.8);
\draw[draw={rgb,255:red,200; green,200; blue,200}](22.5,-17.2)--(22.5,0.8);
\draw[draw={rgb,255:red,200; green,200; blue,200}](23.5,-17.2)--(23.5,0.8);
\draw[draw={rgb,255:red,200; green,200; blue,200}](24.5,-17.2)--(24.5,0.8);
\draw[draw={rgb,255:red,200; green,200; blue,200}](25.5,-17.2)--(25.5,0.8);
\draw[draw={rgb,255:red,200; green,200; blue,200}](26.5,-17.2)--(26.5,0.8);
\draw[draw={rgb,255:red,200; green,200; blue,200}](27.5,-17.2)--(27.5,0.8);
\draw[draw={rgb,255:red,200; green,200; blue,200}](28.5,-17.2)--(28.5,0.8);
\draw[draw={rgb,255:red,200; green,200; blue,200}](4.5,0.8)--(29.5,0.8);
\draw[draw={rgb,255:red,200; green,200; blue,200}](4.5,-0.2)--(29.5,-0.2);
\draw[draw={rgb,255:red,200; green,200; blue,200}](4.5,-1.2)--(29.5,-1.2);
\draw[draw={rgb,255:red,200; green,200; blue,200}](4.5,-2.2)--(29.5,-2.2);
\draw[draw={rgb,255:red,200; green,200; blue,200}](4.5,-3.2)--(29.5,-3.2);
\draw[draw={rgb,255:red,200; green,200; blue,200}](4.5,-4.2)--(29.5,-4.2);
\draw[draw={rgb,255:red,200; green,200; blue,200}](4.5,-5.2)--(29.5,-5.2);
\draw[draw={rgb,255:red,200; green,200; blue,200}](4.5,-6.2)--(29.5,-6.2);
\draw[draw={rgb,255:red,200; green,200; blue,200}](4.5,-7.2)--(29.5,-7.2);
\draw[draw={rgb,255:red,200; green,200; blue,200}](4.5,-8.2)--(29.5,-8.2);
\draw[draw={rgb,255:red,200; green,200; blue,200}](4.5,-9.2)--(29.5,-9.2);
\draw[draw={rgb,255:red,200; green,200; blue,200}](4.5,-10.2)--(29.5,-10.2);
\draw[draw={rgb,255:red,200; green,200; blue,200}](4.5,-11.2)--(29.5,-11.2);
\draw[draw={rgb,255:red,200; green,200; blue,200}](4.5,-12.2)--(29.5,-12.2);
\draw[draw={rgb,255:red,200; green,200; blue,200}](4.5,-13.2)--(29.5,-13.2);
\draw[draw={rgb,255:red,200; green,200; blue,200}](4.5,-14.2)--(29.5,-14.2);
\draw[draw={rgb,255:red,200; green,200; blue,200}](4.5,-15.2)--(29.5,-15.2);
\draw[draw={rgb,255:red,200; green,200; blue,200}](4.5,-16.2)--(29.5,-16.2);
\draw[draw={rgb,255:red,200; green,113; blue,55},opacity=0.5,very thick](6,-13.7)--(8,-13.7)--(8,-12.7)--(10,-12.7)--(10,-9.7)--(17,-9.7)--(17,-12.7)--(19,-12.7)--(19,-9.7)--(20,-9.7)--(20,-7.7)--(23,-7.7)--(23,-10.7)--(25,-10.7)--(25,-5.7)--(20,-5.7)--(20,-4.7)--(17,-4.7)--(17,-3.7)--(19,-3.7)--(19,-1.7)--(16,-1.7)--(16,-5.7)--(13,-5.7);
\draw[draw={rgb,255:red,200; green,113; blue,55},fill={rgb,255:red,0; green,0; blue,0},opacity=0.1](13.5,-1.7)--(13.5,-5.7)--(16,-5.7)--(16,-1.7)--(19,-1.7)--(19,-3.7)--(17,-3.7)--(17,-4.7)--(20,-4.7)--(20,-5.7)--(25,-5.7)--(25,-10.7)--(23,-10.7)--(23,-7.7)--(20,-7.7)--(20,-9.7)--(19,-9.7)--(19,-12.7)--(17,-12.7)--(17,-9.7)--(10,-9.7) .. controls (10,-11.7) and (10,-13.7) .. (10,-15.7)--(27,-15.7)--(27,-1.7)--(26,-1.7)--(26,-0.7)--(15,-0.7)--(15,-1.7)-- cycle;
\draw[draw={rgb,255:red,255; green,0; blue,0},opacity=0.5,->,very thick](10,-12.7)--(10,-15.7)--(27,-15.7)--(27,-1.7)--(26,-1.7)--(26,-0.7)--(15,-0.7)--(15,-1.7)--(13.5,-1.7);
\draw(25, -3.7) node[anchor=south west] {$s$};
\draw[draw={rgb,255:red,255; green,0; blue,0},opacity=0.5,->,very thick](5.5,-13.7)--(9,-13.7);
\draw(5.02, -13.78) node[anchor=south west] {$r_0$};
\draw[draw=none,fill={rgb,255:red,28; green,36; blue,31},thin](6.03, -13.7) ellipse (0.15cm and 0.15cm);\draw[draw={rgb,255:red,255; green,0; blue,0},opacity=0.5,<-,very thick](10,-12.7)--(13,-12.7)--(13,-13.7)--(15,-13.7)--(15,-10.7)--(16,-10.7)--(16,-9.7);
\draw[draw={rgb,255:red,200; green,113; blue,55},opacity=0.5,->,very thick](10,-12.7)--(10,-9.7)--(16,-9.7);
\draw(4, -11.44) node[anchor=south west] {$\pos{P_e} = r_b$};
\draw(15.68, -9.84) node[anchor=south west] {$\pos{P_d} = r_a$};
\draw[draw=none,fill={rgb,255:red,28; green,36; blue,31},thin](16, -9.7) ellipse (0.15cm and 0.15cm);\draw[draw=none,fill={rgb,255:red,28; green,36; blue,31},thin](10, -12.7) ellipse (0.15cm and 0.15cm);\draw[draw={rgb,255:red,255; green,0; blue,0},dashed,->,thick](8.5,-13.7) .. controls (11,-13.7) and (11.87,-13.62) .. (12.5,-14.7) .. controls (13.13,-15.78) and (14.61,-15.26) .. (15.5,-14.7) .. controls (16.39,-14.14) and (17.15,-14.02) .. (16.71,-12.56) .. controls (16.41,-11.53) and (17.16,-10.37) .. (16.21,-9.81);
\draw[draw={rgb,255:red,0; green,0; blue,0}](5.5,-13.7)--(5.5,-16.7);
\draw(4.32, -16.93) node[anchor=south west] {$l_i$};
\draw[draw={rgb,255:red,0; green,0; blue,0}](13.5,-5.7)--(13.5,0.3);
\draw(12.34, -0.81) node[anchor=south west] {$l^k$};
\draw[draw=none,fill={rgb,255:red,28; green,36; blue,31},thin](13.5, -5.7) ellipse (0.1cm and 0.1cm);\draw[draw=none,fill={rgb,255:red,28; green,36; blue,31},thin](5.5, -13.7) ellipse (0.1cm and 0.1cm);
\end{tikzpicture}
    \caption{Illustration for the proof of \subl{lem:order}. $\pos{ P_{\range{i+1}{i+2}{k+1}} }$  is in brown, $r$ is in red, and we assume that there are two indices $a<b$ such that $r_a = \pos{P_d}$ and $r_b = \pos{P_e}$ and such that (for the sake of contradiction) we have $d>e$.
      If we draw $r_{\rng b {|r|-1}}$ after $r_{\rng a b}$, then since $r$ is simple, then although $r_{\range 0 1 a}$ must enter the grey zone where $r_a$ is, any attempt by $r_{\range 0 1 a}$ to do so leads it having position(s) outside $\mathcal C$ (contradicting that $\embed{r}$ stays in $\mathcal C$ by definition) or 
  $r$ intersecting itself (contradicting that $r$ is a simple binding path -- a situation shown by the dashed curve).}\label{fig:shield-order1}
  \end{figure}

  Let $s$ be the curve defined by:
  $$s = \concat{
  \reverse{l^i},\,
  \gs{l^i(0)}{\pos{P_{i+1}}}, \,
  \embed{P_{\range{i+1}{i+2}e}}, \,
  \embed{r_{\rng b {|r|-1}}},\, 
  \gs{r_{|r|-1}}{l^k(z)},\,
  l^k_z
}$$
where $l^k(z) \in\mathbb{R}^2$ is the point of $l^k$ at the same y-coordinate as $r_{|r|-1}$ and $\ensuremath{l^k_z}$ is the ray to the north starting at $l^k(z)$.
The curve $s$ is entirely in $\mathcal C$ because the six curves used to define $s$ are in $\mathcal C$.
 The six curves used to define $s$ only intersect pairwise at their endpoints meaning that $s$ is simple (in particular by the minimality of $e$, $P_{\range {i+1}{i+2}{e-1}}$ does not intersect $r_{\rng b {|r|-1}}$).
  Therefore, since $s$ starts and ends with vertical rays, and is otherwise made of a finite number of horizontal and vertical segments, $s$ cuts the plane into two connected components, by Theorem~\ref{thm:infinite-jordan}.
  We claim that $\embed{P_{\range{e}{e+1}k}}$ does not turn right from~$s$, by looking at each part of $s$:
  \begin{itemize}
  \item $\glu P i$ and $\glu P k$ are visible relative to $P$, hence $\embed{P_{\range{e}{e+1}k}}$ does not intersect $\reverse{l^i}$ nor $l^k$.
  \item $P$ is simple, hence $\embed{P_{\range{e} {e+1} k}}$ cannot intersect $\embed{P_{\range {i+1} {i+2} {e-1}}}$.
  \item If $\embed{P_{\range{e}{e+1}k}}$ turns right from a point of $\gs{l^i(0)}{ \pos{P_{i+1}}}$, then
    $\embed{P_{\range{e}{e+1}k}}$ must intersect $P_i$ or $l^i(0)$ (or both), which is already covered by the two previous points.
  \item $\embed{P_{\range{e}{e+1}k}}$ (which is in $\mathcal C$) does not turn right from $\embed{r_{\rng b {|r|-1}}}$. There are two subcases:
    \begin{itemize}
    \item If that right turn happens on a point of $r$ strictly before $|r|-1$, then this contradicts the definition of $r$ as a most right-priority path.
    \item Else, that right turn happens at $r_{|r|-1} = \pos{P_f}$ for some integer $f$, which means that $f<k$.
      This means that $\embed{P}$ also turns right from the curve defined as $\rho = \concat{\reverse{l^i}, \gs{l^i(0)}{\pos{r_0}}, \embed{r},   \gs{r_{|r|-1}}{l^k(z)}, l^k_z}$. However, $P_k$ is not in the connected component on the right-hand side of $\rho$, since $f < k$. Therefore, $P$ must intersect $\rho$, and that can only happen on $r$. However, this implies that not all of $r$ is on the same side of $P$, contradicting the definition of $r$ as being in $\mathcal C$.
    \end{itemize}
  \item if $\embed{P_{\range{e}{e+1}k}}$ turns right from a point of $\gs{r_{|r|-1}}{l^k(z)}$ then this can only occur at one of the two endpoints of this segment (which is a horizontal line segment of length 0.5) which is already covered in the first or fourth point above.
  \end{itemize}
  This means that $\embed{P_{\range{e}{e+1}k}}$ does not turn right from~$s$, and therefore, $r_a = \pos{P_d}$ is either on $s$, or in the component of $\mathbb{R}^2$ that is on the left-hand side of $s$.

  If $r_a$ is on $s$, then since $d > e$ and $P$ is simple, $r_a$ cannot be on $P_{\range{i+1}{i+2}e}$. But $r_a$ cannot be on $r_{\rng b {|r|-1}}$ either, since $r$ is simple and $a < b$.
  Therefore, $r_a$ is not on $s$, and since $r_0 = \pos{P_{i+1}}$, this implies that $\embed{r_{\range 0 1 a}}$ turns left from $\embed{s}$. The position of that left turn from $\embed{s}$ must be on one of the six curves used to define~$s$ (see Figure~\ref{fig:shield-order1} for an example):
  \begin{itemize}
  \item If the left turn is from a point on $\embed{r_{\range b {b+1} {|r|-1}}}$, this contradicts that $r$ is simple.
  \item If the left turn is from a point on $\reverse{l^i}$,  $\gs{l^i(0)}{\pos{P_{i+1}}}$, $\embed{P_{\range {i+1} {i+2} e}}$, $\gs{r_{|r|-1}}{l^k(z)}$ or $\ensuremath{l^k_z}$, then either we are also in the previous case (i.e. the left turn is on $\embed{r_{\range b {b+1}{|r|-1}}}$, or we contradict the definition of $\embed{r}$ as being in $\mathcal C$.
  \end{itemize}
  In all cases, we get a contradiction.
  Therefore, $d<e$.
\end{proof}
\vspace{\baselineskip}

Next we show that if $(i,j,j)$ is a shield for $P$, then  $P$ is pumpable with pumping vector $\vpij$.
\begin{sublemma}\label{lem:if j=k then P pumpable}
Let $(i,j,k)$ be a shield for $P$ (as in Definition~\ref{def:shield}). If $j=k$ then $P$ is pumpable.
\end{sublemma}
\begin{proof}
Since $j=k$, the ray $l^k+\vpji = \pos{\glu P k}+\vpji = \pos{\glu P j}+\vpji = \pos{\glu P i} = l^i(0)$, and thus, by Definition~\ref{def:shield}, $\glu P i$ is visible from both the north and the south.
Moreover, $\glu P i$ points east   (Definition~\ref{def:shield}), hence all positions of $\sigma \cup \asm{P_{\range 0 1 i}}$ are strictly to the west of the vertical line through $\pos{\glu P i}$.

Since $j=k$, then $l^j(0)=l^k(0)$ and by Definition~\ref{def:shield} $\glu P k$ is visible from both the north and the south. 
Moreover, by the visibility of $\glu P i$, the path $P_{i+1,i+2,\ldots,k}$ is entirely positioned between the two vertical lines that run through $\glu P i$ and  $\glu P k$. 
As noted after Definition~\ref{def:shield}, $\xcoord{\vect{P_i P_j}} = \xcoord{\vect{P_i P_k}}> 0$.
Hence, for all $n\in\mathbb{N}$, the path $P_{\range {i+1} {i+2} k} + (n+1)\vpij$ is strictly to the east of $P_{\range {i+1} {i+2} k}+ n \vpij$. 
Moreover, for all $n\in\mathbb{N}$, the pair of tiles
$P_{k} + n\vpij$ and
$P_{i+1} + (n+1)\vpij$ interact, 
 hence the pumping of $P$ between $i$ and $k$ is simple and producible (Definition~\ref{def:pumpingPbetweeniandj}). 
 Hence  $P$ is a pumpable path (Definition~\ref{def:pumpable path}) with pumping vector~$\vpij$.
\end{proof}

From now on, we assume that $j<k$, otherwise $P$ is pumpable by the previous claim. We use this to define a second connected component $\mathcal D\subseteq\mathcal C$ (see Figure~\ref{fig:shield-d} for an example), since $j<k$ makes the following curve simple:
\begin{definition}\label{def:d}
  Let $d$ be the curve defined as the concatenation of five curves:
  $$d = \concat{
    \reverse{l^j},  \,\gs{l^j(0)}{\pos{P_{j+1}}},\, \embed{P_{\range {j+1}{j+2}k} }, \,\gs{\pos{P_{k} }}{l^k(0)},\, l^k
  }$$
\end{definition}

  By a similar argument to 
  \subl{lem:c-cuts} for curve $c$ (Definition~\ref{def:c}), 
  here $d$ is a simple connected curve starting with a vertical ray to the south ($l^j$) and ending with a vertical ray to the north ($l^k$). Therefore, by Theorem~\ref{thm:infinite-jordan}, $d$ cuts $\R^2$ into two infinite connected components.   Let $\mathcal{D} \subsetneq \mathbb{R}^2$ be the connected component on the right-hand side of $d$ (intuitively, the component connected to the east of $l^i$ and $l^k$), including $d$ itself.

  Moreover, since 
   each of the five curves used to define $d$  are in $\mathcal C$ (Definition~\ref{def:C}) we get that 
  $d$ is entirely in $\mathcal C$. 
  Hence $\mathcal D$ is a subset of $\mathcal C$.

The following two claims are illustrated in Figure~\ref{fig:shield-d}.

\begin{figure}[ht]
  \centering
  \begin{tikzpicture}[scale=\scale]\draw[draw={rgb,255:red,200; green,200; blue,200}](4.5,0.8) rectangle (30.5, -15.2);
\draw[draw={rgb,255:red,200; green,200; blue,200}](4.5,-15.2)--(4.5,0.8);
\draw[draw={rgb,255:red,200; green,200; blue,200}](5.5,-15.2)--(5.5,0.8);
\draw[draw={rgb,255:red,200; green,200; blue,200}](6.5,-15.2)--(6.5,0.8);
\draw[draw={rgb,255:red,200; green,200; blue,200}](7.5,-15.2)--(7.5,0.8);
\draw[draw={rgb,255:red,200; green,200; blue,200}](8.5,-15.2)--(8.5,0.8);
\draw[draw={rgb,255:red,200; green,200; blue,200}](9.5,-15.2)--(9.5,0.8);
\draw[draw={rgb,255:red,200; green,200; blue,200}](10.5,-15.2)--(10.5,0.8);
\draw[draw={rgb,255:red,200; green,200; blue,200}](11.5,-15.2)--(11.5,0.8);
\draw[draw={rgb,255:red,200; green,200; blue,200}](12.5,-15.2)--(12.5,0.8);
\draw[draw={rgb,255:red,200; green,200; blue,200}](13.5,-15.2)--(13.5,0.8);
\draw[draw={rgb,255:red,200; green,200; blue,200}](14.5,-15.2)--(14.5,0.8);
\draw[draw={rgb,255:red,200; green,200; blue,200}](15.5,-15.2)--(15.5,0.8);
\draw[draw={rgb,255:red,200; green,200; blue,200}](16.5,-15.2)--(16.5,0.8);
\draw[draw={rgb,255:red,200; green,200; blue,200}](17.5,-15.2)--(17.5,0.8);
\draw[draw={rgb,255:red,200; green,200; blue,200}](18.5,-15.2)--(18.5,0.8);
\draw[draw={rgb,255:red,200; green,200; blue,200}](19.5,-15.2)--(19.5,0.8);
\draw[draw={rgb,255:red,200; green,200; blue,200}](20.5,-15.2)--(20.5,0.8);
\draw[draw={rgb,255:red,200; green,200; blue,200}](21.5,-15.2)--(21.5,0.8);
\draw[draw={rgb,255:red,200; green,200; blue,200}](22.5,-15.2)--(22.5,0.8);
\draw[draw={rgb,255:red,200; green,200; blue,200}](23.5,-15.2)--(23.5,0.8);
\draw[draw={rgb,255:red,200; green,200; blue,200}](24.5,-15.2)--(24.5,0.8);
\draw[draw={rgb,255:red,200; green,200; blue,200}](25.5,-15.2)--(25.5,0.8);
\draw[draw={rgb,255:red,200; green,200; blue,200}](26.5,-15.2)--(26.5,0.8);
\draw[draw={rgb,255:red,200; green,200; blue,200}](27.5,-15.2)--(27.5,0.8);
\draw[draw={rgb,255:red,200; green,200; blue,200}](28.5,-15.2)--(28.5,0.8);
\draw[draw={rgb,255:red,200; green,200; blue,200}](29.5,-15.2)--(29.5,0.8);
\draw[draw={rgb,255:red,200; green,200; blue,200}](4.5,0.8)--(30.5,0.8);
\draw[draw={rgb,255:red,200; green,200; blue,200}](4.5,-0.2)--(30.5,-0.2);
\draw[draw={rgb,255:red,200; green,200; blue,200}](4.5,-1.2)--(30.5,-1.2);
\draw[draw={rgb,255:red,200; green,200; blue,200}](4.5,-2.2)--(30.5,-2.2);
\draw[draw={rgb,255:red,200; green,200; blue,200}](4.5,-3.2)--(30.5,-3.2);
\draw[draw={rgb,255:red,200; green,200; blue,200}](4.5,-4.2)--(30.5,-4.2);
\draw[draw={rgb,255:red,200; green,200; blue,200}](4.5,-5.2)--(30.5,-5.2);
\draw[draw={rgb,255:red,200; green,200; blue,200}](4.5,-6.2)--(30.5,-6.2);
\draw[draw={rgb,255:red,200; green,200; blue,200}](4.5,-7.2)--(30.5,-7.2);
\draw[draw={rgb,255:red,200; green,200; blue,200}](4.5,-8.2)--(30.5,-8.2);
\draw[draw={rgb,255:red,200; green,200; blue,200}](4.5,-9.2)--(30.5,-9.2);
\draw[draw={rgb,255:red,200; green,200; blue,200}](4.5,-10.2)--(30.5,-10.2);
\draw[draw={rgb,255:red,200; green,200; blue,200}](4.5,-11.2)--(30.5,-11.2);
\draw[draw={rgb,255:red,200; green,200; blue,200}](4.5,-12.2)--(30.5,-12.2);
\draw[draw={rgb,255:red,200; green,200; blue,200}](4.5,-13.2)--(30.5,-13.2);
\draw[draw={rgb,255:red,200; green,200; blue,200}](4.5,-14.2)--(30.5,-14.2);
\draw[draw=none,fill={rgb,255:red,179; green,179; blue,179},opacity=0.25](13.5,-5.7)--(13.5,-0.2)--(30.5,-0.2)--(30.5,-15.2)--(9.5,-15.2)--(9.5,-12.7)--(10,-12.7)--(10,-9.7)--(17,-9.7)--(17,-12.7)--(19,-12.7)--(19,-9.7)--(20,-9.7)--(20,-7.7)--(23,-7.7)--(23,-10.7)--(25,-10.7)--(25,-5.7)--(20,-5.7)--(20,-4.7)--(17,-4.7)--(17,-3.7)--(19,-3.7)--(19,-1.7)--(16,-1.7)--(16,-5.7)--(13.5,-5.7);
\draw(27.26, -3.7) node[anchor=south west] {$\mathcal D$};
\draw(26.59, -11.3) node[anchor=south west] {$\embed{r}+\vpij$};
\draw[draw={rgb,255:red,200; green,113; blue,55},opacity=0.5,very thick](6,-13.7)--(8,-13.7)--(8,-12.7)--(10,-12.7)--(10,-9.7)--(17,-9.7)--(17,-12.7)--(19,-12.7)--(19,-9.7)--(20,-9.7)--(20,-7.7)--(23,-7.7)--(23,-10.7)--(25,-10.7)--(25,-5.7)--(20,-5.7)--(20,-4.7)--(17,-4.7)--(17,-3.7)--(19,-3.7)--(19,-1.7)--(16,-1.7)--(16,-5.7)--(13,-5.7);
\draw[draw={rgb,255:red,255; green,0; blue,0},opacity=0.5,very thick](9.5,-12.7)--(12,-12.7)--(12,-11.7)--(14,-11.7)--(14,-9.7)--(17,-9.7)--(17,-12.7)--(19,-12.7)--(19,-9.7)--(20,-9.7)--(20,-8.7)--(21,-8.7)--(21,-11.7)--(23,-11.7)--(23,-10.7)--(25,-10.7)--(25,-6.7)--(27,-6.7)--(27,-9.7)--(29,-9.7)--(29,-4.7)--(24,-4.7)--(24,-3.7)--(21,-3.7)--(21,-2.7)--(23,-2.7)--(23,-0.7)--(20,-0.7)--(20,-4.7)--(19,-4.7);
\draw[draw={rgb,255:red,0; green,0; blue,0}](9.5,-12.69)--(9.5,-16.7);
\draw(8.23, -16.93) node[anchor=south west] {$l_j$};
\draw[draw=none,fill={rgb,255:red,28; green,36; blue,31},thin](9.5, -12.69) ellipse (0.1cm and 0.1cm);\draw[draw={rgb,255:red,0; green,0; blue,0}](5.5,-13.7)--(5.5,-16.7);
\draw(4.32, -16.93) node[anchor=south west] {$l_i$};
\draw[draw={rgb,255:red,0; green,0; blue,0}](13.5,-5.7)--(13.5,0.3);
\draw(12.34, -0.81) node[anchor=south west] {$l^k$};
\draw[draw=none,fill={rgb,255:red,28; green,36; blue,31},thin](13.5, -5.7) ellipse (0.1cm and 0.1cm);\draw[draw=none,fill={rgb,255:red,28; green,36; blue,31},thin](5.5, -13.7) ellipse (0.1cm and 0.1cm);
\end{tikzpicture}
  \caption{The component $\mathcal D \subset \mathbb{R}^2$ (shaded) and the embedded binding path $\embed{r}+\vpij$. 
  \subl{lem:P0j is not in D} asserts that no tile of $\sigma \cup \asm{P_{0,1,\ldots, j}}$ is in $\mathcal D$ and 
  \subl{lem:d} asserts that $\embed{r}+\vpij$ is entirely contained in $\mathcal D$.}\label{fig:shield-d}
\end{figure}

\enlargethispage{-2\baselineskip}
\begin{sublemma}\label{lem:P0j is not in D}
  $\dom{\sigma \cup \asm{P_{0,1,\ldots, j}}} \subset (\R^2 \setminus \mathcal{D})$.
\end{sublemma}
\begin{proof} Similar argument to \subl{lem:c}, but respectively substitute $j,d,\mathcal{D}$ for $i,c,\mathcal{C}$ in the proof of \subl{lem:c}.
\end{proof}

\begin{sublemma}
  \label{lem:d}
  $\embed{r}+\vect{P_iP_j}$ is entirely in $\mathcal{D}$.
\end{sublemma}

\begin{proof}
  Assume otherwise, for the sake of contradiction. Since $r_0+\vpij = \pos{P_{j+1}}$ is on the border $d$ of $\mathcal D$ and thus in $\mathcal D$, this would mean that $\embed r +\vpij$ turns left from~$d$ (see Definition~\ref{def:curve turn} for one curve turning from another). That left turn is on one of the five curves that define $d$ (Definition~\ref{def:d}), which we handle in four cases:

  \begin{enumerate}
  \item $\embed r +\vpij$ turns left from $d$ at a position on~$\reverse{l^j}$:  But since $\glu P i$ and $\glu P j$ are both visible relative to $P$, and $r$ is made of positions taken from $P_{i+1,i+2,\ldots,k}$ and $P_{j+1,j+2,\ldots,k}+\vpji$, this implies $\embed{r}$ does not intersect the visibility ray $l^i$ of $\glu P i$, and hence $\embed{r}+\vpij$ does not intersect $l^j$. Thus having the left turn be from $\reverse{l^j}$ is a contradiction.

\newcommand{\lkz}{\ensuremath{l^k_z}}
\item \label{lem:d:proof:case:Pjk} $\embed r +\vpij$ turns left from $d$ at a position on  $\embed{P_{\range {j+1}{j+2}k}}$. Since this case is rather involved, we split it into a number of paragraphs, and each paragraph header gives the core argument being proven in the paragraph.
  We first define the curve $\rho$ as:
  $$
  \rho = \concat{l^i,\gs{l^i(0)}{r_0} , \embed{r} , \gs{ r_{|r|-1} }{ \lkz(0)} , \lkz}
  $$
  where $l^k_z = l^k([z,+\infty[)$ is the ray to the north starting at  $ \lkz(0) = l^k(z) \in\mathbb{R}^2$   which is the point of $l^k$ at the same y-coordinate as $r_{|r|-1}$. 
    
  \paragraph{Setup: Indices $a$ and $b$ such that $\pos{P_a} = r_b+\vpij$ and $\embed{r_{\range 0 1 {b+1}}}+\vpij$ turns left from~$d$.}
  First, $\embed r$ turns left from $d+\vpji $ at a position on 
  $\embed{P_{\range {j+1}{j+2}k} +\vpji }$.
The curve $\rho$ is a simple infinite almost-vertical polygonal curve (Definition~\ref{def:simple infinite almost-vertical polygonal curve}) -- in particular, it is simple because its five component curves intersect each other only at their endpoints in the order given.
Using Definition~\ref{def:curve LHS RHS},
let $\rhsr \subsetneq \mathbb R$ be the right-hand side of $\rho$, and 
let $\srhsr$ denote the strict right-hand side 
of~$\rho$.

Since we are in Case~\ref{lem:d:proof:case:Pjk}, $\embed r$ has a point in the strict left hand side (Definition~\ref{def:curve LHS RHS}) of the simple infinite almost-vertical polygonal curve $d+\vpji$, with the left turn being on \embed{P_{\range {j+1}{j+2}k}}.
 Hence, let $b \in \{\range{0}{1}{|r|-1} \}$ be the smallest integer such that 
 there is an $a\in \{\range{j+1}{j+2}{k}\}$ where $r_b+\vpij = \pos{P_a}$ 
 and $\embed{ r_{0,1,\ldots,b+1} } +\vpij$ has a point on the strict left hand side of curve $d$.
 Hence, $r_{0,1,\ldots,b}$ is entirely in the right hand side of $d + \vpji$, 
 and the point $\hp{r}{b}{b+1}$ is in the strict left hand side of $d + \vpji$.  

 \paragraph{We claim that $a<k$.}
 Suppose for the sake of contradiction that $a=k$. Then $r_b=\pos{P_k}+\vpji$. 
Since $\embed{r}$ makes a left turn from $d+\vpji$ at $\pos{P_k}+\vpji$, 
and since $\embed{r_{0,1,\ldots,b}}$ is on the right hand side of $d+\vpji$, 
$\embed{r}$ makes a left turn from the curve
   $$
   \rho' = \concat{l^i,\gs{l^i(0)}{r_0} , \embed{r_{0,1,\ldots,b}}, \gs{ \pos{P_k}+\vpji }{ l^k(0) +\vpji}, l^k+\vpji}  
   $$
   at the position $r_b$.
 By the definition of $r$,  we know that $r_{b+1,b+2,\ldots,|r|-1}$ has its last point within horizontal distance $\pm 0.5$ of the ray $l^k$. 
But any such position is in the strict right hand side of $\rho'$ because:
 (a) since $i<j$, $l^k$ is at least distance $1$ to the east of $l^k+\vpji$, 
  (b) since $i<k$ (because $i< j \leq k$), and by visibility, $l^k$ does not intersect $l^i$, and
 (c) $\embed{r}$ does not intersect $l^k$ by definition of $r$.  Hence $\embed{r_{b+1,b+2,\ldots,|r|-1}}$ intersects $\rho'$, 
 but that intersection can not happen along $l^i$ (by its definition $\embed{r}$ does not intersect $l^i$)
 nor along $\embed{r_{0,1,\ldots,b}}$ (because $r$ is simple), 
 nor at $l^k(z) + \vpji$ for all $z\in\mathbb{R}, z>0$ (because, by Definition~\ref{def:shield}, the curves $\embed{P_{i+1,i+2,\ldots,k}}$ and $\embed{P_{j+1,j+2,\ldots,k}}+\vpji$ from which $\embed{r}$ is composed do not touch $l^k(z) + \vpji$),
 nor $\gs{ \pos{P_k}+\vpji }{ l^k(0) +\vpji} $because $r$ turns left from $d+\vpji$ at $r_b= \pos{P_k}+\vpji$. 
 Thus we get a contraction. 
Hence  $a<k$ as claimed.

\paragraph{We claim that $\embed{ P_{j+1,j+2,\ldots,k}+\vpji }$ has a point in $\srhsr$ (the strict right hand side of $\rho$).}

We first treat the case $b=0$ as a special case: in this case, $r_b=r_0 = \pos{P_{j+1}} +\vpji $ and, by the statement of the case we are in (Case~\ref{lem:d:proof:case:Pjk}), the point $\hp{P}{j+1}{j+2}+\vpji$ is in the right hand side of $\rho$.
The point $\hp{P}{j+1}{j+2}+\vpji$ is not on $\rho$, because $\hp{r}{0}{1}\neq \hp{P}{j+1}{j+2}+\vpji$ (because of the turn at $r_0$)
and because if it were $r$ would not be simple (there would be an intersection of $r_{1,2,\ldots,|r|-1}$ with $r_0$). 
Hence  $\hp{P}{j+1}{j+2}+\vpji$ is in the strict right hand side of~$\rho$.
Thus if $b=0$, we are done with our current argument, i.e. $\embed{ P_{j+1,j+2,\ldots,k}+\vpji }$ has a point in $\srhsr$.

Else $b>0$. 
There are two possibilities for $r_{b-1}$. First, suppose  
 $r_{b-1} = \pos{P_{a-1}}+\vpji$. 
 Since $r_{b-1,b} = \pos{P_{a-1,a}}+\vpji$, we get that $r_{b-1,b,b+1}$ turns left from\footnote{\label{fn:a b > 0}We've already shown that $a<k$, hence the tile $P_{a+1} +\vpji$ is a tile of $P_{j+1,j+2,\ldots,k} +\vpji$ and thus is well-defined. Also, $b>0$ implies $a \neq j+1$ which in turn implies $a>j+1$, hence $P_{a-1} +\vpji$ is a tile of $P_{j+1,j+2,\ldots,k} +\vpji$ and thus is well-defined.} 
 $\pos{P_{a-1,a,a+1}+\vpji}$ and thus 
 the point $\hp{P}{a}{a+1}+\vpji$ is in $\srhsr$ 
 and we are done with our current argument, i.e. that $\embed{ P_{j+1,j+2,\ldots,k}+\vpji }$ has a point in $\srhsr$.

Otherwise, we have $r_{b-1} \neq \pos{P_{a-1}}+\vpji$, which we split into two cases:
\begin{itemize} 
\item $r_{b-1} \neq \pos{P_{a+1}}+\vpji$. 

\newcommand{\Rset}{\mathbb{R}}
In this case, the four positions around $P_a+\vpji=r_b$ are each occupied by one of the following tiles $P_{a-1}+\vpji$, $P_{a+1}+\vpji$, $r_{b-1}$ and $r_{b+1}$ (since $r$ and $P$ are simple and by the hypothesis of this case). Consider the enumeration of these four tiles around $P_a+\vpji=r_b$ in clockwise order starting with $P_{a-1}+\vpij$. By definition of $b$, $r_{b+1}$ comes before $P_{a+1}+\vpji$ and by minimality of $b$, $r_{b-1}$ cannot come before $P_{a+1}+\vpij$, then the four tiles are ordered as follow: $$P_{a-1}+\vpji, r_{b+1}, P_{a+1}+\vpji,r_{b-1}$$ which means that $P_{a-1,a,a+1}+\vpji$
turns right form $r$ at $P_a+\vpji=r_b$.

Then, the point $\hp{P}{a}{a+1}+\vpji$ is on the strict right hand side of $\rho$ (note in particular that $\hp{P}{a}{a+1}+\vpji$ is not on $l^k_z$, for else we would get an extension of $r$ that would be in $\overline{\mathbb S}$ and end higher than $r$, contradicting the fact that by definition, $r\in\mathbb S$).
Thus $\embed{ P_{j+1,j+2,\ldots,k}+\vpji }$ has a point in $\srhsr$, which is the claim of this paragraph.

\item $r_{b-1} = \pos{P_{a+1}}+\vpji$.
There are two cases:

\begin{list}{$\ast$}{} 
\item $b - (k-a) > 0$.
We claim that  $\pos{ P_{a,a+1,\ldots,k} + \vpji } = r_{b,b-1,\ldots,b - (k-a)}$ (intuitively, we are claiming that $r_{b,b-1,\ldots,b - (k-a)}$ tracks ``backwards'' along  $\pos{ P_{a,a+1,\ldots,k} + \vpji }$, without either turning from the other). 
Since $r_{0,1,\ldots b}$ is on the right hand side of $d+\vpji$, $\embed{r_{b,b-1,\ldots,b - (k-a)}}$ has all its points on the right hand side of $d+\vpji$ (and none on the strict left hand side of $d+\vpji$).
Hence if  $\embed{r_{b,b-1,\ldots,b - (k-a)}}$ turns from $d+\vpji$, it turns to the right from $d+\vpji$ (and in particular has points in the strict right hand side of $d+\vpji$). 

But this means that $d+\vpji$ has points on the strict left hand side of $\reverse{\rho}$, or in other words $d+\vpji$ has points
on the strict right hand side of $\rho$. This yields the claim of this paragraph.

\item $b - (k-a) \leq 0$. There are two cases, depending on whether or not  $r_{b,b-1,\ldots,0}$  turns from $\pos{ P_{a,a+1,\ldots,k} + \vpji }$. First, if $r_{b,b-1,\ldots,0}$  turns (right or left) from $\pos{ P_{a,a+1,\ldots,k} + \vpji }$, then
by using the same argument as the previous bullet, we get the claim of this paragraph.

Finally, assume for the sake of contradiction, that $r_{b,b-1,\ldots,0}$ does not turn from $\pos{ P_{a,a+1,\ldots,k} + \vpji }$. We have $r_0 = P_{m}+\vpji$ for some $m\in \{ a,a+1,\ldots,k \}$.
Since we are in a case where $b>0$, we also get\footnoteref{fn:a b > 0} that $j+1 < a $,
hence $j+1 \neq m$.
But, by the definition of $r$, $r_0 =  \pos{P_{j+1}}+\vpji$, hence $j+1= m$, which is a contradiction, hence $r_{b,b-1,\ldots,0}$ must turn from $\pos{ P_{a,a+1,\ldots,k} + \vpji }$ we get a contradiction.
\end{list}
\end{itemize}

Hence we have proven that $\embed{ P_{j+1,j+2,\ldots,k}+\vpji }$ has a point in $\srhsr$.

\paragraph{$P_{\range{j+1}{j+2}k}+\vpji$ does not intersect $\rho$ again after entering $\srhsr$, yielding a contradiction.}
We have already proven that 
$\embed{ P_{j+1,j+2,\ldots,k}+\vpji }$ has a point in $\srhsr$ (previous paragraph). 
Using that fact, let $a' \in \{j+1,j+2,\ldots, k-1\}$, $b' \in \{0,1,\ldots,|r|-1\}$ be so that $b'$ is the smallest integer such that 
$r_{b'} = \pos{P_{a'}} + \vpji$ 
and $\pos{P_{j+1,j+2,\ldots,a'+1}} + \vpji$ turns right from $\rho$ at position $r_{b'}$. Hence $\pos{\glue{P}{a'}{a'+1}+\vpji}$ is in $\srhsr$ (since $\embed{P_{j+1,j+2,\ldots,a'+1}} + \vpji$ cannot intersect with $l^i$ and $l^k$).

Claim~\ref{lem:c-lk},  states that
$
 (( l^k+\vpji)  \cap \mathcal{C} ) 
  = 
 ( (l^k+\vpji) \cap c )
 \subseteq 
 \{l^k(0)+\vpji \}
 $ 
which in turn implies that the point $\pos{\glu P k}+\vpji = l^k(0)+\vpji$ is not in $\mathcal C\setminus c$, and therefore not in $\mathcal R\setminus c$ (recall that $\mathcal R \subset \mathcal C$).
Therefore, $\embed {P_{\range{a'}{a'+1}{k+1}}+\vpji}$ intersects $\rho$
for some real number $t>0$ at a point $Z = \embed {P_{\rng {a'} k}+\vpji}(t) \in\mathbb{R}^2$.
\footnote{Remember that $\embed {P_{\range{a'}{a'+1}{k+1}}+\vpji}$ has domain $[0,k+1-a']$.}
We analyse where such an intersection point $Z$ might occur on $\rho$:
\begin{itemize}
  \item $Z$ is not on $l^i$ nor $\gs{l^i(0)}{r_0}$:  
    Since $j+1\leq a'$, the curve $\embed{P_{\rng {a'} k}}$ does not intersect $l^j$ by visibility, and thus $\embed{ P_{\rng {a'} k}+\vpji}$ does not intersect $l^i$. 
    Also, $Z$ is not on $\gs{l^i(0)}{r_0} = \gs{l^i(0)}{\pos{P_{j+1}+\vpji}}$, since $j+1\leq a'$, $t>0$ and $P$ is simple.
  \item If the first such intersection point $Z$ along $\embed{P_{\rng {a'} k}+\vpji}$ (i.e., smallest $t>0$)
   is with $l^k$ or $\gs{r_{|r|-1}}{l^k_z(0)}$,
    then let $r' = r_{\range 0 1 {b'}}\,\pos{P_{\range {a'+1}{a'+2}{g}}+\vpji}$ where the index $g \in \{ a'+1,a'+2,\ldots,k \}$ is defined such that $\gs{\pos{P_g}}{\pos{P_{g+1}}} + \vpji$ intersects $\gs{\pos{P_k}}{l^k(z)}\cup l^k$.

    Then, by the definition of $a'$ and $b'$, either $r$ is a prefix of the binding path $r'$ (if $b'=|r|-1$) or $r'$ turns right from $r$ at $r_{b'}=\pos{P_{a'}+\vpji}$ (if $b'<|r|-1$). In the former case, this contradicts the definition of $r$, since then $r'\in\mathbb S$ and $r\not\in\mathbb S$. In the latter case, this contradicts that $r$ is the most right-priority path of $\mathbb S$, since $r'\in\mathbb S$ is more right-priority.

  \item Else, the first (i.e., smallest $t > 0$) such intersection point $Z$ along $\embed{P_{\rng {a'} k}+\vpji}$ is on $\embed{r}$, in other words
   $P_{\range{a'+1}{a'+2}{k}}+\vpji$ has a position on $r$.
	 Since $\pos{P_{a'}} +\vpji = r_{b'}$, we know that  $P_{\range{a'+1}{a'+2}{k}}+\vpji$ does not intersect $r_{b'}$, and we have two subcases: 
        \begin{itemize}
    \item If $P_{\range {a'+1}{a'+2}{k}}+\vpji$ intersects $r_{\range {b'+1}{b'+2}{|r|-1}}$, let 
    $g \in \{ \range{a'+1}{a'+2}{k} \}$ be the smallest index such that $\pos{P_g+\vpji} = r_h$ for some $h \in \{ \range{b'+1}{b'+2}{|r|-1} \}$.

      Then $r' = r_{\range 0 1 {b'}}\, \pos{ P_{\range {a'+1}{a'+2}{g-1}}+\vpji}\, r_{\range{h}{h+1}{|r|-1}}$ 
      turns right from $r$ by definition of $a'$ and $b'$.\footnote{In the special case of $g=a'+1$, 
      although $\pos{ P_{\range{a'+1}{a'+2}{g-1}}+\vpji}$ is empty, 
      the segment $\gs{r_{b'}}{r_h}$ is not entirely on $\embed{r}$ (only its endpoints are). 
      This follows from 
      the fact that $b'<h$, i.e., $b' \neq h$ (due to the existence of the turn from $r$ at $r_{b'}$).
      Moreover,  $\gs{r_{b'}}{r_h}$ is entirely in $\mathcal R \subset \mathcal C$ 
      since $r_{b'-1,b',h}$ is a right turn from $r$ at $r_{b'}$.}
      This right turn contradicts the definition of $r$ as the most right-priority binding-graph path that starts at $r_0$, ends adjacent to $l^k$ ($r_{|r|-1}$ is horizontal distance $\pm 0.5$ from $l^k$), and whose embedding stays entirely in~$\mathcal C$.

    \item If $P_{\range {a'+1}{a'+2}{k}}+\vpji$ intersects $r_{\range 0 1 {b'-1}}$,
      then let $g \in \{ a'+1,a'+2, \ldots, k\}$ be the smallest index such that $\pos{P_g +\vpji}= r_h$ for some index $h \in \{0,1,\ldots,b'-1\} $, and let
      $r' = r_{\range 0 1 h}\pos{ (P_{\range{a'+1}{a'+2}{g-1}}+\vpji)^\leftarrow}r_{\range {b'} {b'+1} {|r|-1}}$.\footnote{Here the ``reverse arrow'' notation applied to a path denotes the path in reverse order, i.e. $ (P_{\range {a'}{a'+1}{g-1}})^\leftarrow+\vpji = P_{g-1}  P_{g-2}  \ldots P_{a'+1} P_{a'} +\vpji $. 
      }
      \footnote{In the special case of $g=a'+1$, 
      although $\pos{ (P_{\range{a'+1}{a'+2}{g-1}}+\vpji)^\leftarrow}$ is empty, 
      the segment $\gs{r_h}{r_b}$ is not entirely on $\embed{r}$ (although its endpoints are).
      This follows from the fact that $h<b'$, i.e., $h \neq b'$ (due to the existence of the turn from $r$ at $r_h$). 
            Moreover, the segment $\gs{r_h}{r_{b'}}$ is entirely in $\mathbb R \subset \mathbb C$ since $r_{h-1,h,{b'}}$ is a right turn from $r$ at $r_h$.}

      Since $\pos{\glue P {a'} {a'+1}+\vpji}\in\mathcal R\setminus \embed{r}$, then $r'$ turns right from $r$. 
      Also, $\embed{r'}$ is in $\mathcal R\subseteq \mathcal C$, since $\rev{\embed{P_{\range {a'+1}{a'+2}{g-1}}+\vpji}}$  does not intersect $\embed{r}$ and all other points of $\embed{r'}$ are in $\mathcal R$ since they are taken from $\embed{r}$ (plus two length 1 line segments that are in $\mathcal R$). 
      Moreover, $r'$ starts and ends at the same positions as $r$ does, which contradicts the definition of $r$ as the most right priority such path.
    \end{itemize}
  \end{itemize}

  In all cases we contradict the existence of the intersection point $Z$ on $\rho$, and hence
  $\embed r +\vpij$ does not turn left from $d$ at a position on  $\embed{P_{\range {j+1}{j+2}k}}$.

   \item $\embed r +\vpij$ turns left from $d$ at a position on~$\gs{l^j(0)}{\pos{P_{j+1}}}$: since $r$ is a binding path in graph $G$, this can only happen at a tile position, i.e., at $\pos{P_{j+1}}$, and this case was already covered in  Case~\ref{lem:d:proof:case:Pjk}.

  \item\label{r+PiPj is in D: case 4 of DW} $\embed r +\vpij$ turns left from $d$ at a position on~$\concat{ \gs{\pos{P_k}}{l^k(0)} , l^k }$.
       If $\embed r +\vpij$ turns left from $d$ at~$\pos{P_k}$  this was already covered in Case~\ref{lem:d:proof:case:Pjk}. 
  
  By Hypothesis~\ref{lem:hp:shield backup} of Definition~\ref{def:shield}, $P$ may only intersect $l^k+\vpji$ at $l^k(0)+\vpji$, and 
  by Hypothesis~\ref{lem:hp:shield k} of Definition~\ref{def:shield} $P$ may only intersect $l^k$ at $l^k(0)$.  
 Therefore, $P$ and $P+\vpji$ can only intersect $l^k+\vpji$ at $l^k(0)+\vpji$, and $\embed{r}+\vpij$ can only intersect $l^k$ at $l^k(0)$,
 thus for all $z \in\mathbb{R}, z>0$ the left turn does {\em not} occur at $l^k(z)$.
 
It remains to handle the case of the left turn of $\embed{r}+\vpij$ from $d$ being along the half-open segment $\left( \pos{P_k} , l^k(0)  \right]$. 
  Since $r$ is a binding path in graph $G$,  an intersection of $\embed{r} +\vpij$ and $l^k(0) = \pos{\glu P k}$
  can only occur if $\embed{r} +\vpij$
  contains exactly one of the unit-length horizontal line segments  $\gs { \pos{P_{k+1}}}{\pos{P_k} }$ or $\gs { \pos{P_k}}{\pos{P_{k+1}} }$. 
  There are two cases. 
  \begin{itemize}
    \item If $\glu P k$ points to the east, then 
             $\gs { \pos{P_{k+1}}}{\pos{P_k} }$ is not (does not contain) a left turn from $\concat{ \gs{\pos{P_k}}{l^k(0)} , l^k }$, 
             and $\gs { \pos{P_k}}{\pos{P_{k+1}} }$ is a right turn, from $\concat{ \gs{\pos{P_k}}{l^k(0)} , l^k }$.
             Thus, in this case  we contradict that $\embed r +\vpij$ turns left from $d$ at a point on~$\concat{ \gs{\pos{P_k}}{l^k(0)} , l^k }$.
    \item Else $\glu P k$ points to the west. 
    Since $\gs { \pos{P_{k+1}}}{\pos{P_k} }$ is not (does not contain) a left turn from $\concat{ \gs{\pos{P_k}}{l^k(0)} , l^k }$
    and since 
  $\gs { \pos{P_k}}{\pos{P_{k+1}} }$ is a left turn from $\concat{ \gs{\pos{P_k}}{l^k(0)} , l^k }$, we know that $\embed{r}+\vpij$
  contains $\gs { \pos{P_k}}{\pos{P_{k+1}} }$. Thus  $\embed{r}$ contains the segment $s = \gs { \pos{P_k}}{\pos{P_{k+1}} } + \vpji$. 
  The segment $s$ is entirely in $\mathcal C$ because $\embed{r}$ is (by definition of $r$). 
  We also claim that the segment $s$ has no points in $\mathcal C\setminus c$: 
        By Claim~\ref{lem:c-lk},  
      the midpoint of $s$ is not in $\mathcal C \setminus c$, 
      and hence that midpoint is on $c$.
   The segment $s$ is not on $l^i$ nor $l^k$ (because $\embed{r}$ does not intersect $l^i$ nor $l^k$).
  Hence the midpoint of $s$ is on $\embed{P_{i+1,i+2,\ldots,k}}$, 
  a curve whose unit length horizontal segments start and end at integer coordinates and   
  thus $s$ has all of its points on $c$.
    By \subl{lem:order}, $s$ (a segment of $\embed{r}$) and $\embed{P_{i+1,...,k}}$ have the same curve-direction (points along both curves occur in the same order).
    
  We claim that  $s$ being a segment of $\embed{P_{i+1,i+2,\ldots,k}}$ yields a contradiction: 
Since $s$ is a horizontal segment ``pointed'' to the west, $l^k+\vpji$ is a right turn from $s$.  
But  $l^k(z)+\vpji$, for all $z\in\mathbb{R}, z>0$ is on the left hand side of $c$, and thus   $l^k+\vpji$ is a left turn from $\embed{P_{i+1,i+2,\ldots,k}}$, yielding the claimed contradiction.
Thus $\embed r +\vpij$ does not turn left from $d$ at a point on~$\concat{ \gs{\pos{P_k}}{l^k(0)} , l^k }$.
  \end{itemize}
 
     \end{enumerate}
Each of the four cases contradicts the claim that $\embed r +\vpij$ turns left from $d$.  
\end{proof}

\paragraph{The path $R$.}
We now define a path $R$ capable of growing in two translations, as stated in the following \sublname:

\begin{sublemma}
  \label{lem:r}
  There is a path $R$ such that all of the following hold:
  \begin{itemize}
  \item $P_{\range 0 1 i}R$ is a producible path, and
  \item $P_{\range 0 1 j}(R+\vpij)$ is a producible path, and
  \item Exactly one of the following is the case:
    \begin{itemize}
    \item $\pos R = r$, and $R$ does not conflict with
      $P_{i+1,i+2,\ldots,k}$ nor $P_{j+1,j+2,\ldots,k} + \vpji$.
    \item $\pos R$ is a strict prefix of $r$ and $R$ or $R+\vpij$ (or both) conflict with $P$, meaning in particular that $P$ is fragile.
    \end{itemize}
  \end{itemize}
\end{sublemma}
\begin{proof}
  We first define the length $s_0$ of $R$.
  If there is an integer $s\in\{\range 0 1 {|r|-1}\}$ such that
  $r_s=\pos{P_{a}} = \pos{P_{b}+\vpji}$ for some $a\in\{\range{i+1}{i+2}k\}$ and $b\in\{\range{j+1}{j+2}k\}$,
  and $\type{P_a}\neq\type{P_b}$, then let $s_0$ be the smallest such $s$ (and note that, in this case $s_0$ is an index on $r$ and $r_{s_0}$ is the position of the conflict).
  Else, we simply let $s_0 = |r|$\footnote{Note that, in particular, $s_0$ is not an index on $r$ (nor $R$) since the maximum index on $R$ is $|r|-1$}.

  We next define $R$ to be of length $s_0$, and for all $s\in\{\range 0 1 {s_0-1}\}$, we define the tile $R_s$ as follows:
  \begin{itemize}
  \item If there is an index $a\in\{\range{i+1}{i+2}k\}$ such that $r_s = \pos{P_a}$ then we let $R_s = P_a$.
  \item Else, by definition of $r$, there is an index $a$ such that $r_s = \pos{P_{a}}+\vpji$. In this case, we let $R_s=P_a+\vpji$.
  \end{itemize}
  
  \vspace{0.5\baselineskip}
  \noindent Regardless of whether or not $s_0<|r|$, we prove the following two claims:
  \begin{itemize}
  \item First, we claim that $P_{\range 0 1 i}R$ is a producible path:
    \begin{itemize}
    \item We  claim that the glues along $R$ match: indeed, by definition of the graph $G$ in which $r$ is a path, for any $a\in\{\range 0 1 {|r|-2}\}$, there is an edge $\{ r_a, r_{a+1} \}$ in~$G$.
      Moreover, since $\asm{P_{\range {i+1}{i+2}k}}$ and $\asm{P_{\range {j+1}{j+2}k}+\vpji}$ agree on all tiles that happen to share positions of $r_{\range 0 1{s_0-1}}$ the glues along $R$ match.

    \item Moreover, we claim that $P_{\range 0 1 i}R$ is simple and producible: indeed, by definition of $r$, and since the positions of $R$ are exactly $r$ (i.e., $\pos R = r$), $R$ is entirely in $\mathcal C$, and since neither $\sigma \cup \asm{P_{\range 0 1 i}}$  has no tile in $\mathcal C$ (\subl{lem:c}),  $\sigma\cup P_{\range 0 1 i}$ does not conflict with $R$. Finally, 
          since $P_i$ and $R_0 = P_{i+1}$ interact, $P_{\range 0 1 i}R$ is a producible path.
    \end{itemize}
  \item Next, we claim that $P_{\range 0 1 j}(R+\vpij)$ is also a producible path.
    First, we already proved in the previous bullet that $R$ has matching glues and is simple.
    By \subl{lem:d},
        $R$ is in $\mathcal D$ (note that if $\mathcal{D}$ is not correctly defined, \emph{i.e.} $j=k$, then $P$ is pumpable by lemma \ref{lem:if j=k then P pumpable}), and since $\sigma\cup P_{\range 0 1 j}$ is in $\R^2\setminus\mathcal D$ (Claim~\ref{lem:P0j is not in D}), $R+\vpij$ does not intersect $\sigma\cup \asm{ P_{\range 0 1 j}}$.     Moreover this implies that $P_{\range 0 1 j}(R+\vpij)$ is simple. Finally, since $P_j$ and $R_0+\vpij=P_{j+1}$ interact, $P_{\range 0 1 j}(R+\vpij)$ is producible.
    \item Finally, $R$ conflicts with neither $P_{i+1,i+2,\ldots,k}$ and $P_{j+1,j+2,\ldots,k} + \vpji$ since by the definition of $R$ the index $s_0$ (the first index of a potential tile conflict along the positions of $r$) is not an index on $R$.
  \end{itemize}

\vspace{0.5\baselineskip}
  \noindent 
  Then, there are two cases, depending on whether $s_0=|r|$ or $s_0<|r|$:
  \begin{itemize}
  \item If $s_0 = |r|$, we are done, since in this case $\pos R = r$, and we have already proved the other three conclusions of this lemma for the case where we are not showing $P$ to be fragile.
  \item Else $s_0 < |r|$,  and we claim that $P$ is fragile. By the definition of $r_s$ we have that $P_{i+1,i+2,\ldots,k}$ and $P_{j+1,j+2,\ldots,k} +\vect{P_j P_i}$ have a conflict at position $r_{s_0}$.
  There are two cases:

    \begin{figure}[ht]
      \centering
      \begin{tikzpicture}[scale=\scale]\draw[draw={rgb,255:red,200; green,200; blue,200}](4.5,-0.2) rectangle (26.5, -15.2);
\draw[draw={rgb,255:red,200; green,200; blue,200}](4.5,-15.2)--(4.5,-0.2);
\draw[draw={rgb,255:red,200; green,200; blue,200}](5.5,-15.2)--(5.5,-0.2);
\draw[draw={rgb,255:red,200; green,200; blue,200}](6.5,-15.2)--(6.5,-0.2);
\draw[draw={rgb,255:red,200; green,200; blue,200}](7.5,-15.2)--(7.5,-0.2);
\draw[draw={rgb,255:red,200; green,200; blue,200}](8.5,-15.2)--(8.5,-0.2);
\draw[draw={rgb,255:red,200; green,200; blue,200}](9.5,-15.2)--(9.5,-0.2);
\draw[draw={rgb,255:red,200; green,200; blue,200}](10.5,-15.2)--(10.5,-0.2);
\draw[draw={rgb,255:red,200; green,200; blue,200}](11.5,-15.2)--(11.5,-0.2);
\draw[draw={rgb,255:red,200; green,200; blue,200}](12.5,-15.2)--(12.5,-0.2);
\draw[draw={rgb,255:red,200; green,200; blue,200}](13.5,-15.2)--(13.5,-0.2);
\draw[draw={rgb,255:red,200; green,200; blue,200}](14.5,-15.2)--(14.5,-0.2);
\draw[draw={rgb,255:red,200; green,200; blue,200}](15.5,-15.2)--(15.5,-0.2);
\draw[draw={rgb,255:red,200; green,200; blue,200}](16.5,-15.2)--(16.5,-0.2);
\draw[draw={rgb,255:red,200; green,200; blue,200}](17.5,-15.2)--(17.5,-0.2);
\draw[draw={rgb,255:red,200; green,200; blue,200}](18.5,-15.2)--(18.5,-0.2);
\draw[draw={rgb,255:red,200; green,200; blue,200}](19.5,-15.2)--(19.5,-0.2);
\draw[draw={rgb,255:red,200; green,200; blue,200}](20.5,-15.2)--(20.5,-0.2);
\draw[draw={rgb,255:red,200; green,200; blue,200}](21.5,-15.2)--(21.5,-0.2);
\draw[draw={rgb,255:red,200; green,200; blue,200}](22.5,-15.2)--(22.5,-0.2);
\draw[draw={rgb,255:red,200; green,200; blue,200}](23.5,-15.2)--(23.5,-0.2);
\draw[draw={rgb,255:red,200; green,200; blue,200}](24.5,-15.2)--(24.5,-0.2);
\draw[draw={rgb,255:red,200; green,200; blue,200}](25.5,-15.2)--(25.5,-0.2);
\draw[draw={rgb,255:red,200; green,200; blue,200}](4.5,-0.2)--(26.5,-0.2);
\draw[draw={rgb,255:red,200; green,200; blue,200}](4.5,-1.2)--(26.5,-1.2);
\draw[draw={rgb,255:red,200; green,200; blue,200}](4.5,-2.2)--(26.5,-2.2);
\draw[draw={rgb,255:red,200; green,200; blue,200}](4.5,-3.2)--(26.5,-3.2);
\draw[draw={rgb,255:red,200; green,200; blue,200}](4.5,-4.2)--(26.5,-4.2);
\draw[draw={rgb,255:red,200; green,200; blue,200}](4.5,-5.2)--(26.5,-5.2);
\draw[draw={rgb,255:red,200; green,200; blue,200}](4.5,-6.2)--(26.5,-6.2);
\draw[draw={rgb,255:red,200; green,200; blue,200}](4.5,-7.2)--(26.5,-7.2);
\draw[draw={rgb,255:red,200; green,200; blue,200}](4.5,-8.2)--(26.5,-8.2);
\draw[draw={rgb,255:red,200; green,200; blue,200}](4.5,-9.2)--(26.5,-9.2);
\draw[draw={rgb,255:red,200; green,200; blue,200}](4.5,-10.2)--(26.5,-10.2);
\draw[draw={rgb,255:red,200; green,200; blue,200}](4.5,-11.2)--(26.5,-11.2);
\draw[draw={rgb,255:red,200; green,200; blue,200}](4.5,-12.2)--(26.5,-12.2);
\draw[draw={rgb,255:red,200; green,200; blue,200}](4.5,-13.2)--(26.5,-13.2);
\draw[draw={rgb,255:red,200; green,200; blue,200}](4.5,-14.2)--(26.5,-14.2);
\draw[draw={rgb,255:red,200; green,113; blue,55},opacity=0.5,thick](6,-13.7)--(8,-13.7)--(8,-12.7)--(10,-12.7)--(10,-9.7)--(17,-9.7)--(17,-12.7)--(19,-12.7)--(19,-9.7)--(20,-9.7)--(20,-7.7)--(23,-7.7)--(23,-10.7)--(25,-10.7)--(25,-5.7)--(20,-5.7)--(20,-4.7)--(17,-4.7)--(17,-3.7)--(19,-3.7)--(19,-1.7)--(16,-1.7)--(16,-5.7)--(13,-5.7);
\draw[draw={rgb,255:red,0; green,0; blue,0},fill={rgb,255:red,200; green,113; blue,55},opacity=0.5,fill opacity=0.5](5.65,-13.35) rectangle (6.35, -14.05);
\draw[draw={rgb,255:red,0; green,0; blue,0},fill={rgb,255:red,200; green,113; blue,55},opacity=0.5,fill opacity=0.5](6.65,-13.35) rectangle (7.35, -14.05);
\draw[draw={rgb,255:red,0; green,0; blue,0},fill={rgb,255:red,200; green,113; blue,55},opacity=0.5,fill opacity=0.5](7.65,-13.35) rectangle (8.35, -14.05);
\draw[draw={rgb,255:red,0; green,0; blue,0},fill={rgb,255:red,200; green,113; blue,55},opacity=0.5,fill opacity=0.5](7.65,-12.35) rectangle (8.35, -13.05);
\draw[draw={rgb,255:red,0; green,0; blue,0},fill={rgb,255:red,200; green,113; blue,55},opacity=0.5,fill opacity=0.5](8.65,-12.35) rectangle (9.35, -13.05);
\draw[draw={rgb,255:red,0; green,0; blue,0},fill={rgb,255:red,200; green,113; blue,55},opacity=0.5,fill opacity=0.5](9.65,-12.35) rectangle (10.35, -13.05);
\draw[draw={rgb,255:red,0; green,0; blue,0},fill={rgb,255:red,200; green,113; blue,55},opacity=0.5,fill opacity=0.5](9.65,-11.35) rectangle (10.35, -12.05);
\draw[draw={rgb,255:red,0; green,0; blue,0},fill={rgb,255:red,200; green,113; blue,55},opacity=0.5,fill opacity=0.5](9.65,-10.35) rectangle (10.35, -11.05);
\draw[draw={rgb,255:red,0; green,0; blue,0},fill={rgb,255:red,200; green,113; blue,55},opacity=0.5,fill opacity=0.5](9.65,-9.35) rectangle (10.35, -10.05);
\draw[draw={rgb,255:red,0; green,0; blue,0},fill={rgb,255:red,200; green,113; blue,55},opacity=0.5,fill opacity=0.5](10.65,-9.35) rectangle (11.35, -10.05);
\draw[draw={rgb,255:red,0; green,0; blue,0},fill={rgb,255:red,200; green,113; blue,55},opacity=0.5,fill opacity=0.5](11.65,-9.35) rectangle (12.35, -10.05);
\draw[draw={rgb,255:red,0; green,0; blue,0},fill={rgb,255:red,200; green,113; blue,55},opacity=0.5,fill opacity=0.5](12.65,-9.35) rectangle (13.35, -10.05);
\draw[draw={rgb,255:red,0; green,0; blue,0},fill={rgb,255:red,200; green,113; blue,55},opacity=0.5,fill opacity=0.5](13.65,-9.35) rectangle (14.35, -10.05);
\draw[draw={rgb,255:red,0; green,0; blue,0},fill={rgb,255:red,200; green,113; blue,55},opacity=0.5,fill opacity=0.5](14.65,-9.35) rectangle (15.35, -10.05);
\draw[draw={rgb,255:red,200; green,113; blue,55},opacity=0.5](6,-13.7)--(8,-13.7)--(8,-12.7)--(10,-12.7)--(10,-9.7)--(15,-9.7);
\draw[draw={rgb,255:red,0; green,0; blue,0},fill={rgb,255:red,255; green,0; blue,0},opacity=0.5,fill opacity=0.5](5.65,-13.35) rectangle (6.35, -14.05);
\draw[draw={rgb,255:red,0; green,0; blue,0},fill={rgb,255:red,255; green,0; blue,0},opacity=0.5,fill opacity=0.5](6.65,-13.35) rectangle (7.35, -14.05);
\draw[draw={rgb,255:red,0; green,0; blue,0},fill={rgb,255:red,255; green,0; blue,0},opacity=0.5,fill opacity=0.5](7.65,-13.35) rectangle (8.35, -14.05);
\draw[draw={rgb,255:red,0; green,0; blue,0},fill={rgb,255:red,255; green,0; blue,0},opacity=0.5,fill opacity=0.5](7.65,-12.35) rectangle (8.35, -13.05);
\draw[draw={rgb,255:red,0; green,0; blue,0},fill={rgb,255:red,255; green,0; blue,0},opacity=0.5,fill opacity=0.5](8.65,-12.35) rectangle (9.35, -13.05);
\draw[draw={rgb,255:red,0; green,0; blue,0},fill={rgb,255:red,255; green,0; blue,0},opacity=0.5,fill opacity=0.5](9.65,-12.35) rectangle (10.35, -13.05);
\draw[draw={rgb,255:red,0; green,0; blue,0},fill={rgb,255:red,255; green,0; blue,0},opacity=0.5,fill opacity=0.5](9.65,-11.35) rectangle (10.35, -12.05);
\draw[draw={rgb,255:red,0; green,0; blue,0},fill={rgb,255:red,255; green,0; blue,0},opacity=0.5,fill opacity=0.5](9.65,-10.35) rectangle (10.35, -11.05);
\draw[draw={rgb,255:red,0; green,0; blue,0},fill={rgb,255:red,255; green,0; blue,0},opacity=0.5,fill opacity=0.5](10.65,-10.35) rectangle (11.35, -11.05);
\draw[draw={rgb,255:red,0; green,0; blue,0},fill={rgb,255:red,255; green,0; blue,0},opacity=0.5,fill opacity=0.5](11.65,-10.35) rectangle (12.35, -11.05);
\draw[draw={rgb,255:red,0; green,0; blue,0},fill={rgb,255:red,255; green,0; blue,0},opacity=0.5,fill opacity=0.5](12.65,-10.35) rectangle (13.35, -11.05);
\draw[draw={rgb,255:red,0; green,0; blue,0},fill={rgb,255:red,255; green,0; blue,0},opacity=0.5,fill opacity=0.5](12.65,-11.35) rectangle (13.35, -12.05);
\draw[draw={rgb,255:red,0; green,0; blue,0},fill={rgb,255:red,255; green,0; blue,0},opacity=0.5,fill opacity=0.5](12.65,-12.35) rectangle (13.35, -13.05);
\draw[draw={rgb,255:red,0; green,0; blue,0},fill={rgb,255:red,255; green,0; blue,0},opacity=0.5,fill opacity=0.5](12.65,-13.35) rectangle (13.35, -14.05);
\draw[draw={rgb,255:red,0; green,0; blue,0},fill={rgb,255:red,255; green,0; blue,0},opacity=0.5,fill opacity=0.5](13.65,-13.35) rectangle (14.35, -14.05);
\draw[draw={rgb,255:red,0; green,0; blue,0},fill={rgb,255:red,255; green,0; blue,0},opacity=0.5,fill opacity=0.5](14.65,-13.35) rectangle (15.35, -14.05);
\draw[draw={rgb,255:red,0; green,0; blue,0},fill={rgb,255:red,255; green,0; blue,0},opacity=0.5,fill opacity=0.5](14.65,-12.35) rectangle (15.35, -13.05);
\draw[draw={rgb,255:red,0; green,0; blue,0},fill={rgb,255:red,255; green,0; blue,0},opacity=0.5,fill opacity=0.5](14.65,-11.35) rectangle (15.35, -12.05);
\draw[draw={rgb,255:red,0; green,0; blue,0},fill={rgb,255:red,255; green,0; blue,0},opacity=0.5,fill opacity=0.5](14.65,-10.35) rectangle (15.35, -11.05);
\draw[draw={rgb,255:red,0; green,0; blue,0},fill={rgb,255:red,255; green,0; blue,0},opacity=0.5,fill opacity=0.5](15.65,-10.35) rectangle (16.35, -11.05);
\draw[draw={rgb,255:red,0; green,0; blue,0},fill={rgb,255:red,255; green,0; blue,0},opacity=0.5,fill opacity=0.5](15.65,-9.35) rectangle (16.35, -10.05);
\draw[draw={rgb,255:red,255; green,0; blue,0},opacity=0.5,very thick](6,-13.7)--(8,-13.7)--(8,-12.7)--(10,-12.7)--(10,-10.7)--(13,-10.7)--(13,-13.7)--(15,-13.7)--(15,-10.7)--(16,-10.7)--(16,-9.7);
\draw[draw={rgb,255:red,255; green,0; blue,0},opacity=0.5,very thick](6,-13.7)--(8,-13.7)--(8,-12.7)--(10,-12.7)--(10,-10.7)--(13,-10.7)--(13,-13.7)--(15,-13.7)--(15,-10.7)--(16,-10.7)--(16,-9.7)--(17,-9.7)--(17,-12.7)--(19,-12.7)--(19,-11.7)--(21,-11.7)--(21,-7.7)--(23,-7.7)--(23,-10.7)--(25,-10.7)--(25,-5.7)--(20,-5.7)--(20,-4.7)--(17,-4.7)--(17,-3.7)--(19,-3.7)--(19,-1.7)--(16,-1.7) .. controls (16,-3.03) and (16,-4.37) .. (16,-5.7)--(14,-5.7);
\draw[draw={rgb,255:red,0; green,0; blue,0},very thick](15.72,-9.42)--(16.28,-9.98);
\draw[draw={rgb,255:red,0; green,0; blue,0},very thick](15.72,-9.98)--(16.28,-9.42);
\draw[draw={rgb,255:red,0; green,0; blue,0}](9.5,-12.69)--(9.5,-16.7);
\draw(8.23, -16.93) node[anchor=south west] {$l_j$};
\draw[draw=none,fill={rgb,255:red,28; green,36; blue,31},thin](9.5, -12.69) ellipse (0.1cm and 0.1cm);\draw[draw={rgb,255:red,0; green,0; blue,0}](5.5,-13.7)--(5.5,-16.7);
\draw(4.32, -16.93) node[anchor=south west] {$l_i$};
\draw[draw={rgb,255:red,0; green,0; blue,0}](13.5,-5.7)--(13.5,0.3);
\draw(12.34, -0.81) node[anchor=south west] {$l^k$};
\draw[draw=none,fill={rgb,255:red,28; green,36; blue,31},thin](13.5, -5.7) ellipse (0.1cm and 0.1cm);\draw[draw=none,fill={rgb,255:red,28; green,36; blue,31},thin](5.5, -13.7) ellipse (0.1cm and 0.1cm);
\end{tikzpicture}
      \caption{We are in the case where $s_0<|r|$ and $r_{s_0-1}=\pos{P_a}+\vpji$. To show that $P$ is fragile, we first produce $\sigma \cup \asm{P_{0,1,\ldots,i}}$ (not shown), then produce $R$ (red tiles) and then we place  the tile $P_{a+1}+\vpji$ (marked  $\times$) at the position $r_{s_0}$. The path $P$, shown in brown, can not be grown from this assembly due to the conflict at $\times$.}\label{fig:shield-R1}
    \end{figure}

    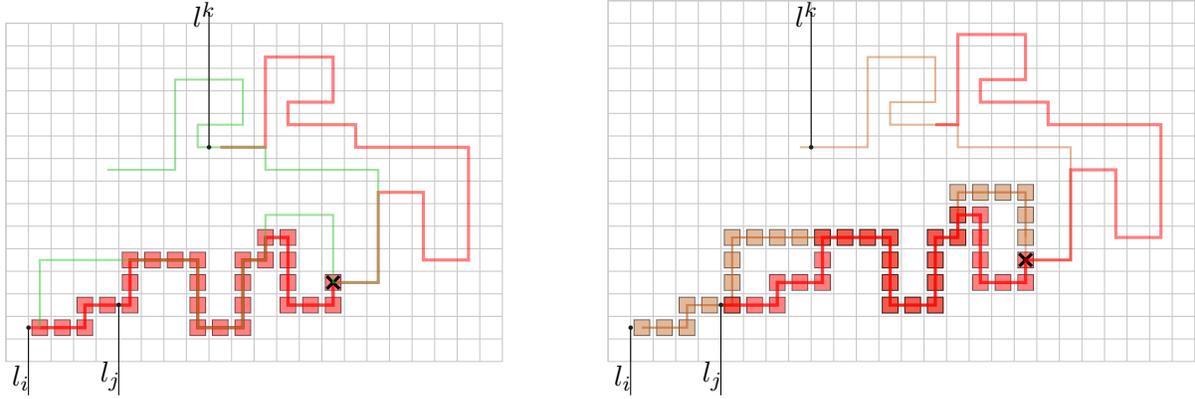
\begin{figure}[ht]
      \renewcommand\scale{0.3}
      \begin{tikzpicture}[scale=\scale]\draw[draw={rgb,255:red,200; green,200; blue,200}](4.5,-0.2) rectangle (26.5, -15.2);
\draw[draw={rgb,255:red,200; green,200; blue,200}](4.5,-15.2)--(4.5,-0.2);
\draw[draw={rgb,255:red,200; green,200; blue,200}](5.5,-15.2)--(5.5,-0.2);
\draw[draw={rgb,255:red,200; green,200; blue,200}](6.5,-15.2)--(6.5,-0.2);
\draw[draw={rgb,255:red,200; green,200; blue,200}](7.5,-15.2)--(7.5,-0.2);
\draw[draw={rgb,255:red,200; green,200; blue,200}](8.5,-15.2)--(8.5,-0.2);
\draw[draw={rgb,255:red,200; green,200; blue,200}](9.5,-15.2)--(9.5,-0.2);
\draw[draw={rgb,255:red,200; green,200; blue,200}](10.5,-15.2)--(10.5,-0.2);
\draw[draw={rgb,255:red,200; green,200; blue,200}](11.5,-15.2)--(11.5,-0.2);
\draw[draw={rgb,255:red,200; green,200; blue,200}](12.5,-15.2)--(12.5,-0.2);
\draw[draw={rgb,255:red,200; green,200; blue,200}](13.5,-15.2)--(13.5,-0.2);
\draw[draw={rgb,255:red,200; green,200; blue,200}](14.5,-15.2)--(14.5,-0.2);
\draw[draw={rgb,255:red,200; green,200; blue,200}](15.5,-15.2)--(15.5,-0.2);
\draw[draw={rgb,255:red,200; green,200; blue,200}](16.5,-15.2)--(16.5,-0.2);
\draw[draw={rgb,255:red,200; green,200; blue,200}](17.5,-15.2)--(17.5,-0.2);
\draw[draw={rgb,255:red,200; green,200; blue,200}](18.5,-15.2)--(18.5,-0.2);
\draw[draw={rgb,255:red,200; green,200; blue,200}](19.5,-15.2)--(19.5,-0.2);
\draw[draw={rgb,255:red,200; green,200; blue,200}](20.5,-15.2)--(20.5,-0.2);
\draw[draw={rgb,255:red,200; green,200; blue,200}](21.5,-15.2)--(21.5,-0.2);
\draw[draw={rgb,255:red,200; green,200; blue,200}](22.5,-15.2)--(22.5,-0.2);
\draw[draw={rgb,255:red,200; green,200; blue,200}](23.5,-15.2)--(23.5,-0.2);
\draw[draw={rgb,255:red,200; green,200; blue,200}](24.5,-15.2)--(24.5,-0.2);
\draw[draw={rgb,255:red,200; green,200; blue,200}](25.5,-15.2)--(25.5,-0.2);
\draw[draw={rgb,255:red,200; green,200; blue,200}](4.5,-0.2)--(26.5,-0.2);
\draw[draw={rgb,255:red,200; green,200; blue,200}](4.5,-1.2)--(26.5,-1.2);
\draw[draw={rgb,255:red,200; green,200; blue,200}](4.5,-2.2)--(26.5,-2.2);
\draw[draw={rgb,255:red,200; green,200; blue,200}](4.5,-3.2)--(26.5,-3.2);
\draw[draw={rgb,255:red,200; green,200; blue,200}](4.5,-4.2)--(26.5,-4.2);
\draw[draw={rgb,255:red,200; green,200; blue,200}](4.5,-5.2)--(26.5,-5.2);
\draw[draw={rgb,255:red,200; green,200; blue,200}](4.5,-6.2)--(26.5,-6.2);
\draw[draw={rgb,255:red,200; green,200; blue,200}](4.5,-7.2)--(26.5,-7.2);
\draw[draw={rgb,255:red,200; green,200; blue,200}](4.5,-8.2)--(26.5,-8.2);
\draw[draw={rgb,255:red,200; green,200; blue,200}](4.5,-9.2)--(26.5,-9.2);
\draw[draw={rgb,255:red,200; green,200; blue,200}](4.5,-10.2)--(26.5,-10.2);
\draw[draw={rgb,255:red,200; green,200; blue,200}](4.5,-11.2)--(26.5,-11.2);
\draw[draw={rgb,255:red,200; green,200; blue,200}](4.5,-12.2)--(26.5,-12.2);
\draw[draw={rgb,255:red,200; green,200; blue,200}](4.5,-13.2)--(26.5,-13.2);
\draw[draw={rgb,255:red,200; green,200; blue,200}](4.5,-14.2)--(26.5,-14.2);
\draw[draw={rgb,255:red,0; green,0; blue,0},fill={rgb,255:red,255; green,0; blue,0},opacity=0.5,fill opacity=0.5](5.65,-13.35) rectangle (6.35, -14.05);
\draw[draw={rgb,255:red,0; green,0; blue,0},fill={rgb,255:red,255; green,0; blue,0},opacity=0.5,fill opacity=0.5](6.65,-13.35) rectangle (7.35, -14.05);
\draw[draw={rgb,255:red,0; green,0; blue,0},fill={rgb,255:red,255; green,0; blue,0},opacity=0.5,fill opacity=0.5](7.65,-13.35) rectangle (8.35, -14.05);
\draw[draw={rgb,255:red,0; green,0; blue,0},fill={rgb,255:red,255; green,0; blue,0},opacity=0.5,fill opacity=0.5](7.65,-12.35) rectangle (8.35, -13.05);
\draw[draw={rgb,255:red,0; green,0; blue,0},fill={rgb,255:red,255; green,0; blue,0},opacity=0.5,fill opacity=0.5](8.65,-12.35) rectangle (9.35, -13.05);
\draw[draw={rgb,255:red,0; green,0; blue,0},fill={rgb,255:red,255; green,0; blue,0},opacity=0.5,fill opacity=0.5](9.65,-12.35) rectangle (10.35, -13.05);
\draw[draw={rgb,255:red,0; green,0; blue,0},fill={rgb,255:red,255; green,0; blue,0},opacity=0.5,fill opacity=0.5](9.65,-11.35) rectangle (10.35, -12.05);
\draw[draw={rgb,255:red,0; green,0; blue,0},fill={rgb,255:red,255; green,0; blue,0},opacity=0.5,fill opacity=0.5](9.65,-10.35) rectangle (10.35, -11.05);
\draw[draw={rgb,255:red,0; green,0; blue,0},fill={rgb,255:red,255; green,0; blue,0},opacity=0.5,fill opacity=0.5](10.65,-10.35) rectangle (11.35, -11.05);
\draw[draw={rgb,255:red,0; green,0; blue,0},fill={rgb,255:red,255; green,0; blue,0},opacity=0.5,fill opacity=0.5](11.65,-10.35) rectangle (12.35, -11.05);
\draw[draw={rgb,255:red,0; green,0; blue,0},fill={rgb,255:red,255; green,0; blue,0},opacity=0.5,fill opacity=0.5](12.65,-10.35) rectangle (13.35, -11.05);
\draw[draw={rgb,255:red,0; green,0; blue,0},fill={rgb,255:red,255; green,0; blue,0},opacity=0.5,fill opacity=0.5](12.65,-11.35) rectangle (13.35, -12.05);
\draw[draw={rgb,255:red,0; green,0; blue,0},fill={rgb,255:red,255; green,0; blue,0},opacity=0.5,fill opacity=0.5](12.65,-12.35) rectangle (13.35, -13.05);
\draw[draw={rgb,255:red,0; green,0; blue,0},fill={rgb,255:red,255; green,0; blue,0},opacity=0.5,fill opacity=0.5](12.65,-13.35) rectangle (13.35, -14.05);
\draw[draw={rgb,255:red,0; green,0; blue,0},fill={rgb,255:red,255; green,0; blue,0},opacity=0.5,fill opacity=0.5](13.65,-13.35) rectangle (14.35, -14.05);
\draw[draw={rgb,255:red,0; green,0; blue,0},fill={rgb,255:red,255; green,0; blue,0},opacity=0.5,fill opacity=0.5](14.65,-13.35) rectangle (15.35, -14.05);
\draw[draw={rgb,255:red,0; green,0; blue,0},fill={rgb,255:red,255; green,0; blue,0},opacity=0.5,fill opacity=0.5](14.65,-12.35) rectangle (15.35, -13.05);
\draw[draw={rgb,255:red,0; green,0; blue,0},fill={rgb,255:red,255; green,0; blue,0},opacity=0.5,fill opacity=0.5](14.65,-11.35) rectangle (15.35, -12.05);
\draw[draw={rgb,255:red,0; green,0; blue,0},fill={rgb,255:red,255; green,0; blue,0},opacity=0.5,fill opacity=0.5](14.65,-10.35) rectangle (15.35, -11.05);
\draw[draw={rgb,255:red,0; green,0; blue,0},fill={rgb,255:red,255; green,0; blue,0},opacity=0.5,fill opacity=0.5](15.65,-10.35) rectangle (16.35, -11.05);
\draw[draw={rgb,255:red,0; green,0; blue,0},fill={rgb,255:red,255; green,0; blue,0},opacity=0.5,fill opacity=0.5](15.65,-9.35) rectangle (16.35, -10.05);
\draw[draw={rgb,255:red,0; green,0; blue,0},fill={rgb,255:red,255; green,0; blue,0},opacity=0.5,fill opacity=0.5](16.65,-9.35) rectangle (17.35, -10.05);
\draw[draw={rgb,255:red,0; green,0; blue,0},fill={rgb,255:red,255; green,0; blue,0},opacity=0.5,fill opacity=0.5](16.65,-10.35) rectangle (17.35, -11.05);
\draw[draw={rgb,255:red,0; green,0; blue,0},fill={rgb,255:red,255; green,0; blue,0},opacity=0.5,fill opacity=0.5](16.65,-11.35) rectangle (17.35, -12.05);
\draw[draw={rgb,255:red,0; green,0; blue,0},fill={rgb,255:red,255; green,0; blue,0},opacity=0.5,fill opacity=0.5](16.65,-12.35) rectangle (17.35, -13.05);
\draw[draw={rgb,255:red,0; green,0; blue,0},fill={rgb,255:red,255; green,0; blue,0},opacity=0.5,fill opacity=0.5](17.65,-12.35) rectangle (18.35, -13.05);
\draw[draw={rgb,255:red,0; green,0; blue,0},fill={rgb,255:red,255; green,0; blue,0},opacity=0.5,fill opacity=0.5](18.65,-12.35) rectangle (19.35, -13.05);
\draw[draw={rgb,255:red,0; green,0; blue,0},fill={rgb,255:red,255; green,0; blue,0},opacity=0.5,fill opacity=0.5](18.65,-11.35) rectangle (19.35, -12.05);
\draw[draw={rgb,255:red,255; green,0; blue,0},opacity=0.5,very thick](5.5,-13.7)--(8,-13.7)--(8,-12.7)--(10,-12.7)--(10,-10.7)--(13,-10.7)--(13,-13.7)--(15,-13.7)--(15,-10.7)--(16,-10.7)--(16,-9.7)--(17,-9.7)--(17,-12.7)--(19,-12.7)--(19,-11.7);
\draw[draw={rgb,255:red,255; green,0; blue,0},opacity=0.5,very thick](6,-13.7)--(8,-13.7)--(8,-12.7)--(10,-12.7)--(10,-10.7)--(13,-10.7)--(13,-13.7)--(15,-13.7)--(15,-10.7)--(16,-10.7)--(16,-9.7)--(17,-9.7)--(17,-12.7)--(19,-12.7)--(19,-11.7)--(21,-11.7)--(21,-7.7)--(23,-7.7)--(23,-10.7)--(25,-10.7)--(25,-5.7)--(20,-5.7)--(20,-4.7)--(17,-4.7)--(17,-3.7)--(19,-3.7)--(19,-1.7)--(16,-1.7) .. controls (16,-3.03) and (16,-4.37) .. (16,-5.7)--(14,-5.7);
\draw[draw={rgb,255:red,0; green,0; blue,0},very thick](18.72,-11.42)--(19.28,-11.98);
\draw[draw={rgb,255:red,0; green,0; blue,0},very thick](18.72,-11.98)--(19.28,-11.42);
\draw[draw={rgb,255:red,55; green,200; blue,55},opacity=0.5,thick](6,-13.7)--(6,-10.7)--(13,-10.7)--(13,-13.7)--(15,-13.7)--(15,-10.7)--(16,-10.7)--(16,-8.7)--(19,-8.7)--(19,-11.7)--(21,-11.7)--(21,-6.7)--(16,-6.7)--(16,-5.7)--(13,-5.7)--(13,-4.7)--(15,-4.7)--(15,-2.7)--(12,-2.7)--(12,-6.7)--(9,-6.7);
\draw[draw={rgb,255:red,0; green,0; blue,0}](9.5,-12.69)--(9.5,-16.7);
\draw(8.23, -16.93) node[anchor=south west] {$l_j$};
\draw[draw=none,fill={rgb,255:red,28; green,36; blue,31},thin](9.5, -12.69) ellipse (0.1cm and 0.1cm);\draw[draw={rgb,255:red,0; green,0; blue,0}](5.5,-13.7)--(5.5,-16.7);
\draw(4.32, -16.93) node[anchor=south west] {$l_i$};
\draw[draw={rgb,255:red,0; green,0; blue,0}](13.5,-5.7)--(13.5,0.3);
\draw(12.34, -0.81) node[anchor=south west] {$l^k$};
\draw[draw=none,fill={rgb,255:red,28; green,36; blue,31},thin](13.5, -5.7) ellipse (0.1cm and 0.1cm);\draw[draw=none,fill={rgb,255:red,28; green,36; blue,31},thin](5.5, -13.7) ellipse (0.1cm and 0.1cm);
\end{tikzpicture}\hfill
      \begin{tikzpicture}[scale=\scale]\draw[draw={rgb,255:red,200; green,200; blue,200}](4.5,0.8) rectangle (30.5, -15.2);
\draw[draw={rgb,255:red,200; green,200; blue,200}](4.5,-15.2)--(4.5,0.8);
\draw[draw={rgb,255:red,200; green,200; blue,200}](5.5,-15.2)--(5.5,0.8);
\draw[draw={rgb,255:red,200; green,200; blue,200}](6.5,-15.2)--(6.5,0.8);
\draw[draw={rgb,255:red,200; green,200; blue,200}](7.5,-15.2)--(7.5,0.8);
\draw[draw={rgb,255:red,200; green,200; blue,200}](8.5,-15.2)--(8.5,0.8);
\draw[draw={rgb,255:red,200; green,200; blue,200}](9.5,-15.2)--(9.5,0.8);
\draw[draw={rgb,255:red,200; green,200; blue,200}](10.5,-15.2)--(10.5,0.8);
\draw[draw={rgb,255:red,200; green,200; blue,200}](11.5,-15.2)--(11.5,0.8);
\draw[draw={rgb,255:red,200; green,200; blue,200}](12.5,-15.2)--(12.5,0.8);
\draw[draw={rgb,255:red,200; green,200; blue,200}](13.5,-15.2)--(13.5,0.8);
\draw[draw={rgb,255:red,200; green,200; blue,200}](14.5,-15.2)--(14.5,0.8);
\draw[draw={rgb,255:red,200; green,200; blue,200}](15.5,-15.2)--(15.5,0.8);
\draw[draw={rgb,255:red,200; green,200; blue,200}](16.5,-15.2)--(16.5,0.8);
\draw[draw={rgb,255:red,200; green,200; blue,200}](17.5,-15.2)--(17.5,0.8);
\draw[draw={rgb,255:red,200; green,200; blue,200}](18.5,-15.2)--(18.5,0.8);
\draw[draw={rgb,255:red,200; green,200; blue,200}](19.5,-15.2)--(19.5,0.8);
\draw[draw={rgb,255:red,200; green,200; blue,200}](20.5,-15.2)--(20.5,0.8);
\draw[draw={rgb,255:red,200; green,200; blue,200}](21.5,-15.2)--(21.5,0.8);
\draw[draw={rgb,255:red,200; green,200; blue,200}](22.5,-15.2)--(22.5,0.8);
\draw[draw={rgb,255:red,200; green,200; blue,200}](23.5,-15.2)--(23.5,0.8);
\draw[draw={rgb,255:red,200; green,200; blue,200}](24.5,-15.2)--(24.5,0.8);
\draw[draw={rgb,255:red,200; green,200; blue,200}](25.5,-15.2)--(25.5,0.8);
\draw[draw={rgb,255:red,200; green,200; blue,200}](26.5,-15.2)--(26.5,0.8);
\draw[draw={rgb,255:red,200; green,200; blue,200}](27.5,-15.2)--(27.5,0.8);
\draw[draw={rgb,255:red,200; green,200; blue,200}](28.5,-15.2)--(28.5,0.8);
\draw[draw={rgb,255:red,200; green,200; blue,200}](29.5,-15.2)--(29.5,0.8);
\draw[draw={rgb,255:red,200; green,200; blue,200}](4.5,0.8)--(30.5,0.8);
\draw[draw={rgb,255:red,200; green,200; blue,200}](4.5,-0.2)--(30.5,-0.2);
\draw[draw={rgb,255:red,200; green,200; blue,200}](4.5,-1.2)--(30.5,-1.2);
\draw[draw={rgb,255:red,200; green,200; blue,200}](4.5,-2.2)--(30.5,-2.2);
\draw[draw={rgb,255:red,200; green,200; blue,200}](4.5,-3.2)--(30.5,-3.2);
\draw[draw={rgb,255:red,200; green,200; blue,200}](4.5,-4.2)--(30.5,-4.2);
\draw[draw={rgb,255:red,200; green,200; blue,200}](4.5,-5.2)--(30.5,-5.2);
\draw[draw={rgb,255:red,200; green,200; blue,200}](4.5,-6.2)--(30.5,-6.2);
\draw[draw={rgb,255:red,200; green,200; blue,200}](4.5,-7.2)--(30.5,-7.2);
\draw[draw={rgb,255:red,200; green,200; blue,200}](4.5,-8.2)--(30.5,-8.2);
\draw[draw={rgb,255:red,200; green,200; blue,200}](4.5,-9.2)--(30.5,-9.2);
\draw[draw={rgb,255:red,200; green,200; blue,200}](4.5,-10.2)--(30.5,-10.2);
\draw[draw={rgb,255:red,200; green,200; blue,200}](4.5,-11.2)--(30.5,-11.2);
\draw[draw={rgb,255:red,200; green,200; blue,200}](4.5,-12.2)--(30.5,-12.2);
\draw[draw={rgb,255:red,200; green,200; blue,200}](4.5,-13.2)--(30.5,-13.2);
\draw[draw={rgb,255:red,200; green,200; blue,200}](4.5,-14.2)--(30.5,-14.2);
\draw[draw={rgb,255:red,200; green,113; blue,55},opacity=0.5,thick](6,-13.7)--(8,-13.7)--(8,-12.7)--(10,-12.7)--(10,-9.7)--(17,-9.7)--(17,-12.7)--(19,-12.7)--(19,-9.7)--(20,-9.7)--(20,-7.7)--(23,-7.7)--(23,-10.7)--(25,-10.7)--(25,-5.7)--(20,-5.7)--(20,-4.7)--(17,-4.7)--(17,-3.7)--(19,-3.7)--(19,-1.7)--(16,-1.7)--(16,-5.7)--(13,-5.7);
\draw[draw={rgb,255:red,0; green,0; blue,0},fill={rgb,255:red,200; green,113; blue,55},opacity=0.5,fill opacity=0.5](5.65,-13.35) rectangle (6.35, -14.05);
\draw[draw={rgb,255:red,0; green,0; blue,0},fill={rgb,255:red,200; green,113; blue,55},opacity=0.5,fill opacity=0.5](6.65,-13.35) rectangle (7.35, -14.05);
\draw[draw={rgb,255:red,0; green,0; blue,0},fill={rgb,255:red,200; green,113; blue,55},opacity=0.5,fill opacity=0.5](7.65,-13.35) rectangle (8.35, -14.05);
\draw[draw={rgb,255:red,0; green,0; blue,0},fill={rgb,255:red,200; green,113; blue,55},opacity=0.5,fill opacity=0.5](7.65,-12.35) rectangle (8.35, -13.05);
\draw[draw={rgb,255:red,0; green,0; blue,0},fill={rgb,255:red,200; green,113; blue,55},opacity=0.5,fill opacity=0.5](8.65,-12.35) rectangle (9.35, -13.05);
\draw[draw={rgb,255:red,0; green,0; blue,0},fill={rgb,255:red,200; green,113; blue,55},opacity=0.5,fill opacity=0.5](9.65,-12.35) rectangle (10.35, -13.05);
\draw[draw={rgb,255:red,0; green,0; blue,0},fill={rgb,255:red,200; green,113; blue,55},opacity=0.5,fill opacity=0.5](9.65,-11.35) rectangle (10.35, -12.05);
\draw[draw={rgb,255:red,0; green,0; blue,0},fill={rgb,255:red,200; green,113; blue,55},opacity=0.5,fill opacity=0.5](9.65,-10.35) rectangle (10.35, -11.05);
\draw[draw={rgb,255:red,0; green,0; blue,0},fill={rgb,255:red,200; green,113; blue,55},opacity=0.5,fill opacity=0.5](9.65,-9.35) rectangle (10.35, -10.05);
\draw[draw={rgb,255:red,0; green,0; blue,0},fill={rgb,255:red,200; green,113; blue,55},opacity=0.5,fill opacity=0.5](10.65,-9.35) rectangle (11.35, -10.05);
\draw[draw={rgb,255:red,0; green,0; blue,0},fill={rgb,255:red,200; green,113; blue,55},opacity=0.5,fill opacity=0.5](11.65,-9.35) rectangle (12.35, -10.05);
\draw[draw={rgb,255:red,0; green,0; blue,0},fill={rgb,255:red,200; green,113; blue,55},opacity=0.5,fill opacity=0.5](12.65,-9.35) rectangle (13.35, -10.05);
\draw[draw={rgb,255:red,0; green,0; blue,0},fill={rgb,255:red,200; green,113; blue,55},opacity=0.5,fill opacity=0.5](13.65,-9.35) rectangle (14.35, -10.05);
\draw[draw={rgb,255:red,0; green,0; blue,0},fill={rgb,255:red,200; green,113; blue,55},opacity=0.5,fill opacity=0.5](14.65,-9.35) rectangle (15.35, -10.05);
\draw[draw={rgb,255:red,0; green,0; blue,0},fill={rgb,255:red,200; green,113; blue,55},opacity=0.5,fill opacity=0.5](15.65,-9.35) rectangle (16.35, -10.05);
\draw[draw={rgb,255:red,0; green,0; blue,0},fill={rgb,255:red,200; green,113; blue,55},opacity=0.5,fill opacity=0.5](16.65,-9.35) rectangle (17.35, -10.05);
\draw[draw={rgb,255:red,0; green,0; blue,0},fill={rgb,255:red,200; green,113; blue,55},opacity=0.5,fill opacity=0.5](16.65,-10.35) rectangle (17.35, -11.05);
\draw[draw={rgb,255:red,0; green,0; blue,0},fill={rgb,255:red,200; green,113; blue,55},opacity=0.5,fill opacity=0.5](16.65,-11.35) rectangle (17.35, -12.05);
\draw[draw={rgb,255:red,0; green,0; blue,0},fill={rgb,255:red,200; green,113; blue,55},opacity=0.5,fill opacity=0.5](16.65,-12.35) rectangle (17.35, -13.05);
\draw[draw={rgb,255:red,0; green,0; blue,0},fill={rgb,255:red,200; green,113; blue,55},opacity=0.5,fill opacity=0.5](17.65,-12.35) rectangle (18.35, -13.05);
\draw[draw={rgb,255:red,0; green,0; blue,0},fill={rgb,255:red,200; green,113; blue,55},opacity=0.5,fill opacity=0.5](18.65,-12.35) rectangle (19.35, -13.05);
\draw[draw={rgb,255:red,0; green,0; blue,0},fill={rgb,255:red,200; green,113; blue,55},opacity=0.5,fill opacity=0.5](18.65,-11.35) rectangle (19.35, -12.05);
\draw[draw={rgb,255:red,0; green,0; blue,0},fill={rgb,255:red,200; green,113; blue,55},opacity=0.5,fill opacity=0.5](18.65,-10.35) rectangle (19.35, -11.05);
\draw[draw={rgb,255:red,0; green,0; blue,0},fill={rgb,255:red,200; green,113; blue,55},opacity=0.5,fill opacity=0.5](18.65,-9.35) rectangle (19.35, -10.05);
\draw[draw={rgb,255:red,0; green,0; blue,0},fill={rgb,255:red,200; green,113; blue,55},opacity=0.5,fill opacity=0.5](19.65,-9.35) rectangle (20.35, -10.05);
\draw[draw={rgb,255:red,0; green,0; blue,0},fill={rgb,255:red,200; green,113; blue,55},opacity=0.5,fill opacity=0.5](19.65,-8.35) rectangle (20.35, -9.05);
\draw[draw={rgb,255:red,0; green,0; blue,0},fill={rgb,255:red,200; green,113; blue,55},opacity=0.5,fill opacity=0.5](19.65,-7.35) rectangle (20.35, -8.05);
\draw[draw={rgb,255:red,0; green,0; blue,0},fill={rgb,255:red,200; green,113; blue,55},opacity=0.5,fill opacity=0.5](20.65,-7.35) rectangle (21.35, -8.05);
\draw[draw={rgb,255:red,0; green,0; blue,0},fill={rgb,255:red,200; green,113; blue,55},opacity=0.5,fill opacity=0.5](21.65,-7.35) rectangle (22.35, -8.05);
\draw[draw={rgb,255:red,0; green,0; blue,0},fill={rgb,255:red,200; green,113; blue,55},opacity=0.5,fill opacity=0.5](22.65,-7.35) rectangle (23.35, -8.05);
\draw[draw={rgb,255:red,0; green,0; blue,0},fill={rgb,255:red,200; green,113; blue,55},opacity=0.5,fill opacity=0.5](22.65,-8.35) rectangle (23.35, -9.05);
\draw[draw={rgb,255:red,0; green,0; blue,0},fill={rgb,255:red,200; green,113; blue,55},opacity=0.5,fill opacity=0.5](22.65,-9.35) rectangle (23.35, -10.05);
\draw[draw={rgb,255:red,200; green,113; blue,55},opacity=0.5,thick](6,-13.7)--(8,-13.7)--(8,-12.7)--(10,-12.7)--(10,-9.7)--(17,-9.7)--(17,-12.7)--(19,-12.7)--(19,-9.7)--(20,-9.7)--(20,-7.7)--(23,-7.7)--(23,-9.7);
\draw[draw={rgb,255:red,0; green,0; blue,0},fill={rgb,255:red,255; green,0; blue,0},opacity=0.5,fill opacity=0.5](9.65,-12.35) rectangle (10.35, -13.05);
\draw[draw={rgb,255:red,0; green,0; blue,0},fill={rgb,255:red,255; green,0; blue,0},opacity=0.5,fill opacity=0.5](10.65,-12.35) rectangle (11.35, -13.05);
\draw[draw={rgb,255:red,0; green,0; blue,0},fill={rgb,255:red,255; green,0; blue,0},opacity=0.5,fill opacity=0.5](11.65,-12.35) rectangle (12.35, -13.05);
\draw[draw={rgb,255:red,0; green,0; blue,0},fill={rgb,255:red,255; green,0; blue,0},opacity=0.5,fill opacity=0.5](11.65,-11.35) rectangle (12.35, -12.05);
\draw[draw={rgb,255:red,0; green,0; blue,0},fill={rgb,255:red,255; green,0; blue,0},opacity=0.5,fill opacity=0.5](12.65,-11.35) rectangle (13.35, -12.05);
\draw[draw={rgb,255:red,0; green,0; blue,0},fill={rgb,255:red,255; green,0; blue,0},opacity=0.5,fill opacity=0.5](13.65,-11.35) rectangle (14.35, -12.05);
\draw[draw={rgb,255:red,0; green,0; blue,0},fill={rgb,255:red,255; green,0; blue,0},opacity=0.5,fill opacity=0.5](13.65,-10.35) rectangle (14.35, -11.05);
\draw[draw={rgb,255:red,0; green,0; blue,0},fill={rgb,255:red,255; green,0; blue,0},opacity=0.5,fill opacity=0.5](13.65,-9.35) rectangle (14.35, -10.05);
\draw[draw={rgb,255:red,0; green,0; blue,0},fill={rgb,255:red,255; green,0; blue,0},opacity=0.5,fill opacity=0.5](14.65,-9.35) rectangle (15.35, -10.05);
\draw[draw={rgb,255:red,0; green,0; blue,0},fill={rgb,255:red,255; green,0; blue,0},opacity=0.5,fill opacity=0.5](15.65,-9.35) rectangle (16.35, -10.05);
\draw[draw={rgb,255:red,0; green,0; blue,0},fill={rgb,255:red,255; green,0; blue,0},opacity=0.5,fill opacity=0.5](16.65,-9.35) rectangle (17.35, -10.05);
\draw[draw={rgb,255:red,0; green,0; blue,0},fill={rgb,255:red,255; green,0; blue,0},opacity=0.5,fill opacity=0.5](16.65,-10.35) rectangle (17.35, -11.05);
\draw[draw={rgb,255:red,0; green,0; blue,0},fill={rgb,255:red,255; green,0; blue,0},opacity=0.5,fill opacity=0.5](16.65,-11.35) rectangle (17.35, -12.05);
\draw[draw={rgb,255:red,0; green,0; blue,0},fill={rgb,255:red,255; green,0; blue,0},opacity=0.5,fill opacity=0.5](16.65,-12.35) rectangle (17.35, -13.05);
\draw[draw={rgb,255:red,0; green,0; blue,0},fill={rgb,255:red,255; green,0; blue,0},opacity=0.5,fill opacity=0.5](17.65,-12.35) rectangle (18.35, -13.05);
\draw[draw={rgb,255:red,0; green,0; blue,0},fill={rgb,255:red,255; green,0; blue,0},opacity=0.5,fill opacity=0.5](18.65,-12.35) rectangle (19.35, -13.05);
\draw[draw={rgb,255:red,0; green,0; blue,0},fill={rgb,255:red,255; green,0; blue,0},opacity=0.5,fill opacity=0.5](18.65,-11.35) rectangle (19.35, -12.05);
\draw[draw={rgb,255:red,0; green,0; blue,0},fill={rgb,255:red,255; green,0; blue,0},opacity=0.5,fill opacity=0.5](18.65,-10.35) rectangle (19.35, -11.05);
\draw[draw={rgb,255:red,0; green,0; blue,0},fill={rgb,255:red,255; green,0; blue,0},opacity=0.5,fill opacity=0.5](18.65,-9.35) rectangle (19.35, -10.05);
\draw[draw={rgb,255:red,0; green,0; blue,0},fill={rgb,255:red,255; green,0; blue,0},opacity=0.5,fill opacity=0.5](19.65,-9.35) rectangle (20.35, -10.05);
\draw[draw={rgb,255:red,0; green,0; blue,0},fill={rgb,255:red,255; green,0; blue,0},opacity=0.5,fill opacity=0.5](19.65,-8.35) rectangle (20.35, -9.05);
\draw[draw={rgb,255:red,0; green,0; blue,0},fill={rgb,255:red,255; green,0; blue,0},opacity=0.5,fill opacity=0.5](20.65,-8.35) rectangle (21.35, -9.05);
\draw[draw={rgb,255:red,0; green,0; blue,0},fill={rgb,255:red,255; green,0; blue,0},opacity=0.5,fill opacity=0.5](20.65,-9.35) rectangle (21.35, -10.05);
\draw[draw={rgb,255:red,0; green,0; blue,0},fill={rgb,255:red,255; green,0; blue,0},opacity=0.5,fill opacity=0.5](20.65,-10.35) rectangle (21.35, -11.05);
\draw[draw={rgb,255:red,0; green,0; blue,0},fill={rgb,255:red,255; green,0; blue,0},opacity=0.5,fill opacity=0.5](20.65,-11.35) rectangle (21.35, -12.05);
\draw[draw={rgb,255:red,0; green,0; blue,0},fill={rgb,255:red,255; green,0; blue,0},opacity=0.5,fill opacity=0.5](21.65,-11.35) rectangle (22.35, -12.05);
\draw[draw={rgb,255:red,0; green,0; blue,0},fill={rgb,255:red,255; green,0; blue,0},opacity=0.5,fill opacity=0.5](22.65,-11.35) rectangle (23.35, -12.05);
\draw[draw={rgb,255:red,0; green,0; blue,0},fill={rgb,255:red,255; green,0; blue,0},opacity=0.5,fill opacity=0.5](22.65,-10.35) rectangle (23.35, -11.05);
\draw[draw={rgb,255:red,255; green,0; blue,0},opacity=0.5,very thick](9.5,-12.7)--(12,-12.7)--(12,-11.7)--(14,-11.7)--(14,-9.7)--(17,-9.7)--(17,-12.7)--(19,-12.7)--(19,-9.7)--(20,-9.7)--(20,-8.7)--(21,-8.7)--(21,-11.7)--(23,-11.7)--(23,-10.7);
\draw[draw={rgb,255:red,0; green,0; blue,0},very thick](22.72,-10.42)--(23.28,-10.98);
\draw[draw={rgb,255:red,0; green,0; blue,0},very thick](22.72,-10.98)--(23.28,-10.42);
\draw[draw={rgb,255:red,255; green,0; blue,0},opacity=0.5,very thick](9.5,-12.7)--(12,-12.7)--(12,-11.7)--(14,-11.7)--(14,-9.7)--(17,-9.7)--(17,-12.7)--(19,-12.7)--(19,-9.7)--(20,-9.7)--(20,-8.7)--(21,-8.7)--(21,-11.7)--(23,-11.7)--(23,-10.7)--(25,-10.7)--(25,-6.7)--(27,-6.7)--(27,-9.7)--(29,-9.7)--(29,-4.7)--(24,-4.7)--(24,-3.7)--(21,-3.7)--(21,-2.7)--(23,-2.7)--(23,-0.7)--(20,-0.7)--(20,-4.7)--(19,-4.7);
\draw[draw={rgb,255:red,0; green,0; blue,0}](9.5,-12.69)--(9.5,-16.7);
\draw(8.23, -16.93) node[anchor=south west] {$l_j$};
\draw[draw=none,fill={rgb,255:red,28; green,36; blue,31},thin](9.5, -12.69) ellipse (0.1cm and 0.1cm);\draw[draw={rgb,255:red,0; green,0; blue,0}](5.5,-13.7)--(5.5,-16.7);
\draw(4.32, -16.93) node[anchor=south west] {$l_i$};
\draw[draw={rgb,255:red,0; green,0; blue,0}](13.5,-5.7)--(13.5,0.3);
\draw(12.34, -0.81) node[anchor=south west] {$l^k$};
\draw[draw=none,fill={rgb,255:red,28; green,36; blue,31},thin](13.5, -5.7) ellipse (0.1cm and 0.1cm);\draw[draw=none,fill={rgb,255:red,28; green,36; blue,31},thin](5.5, -13.7) ellipse (0.1cm and 0.1cm);
\end{tikzpicture}
      \renewcommand\scale{0.4}
      \caption{Left: We are in the case where $s_0<|r|$ and $r_{s_0-1}=\pos{P_a}$.
      Right: We first grow $\sigma\cup \asm{P_{\range 0 1 i}}$ (not shown), then we grow $\asm{P_{\range {i+1}{i+2} j}(R+\vpij)}$ (red tiles), and then $P_{a+1}+\vpij$ (marked with a $\times$).  $P_{\range{j+1}{j+2}k}$ is shown in brown and can not grow from this assembly due to the conflict at $\times$.
      }\label{fig:shield-R2}
    \end{figure}

    \begin{enumerate}
    \item\label{use R to break P:case1} $r_{s_0-1}=\pos{P_a}+\vpji$ for some $a\in\{\range{j+1}{j+2}k\}$ (an example is shown in Figure~\ref{fig:shield-R1}).
      In this case, we first grow $\sigma\cup \asm{P_{\range 0 1 i}R}$ (we have already proved that $P_{\range 0 1 i}R$ is a producible path). We claim we can producibly place the tile  $P_{a+1}+\vpji$:
      \begin{itemize}
      \item $\pos{P_{a+1}+\vpji} = r_{s_0}$, and thus $P_{a+1}+\vpji$ interacts (has a matching abutting glue) with $P_{a}+\vpji = R_{|R|-1}$  which is at position $r_{s_0-1}$,
      \item $P_{a+1}+\vpji$ conflicts with $P_{\range{i+1}{i+2}k}$ and hence does not share any position of a tile of $\sigma\cup\asm{P_{\range 0 1 i}}$
        (since $P_{\range{i+1}{i+2}k}$ does not intersect $\sigma\cup\asm{P_{\range 0 1 i}}$ by the definition of producible path), and
      \item since $r$ is simple that position is not occupied by any other tile of $R$.
      \end{itemize}

      By the definition of ${s_0}$, the tile  $P_{a+1}+\vpji$  conflicts with $P_{\range{i+1}{i+2}k}$, which shows that $P$ is fragile.

    \item\label{use R to break P:case2} Else, $r_{s_0-1}=\pos{P_a}$ for some $a\in\{\range{i+1}{i+2}k\}$  (an example is shown in Figure~\ref{fig:shield-R2}).
      We first grow $\sigma\cup \asm{P_{\range 0 1 j}(R+\vpij)}$ (which we have shown is producible). We claim we can then producibly place the tile $P_{a+1}+\vpij$:
      \begin{itemize}
      \item $\pos{P_{a+1}+\vpij} = r_{s_0} +\vpij$ and
        thus $P_{a+1}+\vpij$ interacts (has a matching abutting glue)  with $P_{a}+\vpij = R_{|R|-1} +\vpij$ which is at position $r_{s_0-1}$,

     \item $P_{a+1}+\vpij$ conflicts with  $P_{\range{j+1}{j+2}k}$ and hence does not share any position of a tile of $\sigma\cup\asm{P_{\range 0 1 j}}$,
        (since $P_{\range{j+1}{j+2}k}$ does not intersect $\sigma\cup\asm{P_{\range 0 1 j}}$ by the definition of a producible path), and
      \item since $r$ is simple $\pos{P_{a+1}}+\vpij$ is not occupied by any other tile of $R + \vpij$.
      \end{itemize}
      By the definition of ${s_0}$, the tile
      $P_{a+1}$ conflicts with $P_{j+1,j+2,\ldots,k}+\vect{P_j P_i}$
      thus $P_{a+1}+\vpij$  conflicts with $P_{j+1,j+2,\ldots,k}$, which shows that $P$ is fragile.
    \end{enumerate}
    In either Case~\ref{use R to break P:case1} or~\ref{use R to break P:case2},  $P$ is fragile.
  \end{itemize}
\end{proof}

\subsubsection{Using $m_0$ and $R$ to define $u_0$ and $v_0$}
\label{subsec:u0v0}

Using the index $m_0$ from Subsection~\ref{subsec:first}, and the path $R$ from Subsection~\ref{subsec:r} we prove the following claim (illustrated with an example in Figure~\ref{fig:shield-u0v0}):
\begin{sublemma}
  \label{lem:u0v0}
  If $|R| = |r|$ (i.e. if $R$ does not conflict with $P_{\rng i k}$ nor with $P_{\rng j k}+\vpij$), there are two indices $u_0$ and $v_0$ satisfying all of the following conditions:
\begin{enumerate}
\item\label{c1} $i+1 \leq u_0 \leq m_0\leq v_0$, and 
\item\label{c2} $P_{u_0}=P_{v_0}+\vect{P_jP_i}$, and 
\item\label{c3} $P_{\rng{u_0}{m_0}}+\vect{P_iP_j}$ is entirely in $\mathcal{C^+}$, and 
\item\label{c4} $(P_{\rng{u_0}{m_0}}+\vect{P_iP_j}) \cap P_{i+1, i+2, \ldots, k} = \{ P_{v_0} \}$.
\end{enumerate}
\end{sublemma}
\begin{proof}
    We first claim that $P_{m_0}$ is a tile of $R$, and moreover that $\embed R\cap \lmz = \{\lmz(0)\} = \{\pos{P_{m_0}}\}$.
  By definition of $R$, $R_0=P_{i+1}$, and since $l^i$ is strictly to the west of $\lmz$ (by \subl{lem:lmzEast} and since $x_{\vect{P_iP_j}}>0$ meaning $l^i$ is to the west of $l^j$),
  since $m_0>i+1$ (\subl{lem:m0i1}),
  and by the definition of $\cp$ (Subsection~\ref{shield:split}), $\pos{P_{i+1}} \not\in \cp$.
  By definition of $r$, $\pos{R_{|R|-1}} \in \mathcal{C}$.

  We also claim that $\pos{R_{|R|-1}} \in  \cp$. Indeed, assume for the sake of contradiction that $\pos{R_{|R|-1}}\in\cm$. However, $l^k$ is in $\cp$, and $\pos{R_{|R|-1}}$ is at a horizontal distance 0.5 to the east or to the west of $l^k$ and is in $\mathcal C$ (by definition of $r$).
  Then the horizontal segment from $R_{|R|-1}$ to $l^k$ crosses $\lmz$, and since $\lmz$ is on a tile column, the only position where that can happen is at $R_{|R|-1}$, contradicting our assumption. Therefore, $R_{|R|-1} \in \cp$.

  Now, since $\embed R$ is entirely in $\mathcal{C}$ (by definition of $r$) and has a position in $\cm$ and another one in $\cp$, $R$ intersects the border of $\cp$.
  Moreover, that intersection happens on $\lmz$ since $\lmz$ partitions $\mathcal C$ into $\cm$ and $\cp$ (Subsection~\ref{shield:split}).
  However, since $R$ is composed only of segments of $P_{\range{i+1}{i+2}k}$ and $P_{\range{j+1}{j+2}k}+\vpji$, $\embed R$ cannot intersect $\lmz$ strictly below $\lmz(0)=\pos{P_{m_0}}$, because by \subl{lem:cmz}, $\embed{P_{\range{i+1}{i+2}k}}$ can only intersect $\lmz$ at $\lmz(0)$, and $\embed{P_{\range{j+1}{j+2}k}+\vpji}$ does not intersect $\lmz$ (because by \subl{lem:cmz}, $\embed{P_{\range{i+1}{i+2}k}}$ does not intersect $\lmz+\vpij$).
  Therefore, $R$ intersects $\lmz$ at position $\lmz(0)$, and hence $P_{m_0}$ is also a tile of $R$. Let $b\in\{0,1,\ldots, |R|-1 \}$ be the index such that $R_b = P_{m_0}$.

    We now define the index $v_0$ to be used in the statement of this \sublname:
  let $a \in \{0,1,\ldots,b-1\}$ be the largest index such that there is an index $v_0\in\{\range{i+1}{i+2}k\}$ such that $R_a+\vpij = P_{v_0}$.
  Indices $a$ and $v_0$ exist because with $a = 0$ and $v_0 = j+1$, we have $R_0+\vpij = P_{i+1} +\vpij = P_{j+1}$
  (this equality holds by the definition of $R$ at the beginning of the proof of \subl{lem:r}, and by the final conclusion of \subl{lem:r} which states that there are no conflicts between $R$ and $P_{i+1,i+2,\ldots,k}$, and no conflicts between $R$ and $P_{j+1,j+2,\ldots,k} +\vpji$).

    Next, we define the index $u_0$ to be used in the statement of this \sublname: $R_{\range{a+1}{a+2}b}+\vpij$ does not intersect $\Pik$ (by definition of $a$ and $b$), therefore $R_{\range{a+1}{a+2}b}$ does not intersect $\Pjk$ either. Hence, $R_{\range{a+1}{a+2}b}$ is a segment of $P$ (because $R$ is only composed of segments of $\Pik$ and segments of $\Pjk$), and moreover $R_a$ is both a tile of $\Pik$ and a tile of $\Pjk$. We then define $u_0 \in \{ i+1,i+2,\ldots,k \}$ to be the index such that $R_a = P_{u_0}$. By \subl{lem:order}, since $a<b$ and $R_b = P_{m_0}$, this means that $u_0 < m_0$, and we get Conclusion~\ref{c2}.

  \begin{figure}[ht]
    \centering
    \begin{tikzpicture}[scale=\scale]\draw[draw={rgb,255:red,200; green,200; blue,200}](0.5,-0.2) rectangle (26.5, -15.2);
\draw[draw={rgb,255:red,200; green,200; blue,200}](0.5,-15.2)--(0.5,-0.2);
\draw[draw={rgb,255:red,200; green,200; blue,200}](1.5,-15.2)--(1.5,-0.2);
\draw[draw={rgb,255:red,200; green,200; blue,200}](2.5,-15.2)--(2.5,-0.2);
\draw[draw={rgb,255:red,200; green,200; blue,200}](3.5,-15.2)--(3.5,-0.2);
\draw[draw={rgb,255:red,200; green,200; blue,200}](4.5,-15.2)--(4.5,-0.2);
\draw[draw={rgb,255:red,200; green,200; blue,200}](5.5,-15.2)--(5.5,-0.2);
\draw[draw={rgb,255:red,200; green,200; blue,200}](6.5,-15.2)--(6.5,-0.2);
\draw[draw={rgb,255:red,200; green,200; blue,200}](7.5,-15.2)--(7.5,-0.2);
\draw[draw={rgb,255:red,200; green,200; blue,200}](8.5,-15.2)--(8.5,-0.2);
\draw[draw={rgb,255:red,200; green,200; blue,200}](9.5,-15.2)--(9.5,-0.2);
\draw[draw={rgb,255:red,200; green,200; blue,200}](10.5,-15.2)--(10.5,-0.2);
\draw[draw={rgb,255:red,200; green,200; blue,200}](11.5,-15.2)--(11.5,-0.2);
\draw[draw={rgb,255:red,200; green,200; blue,200}](12.5,-15.2)--(12.5,-0.2);
\draw[draw={rgb,255:red,200; green,200; blue,200}](13.5,-15.2)--(13.5,-0.2);
\draw[draw={rgb,255:red,200; green,200; blue,200}](14.5,-15.2)--(14.5,-0.2);
\draw[draw={rgb,255:red,200; green,200; blue,200}](15.5,-15.2)--(15.5,-0.2);
\draw[draw={rgb,255:red,200; green,200; blue,200}](16.5,-15.2)--(16.5,-0.2);
\draw[draw={rgb,255:red,200; green,200; blue,200}](17.5,-15.2)--(17.5,-0.2);
\draw[draw={rgb,255:red,200; green,200; blue,200}](18.5,-15.2)--(18.5,-0.2);
\draw[draw={rgb,255:red,200; green,200; blue,200}](19.5,-15.2)--(19.5,-0.2);
\draw[draw={rgb,255:red,200; green,200; blue,200}](20.5,-15.2)--(20.5,-0.2);
\draw[draw={rgb,255:red,200; green,200; blue,200}](21.5,-15.2)--(21.5,-0.2);
\draw[draw={rgb,255:red,200; green,200; blue,200}](22.5,-15.2)--(22.5,-0.2);
\draw[draw={rgb,255:red,200; green,200; blue,200}](23.5,-15.2)--(23.5,-0.2);
\draw[draw={rgb,255:red,200; green,200; blue,200}](24.5,-15.2)--(24.5,-0.2);
\draw[draw={rgb,255:red,200; green,200; blue,200}](25.5,-15.2)--(25.5,-0.2);
\draw[draw={rgb,255:red,200; green,200; blue,200}](0.5,-0.2)--(26.5,-0.2);
\draw[draw={rgb,255:red,200; green,200; blue,200}](0.5,-1.2)--(26.5,-1.2);
\draw[draw={rgb,255:red,200; green,200; blue,200}](0.5,-2.2)--(26.5,-2.2);
\draw[draw={rgb,255:red,200; green,200; blue,200}](0.5,-3.2)--(26.5,-3.2);
\draw[draw={rgb,255:red,200; green,200; blue,200}](0.5,-4.2)--(26.5,-4.2);
\draw[draw={rgb,255:red,200; green,200; blue,200}](0.5,-5.2)--(26.5,-5.2);
\draw[draw={rgb,255:red,200; green,200; blue,200}](0.5,-6.2)--(26.5,-6.2);
\draw[draw={rgb,255:red,200; green,200; blue,200}](0.5,-7.2)--(26.5,-7.2);
\draw[draw={rgb,255:red,200; green,200; blue,200}](0.5,-8.2)--(26.5,-8.2);
\draw[draw={rgb,255:red,200; green,200; blue,200}](0.5,-9.2)--(26.5,-9.2);
\draw[draw={rgb,255:red,200; green,200; blue,200}](0.5,-10.2)--(26.5,-10.2);
\draw[draw={rgb,255:red,200; green,200; blue,200}](0.5,-11.2)--(26.5,-11.2);
\draw[draw={rgb,255:red,200; green,200; blue,200}](0.5,-12.2)--(26.5,-12.2);
\draw[draw={rgb,255:red,200; green,200; blue,200}](0.5,-13.2)--(26.5,-13.2);
\draw[draw={rgb,255:red,200; green,200; blue,200}](0.5,-14.2)--(26.5,-14.2);
\draw[draw=none,fill={rgb,255:red,0; green,102; blue,255},opacity=0.1](13.45,-5.7)--(13.5,-0.2)--(26.5,-0.2)--(26.45,-15.2)--(18.95,-15.2)--(18.95,-9.7)--(19.95,-9.7)--(19.95,-7.7)--(22.95,-7.7)--(22.95,-10.7)--(24.95,-10.7)--(24.95,-5.7)--(19.95,-5.7)--(19.95,-4.7)--(17,-4.7)--(17,-3.7)--(19,-3.7)--(19,-1.7)--(16,-1.7)--(15.95,-3.7)--(15.95,-5.7)--(13.45,-5.7);
\draw(22.87, -4.18) node[anchor=south west] {$\cp$};
\draw[draw={rgb,255:red,0; green,0; blue,0},fill={rgb,255:red,200; green,113; blue,55},opacity=0.29899997,fill opacity=0.29899997](1.65,-8.35) rectangle (2.35, -9.05);
\draw[draw={rgb,255:red,0; green,0; blue,0},fill={rgb,255:red,200; green,113; blue,55},opacity=0.29899997,fill opacity=0.29899997](2.65,-8.35) rectangle (3.35, -9.05);
\draw[draw={rgb,255:red,0; green,0; blue,0},fill={rgb,255:red,200; green,113; blue,55},opacity=0.29899997,fill opacity=0.29899997](3.65,-8.35) rectangle (4.35, -9.05);
\draw[draw={rgb,255:red,0; green,0; blue,0},fill={rgb,255:red,200; green,113; blue,55},opacity=0.29899997,fill opacity=0.29899997](4.65,-8.35) rectangle (5.35, -9.05);
\draw[draw={rgb,255:red,0; green,0; blue,0},fill={rgb,255:red,200; green,113; blue,55},opacity=0.29899997,fill opacity=0.29899997](5.65,-8.35) rectangle (6.35, -9.05);
\draw[draw={rgb,255:red,0; green,0; blue,0},fill={rgb,255:red,200; green,113; blue,55},opacity=0.29899997,fill opacity=0.29899997](6.65,-8.35) rectangle (7.35, -9.05);
\draw[draw={rgb,255:red,0; green,0; blue,0},fill={rgb,255:red,200; green,113; blue,55},opacity=0.29899997,fill opacity=0.29899997](7.65,-8.35) rectangle (8.35, -9.05);
\draw[draw={rgb,255:red,0; green,0; blue,0},fill={rgb,255:red,200; green,113; blue,55},opacity=0.29899997,fill opacity=0.29899997](7.65,-9.35) rectangle (8.35, -10.05);
\draw[draw={rgb,255:red,0; green,0; blue,0},fill={rgb,255:red,200; green,113; blue,55},opacity=0.29899997,fill opacity=0.29899997](7.65,-10.35) rectangle (8.35, -11.05);
\draw[draw={rgb,255:red,0; green,0; blue,0},fill={rgb,255:red,200; green,113; blue,55},opacity=0.29899997,fill opacity=0.29899997](7.65,-11.35) rectangle (8.35, -12.05);
\draw[draw={rgb,255:red,0; green,0; blue,0},fill={rgb,255:red,200; green,113; blue,55},opacity=0.29899997,fill opacity=0.29899997](6.65,-11.35) rectangle (7.35, -12.05);
\draw[draw={rgb,255:red,0; green,0; blue,0},fill={rgb,255:red,200; green,113; blue,55},opacity=0.29899997,fill opacity=0.29899997](5.65,-11.35) rectangle (6.35, -12.05);
\draw[draw={rgb,255:red,0; green,0; blue,0},fill={rgb,255:red,200; green,113; blue,55},opacity=0.29899997,fill opacity=0.29899997](4.65,-11.35) rectangle (5.35, -12.05);
\draw[draw={rgb,255:red,0; green,0; blue,0},fill={rgb,255:red,200; green,113; blue,55},opacity=0.29899997,fill opacity=0.29899997](3.65,-11.35) rectangle (4.35, -12.05);
\draw[draw={rgb,255:red,0; green,0; blue,0},fill={rgb,255:red,200; green,113; blue,55},opacity=0.29899997,fill opacity=0.29899997](3.65,-12.35) rectangle (4.35, -13.05);
\draw[draw={rgb,255:red,0; green,0; blue,0},fill={rgb,255:red,200; green,113; blue,55},opacity=0.29899997,fill opacity=0.29899997](3.65,-13.35) rectangle (4.35, -14.05);
\draw[draw={rgb,255:red,0; green,0; blue,0},fill={rgb,255:red,200; green,113; blue,55},opacity=0.29899997,fill opacity=0.29899997](4.65,-13.35) rectangle (5.35, -14.05);
\draw[draw={rgb,255:red,200; green,113; blue,55},opacity=0.29899997,thick](2,-8.7)--(8,-8.7)--(8,-11.7)--(4,-11.7)--(4,-13.7)--(5.5,-13.7);
\draw[draw={rgb,255:red,0; green,0; blue,0},fill={rgb,255:red,200; green,113; blue,55},opacity=0.5,fill opacity=0.5](5.65,-13.35) rectangle (6.35, -14.05);
\draw[draw={rgb,255:red,0; green,0; blue,0},fill={rgb,255:red,200; green,113; blue,55},opacity=0.5,fill opacity=0.5](6.65,-13.35) rectangle (7.35, -14.05);
\draw[draw={rgb,255:red,0; green,0; blue,0},fill={rgb,255:red,200; green,113; blue,55},opacity=0.5,fill opacity=0.5](7.65,-13.35) rectangle (8.35, -14.05);
\draw[draw={rgb,255:red,0; green,0; blue,0},fill={rgb,255:red,200; green,113; blue,55},opacity=0.5,fill opacity=0.5](7.65,-12.35) rectangle (8.35, -13.05);
\draw[draw={rgb,255:red,0; green,0; blue,0},fill={rgb,255:red,200; green,113; blue,55},opacity=0.5,fill opacity=0.5](8.65,-12.35) rectangle (9.35, -13.05);
\draw[draw={rgb,255:red,0; green,0; blue,0},fill={rgb,255:red,200; green,113; blue,55},opacity=0.5,fill opacity=0.5](9.65,-12.35) rectangle (10.35, -13.05);
\draw[draw={rgb,255:red,0; green,0; blue,0},fill={rgb,255:red,200; green,113; blue,55},opacity=0.5,fill opacity=0.5](9.65,-11.35) rectangle (10.35, -12.05);
\draw[draw={rgb,255:red,0; green,0; blue,0},fill={rgb,255:red,200; green,113; blue,55},opacity=0.5,fill opacity=0.5](9.65,-10.35) rectangle (10.35, -11.05);
\draw[draw={rgb,255:red,0; green,0; blue,0},fill={rgb,255:red,200; green,113; blue,55},opacity=0.5,fill opacity=0.5](9.65,-9.35) rectangle (10.35, -10.05);
\draw[draw={rgb,255:red,0; green,0; blue,0},fill={rgb,255:red,200; green,113; blue,55},opacity=0.5,fill opacity=0.5](10.65,-9.35) rectangle (11.35, -10.05);
\draw[draw={rgb,255:red,0; green,0; blue,0},fill={rgb,255:red,200; green,113; blue,55},opacity=0.5,fill opacity=0.5](11.65,-9.35) rectangle (12.35, -10.05);
\draw[draw={rgb,255:red,0; green,0; blue,0},fill={rgb,255:red,200; green,113; blue,55},opacity=0.5,fill opacity=0.5](12.65,-9.35) rectangle (13.35, -10.05);
\draw[draw={rgb,255:red,0; green,0; blue,0},fill={rgb,255:red,200; green,113; blue,55},opacity=0.5,fill opacity=0.5](13.65,-9.35) rectangle (14.35, -10.05);
\draw[draw={rgb,255:red,0; green,0; blue,0},fill={rgb,255:red,200; green,113; blue,55},opacity=0.5,fill opacity=0.5](14.65,-9.35) rectangle (15.35, -10.05);
\draw[draw={rgb,255:red,0; green,0; blue,0},fill={rgb,255:red,200; green,113; blue,55},opacity=0.5,fill opacity=0.5](15.65,-9.35) rectangle (16.35, -10.05);
\draw[draw={rgb,255:red,0; green,0; blue,0},fill={rgb,255:red,200; green,113; blue,55},opacity=0.5,fill opacity=0.5](16.65,-9.35) rectangle (17.35, -10.05);
\draw[draw={rgb,255:red,0; green,0; blue,0},fill={rgb,255:red,200; green,113; blue,55},opacity=0.5,fill opacity=0.5](16.65,-10.35) rectangle (17.35, -11.05);
\draw[draw={rgb,255:red,0; green,0; blue,0},fill={rgb,255:red,200; green,113; blue,55},opacity=0.5,fill opacity=0.5](16.65,-11.35) rectangle (17.35, -12.05);
\draw[draw={rgb,255:red,0; green,0; blue,0},fill={rgb,255:red,200; green,113; blue,55},opacity=0.5,fill opacity=0.5](16.65,-12.35) rectangle (17.35, -13.05);
\draw[draw={rgb,255:red,0; green,0; blue,0},fill={rgb,255:red,200; green,113; blue,55},opacity=0.5,fill opacity=0.5](17.65,-12.35) rectangle (18.35, -13.05);
\draw[draw={rgb,255:red,0; green,0; blue,0},fill={rgb,255:red,200; green,113; blue,55},opacity=0.5,fill opacity=0.5](18.65,-12.35) rectangle (19.35, -13.05);
\draw[draw={rgb,255:red,0; green,0; blue,0},fill={rgb,255:red,200; green,113; blue,55},opacity=0.5,fill opacity=0.5](18.65,-11.35) rectangle (19.35, -12.05);
\draw[draw={rgb,255:red,0; green,0; blue,0},fill={rgb,255:red,200; green,113; blue,55},opacity=0.5,fill opacity=0.5](18.65,-10.35) rectangle (19.35, -11.05);
\draw[draw={rgb,255:red,0; green,0; blue,0},fill={rgb,255:red,200; green,113; blue,55},opacity=0.5,fill opacity=0.5](18.65,-9.35) rectangle (19.35, -10.05);
\draw[draw={rgb,255:red,0; green,0; blue,0},fill={rgb,255:red,200; green,113; blue,55},opacity=0.5,fill opacity=0.5](19.65,-9.35) rectangle (20.35, -10.05);
\draw[draw={rgb,255:red,0; green,0; blue,0},fill={rgb,255:red,200; green,113; blue,55},opacity=0.5,fill opacity=0.5](19.65,-8.35) rectangle (20.35, -9.05);
\draw[draw={rgb,255:red,0; green,0; blue,0},fill={rgb,255:red,200; green,113; blue,55},opacity=0.5,fill opacity=0.5](19.65,-7.35) rectangle (20.35, -8.05);
\draw[draw={rgb,255:red,0; green,0; blue,0},fill={rgb,255:red,200; green,113; blue,55},opacity=0.5,fill opacity=0.5](20.65,-7.35) rectangle (21.35, -8.05);
\draw[draw={rgb,255:red,0; green,0; blue,0},fill={rgb,255:red,200; green,113; blue,55},opacity=0.5,fill opacity=0.5](21.65,-7.35) rectangle (22.35, -8.05);
\draw[draw={rgb,255:red,0; green,0; blue,0},fill={rgb,255:red,200; green,113; blue,55},opacity=0.5,fill opacity=0.5](22.65,-7.35) rectangle (23.35, -8.05);
\draw[draw={rgb,255:red,0; green,0; blue,0},fill={rgb,255:red,200; green,113; blue,55},opacity=0.5,fill opacity=0.5](22.65,-8.35) rectangle (23.35, -9.05);
\draw[draw={rgb,255:red,0; green,0; blue,0},fill={rgb,255:red,200; green,113; blue,55},opacity=0.5,fill opacity=0.5](22.65,-9.35) rectangle (23.35, -10.05);
\draw[draw={rgb,255:red,0; green,0; blue,0},fill={rgb,255:red,200; green,113; blue,55},opacity=0.5,fill opacity=0.5](22.65,-10.35) rectangle (23.35, -11.05);
\draw[draw={rgb,255:red,0; green,0; blue,0},fill={rgb,255:red,200; green,113; blue,55},opacity=0.5,fill opacity=0.5](23.65,-10.35) rectangle (24.35, -11.05);
\draw[draw={rgb,255:red,0; green,0; blue,0},fill={rgb,255:red,200; green,113; blue,55},opacity=0.5,fill opacity=0.5](24.65,-10.35) rectangle (25.35, -11.05);
\draw[draw={rgb,255:red,0; green,0; blue,0},fill={rgb,255:red,200; green,113; blue,55},opacity=0.5,fill opacity=0.5](24.65,-9.35) rectangle (25.35, -10.05);
\draw[draw={rgb,255:red,0; green,0; blue,0},fill={rgb,255:red,200; green,113; blue,55},opacity=0.5,fill opacity=0.5](24.65,-8.35) rectangle (25.35, -9.05);
\draw[draw={rgb,255:red,0; green,0; blue,0},fill={rgb,255:red,200; green,113; blue,55},opacity=0.5,fill opacity=0.5](24.65,-7.35) rectangle (25.35, -8.05);
\draw[draw={rgb,255:red,0; green,0; blue,0},fill={rgb,255:red,200; green,113; blue,55},opacity=0.5,fill opacity=0.5](24.65,-6.35) rectangle (25.35, -7.05);
\draw[draw={rgb,255:red,0; green,0; blue,0},fill={rgb,255:red,200; green,113; blue,55},opacity=0.5,fill opacity=0.5](24.65,-5.35) rectangle (25.35, -6.05);
\draw[draw={rgb,255:red,0; green,0; blue,0},fill={rgb,255:red,200; green,113; blue,55},opacity=0.5,fill opacity=0.5](23.65,-5.35) rectangle (24.35, -6.05);
\draw[draw={rgb,255:red,0; green,0; blue,0},fill={rgb,255:red,200; green,113; blue,55},opacity=0.5,fill opacity=0.5](22.65,-5.35) rectangle (23.35, -6.05);
\draw[draw={rgb,255:red,0; green,0; blue,0},fill={rgb,255:red,200; green,113; blue,55},opacity=0.5,fill opacity=0.5](21.65,-5.35) rectangle (22.35, -6.05);
\draw[draw={rgb,255:red,0; green,0; blue,0},fill={rgb,255:red,200; green,113; blue,55},opacity=0.5,fill opacity=0.5](20.65,-5.35) rectangle (21.35, -6.05);
\draw[draw={rgb,255:red,0; green,0; blue,0},fill={rgb,255:red,200; green,113; blue,55},opacity=0.5,fill opacity=0.5](19.65,-5.35) rectangle (20.35, -6.05);
\draw[draw={rgb,255:red,0; green,0; blue,0},fill={rgb,255:red,200; green,113; blue,55},opacity=0.5,fill opacity=0.5](19.65,-4.35) rectangle (20.35, -5.05);
\draw[draw={rgb,255:red,0; green,0; blue,0},fill={rgb,255:red,200; green,113; blue,55},opacity=0.5,fill opacity=0.5](18.65,-4.35) rectangle (19.35, -5.05);
\draw[draw={rgb,255:red,0; green,0; blue,0},fill={rgb,255:red,200; green,113; blue,55},opacity=0.5,fill opacity=0.5](17.65,-4.35) rectangle (18.35, -5.05);
\draw[draw={rgb,255:red,0; green,0; blue,0},fill={rgb,255:red,200; green,113; blue,55},opacity=0.5,fill opacity=0.5](16.65,-4.35) rectangle (17.35, -5.05);
\draw[draw={rgb,255:red,0; green,0; blue,0},fill={rgb,255:red,200; green,113; blue,55},opacity=0.5,fill opacity=0.5](16.65,-3.35) rectangle (17.35, -4.05);
\draw[draw={rgb,255:red,0; green,0; blue,0},fill={rgb,255:red,200; green,113; blue,55},opacity=0.5,fill opacity=0.5](17.65,-3.35) rectangle (18.35, -4.05);
\draw[draw={rgb,255:red,0; green,0; blue,0},fill={rgb,255:red,200; green,113; blue,55},opacity=0.5,fill opacity=0.5](18.65,-3.35) rectangle (19.35, -4.05);
\draw[draw={rgb,255:red,0; green,0; blue,0},fill={rgb,255:red,200; green,113; blue,55},opacity=0.5,fill opacity=0.5](18.65,-2.35) rectangle (19.35, -3.05);
\draw[draw={rgb,255:red,0; green,0; blue,0},fill={rgb,255:red,200; green,113; blue,55},opacity=0.5,fill opacity=0.5](18.65,-1.35) rectangle (19.35, -2.05);
\draw[draw={rgb,255:red,0; green,0; blue,0},fill={rgb,255:red,200; green,113; blue,55},opacity=0.5,fill opacity=0.5](17.65,-1.35) rectangle (18.35, -2.05);
\draw[draw={rgb,255:red,0; green,0; blue,0},fill={rgb,255:red,200; green,113; blue,55},opacity=0.5,fill opacity=0.5](16.65,-1.35) rectangle (17.35, -2.05);
\draw[draw={rgb,255:red,0; green,0; blue,0},fill={rgb,255:red,200; green,113; blue,55},opacity=0.5,fill opacity=0.5](15.65,-1.35) rectangle (16.35, -2.05);
\draw[draw={rgb,255:red,0; green,0; blue,0},fill={rgb,255:red,200; green,113; blue,55},opacity=0.5,fill opacity=0.5](15.65,-2.35) rectangle (16.35, -3.05);
\draw[draw={rgb,255:red,0; green,0; blue,0},fill={rgb,255:red,200; green,113; blue,55},opacity=0.5,fill opacity=0.5](15.65,-3.35) rectangle (16.35, -4.05);
\draw[draw={rgb,255:red,0; green,0; blue,0},fill={rgb,255:red,200; green,113; blue,55},opacity=0.5,fill opacity=0.5](15.65,-4.35) rectangle (16.35, -5.05);
\draw[draw={rgb,255:red,0; green,0; blue,0},fill={rgb,255:red,200; green,113; blue,55},opacity=0.5,fill opacity=0.5](15.65,-5.35) rectangle (16.35, -6.05);
\draw[draw={rgb,255:red,0; green,0; blue,0},fill={rgb,255:red,200; green,113; blue,55},opacity=0.5,fill opacity=0.5](14.65,-5.35) rectangle (15.35, -6.05);
\draw[draw={rgb,255:red,0; green,0; blue,0},fill={rgb,255:red,200; green,113; blue,55},opacity=0.5,fill opacity=0.5](13.65,-5.35) rectangle (14.35, -6.05);
\draw[draw={rgb,255:red,0; green,0; blue,0},fill={rgb,255:red,200; green,113; blue,55},opacity=0.5,fill opacity=0.5](12.65,-5.35) rectangle (13.35, -6.05);
\draw[draw={rgb,255:red,200; green,113; blue,55},opacity=0.5,thick](6,-13.7)--(8,-13.7)--(8,-12.7)--(10,-12.7)--(10,-9.7)--(17,-9.7)--(17,-12.7)--(19,-12.7)--(19,-9.7)--(20,-9.7)--(20,-7.7)--(23,-7.7)--(23,-10.7)--(25,-10.7)--(25,-5.7)--(20,-5.7)--(20,-4.7)--(17,-4.7)--(17,-3.7)--(19,-3.7)--(19,-1.7)--(16,-1.7)--(16,-5.7)--(13,-5.7);
\draw[draw={rgb,255:red,0; green,0; blue,0},fill={rgb,255:red,255; green,0; blue,0},opacity=0.5,fill opacity=0.5](19.65,-8.35) rectangle (20.35, -9.05);
\draw[draw={rgb,255:red,0; green,0; blue,0},fill={rgb,255:red,255; green,0; blue,0},opacity=0.5,fill opacity=0.5](20.65,-8.35) rectangle (21.35, -9.05);
\draw[draw={rgb,255:red,0; green,0; blue,0},fill={rgb,255:red,255; green,0; blue,0},opacity=0.5,fill opacity=0.5](20.65,-9.35) rectangle (21.35, -10.05);
\draw[draw={rgb,255:red,0; green,0; blue,0},fill={rgb,255:red,255; green,0; blue,0},opacity=0.5,fill opacity=0.5](20.65,-10.35) rectangle (21.35, -11.05);
\draw[draw={rgb,255:red,0; green,0; blue,0},fill={rgb,255:red,255; green,0; blue,0},opacity=0.5,fill opacity=0.5](20.65,-11.35) rectangle (21.35, -12.05);
\draw[draw={rgb,255:red,0; green,0; blue,0},fill={rgb,255:red,255; green,0; blue,0},opacity=0.5,fill opacity=0.5](21.65,-11.35) rectangle (22.35, -12.05);
\draw[draw={rgb,255:red,0; green,0; blue,0},fill={rgb,255:red,255; green,0; blue,0},opacity=0.5,fill opacity=0.5](22.65,-11.35) rectangle (23.35, -12.05);
\draw[draw={rgb,255:red,255; green,0; blue,0},opacity=0.5,thick](20,-8.7)--(21,-8.7)--(21,-11.7)--(23,-11.7);
\draw(23.26, -12.94) node[anchor=south west] {$R_{\rng a b}+\vpij$};
\draw[draw=none,fill={rgb,255:red,28; green,36; blue,31},thin](20, -8.7) ellipse (0.15cm and 0.15cm);\draw(17.76, -9.18) node[anchor=south west] {$P_{v_0}$};
\draw(14.74, -9.55) node[anchor=south west] {$P_{u_0}$};
\draw[draw=none,fill={rgb,255:red,28; green,36; blue,31},thin](16, -9.7) ellipse (0.15cm and 0.15cm);\draw[draw=none,fill={rgb,255:red,28; green,36; blue,31},thin](23, -11.7) ellipse (0.1cm and 0.1cm);\draw[draw={rgb,255:red,0; green,0; blue,0}](9.5,-12.69)--(9.5,-16.7);
\draw(8.23, -16.93) node[anchor=south west] {$l_j$};
\draw[draw=none,fill={rgb,255:red,28; green,36; blue,31},thin](9.5, -12.69) ellipse (0.1cm and 0.1cm);\draw[draw={rgb,255:red,0; green,0; blue,0}](5.5,-13.7)--(5.5,-16.7);
\draw(4.32, -16.93) node[anchor=south west] {$l_i$};
\draw[draw={rgb,255:red,0; green,0; blue,0}](13.5,-5.7)--(13.5,0.3);
\draw(12.34, -0.81) node[anchor=south west] {$l^k$};
\draw[draw=none,fill={rgb,255:red,28; green,36; blue,31},thin](13.5, -5.7) ellipse (0.1cm and 0.1cm);\draw[draw=none,fill={rgb,255:red,28; green,36; blue,31},thin](5.5, -13.7) ellipse (0.1cm and 0.1cm);\draw[draw={rgb,255:red,0; green,0; blue,0}](19,-12.7)--(19,-16.7);
\draw(17.76, -17.73) node[anchor=south west] {$\lmz$};
\draw[draw=none,fill={rgb,255:red,28; green,36; blue,31},thin](19, -12.7) ellipse (0.1cm and 0.1cm);
\end{tikzpicture}
    \caption{Definition of $u_0$ and  $v_0$. The path $R_{\range{a}{a+1}b}+\vpij$ is shown in red and shares its start position $\pos{R_a}$  with the tile $P_{v_0}$. Also, in the example, it can be seen that $\pos{P_{u_0}} = \pos{R_a}$. In the proof of \subl{lem:u0v0} we find that $R_{\range{a}{a+1}b} = P_{\range{u_0}{u_0+1}m_0}$, and we go on to show that $u_0 \leq m_0 \leq v_0$, 
    $P_{u_0}=P_{v_0}+\vpji$, 
    and $P_{\range{u_0}{u_0+1}m_0} + \vpij$ is in $\cp$.}\label{fig:shield-u0v0}
  \end{figure}
    Then, we claim that $v_0\geq m_0$ and that $R_{\rng a b}+\vpij = P_{\rng{u_0}{m_0}}+\vpij$ is in $\cp$.
  There are two cases:

  \begin{itemize}
  \item If $m_0>j$, $\lmz$ is entirely in $\mathcal D$, since $\lmz$ is then a right turn from curve $d$. Therefore, $\cp\subseteq\mathcal D$.
    By \subl{lem:cmz}, $P_{m_0}+\vpij$ is in $\cp$. Moreover, by \subl{lem:lmzEast}, $\lmz+\vpji$ does not intersect $\Pik$, except possibly at $\lmz(0)+\vpji$, hence $\lmz$ does not intersect $P_{\rng{u_0}{m_0}}+\vpij$, except possibly at $\lmz(0) = \pos{P_{m_0}}$ (and in this case, $\pos{P_{m_0}}=\pos{P_{u_0}}+\vpij$ and $m_0=v_0$). Thus, $P_{\rng{u_0}{m_0}}+\vpij$ does not turn left from $\lmz$.

    Therefore, since by \subl{lem:d}, $R_{\rng a b}+\vpij$ is entirely in $\mathcal D$, then $P_{\rng{u_0}{m_0}}+\vpij$ is entirely in $\mathcal D$. Then, $P_{\rng{u_0}{m_0}}+\vpij$ does not turn left from $\embed{P_{\rng {m_0}{k}}}$ or from $l^k$.
    This implies that $P_{\rng{u_0}{m_0}}+\vpij$ is inside $\cp$ and that $P_{v_0}$ is a tile of $P_{\rng{m_0}k}$ (which is the only part of $P$ on the border of $\cp$), and hence that $v_0\geq m_0$, showing Conclusions~\ref{c1} and~\ref{c3}. Since $v_0$ was chosen as a largest index, we also get Conclusion~\ref{c4}.

  \item
    If $m_0\leq j$, then $\mathcal D\subseteq\cp$.
    Similarly as the previous case, by \subl{lem:d}, $P_{\rng{u_0}{m_0}}+\vpij$ is inside $D\subset\cp$, hence we have Conclusion~\ref{c3}. In particular, $\pos{P_{u_0}+\vpij}=\pos{P_{v_0}}$ is in $\mathcal D$. Since the only part of $\Pik$ that is in $\mathcal D$ is $P_{\range{j+1}{j+2}k}$, this means that $v_0\geq j+1 > m_0$ showing Conclusion~\ref{c1}.
    Since $v_0$ was chosen as a largest index, we also get Conclusion~\ref{c4}.
  \end{itemize}
\end{proof}

\subsection{Proof of Lemma~\ref{lem:shield}}
\label{subsec:dominant}
We restate Lemma~\ref{lem:shield} and give its proof.

\begin{replemma}{lem:shield}
  \shieldStatement
\end{replemma}
\begin{proof}
As noted at the beginning of Section~\ref{sec:shield main section}, 
without loss of generality, we suppose that the last tile of $P$ is $P_{k+1}$, i.e. that $P=P_{0,1,\ldots,k+1}$.
  We prove this lemma by induction on a triple $(u_n, m_n, v_n)$ of indices of $P$ and a curve~$f_n$.

  We first apply Claim~\ref{lem:r} to $P$ and shield $(i, j, k)$. If $|R|<|r|$, this yields the conclusion that $P$ is fragile, with a path that satisfies the conclusion of this lemma and we are done with the proof.
  Otherwise, we assume until the end of this proof that $|R|=|r|$, and $R$ conflicts with neither $P_{i+1,i+2,\ldots,k}$ nor with $P_{j+1,j+2,\ldots,k}+\vpji$.

  \argument{Induction hypothesis} The induction hypothesis is that for all natural numbers $n\geq 0$, the following four conditions are satisfied
  (we recall some notation: $i,j,k,l^k$ were defined in Definition~\ref{def:shield} and $l^i,l^j$  immediately after it;
  $c$ and $\mathcal{C}$ were defined in Definition~\ref{def:c}; $P_{m_0},\lmz$ were introduced in Subsection~\ref{subsec:initial}):
  \begin{enumerate}[label=H\arabic*]
  \item\label{umv:f} $f_n$ is a simple curve, and is the concatenation of: 
  (a) $\reverse{\lmz}+s_n\vpij$ for some natural number $s_n\geq 0$, and 
  (b) a curve that is the concatenation of embeddings of translations of segments\footnote{We are intentionally not fully describing where these segments appear on $P$; that is one of the goals of the remainder of the proof.} of $P_{\range {i+1} {i+2} k}$ and that connects $\lmz(0)+s_n\vpij$ to $\pos{P_{m_n}}$.

  \item\label{umv:c} $f_n$ is entirely in $\mathcal C$ and intersects $c$ exactly once, at the endpoint $\pos{P_{m_n}}$ of $f_n$.
  \item\label{umv:s} for all integers $u>0$, $f_n+u\vpij$ intersects neither $f_n$ nor $c$.
  \item\label{umv:unvn} $u_n\leq m_n\leq v_n$ and $\embed{P_{\rng {u_n} {m_n}}}+\vpij$ intersects $c$ exactly once, at $\pos{P_{u_n}}+\vpij = \pos{P_{v_n}}$, and intersects $g_n$ only at that position, where $g_n$ is the curve defined by:
    $$g_n = \concat{f_n,\embed{P_{\rng {m_n}k}},\gs{P_k}{l^k(0)},l^k}$$
  \end{enumerate}

  \argument{Special case} Sometimes, we need to consider some cases where \ref{umv:f}, \ref{umv:c}, \ref{umv:s} and $u_n\leq m_n\leq v_n$ hold, but where $v_n = k$, $P_{u_n}+\vpij=P_{k}$ and $\pos{P_{u_{n+1}}}+\vpij=\pos{P_{k+1}}$  (which can occur only if $\glu P k$ points to the east) and in this case, $P_{\rng{u_n}{m_n}}+\vpij$ intersects $c$ only at $\gs{P_k}{l^k(0)}$. We now prove that $P$ is pumpable in such a case by considering the curve:
  $$\rho=\concat{
    f_n,\,
    \reverse{\embed{P_{\range {u_n+1} {u_n+2} {m_n}}}},\,
    \gs{P_{u_n+1}}{l^k(0)+\vpji}, \,
    l^k+\vpij}
  $$
  By Theorem~$\ref{thm:infinite-jordan}$, $\rho$ partitions the plane $\R^2$ into two connected components. Note that for all $t\geq0$, $\rho+t\vpij$ does not intersect $\rho+(t+1)\vpij$. Moreover, since $l^k$ does not intersect $\rho$ then $l^k$ is in the right hand-side of this curve (and then in particular $l^k(0)$ and $P_{k}$). Thus $P_{\range {u_{n+1}} {u_{n+2}} {v_n}}$ is in the right-hand side of $\rho$.  Now note that for all $t>0$, the right-hand side of $\rho+t\vpij$ is included into $\mathcal{C}$ and $\sigma\cup P$ has no tile in the right-hand side of $\rho+t\vpij$. In particular, for $t=1$ this means that $P_{\range {u_{n+1}} {u_{n+2}} {v_n}}$ is in the strict left-hand side of $\rho+\vpij$. Thus for all $t$, $P_{\range {u_{n+1}} {u_{n+2}} {v_n}}+t\vpij$ is in the right hand-side of $\rho$ and in the strict left-hand side of $\rho+(t+1)\vpij$. Since $P_{v_n}+t\vpij$ and $P_{u_{n+1}+(t+1)\vpij}$ interact, this means that $P$ is pumpable.

  \argument{Initialisation ($n=0$)}
  We initialise the induction by applying \subl{lem:u0v0} to $P$ to get $i+1 \leq u_0\leq m_0\leq v_0$,
  and by letting $f_0=\reverse\lmz$.
  We show that the induction hypothesis is indeed satisfied for $(u_0, m_0, v_0)$ and $f_0$.

  \begin{enumerate}[label=H\arabic*]
  \item\label{umv:f zero} $f_0$ is indeed equal to $\reverse{\lmz}$ (i.e. $s_0 = 0$) concatenated with a zero length segment of $P_{\range {i+1} {i+2} k}$ (i.e. $\pos{P_{m_0}}$), and $f_0$ ends  at $\pos{P_{m_0}}$.
  \item\label{umv:c zero} By \subl{lem:lmzInC}, $f_0$ is indeed entirely in $\mathcal C$. By \subl{lem:cmz} and the fact that $\lmz$ is on a column and $l^i$, $l^k$ are on glue columns, $f_0$ intersects $c$ exactly once, at $\pos{P_{m_0}}$ (which is on $c$ since $i+1 \leq m_0 \leq k$). 
  \item\label{umv:s zero} 
    Let $u>0$ be an integer. 
  Since $\xcoord{\vpij}>0$, $f_0 + u\vpij = \lmz + u\vpij$ does not intersect $f_0=\rev{\lmz}$.
   \subl{lem:cmz} shows that  
   $\lmz+u\vpij$ does not intersect $\embed{P_{i+1,i+2,\ldots,k}}$, and since $\lmz+u\vpij$ does not intersect $l^k$ nor $l^i$ (because $\lmz+u\vpij$ is on a tile column, and $l^k$, $l^i$ are on glue columns) we get that 
       $\rev{\lmz}+u\vpij$ does not intersect $c$.

  \item\label{umv:unvn zero}
  By \subl{lem:u0v0}, $u_0\leq m_0\leq v_0$ and $P_{\rng {u_0} {m_0}}+\vpij$ intersects $P$ exactly once, at $P_{u_0}+\vpij = P_{v_0}$.
  Since $m_0 \leq v_0 \leq k$, $P_{\rng {u_0} {m_0}}+\vpij$ intersects $c$ and $g_0$ at $\pos{P_{v_0}}$.

  We claim that this is the only intersection between $P_{\rng {u_0} {m_0}}+\vpij$ and $g_0$: indeed, by Hypothesis~\ref{lem:hp:shield backup} of Definition~\ref{def:shield}, $P_{\rng {u_0} {m_0}}+\vpij$ cannot intersect $l^k$ or $\gs{P_k}{l^k(0)}$ (or we are in the special case where $P$ is pumpable).   Moreover, by definition of $\lmz$, $P_{\rng {u_0} {m_0}}+\vpij$ can only intersect $f_0 = \lmz$ at position $\pos{P_{m_0}}$, and by the previous argument, this can only happen if $m_0 = v_0$.

  Therefore, there is exactly one intersection between $g_0$ and $P_{\rng {u_0} {m_0}}+\vpij$, and that intersection is at $\pos{P_{v_0}}$.

  Finally, by \subl{lem:u0v0}, $P_{\rng {u_0} {m_0}}+\vpij$ is entirely in $\cp$. Therefore, $P_{\rng {u_0} {m_0}}+\vpij$ cannot intersect $l^i$ either, which shows that the only intersection between $P_{\rng {u_0} {m_0}}+\vpij$ and $c$ is at position $\pos{P_{v_0}}$.
  \end{enumerate}

  \begin{figure}[ht]
    \begin{center}
      \begin{tikzpicture}[scale=\scale]\draw[draw={rgb,255:red,200; green,200; blue,200}](2.5,1.8) rectangle (39.5, -18.2);
\draw[draw={rgb,255:red,200; green,200; blue,200}](2.5,-18.2)--(2.5,1.8);
\draw[draw={rgb,255:red,200; green,200; blue,200}](3.5,-18.2)--(3.5,1.8);
\draw[draw={rgb,255:red,200; green,200; blue,200}](4.5,-18.2)--(4.5,1.8);
\draw[draw={rgb,255:red,200; green,200; blue,200}](5.5,-18.2)--(5.5,1.8);
\draw[draw={rgb,255:red,200; green,200; blue,200}](6.5,-18.2)--(6.5,1.8);
\draw[draw={rgb,255:red,200; green,200; blue,200}](7.5,-18.2)--(7.5,1.8);
\draw[draw={rgb,255:red,200; green,200; blue,200}](8.5,-18.2)--(8.5,1.8);
\draw[draw={rgb,255:red,200; green,200; blue,200}](9.5,-18.2)--(9.5,1.8);
\draw[draw={rgb,255:red,200; green,200; blue,200}](10.5,-18.2)--(10.5,1.8);
\draw[draw={rgb,255:red,200; green,200; blue,200}](11.5,-18.2)--(11.5,1.8);
\draw[draw={rgb,255:red,200; green,200; blue,200}](12.5,-18.2)--(12.5,1.8);
\draw[draw={rgb,255:red,200; green,200; blue,200}](13.5,-18.2)--(13.5,1.8);
\draw[draw={rgb,255:red,200; green,200; blue,200}](14.5,-18.2)--(14.5,1.8);
\draw[draw={rgb,255:red,200; green,200; blue,200}](15.5,-18.2)--(15.5,1.8);
\draw[draw={rgb,255:red,200; green,200; blue,200}](16.5,-18.2)--(16.5,1.8);
\draw[draw={rgb,255:red,200; green,200; blue,200}](17.5,-18.2)--(17.5,1.8);
\draw[draw={rgb,255:red,200; green,200; blue,200}](18.5,-18.2)--(18.5,1.8);
\draw[draw={rgb,255:red,200; green,200; blue,200}](19.5,-18.2)--(19.5,1.8);
\draw[draw={rgb,255:red,200; green,200; blue,200}](20.5,-18.2)--(20.5,1.8);
\draw[draw={rgb,255:red,200; green,200; blue,200}](21.5,-18.2)--(21.5,1.8);
\draw[draw={rgb,255:red,200; green,200; blue,200}](22.5,-18.2)--(22.5,1.8);
\draw[draw={rgb,255:red,200; green,200; blue,200}](23.5,-18.2)--(23.5,1.8);
\draw[draw={rgb,255:red,200; green,200; blue,200}](24.5,-18.2)--(24.5,1.8);
\draw[draw={rgb,255:red,200; green,200; blue,200}](25.5,-18.2)--(25.5,1.8);
\draw[draw={rgb,255:red,200; green,200; blue,200}](26.5,-18.2)--(26.5,1.8);
\draw[draw={rgb,255:red,200; green,200; blue,200}](27.5,-18.2)--(27.5,1.8);
\draw[draw={rgb,255:red,200; green,200; blue,200}](28.5,-18.2)--(28.5,1.8);
\draw[draw={rgb,255:red,200; green,200; blue,200}](29.5,-18.2)--(29.5,1.8);
\draw[draw={rgb,255:red,200; green,200; blue,200}](30.5,-18.2)--(30.5,1.8);
\draw[draw={rgb,255:red,200; green,200; blue,200}](31.5,-18.2)--(31.5,1.8);
\draw[draw={rgb,255:red,200; green,200; blue,200}](32.5,-18.2)--(32.5,1.8);
\draw[draw={rgb,255:red,200; green,200; blue,200}](33.5,-18.2)--(33.5,1.8);
\draw[draw={rgb,255:red,200; green,200; blue,200}](34.5,-18.2)--(34.5,1.8);
\draw[draw={rgb,255:red,200; green,200; blue,200}](35.5,-18.2)--(35.5,1.8);
\draw[draw={rgb,255:red,200; green,200; blue,200}](36.5,-18.2)--(36.5,1.8);
\draw[draw={rgb,255:red,200; green,200; blue,200}](37.5,-18.2)--(37.5,1.8);
\draw[draw={rgb,255:red,200; green,200; blue,200}](38.5,-18.2)--(38.5,1.8);
\draw[draw={rgb,255:red,200; green,200; blue,200}](39.5,-18.2)--(39.5,1.8);
\draw[draw={rgb,255:red,200; green,200; blue,200}](2.5,1.8)--(39.5,1.8);
\draw[draw={rgb,255:red,200; green,200; blue,200}](2.5,0.8)--(39.5,0.8);
\draw[draw={rgb,255:red,200; green,200; blue,200}](2.5,-0.2)--(39.5,-0.2);
\draw[draw={rgb,255:red,200; green,200; blue,200}](2.5,-1.2)--(39.5,-1.2);
\draw[draw={rgb,255:red,200; green,200; blue,200}](2.5,-2.2)--(39.5,-2.2);
\draw[draw={rgb,255:red,200; green,200; blue,200}](2.5,-3.2)--(39.5,-3.2);
\draw[draw={rgb,255:red,200; green,200; blue,200}](2.5,-4.2)--(39.5,-4.2);
\draw[draw={rgb,255:red,200; green,200; blue,200}](2.5,-5.2)--(39.5,-5.2);
\draw[draw={rgb,255:red,200; green,200; blue,200}](2.5,-6.2)--(39.5,-6.2);
\draw[draw={rgb,255:red,200; green,200; blue,200}](2.5,-7.2)--(39.5,-7.2);
\draw[draw={rgb,255:red,200; green,200; blue,200}](2.5,-8.2)--(39.5,-8.2);
\draw[draw={rgb,255:red,200; green,200; blue,200}](2.5,-9.2)--(39.5,-9.2);
\draw[draw={rgb,255:red,200; green,200; blue,200}](2.5,-10.2)--(39.5,-10.2);
\draw[draw={rgb,255:red,200; green,200; blue,200}](2.5,-11.2)--(39.5,-11.2);
\draw[draw={rgb,255:red,200; green,200; blue,200}](2.5,-12.2)--(39.5,-12.2);
\draw[draw={rgb,255:red,200; green,200; blue,200}](2.5,-13.2)--(39.5,-13.2);
\draw[draw={rgb,255:red,200; green,200; blue,200}](2.5,-14.2)--(39.5,-14.2);
\draw[draw={rgb,255:red,200; green,200; blue,200}](2.5,-15.2)--(39.5,-15.2);
\draw[draw={rgb,255:red,200; green,200; blue,200}](2.5,-16.2)--(39.5,-16.2);
\draw[draw={rgb,255:red,200; green,200; blue,200}](2.5,-17.2)--(39.5,-17.2);
\draw[draw={rgb,255:red,200; green,200; blue,200}](2.5,-18.2)--(39.5,-18.2);
\draw[draw={rgb,255:red,0; green,0; blue,0},fill={rgb,255:red,200; green,113; blue,55},opacity=0.5,fill opacity=0.5](3.65,-5.35) rectangle (4.35, -6.05);
\draw[draw={rgb,255:red,0; green,0; blue,0},fill={rgb,255:red,200; green,113; blue,55},opacity=0.5,fill opacity=0.5](4.65,-5.35) rectangle (5.35, -6.05);
\draw[draw={rgb,255:red,0; green,0; blue,0},fill={rgb,255:red,200; green,113; blue,55},opacity=0.5,fill opacity=0.5](4.65,-4.35) rectangle (5.35, -5.05);
\draw[draw={rgb,255:red,0; green,0; blue,0},fill={rgb,255:red,200; green,113; blue,55},opacity=0.5,fill opacity=0.5](4.65,-3.35) rectangle (5.35, -4.05);
\draw[draw={rgb,255:red,0; green,0; blue,0},fill={rgb,255:red,200; green,113; blue,55},opacity=0.5,fill opacity=0.5](5.65,-3.35) rectangle (6.35, -4.05);
\draw[draw={rgb,255:red,0; green,0; blue,0},fill={rgb,255:red,200; green,113; blue,55},opacity=0.5,fill opacity=0.5](5.65,-4.35) rectangle (6.35, -5.05);
\draw[draw={rgb,255:red,0; green,0; blue,0},fill={rgb,255:red,200; green,113; blue,55},opacity=0.5,fill opacity=0.5](6.65,-4.35) rectangle (7.35, -5.05);
\draw[draw={rgb,255:red,0; green,0; blue,0},fill={rgb,255:red,200; green,113; blue,55},opacity=0.5,fill opacity=0.5](7.65,-4.35) rectangle (8.35, -5.05);
\draw[draw={rgb,255:red,0; green,0; blue,0},fill={rgb,255:red,200; green,113; blue,55},opacity=0.5,fill opacity=0.5](7.65,-3.35) rectangle (8.35, -4.05);
\draw[draw={rgb,255:red,0; green,0; blue,0},fill={rgb,255:red,200; green,113; blue,55},opacity=0.5,fill opacity=0.5](7.65,-2.35) rectangle (8.35, -3.05);
\draw[draw={rgb,255:red,0; green,0; blue,0},fill={rgb,255:red,200; green,113; blue,55},opacity=0.5,fill opacity=0.5](7.65,-1.35) rectangle (8.35, -2.05);
\draw[draw={rgb,255:red,0; green,0; blue,0},fill={rgb,255:red,200; green,113; blue,55},opacity=0.5,fill opacity=0.5](8.65,-1.35) rectangle (9.35, -2.05);
\draw[draw={rgb,255:red,0; green,0; blue,0},fill={rgb,255:red,200; green,113; blue,55},opacity=0.5,fill opacity=0.5](9.65,-1.35) rectangle (10.35, -2.05);
\draw[draw={rgb,255:red,0; green,0; blue,0},fill={rgb,255:red,200; green,113; blue,55},opacity=0.5,fill opacity=0.5](9.65,-2.35) rectangle (10.35, -3.05);
\draw[draw={rgb,255:red,0; green,0; blue,0},fill={rgb,255:red,200; green,113; blue,55},opacity=0.5,fill opacity=0.5](10.65,-2.35) rectangle (11.35, -3.05);
\draw[draw={rgb,255:red,0; green,0; blue,0},fill={rgb,255:red,200; green,113; blue,55},opacity=0.5,fill opacity=0.5](11.65,-2.35) rectangle (12.35, -3.05);
\draw[draw={rgb,255:red,0; green,0; blue,0},fill={rgb,255:red,200; green,113; blue,55},opacity=0.5,fill opacity=0.5](12.65,-2.35) rectangle (13.35, -3.05);
\draw[draw={rgb,255:red,0; green,0; blue,0},fill={rgb,255:red,200; green,113; blue,55},opacity=0.5,fill opacity=0.5](12.65,-3.35) rectangle (13.35, -4.05);
\draw[draw={rgb,255:red,0; green,0; blue,0},fill={rgb,255:red,200; green,113; blue,55},opacity=0.5,fill opacity=0.5](12.65,-4.35) rectangle (13.35, -5.05);
\draw[draw={rgb,255:red,0; green,0; blue,0},fill={rgb,255:red,200; green,113; blue,55},opacity=0.5,fill opacity=0.5](12.65,-5.35) rectangle (13.35, -6.05);
\draw[draw={rgb,255:red,0; green,0; blue,0},fill={rgb,255:red,200; green,113; blue,55},opacity=0.5,fill opacity=0.5](12.65,-6.35) rectangle (13.35, -7.05);
\draw[draw={rgb,255:red,0; green,0; blue,0},fill={rgb,255:red,200; green,113; blue,55},opacity=0.5,fill opacity=0.5](13.65,-6.35) rectangle (14.35, -7.05);
\draw[draw={rgb,255:red,0; green,0; blue,0},fill={rgb,255:red,200; green,113; blue,55},opacity=0.5,fill opacity=0.5](14.65,-6.35) rectangle (15.35, -7.05);
\draw[draw={rgb,255:red,0; green,0; blue,0},fill={rgb,255:red,200; green,113; blue,55},opacity=0.5,fill opacity=0.5](14.65,-5.35) rectangle (15.35, -6.05);
\draw[draw={rgb,255:red,0; green,0; blue,0},fill={rgb,255:red,200; green,113; blue,55},opacity=0.5,fill opacity=0.5](15.65,-5.35) rectangle (16.35, -6.05);
\draw[draw={rgb,255:red,0; green,0; blue,0},fill={rgb,255:red,200; green,113; blue,55},opacity=0.5,fill opacity=0.5](16.65,-5.35) rectangle (17.35, -6.05);
\draw[draw={rgb,255:red,0; green,0; blue,0},fill={rgb,255:red,200; green,113; blue,55},opacity=0.5,fill opacity=0.5](16.65,-4.35) rectangle (17.35, -5.05);
\draw[draw={rgb,255:red,0; green,0; blue,0},fill={rgb,255:red,200; green,113; blue,55},opacity=0.5,fill opacity=0.5](16.65,-3.35) rectangle (17.35, -4.05);
\draw[draw={rgb,255:red,0; green,0; blue,0},fill={rgb,255:red,200; green,113; blue,55},opacity=0.5,fill opacity=0.5](16.65,-2.35) rectangle (17.35, -3.05);
\draw[draw={rgb,255:red,0; green,0; blue,0},fill={rgb,255:red,200; green,113; blue,55},opacity=0.5,fill opacity=0.5](16.65,-1.35) rectangle (17.35, -2.05);
\draw[draw={rgb,255:red,0; green,0; blue,0},fill={rgb,255:red,200; green,113; blue,55},opacity=0.5,fill opacity=0.5](16.65,-0.35) rectangle (17.35, -1.05);
\draw[draw={rgb,255:red,0; green,0; blue,0},fill={rgb,255:red,200; green,113; blue,55},opacity=0.5,fill opacity=0.5](17.65,-0.35) rectangle (18.35, -1.05);
\draw[draw={rgb,255:red,0; green,0; blue,0},fill={rgb,255:red,200; green,113; blue,55},opacity=0.5,fill opacity=0.5](18.65,-0.35) rectangle (19.35, -1.05);
\draw[draw={rgb,255:red,0; green,0; blue,0},fill={rgb,255:red,200; green,113; blue,55},opacity=0.5,fill opacity=0.5](18.65,0.65) rectangle (19.35, -0.05);
\draw[draw={rgb,255:red,0; green,0; blue,0},fill={rgb,255:red,200; green,113; blue,55},opacity=0.5,fill opacity=0.5](19.65,0.65) rectangle (20.35, -0.05);
\draw[draw={rgb,255:red,0; green,0; blue,0},fill={rgb,255:red,200; green,113; blue,55},opacity=0.5,fill opacity=0.5](20.65,0.65) rectangle (21.35, -0.05);
\draw[draw={rgb,255:red,0; green,0; blue,0},fill={rgb,255:red,200; green,113; blue,55},opacity=0.5,fill opacity=0.5](20.65,-0.35) rectangle (21.35, -1.05);
\draw[draw={rgb,255:red,0; green,0; blue,0},fill={rgb,255:red,200; green,113; blue,55},opacity=0.5,fill opacity=0.5](20.65,-1.35) rectangle (21.35, -2.05);
\draw[draw={rgb,255:red,0; green,0; blue,0},fill={rgb,255:red,200; green,113; blue,55},opacity=0.5,fill opacity=0.5](21.65,-1.35) rectangle (22.35, -2.05);
\draw[draw={rgb,255:red,0; green,0; blue,0},fill={rgb,255:red,200; green,113; blue,55},opacity=0.5,fill opacity=0.5](22.65,-1.35) rectangle (23.35, -2.05);
\draw[draw={rgb,255:red,0; green,0; blue,0},fill={rgb,255:red,200; green,113; blue,55},opacity=0.5,fill opacity=0.5](23.65,-1.35) rectangle (24.35, -2.05);
\draw[draw={rgb,255:red,0; green,0; blue,0},fill={rgb,255:red,200; green,113; blue,55},opacity=0.5,fill opacity=0.5](23.65,-2.35) rectangle (24.35, -3.05);
\draw[draw={rgb,255:red,0; green,0; blue,0},fill={rgb,255:red,200; green,113; blue,55},opacity=0.5,fill opacity=0.5](23.65,-3.35) rectangle (24.35, -4.05);
\draw[draw={rgb,255:red,0; green,0; blue,0},fill={rgb,255:red,200; green,113; blue,55},opacity=0.5,fill opacity=0.5](24.65,-3.35) rectangle (25.35, -4.05);
\draw[draw={rgb,255:red,0; green,0; blue,0},fill={rgb,255:red,200; green,113; blue,55},opacity=0.5,fill opacity=0.5](25.65,-3.35) rectangle (26.35, -4.05);
\draw[draw={rgb,255:red,0; green,0; blue,0},fill={rgb,255:red,200; green,113; blue,55},opacity=0.5,fill opacity=0.5](26.65,-3.35) rectangle (27.35, -4.05);
\draw[draw={rgb,255:red,0; green,0; blue,0},fill={rgb,255:red,200; green,113; blue,55},opacity=0.5,fill opacity=0.5](26.65,-2.35) rectangle (27.35, -3.05);
\draw[draw={rgb,255:red,0; green,0; blue,0},fill={rgb,255:red,200; green,113; blue,55},opacity=0.5,fill opacity=0.5](26.65,-1.35) rectangle (27.35, -2.05);
\draw[draw={rgb,255:red,0; green,0; blue,0},fill={rgb,255:red,200; green,113; blue,55},opacity=0.5,fill opacity=0.5](26.65,-0.35) rectangle (27.35, -1.05);
\draw[draw={rgb,255:red,0; green,0; blue,0},fill={rgb,255:red,200; green,113; blue,55},opacity=0.5,fill opacity=0.5](26.65,0.65) rectangle (27.35, -0.05);
\draw[draw={rgb,255:red,0; green,0; blue,0},fill={rgb,255:red,200; green,113; blue,55},opacity=0.5,fill opacity=0.5](27.65,0.65) rectangle (28.35, -0.05);
\draw[draw={rgb,255:red,0; green,0; blue,0},fill={rgb,255:red,200; green,113; blue,55},opacity=0.5,fill opacity=0.5](28.65,0.65) rectangle (29.35, -0.05);
\draw[draw={rgb,255:red,0; green,0; blue,0},fill={rgb,255:red,200; green,113; blue,55},opacity=0.5,fill opacity=0.5](29.65,0.65) rectangle (30.35, -0.05);
\draw[draw={rgb,255:red,0; green,0; blue,0},fill={rgb,255:red,200; green,113; blue,55},opacity=0.5,fill opacity=0.5](30.65,0.65) rectangle (31.35, -0.05);
\draw[draw={rgb,255:red,0; green,0; blue,0},fill={rgb,255:red,200; green,113; blue,55},opacity=0.5,fill opacity=0.5](31.65,0.65) rectangle (32.35, -0.05);
\draw[draw={rgb,255:red,0; green,0; blue,0},fill={rgb,255:red,200; green,113; blue,55},opacity=0.5,fill opacity=0.5](32.65,0.65) rectangle (33.35, -0.05);
\draw[draw={rgb,255:red,200; green,113; blue,55},opacity=0.5](4,-5.7)--(5,-5.7)--(5,-3.7)--(6,-3.7)--(6,-4.7)--(8,-4.7)--(8,-1.7)--(10,-1.7)--(10,-2.7)--(13,-2.7)--(13,-6.7)--(15,-6.7)--(15,-5.7)--(17,-5.7)--(17,-0.7)--(19,-0.7)--(19,0.3)--(21,0.3)--(21,-1.7)--(24,-1.7)--(24,-3.7)--(27,-3.7)--(27,0.3)--(33,0.3);
\draw[draw=none,fill={rgb,255:red,0; green,0; blue,0},opacity=0.05](27,0.3)--(27,-3.7)--(24,-3.7)--(24,-1.7)--(21,-1.7)--(21,0.3)--(19,0.3)--(19,-0.7)--(17,-0.7)--(17,-5.7)--(15,-5.7)--(15,-6.7)--(13,-6.7)--(13,-2.7)--(10,-2.7)--(10,-1.7)--(8,-1.7)--(8,-4.7)--(6,-4.7)--(6,-3.7)--(5,-3.7)--(5,-5.7)--(2.5,-5.7)--(2.5,-18.2)--(39.5,-18.2)--(39.5,1.8)--(32.5,1.8)--(32.5,0.3)-- cycle;
\draw(37, 0.11) node[anchor=south west] {$\mathcal C$};
\draw[draw={rgb,255:red,0; green,0; blue,0},opacity=0.05](27,0.3)--(27,-3.7)--(24,-3.7)--(24,-1.7)--(21,-1.7)--(21,0.3)--(19,0.3)--(19,-0.7)--(17,-0.7)--(17,-5.7)--(15,-5.7)--(15,-6.7)--(13,-6.7)--(13,-2.7)--(10,-2.7)--(10,-1.7)--(8,-1.7)--(8,-4.7)--(6,-4.7)--(6,-3.7)--(5,-3.7)--(5,-5.7)--(3.5,-5.7)--(3.5,-17.2)--(38.5,-17.2)--(38.5,0.8)--(32.5,0.8)--(32.5,0.3)-- cycle;
\draw[draw=none,fill={rgb,255:red,0; green,0; blue,0},thin](11, -2.7) ellipse (0.1cm and 0.1cm);\draw[draw=none,fill={rgb,255:red,0; green,0; blue,0},thin](17, -2.7) ellipse (0.1cm and 0.1cm);\draw[draw=none,fill={rgb,255:red,0; green,0; blue,0},thin](14, -6.7) ellipse (0.1cm and 0.1cm);\draw(10.24, -2.46) node[anchor=south west] {$u_n$};
\draw(15.26, -3.26) node[anchor=south west] {$v_n$};
\draw(13.25, -8.23) node[anchor=south west] {$m_n$};
\draw[draw={rgb,255:red,0; green,0; blue,0}](14,-6.7)--(14,-8.7)--(16,-8.7)--(16,-10.7)--(15,-10.7)--(15,-11.7)--(18,-11.7)--(18,-9.7)--(20,-9.7)--(20,-16.7);
\draw[draw=none,fill={rgb,255:red,0; green,0; blue,0},thin](32.5, 0.3) ellipse (0.1cm and 0.1cm);\draw(32.63, 4.3) node[anchor=south west] {$\ell^k$};
\draw[draw={rgb,255:red,0; green,0; blue,0}](32.5,0.3)--(32.5,5.3);
\draw[draw={rgb,255:red,0; green,0; blue,0},dashed,thin](20,-18.7)--(20,-16.7);
\draw(18.65, -17.51) node[anchor=south west] {$f_n$};
\end{tikzpicture}
    \end{center}
    \caption{Definition of $u_n \leq m_n \leq v_n$. We have $P_{u_n} = P_{v_n}+\vpji$. The curve $f_n$ ends in $m_n$ which is the only intersection between $f_n$ and $c$.}
    \label{fig:induction-setting}
  \end{figure}

  \argument{Restricting $\mathcal C$}
  We start by defining a connected component $\mathcal H_n\subset\mathcal C$, so that at least one translation of $P_{\rng {u_n} {m_n}}$ by a multiple of $\vpij$ must intersect the border of $\mathcal H_n$.
  Let $h_n$ be the curve defined as:
  $$h_n=\concat{
  f_n+\vpij,\,
  \reverse{\embed{P_{\rng {u_n} {m_n}}}}+\vpij,\,
   \embed{P_{\rng {v_n} k}},\,
    \gs{P_k}{l^k(0)}, \,
    l^k}
  $$
  The endpoints of the successive parts of that concatenation are equal. Moreover, by~\ref{umv:f}, $f_n$ is a curve, hence $h_n$ is also a curve. We now claim that $h_n$ is simple, by considering each of its parts in the order of the concatenation, and checking each time that that part does not intersect the remaining parts.
  First, by~\ref{umv:c}, $\concat{f_n+\vpij,\reverse{\embed{P_{\rng {u_n} {m_n}}}}+\vpij} = \concat{f_n,\reverse{\embed{P_{\rng {u_n} {m_n}}}}}+\vpij$ is a simple curve. Moreover, by~\ref{umv:s}, $f_n+\vpij$ does not intersect $c$, hence intersects neither $\embed{P_{\rng {v_n} k}}$, $\gs{P_k,l^k(0)}$ nor $l^k$. And finally, by~\ref{umv:unvn}, the only intersection between $\embed{P_{\rng {u_n} {m_n}}}+\vpij$ and $c$ is at $\pos{P_{u_n}}+\vpij = \pos{P_{v_n}}$. Finally, since $c$ is a simple curve (by \subl{lem:c-cuts}), so is $\concat{   \embed{P_{\rng {v_n} k}},\,
    \gs{P_k}{l^k(0)}, \,
    l^k}$, and therefore, $h_n$ is simple.

Also, we claim that $h_n$ is entirely in $\mathcal C$. Indeed, $\embed{P_{\rng {v_n} k}},\,
    \gs{P_k}{l^k(0)}, \,
    l^k$ are on $c$, the border of $\mathcal{C}$. Moreover, since $f_n$ is in $\mathcal{C}$ (\ref{umv:c}), and $f_n+\vpij $ is to the east of of $f_n$ and does not intersect $c$ (\ref{umv:s}),
    we get that $f_n+\vpij $ is~in~$\mathcal{C}$.
    Finally, $f_n(0)+\vpij =\pos{P_{m_n}}+\vpij$ (\ref{umv:c}) is~in~$\mathcal{C}$ and since the only intersection between $\reverse{\embed{P_{\rng {u_n} {m_n}}}}+\vpij$ and $c$ is $\pos{P_{u_n}}+\vpij$ (\ref{umv:unvn}), we get that
    $\reverse{\embed{P_{\rng {u_n} {m_n}}}}+\vpij$ is in $\mathcal{C}$.
    
    Since $h_n$ is a simple curve, entirely in $\mathcal{C}$, that begins and ends with a vertical ray, $h_n$ partitions $\mathcal C$ into two connected components. Let $\mathcal H_n$ be the connected component on the right-hand side of $h_n$, including $h_n$ itself, as shown in Figure~\ref{fig:induction-hn}.

  \begin{figure}[t]
    \begin{center}
      \begin{tikzpicture}[scale=\scale]\draw[draw={rgb,255:red,200; green,200; blue,200}](2.5,1.8) rectangle (39.5, -18.2);
\draw[draw={rgb,255:red,200; green,200; blue,200}](2.5,-18.2)--(2.5,1.8);
\draw[draw={rgb,255:red,200; green,200; blue,200}](3.5,-18.2)--(3.5,1.8);
\draw[draw={rgb,255:red,200; green,200; blue,200}](4.5,-18.2)--(4.5,1.8);
\draw[draw={rgb,255:red,200; green,200; blue,200}](5.5,-18.2)--(5.5,1.8);
\draw[draw={rgb,255:red,200; green,200; blue,200}](6.5,-18.2)--(6.5,1.8);
\draw[draw={rgb,255:red,200; green,200; blue,200}](7.5,-18.2)--(7.5,1.8);
\draw[draw={rgb,255:red,200; green,200; blue,200}](8.5,-18.2)--(8.5,1.8);
\draw[draw={rgb,255:red,200; green,200; blue,200}](9.5,-18.2)--(9.5,1.8);
\draw[draw={rgb,255:red,200; green,200; blue,200}](10.5,-18.2)--(10.5,1.8);
\draw[draw={rgb,255:red,200; green,200; blue,200}](11.5,-18.2)--(11.5,1.8);
\draw[draw={rgb,255:red,200; green,200; blue,200}](12.5,-18.2)--(12.5,1.8);
\draw[draw={rgb,255:red,200; green,200; blue,200}](13.5,-18.2)--(13.5,1.8);
\draw[draw={rgb,255:red,200; green,200; blue,200}](14.5,-18.2)--(14.5,1.8);
\draw[draw={rgb,255:red,200; green,200; blue,200}](15.5,-18.2)--(15.5,1.8);
\draw[draw={rgb,255:red,200; green,200; blue,200}](16.5,-18.2)--(16.5,1.8);
\draw[draw={rgb,255:red,200; green,200; blue,200}](17.5,-18.2)--(17.5,1.8);
\draw[draw={rgb,255:red,200; green,200; blue,200}](18.5,-18.2)--(18.5,1.8);
\draw[draw={rgb,255:red,200; green,200; blue,200}](19.5,-18.2)--(19.5,1.8);
\draw[draw={rgb,255:red,200; green,200; blue,200}](20.5,-18.2)--(20.5,1.8);
\draw[draw={rgb,255:red,200; green,200; blue,200}](21.5,-18.2)--(21.5,1.8);
\draw[draw={rgb,255:red,200; green,200; blue,200}](22.5,-18.2)--(22.5,1.8);
\draw[draw={rgb,255:red,200; green,200; blue,200}](23.5,-18.2)--(23.5,1.8);
\draw[draw={rgb,255:red,200; green,200; blue,200}](24.5,-18.2)--(24.5,1.8);
\draw[draw={rgb,255:red,200; green,200; blue,200}](25.5,-18.2)--(25.5,1.8);
\draw[draw={rgb,255:red,200; green,200; blue,200}](26.5,-18.2)--(26.5,1.8);
\draw[draw={rgb,255:red,200; green,200; blue,200}](27.5,-18.2)--(27.5,1.8);
\draw[draw={rgb,255:red,200; green,200; blue,200}](28.5,-18.2)--(28.5,1.8);
\draw[draw={rgb,255:red,200; green,200; blue,200}](29.5,-18.2)--(29.5,1.8);
\draw[draw={rgb,255:red,200; green,200; blue,200}](30.5,-18.2)--(30.5,1.8);
\draw[draw={rgb,255:red,200; green,200; blue,200}](31.5,-18.2)--(31.5,1.8);
\draw[draw={rgb,255:red,200; green,200; blue,200}](32.5,-18.2)--(32.5,1.8);
\draw[draw={rgb,255:red,200; green,200; blue,200}](33.5,-18.2)--(33.5,1.8);
\draw[draw={rgb,255:red,200; green,200; blue,200}](34.5,-18.2)--(34.5,1.8);
\draw[draw={rgb,255:red,200; green,200; blue,200}](35.5,-18.2)--(35.5,1.8);
\draw[draw={rgb,255:red,200; green,200; blue,200}](36.5,-18.2)--(36.5,1.8);
\draw[draw={rgb,255:red,200; green,200; blue,200}](37.5,-18.2)--(37.5,1.8);
\draw[draw={rgb,255:red,200; green,200; blue,200}](38.5,-18.2)--(38.5,1.8);
\draw[draw={rgb,255:red,200; green,200; blue,200}](39.5,-18.2)--(39.5,1.8);
\draw[draw={rgb,255:red,200; green,200; blue,200}](2.5,1.8)--(39.5,1.8);
\draw[draw={rgb,255:red,200; green,200; blue,200}](2.5,0.8)--(39.5,0.8);
\draw[draw={rgb,255:red,200; green,200; blue,200}](2.5,-0.2)--(39.5,-0.2);
\draw[draw={rgb,255:red,200; green,200; blue,200}](2.5,-1.2)--(39.5,-1.2);
\draw[draw={rgb,255:red,200; green,200; blue,200}](2.5,-2.2)--(39.5,-2.2);
\draw[draw={rgb,255:red,200; green,200; blue,200}](2.5,-3.2)--(39.5,-3.2);
\draw[draw={rgb,255:red,200; green,200; blue,200}](2.5,-4.2)--(39.5,-4.2);
\draw[draw={rgb,255:red,200; green,200; blue,200}](2.5,-5.2)--(39.5,-5.2);
\draw[draw={rgb,255:red,200; green,200; blue,200}](2.5,-6.2)--(39.5,-6.2);
\draw[draw={rgb,255:red,200; green,200; blue,200}](2.5,-7.2)--(39.5,-7.2);
\draw[draw={rgb,255:red,200; green,200; blue,200}](2.5,-8.2)--(39.5,-8.2);
\draw[draw={rgb,255:red,200; green,200; blue,200}](2.5,-9.2)--(39.5,-9.2);
\draw[draw={rgb,255:red,200; green,200; blue,200}](2.5,-10.2)--(39.5,-10.2);
\draw[draw={rgb,255:red,200; green,200; blue,200}](2.5,-11.2)--(39.5,-11.2);
\draw[draw={rgb,255:red,200; green,200; blue,200}](2.5,-12.2)--(39.5,-12.2);
\draw[draw={rgb,255:red,200; green,200; blue,200}](2.5,-13.2)--(39.5,-13.2);
\draw[draw={rgb,255:red,200; green,200; blue,200}](2.5,-14.2)--(39.5,-14.2);
\draw[draw={rgb,255:red,200; green,200; blue,200}](2.5,-15.2)--(39.5,-15.2);
\draw[draw={rgb,255:red,200; green,200; blue,200}](2.5,-16.2)--(39.5,-16.2);
\draw[draw={rgb,255:red,200; green,200; blue,200}](2.5,-17.2)--(39.5,-17.2);
\draw[draw={rgb,255:red,200; green,200; blue,200}](2.5,-18.2)--(39.5,-18.2);
\draw[draw={rgb,255:red,0; green,0; blue,0},fill={rgb,255:red,200; green,113; blue,55},opacity=0.5,fill opacity=0.5](3.65,-5.35) rectangle (4.35, -6.05);
\draw[draw={rgb,255:red,0; green,0; blue,0},fill={rgb,255:red,200; green,113; blue,55},opacity=0.5,fill opacity=0.5](4.65,-5.35) rectangle (5.35, -6.05);
\draw[draw={rgb,255:red,0; green,0; blue,0},fill={rgb,255:red,200; green,113; blue,55},opacity=0.5,fill opacity=0.5](4.65,-4.35) rectangle (5.35, -5.05);
\draw[draw={rgb,255:red,0; green,0; blue,0},fill={rgb,255:red,200; green,113; blue,55},opacity=0.5,fill opacity=0.5](4.65,-3.35) rectangle (5.35, -4.05);
\draw[draw={rgb,255:red,0; green,0; blue,0},fill={rgb,255:red,200; green,113; blue,55},opacity=0.5,fill opacity=0.5](5.65,-3.35) rectangle (6.35, -4.05);
\draw[draw={rgb,255:red,0; green,0; blue,0},fill={rgb,255:red,200; green,113; blue,55},opacity=0.5,fill opacity=0.5](5.65,-4.35) rectangle (6.35, -5.05);
\draw[draw={rgb,255:red,0; green,0; blue,0},fill={rgb,255:red,200; green,113; blue,55},opacity=0.5,fill opacity=0.5](6.65,-4.35) rectangle (7.35, -5.05);
\draw[draw={rgb,255:red,0; green,0; blue,0},fill={rgb,255:red,200; green,113; blue,55},opacity=0.5,fill opacity=0.5](7.65,-4.35) rectangle (8.35, -5.05);
\draw[draw={rgb,255:red,0; green,0; blue,0},fill={rgb,255:red,200; green,113; blue,55},opacity=0.5,fill opacity=0.5](7.65,-3.35) rectangle (8.35, -4.05);
\draw[draw={rgb,255:red,0; green,0; blue,0},fill={rgb,255:red,200; green,113; blue,55},opacity=0.5,fill opacity=0.5](7.65,-2.35) rectangle (8.35, -3.05);
\draw[draw={rgb,255:red,0; green,0; blue,0},fill={rgb,255:red,200; green,113; blue,55},opacity=0.5,fill opacity=0.5](7.65,-1.35) rectangle (8.35, -2.05);
\draw[draw={rgb,255:red,0; green,0; blue,0},fill={rgb,255:red,200; green,113; blue,55},opacity=0.5,fill opacity=0.5](8.65,-1.35) rectangle (9.35, -2.05);
\draw[draw={rgb,255:red,0; green,0; blue,0},fill={rgb,255:red,200; green,113; blue,55},opacity=0.5,fill opacity=0.5](9.65,-1.35) rectangle (10.35, -2.05);
\draw[draw={rgb,255:red,0; green,0; blue,0},fill={rgb,255:red,200; green,113; blue,55},opacity=0.5,fill opacity=0.5](9.65,-2.35) rectangle (10.35, -3.05);
\draw[draw={rgb,255:red,0; green,0; blue,0},fill={rgb,255:red,200; green,113; blue,55},opacity=0.5,fill opacity=0.5](10.65,-2.35) rectangle (11.35, -3.05);
\draw[draw={rgb,255:red,0; green,0; blue,0},fill={rgb,255:red,200; green,113; blue,55},opacity=0.5,fill opacity=0.5](11.65,-2.35) rectangle (12.35, -3.05);
\draw[draw={rgb,255:red,0; green,0; blue,0},fill={rgb,255:red,200; green,113; blue,55},opacity=0.5,fill opacity=0.5](12.65,-2.35) rectangle (13.35, -3.05);
\draw[draw={rgb,255:red,0; green,0; blue,0},fill={rgb,255:red,200; green,113; blue,55},opacity=0.5,fill opacity=0.5](12.65,-3.35) rectangle (13.35, -4.05);
\draw[draw={rgb,255:red,0; green,0; blue,0},fill={rgb,255:red,200; green,113; blue,55},opacity=0.5,fill opacity=0.5](12.65,-4.35) rectangle (13.35, -5.05);
\draw[draw={rgb,255:red,0; green,0; blue,0},fill={rgb,255:red,200; green,113; blue,55},opacity=0.5,fill opacity=0.5](12.65,-5.35) rectangle (13.35, -6.05);
\draw[draw={rgb,255:red,0; green,0; blue,0},fill={rgb,255:red,200; green,113; blue,55},opacity=0.5,fill opacity=0.5](12.65,-6.35) rectangle (13.35, -7.05);
\draw[draw={rgb,255:red,0; green,0; blue,0},fill={rgb,255:red,200; green,113; blue,55},opacity=0.5,fill opacity=0.5](13.65,-6.35) rectangle (14.35, -7.05);
\draw[draw={rgb,255:red,0; green,0; blue,0},fill={rgb,255:red,200; green,113; blue,55},opacity=0.5,fill opacity=0.5](14.65,-6.35) rectangle (15.35, -7.05);
\draw[draw={rgb,255:red,0; green,0; blue,0},fill={rgb,255:red,200; green,113; blue,55},opacity=0.5,fill opacity=0.5](14.65,-5.35) rectangle (15.35, -6.05);
\draw[draw={rgb,255:red,0; green,0; blue,0},fill={rgb,255:red,200; green,113; blue,55},opacity=0.5,fill opacity=0.5](15.65,-5.35) rectangle (16.35, -6.05);
\draw[draw={rgb,255:red,0; green,0; blue,0},fill={rgb,255:red,200; green,113; blue,55},opacity=0.5,fill opacity=0.5](16.65,-5.35) rectangle (17.35, -6.05);
\draw[draw={rgb,255:red,0; green,0; blue,0},fill={rgb,255:red,200; green,113; blue,55},opacity=0.5,fill opacity=0.5](16.65,-4.35) rectangle (17.35, -5.05);
\draw[draw={rgb,255:red,0; green,0; blue,0},fill={rgb,255:red,200; green,113; blue,55},opacity=0.5,fill opacity=0.5](16.65,-3.35) rectangle (17.35, -4.05);
\draw[draw={rgb,255:red,0; green,0; blue,0},fill={rgb,255:red,200; green,113; blue,55},opacity=0.5,fill opacity=0.5](16.65,-2.35) rectangle (17.35, -3.05);
\draw[draw={rgb,255:red,0; green,0; blue,0},fill={rgb,255:red,200; green,113; blue,55},opacity=0.5,fill opacity=0.5](16.65,-1.35) rectangle (17.35, -2.05);
\draw[draw={rgb,255:red,0; green,0; blue,0},fill={rgb,255:red,200; green,113; blue,55},opacity=0.5,fill opacity=0.5](16.65,-0.35) rectangle (17.35, -1.05);
\draw[draw={rgb,255:red,0; green,0; blue,0},fill={rgb,255:red,200; green,113; blue,55},opacity=0.5,fill opacity=0.5](17.65,-0.35) rectangle (18.35, -1.05);
\draw[draw={rgb,255:red,0; green,0; blue,0},fill={rgb,255:red,200; green,113; blue,55},opacity=0.5,fill opacity=0.5](18.65,-0.35) rectangle (19.35, -1.05);
\draw[draw={rgb,255:red,0; green,0; blue,0},fill={rgb,255:red,200; green,113; blue,55},opacity=0.5,fill opacity=0.5](18.65,0.65) rectangle (19.35, -0.05);
\draw[draw={rgb,255:red,0; green,0; blue,0},fill={rgb,255:red,200; green,113; blue,55},opacity=0.5,fill opacity=0.5](19.65,0.65) rectangle (20.35, -0.05);
\draw[draw={rgb,255:red,0; green,0; blue,0},fill={rgb,255:red,200; green,113; blue,55},opacity=0.5,fill opacity=0.5](20.65,0.65) rectangle (21.35, -0.05);
\draw[draw={rgb,255:red,0; green,0; blue,0},fill={rgb,255:red,200; green,113; blue,55},opacity=0.5,fill opacity=0.5](20.65,-0.35) rectangle (21.35, -1.05);
\draw[draw={rgb,255:red,0; green,0; blue,0},fill={rgb,255:red,200; green,113; blue,55},opacity=0.5,fill opacity=0.5](20.65,-1.35) rectangle (21.35, -2.05);
\draw[draw={rgb,255:red,0; green,0; blue,0},fill={rgb,255:red,200; green,113; blue,55},opacity=0.5,fill opacity=0.5](21.65,-1.35) rectangle (22.35, -2.05);
\draw[draw={rgb,255:red,0; green,0; blue,0},fill={rgb,255:red,200; green,113; blue,55},opacity=0.5,fill opacity=0.5](22.65,-1.35) rectangle (23.35, -2.05);
\draw[draw={rgb,255:red,0; green,0; blue,0},fill={rgb,255:red,200; green,113; blue,55},opacity=0.5,fill opacity=0.5](23.65,-1.35) rectangle (24.35, -2.05);
\draw[draw={rgb,255:red,0; green,0; blue,0},fill={rgb,255:red,200; green,113; blue,55},opacity=0.5,fill opacity=0.5](23.65,-2.35) rectangle (24.35, -3.05);
\draw[draw={rgb,255:red,0; green,0; blue,0},fill={rgb,255:red,200; green,113; blue,55},opacity=0.5,fill opacity=0.5](23.65,-3.35) rectangle (24.35, -4.05);
\draw[draw={rgb,255:red,0; green,0; blue,0},fill={rgb,255:red,200; green,113; blue,55},opacity=0.5,fill opacity=0.5](24.65,-3.35) rectangle (25.35, -4.05);
\draw[draw={rgb,255:red,0; green,0; blue,0},fill={rgb,255:red,200; green,113; blue,55},opacity=0.5,fill opacity=0.5](25.65,-3.35) rectangle (26.35, -4.05);
\draw[draw={rgb,255:red,0; green,0; blue,0},fill={rgb,255:red,200; green,113; blue,55},opacity=0.5,fill opacity=0.5](26.65,-3.35) rectangle (27.35, -4.05);
\draw[draw={rgb,255:red,0; green,0; blue,0},fill={rgb,255:red,200; green,113; blue,55},opacity=0.5,fill opacity=0.5](26.65,-2.35) rectangle (27.35, -3.05);
\draw[draw={rgb,255:red,0; green,0; blue,0},fill={rgb,255:red,200; green,113; blue,55},opacity=0.5,fill opacity=0.5](26.65,-1.35) rectangle (27.35, -2.05);
\draw[draw={rgb,255:red,0; green,0; blue,0},fill={rgb,255:red,200; green,113; blue,55},opacity=0.5,fill opacity=0.5](26.65,-0.35) rectangle (27.35, -1.05);
\draw[draw={rgb,255:red,0; green,0; blue,0},fill={rgb,255:red,200; green,113; blue,55},opacity=0.5,fill opacity=0.5](26.65,0.65) rectangle (27.35, -0.05);
\draw[draw={rgb,255:red,0; green,0; blue,0},fill={rgb,255:red,200; green,113; blue,55},opacity=0.5,fill opacity=0.5](27.65,0.65) rectangle (28.35, -0.05);
\draw[draw={rgb,255:red,0; green,0; blue,0},fill={rgb,255:red,200; green,113; blue,55},opacity=0.5,fill opacity=0.5](28.65,0.65) rectangle (29.35, -0.05);
\draw[draw={rgb,255:red,0; green,0; blue,0},fill={rgb,255:red,200; green,113; blue,55},opacity=0.5,fill opacity=0.5](29.65,0.65) rectangle (30.35, -0.05);
\draw[draw={rgb,255:red,0; green,0; blue,0},fill={rgb,255:red,200; green,113; blue,55},opacity=0.5,fill opacity=0.5](30.65,0.65) rectangle (31.35, -0.05);
\draw[draw={rgb,255:red,0; green,0; blue,0},fill={rgb,255:red,200; green,113; blue,55},opacity=0.5,fill opacity=0.5](31.65,0.65) rectangle (32.35, -0.05);
\draw[draw={rgb,255:red,0; green,0; blue,0},fill={rgb,255:red,200; green,113; blue,55},opacity=0.5,fill opacity=0.5](32.65,0.65) rectangle (33.35, -0.05);
\draw[draw={rgb,255:red,200; green,113; blue,55},opacity=0.5](4,-5.7)--(5,-5.7)--(5,-3.7)--(6,-3.7)--(6,-4.7)--(8,-4.7)--(8,-1.7)--(10,-1.7)--(10,-2.7)--(13,-2.7)--(13,-6.7)--(15,-6.7)--(15,-5.7)--(17,-5.7)--(17,-0.7)--(19,-0.7)--(19,0.3)--(21,0.3)--(21,-1.7)--(24,-1.7)--(24,-3.7)--(27,-3.7)--(27,0.3)--(33,0.3);
\draw[draw={rgb,255:red,0; green,0; blue,0},fill={rgb,255:red,55; green,200; blue,55},opacity=0.5,fill opacity=0.5](16.65,-2.35) rectangle (17.35, -3.05);
\draw[draw={rgb,255:red,0; green,0; blue,0},fill={rgb,255:red,55; green,200; blue,55},opacity=0.5,fill opacity=0.5](17.65,-2.35) rectangle (18.35, -3.05);
\draw[draw={rgb,255:red,0; green,0; blue,0},fill={rgb,255:red,55; green,200; blue,55},opacity=0.5,fill opacity=0.5](18.65,-2.35) rectangle (19.35, -3.05);
\draw[draw={rgb,255:red,0; green,0; blue,0},fill={rgb,255:red,55; green,200; blue,55},opacity=0.5,fill opacity=0.5](18.65,-3.35) rectangle (19.35, -4.05);
\draw[draw={rgb,255:red,0; green,0; blue,0},fill={rgb,255:red,55; green,200; blue,55},opacity=0.5,fill opacity=0.5](18.65,-4.35) rectangle (19.35, -5.05);
\draw[draw={rgb,255:red,0; green,0; blue,0},fill={rgb,255:red,55; green,200; blue,55},opacity=0.5,fill opacity=0.5](18.65,-5.35) rectangle (19.35, -6.05);
\draw[draw={rgb,255:red,0; green,0; blue,0},fill={rgb,255:red,55; green,200; blue,55},opacity=0.5,fill opacity=0.5](18.65,-6.35) rectangle (19.35, -7.05);
\draw[draw={rgb,255:red,0; green,0; blue,0},fill={rgb,255:red,55; green,200; blue,55},opacity=0.5,fill opacity=0.5](19.65,-6.35) rectangle (20.35, -7.05);
\draw[draw={rgb,255:red,55; green,200; blue,55},opacity=0.5](17,-2.7)--(19,-2.7)--(19,-6.7)--(20,-6.7);
\draw[draw={rgb,255:red,0; green,0; blue,0},opacity=0.05](27,0.3)--(27,-3.7)--(24,-3.7)--(24,-1.7)--(21,-1.7)--(21,0.3)--(19,0.3)--(19,-0.7)--(17,-0.7)--(17,-5.7)--(15,-5.7)--(15,-6.7)--(13,-6.7)--(13,-2.7)--(10,-2.7)--(10,-1.7)--(8,-1.7)--(8,-4.7)--(6,-4.7)--(6,-3.7)--(5,-3.7)--(5,-5.7)--(3.5,-5.7)--(3.5,-17.2)--(38.5,-17.2)--(38.5,0.8)--(32.5,0.8)--(32.5,0.3)-- cycle;
\draw[draw=none,fill={rgb,255:red,0; green,0; blue,0},opacity=0.05](26,-18.2)--(39.5,-18.2)--(39.5,1.8)--(32.5,1.8)--(32.5,0.3)--(27,0.3)--(27,-3.7)--(24,-3.7)--(24,-1.7)--(21,-1.7)--(21,0.3)--(19,0.3)--(19,-0.7)--(17,-0.7)--(17,-2.7)--(19,-2.7)--(19,-6.7)--(20,-6.7)--(20,-8.7)--(22,-8.7)--(22,-10.7)--(21,-10.7)--(21,-11.7)--(24,-11.7)--(24,-9.7)--(26,-9.7);
\draw(37, 0.12) node[anchor=south west] {$\mathcal H_n$};
\draw[draw=none,fill={rgb,255:red,0; green,0; blue,0},thin](11, -2.7) ellipse (0.1cm and 0.1cm);\draw[draw=none,fill={rgb,255:red,0; green,0; blue,0},thin](17, -2.7) ellipse (0.1cm and 0.1cm);\draw[draw=none,fill={rgb,255:red,0; green,0; blue,0},thin](14, -6.7) ellipse (0.1cm and 0.1cm);\draw(10.24, -2.46) node[anchor=south west] {$u_n$};
\draw(15.26, -3.26) node[anchor=south west] {$v_n$};
\draw(13.25, -8.23) node[anchor=south west] {$m_n$};
\draw[draw={rgb,255:red,0; green,0; blue,0}](20,-6.7)--(20,-8.7)--(22,-8.7)--(22,-10.7)--(21,-10.7)--(21,-11.7)--(24,-11.7)--(24,-9.7)--(26,-9.7)--(26,-16.7);
\draw[draw=none,fill={rgb,255:red,0; green,0; blue,0},thin](20, -6.7) ellipse (0.1cm and 0.1cm);\draw[draw=none,fill={rgb,255:red,0; green,0; blue,0},thin](32.5, 0.3) ellipse (0.1cm and 0.1cm);\draw(32.63, 4.3) node[anchor=south west] {$\ell^k$};
\draw[draw={rgb,255:red,0; green,0; blue,0}](32.5,0.3)--(32.5,5.3);
\draw[draw={rgb,255:red,0; green,0; blue,0},dashed,thin](26,-18.7)--(26,-16.7);
\draw(24.65, -17.51) node[anchor=south west] {$f_n+\vect{P_iP_j}$};
\end{tikzpicture}
    \end{center}
    \caption{The component $\mathcal{H}_n$, whose border $h_n$ is defined as the concatenation of $f_n+\vpij$, $P_{\range{u_n}{u_n+1}{m_n}}+\vpij$, $P_{\rng{v_n}{k}}$ and $l^k$ (and a half-length segment between $P_k$ and $l^k(0)$).
}
\label{fig:induction-hn}
  \end{figure}

  \argument{Finding $m_{n+1}$ and $f_{n+1}$}
  Let $a\in\{\rng {u_n} {m_n}\}$ and $t\geq 1$ be the two integers such that $(a, t)$ is the largest pair (ordered first by largest $a$, then by largest $t$) such that $\pos{P_{a}}+t\vpij = \pos{P_{m_{n+1}}}$ for some integer $m_{n+1} \in \{ i+1,i+2, \ldots,k\}$.
  These indices exist because by~\ref{umv:unvn}, with $a = u_n$, $m_{n+1} = v_n$ and $t = 1$, we have $\pos{P_{u_n}}+\vpij=\pos{P_{v_n}}$.

  Using $a$ and $t$, we set $f_{n+1} = \concat{f_n, \reverse {\embed{P_{\rng a {m_n}}}}} + t\vpij$ (see Figure~\ref{fig:induction-mn1}).
  \begin{figure}[ht]
    \begin{center}
      \begin{tikzpicture}[scale=\scale]\draw[draw={rgb,255:red,200; green,200; blue,200}](2.5,1.8) rectangle (39.5, -18.2);
\draw[draw={rgb,255:red,200; green,200; blue,200}](2.5,-18.2)--(2.5,1.8);
\draw[draw={rgb,255:red,200; green,200; blue,200}](3.5,-18.2)--(3.5,1.8);
\draw[draw={rgb,255:red,200; green,200; blue,200}](4.5,-18.2)--(4.5,1.8);
\draw[draw={rgb,255:red,200; green,200; blue,200}](5.5,-18.2)--(5.5,1.8);
\draw[draw={rgb,255:red,200; green,200; blue,200}](6.5,-18.2)--(6.5,1.8);
\draw[draw={rgb,255:red,200; green,200; blue,200}](7.5,-18.2)--(7.5,1.8);
\draw[draw={rgb,255:red,200; green,200; blue,200}](8.5,-18.2)--(8.5,1.8);
\draw[draw={rgb,255:red,200; green,200; blue,200}](9.5,-18.2)--(9.5,1.8);
\draw[draw={rgb,255:red,200; green,200; blue,200}](10.5,-18.2)--(10.5,1.8);
\draw[draw={rgb,255:red,200; green,200; blue,200}](11.5,-18.2)--(11.5,1.8);
\draw[draw={rgb,255:red,200; green,200; blue,200}](12.5,-18.2)--(12.5,1.8);
\draw[draw={rgb,255:red,200; green,200; blue,200}](13.5,-18.2)--(13.5,1.8);
\draw[draw={rgb,255:red,200; green,200; blue,200}](14.5,-18.2)--(14.5,1.8);
\draw[draw={rgb,255:red,200; green,200; blue,200}](15.5,-18.2)--(15.5,1.8);
\draw[draw={rgb,255:red,200; green,200; blue,200}](16.5,-18.2)--(16.5,1.8);
\draw[draw={rgb,255:red,200; green,200; blue,200}](17.5,-18.2)--(17.5,1.8);
\draw[draw={rgb,255:red,200; green,200; blue,200}](18.5,-18.2)--(18.5,1.8);
\draw[draw={rgb,255:red,200; green,200; blue,200}](19.5,-18.2)--(19.5,1.8);
\draw[draw={rgb,255:red,200; green,200; blue,200}](20.5,-18.2)--(20.5,1.8);
\draw[draw={rgb,255:red,200; green,200; blue,200}](21.5,-18.2)--(21.5,1.8);
\draw[draw={rgb,255:red,200; green,200; blue,200}](22.5,-18.2)--(22.5,1.8);
\draw[draw={rgb,255:red,200; green,200; blue,200}](23.5,-18.2)--(23.5,1.8);
\draw[draw={rgb,255:red,200; green,200; blue,200}](24.5,-18.2)--(24.5,1.8);
\draw[draw={rgb,255:red,200; green,200; blue,200}](25.5,-18.2)--(25.5,1.8);
\draw[draw={rgb,255:red,200; green,200; blue,200}](26.5,-18.2)--(26.5,1.8);
\draw[draw={rgb,255:red,200; green,200; blue,200}](27.5,-18.2)--(27.5,1.8);
\draw[draw={rgb,255:red,200; green,200; blue,200}](28.5,-18.2)--(28.5,1.8);
\draw[draw={rgb,255:red,200; green,200; blue,200}](29.5,-18.2)--(29.5,1.8);
\draw[draw={rgb,255:red,200; green,200; blue,200}](30.5,-18.2)--(30.5,1.8);
\draw[draw={rgb,255:red,200; green,200; blue,200}](31.5,-18.2)--(31.5,1.8);
\draw[draw={rgb,255:red,200; green,200; blue,200}](32.5,-18.2)--(32.5,1.8);
\draw[draw={rgb,255:red,200; green,200; blue,200}](33.5,-18.2)--(33.5,1.8);
\draw[draw={rgb,255:red,200; green,200; blue,200}](34.5,-18.2)--(34.5,1.8);
\draw[draw={rgb,255:red,200; green,200; blue,200}](35.5,-18.2)--(35.5,1.8);
\draw[draw={rgb,255:red,200; green,200; blue,200}](36.5,-18.2)--(36.5,1.8);
\draw[draw={rgb,255:red,200; green,200; blue,200}](37.5,-18.2)--(37.5,1.8);
\draw[draw={rgb,255:red,200; green,200; blue,200}](38.5,-18.2)--(38.5,1.8);
\draw[draw={rgb,255:red,200; green,200; blue,200}](39.5,-18.2)--(39.5,1.8);
\draw[draw={rgb,255:red,200; green,200; blue,200}](2.5,1.8)--(39.5,1.8);
\draw[draw={rgb,255:red,200; green,200; blue,200}](2.5,0.8)--(39.5,0.8);
\draw[draw={rgb,255:red,200; green,200; blue,200}](2.5,-0.2)--(39.5,-0.2);
\draw[draw={rgb,255:red,200; green,200; blue,200}](2.5,-1.2)--(39.5,-1.2);
\draw[draw={rgb,255:red,200; green,200; blue,200}](2.5,-2.2)--(39.5,-2.2);
\draw[draw={rgb,255:red,200; green,200; blue,200}](2.5,-3.2)--(39.5,-3.2);
\draw[draw={rgb,255:red,200; green,200; blue,200}](2.5,-4.2)--(39.5,-4.2);
\draw[draw={rgb,255:red,200; green,200; blue,200}](2.5,-5.2)--(39.5,-5.2);
\draw[draw={rgb,255:red,200; green,200; blue,200}](2.5,-6.2)--(39.5,-6.2);
\draw[draw={rgb,255:red,200; green,200; blue,200}](2.5,-7.2)--(39.5,-7.2);
\draw[draw={rgb,255:red,200; green,200; blue,200}](2.5,-8.2)--(39.5,-8.2);
\draw[draw={rgb,255:red,200; green,200; blue,200}](2.5,-9.2)--(39.5,-9.2);
\draw[draw={rgb,255:red,200; green,200; blue,200}](2.5,-10.2)--(39.5,-10.2);
\draw[draw={rgb,255:red,200; green,200; blue,200}](2.5,-11.2)--(39.5,-11.2);
\draw[draw={rgb,255:red,200; green,200; blue,200}](2.5,-12.2)--(39.5,-12.2);
\draw[draw={rgb,255:red,200; green,200; blue,200}](2.5,-13.2)--(39.5,-13.2);
\draw[draw={rgb,255:red,200; green,200; blue,200}](2.5,-14.2)--(39.5,-14.2);
\draw[draw={rgb,255:red,200; green,200; blue,200}](2.5,-15.2)--(39.5,-15.2);
\draw[draw={rgb,255:red,200; green,200; blue,200}](2.5,-16.2)--(39.5,-16.2);
\draw[draw={rgb,255:red,200; green,200; blue,200}](2.5,-17.2)--(39.5,-17.2);
\draw[draw={rgb,255:red,200; green,200; blue,200}](2.5,-18.2)--(39.5,-18.2);
\draw[draw={rgb,255:red,0; green,0; blue,0},fill={rgb,255:red,200; green,113; blue,55},opacity=0.5,fill opacity=0.5](3.65,-5.35) rectangle (4.35, -6.05);
\draw[draw={rgb,255:red,0; green,0; blue,0},fill={rgb,255:red,200; green,113; blue,55},opacity=0.5,fill opacity=0.5](4.65,-5.35) rectangle (5.35, -6.05);
\draw[draw={rgb,255:red,0; green,0; blue,0},fill={rgb,255:red,200; green,113; blue,55},opacity=0.5,fill opacity=0.5](4.65,-4.35) rectangle (5.35, -5.05);
\draw[draw={rgb,255:red,0; green,0; blue,0},fill={rgb,255:red,200; green,113; blue,55},opacity=0.5,fill opacity=0.5](4.65,-3.35) rectangle (5.35, -4.05);
\draw[draw={rgb,255:red,0; green,0; blue,0},fill={rgb,255:red,200; green,113; blue,55},opacity=0.5,fill opacity=0.5](5.65,-3.35) rectangle (6.35, -4.05);
\draw[draw={rgb,255:red,0; green,0; blue,0},fill={rgb,255:red,200; green,113; blue,55},opacity=0.5,fill opacity=0.5](5.65,-4.35) rectangle (6.35, -5.05);
\draw[draw={rgb,255:red,0; green,0; blue,0},fill={rgb,255:red,200; green,113; blue,55},opacity=0.5,fill opacity=0.5](6.65,-4.35) rectangle (7.35, -5.05);
\draw[draw={rgb,255:red,0; green,0; blue,0},fill={rgb,255:red,200; green,113; blue,55},opacity=0.5,fill opacity=0.5](7.65,-4.35) rectangle (8.35, -5.05);
\draw[draw={rgb,255:red,0; green,0; blue,0},fill={rgb,255:red,200; green,113; blue,55},opacity=0.5,fill opacity=0.5](7.65,-3.35) rectangle (8.35, -4.05);
\draw[draw={rgb,255:red,0; green,0; blue,0},fill={rgb,255:red,200; green,113; blue,55},opacity=0.5,fill opacity=0.5](7.65,-2.35) rectangle (8.35, -3.05);
\draw[draw={rgb,255:red,0; green,0; blue,0},fill={rgb,255:red,200; green,113; blue,55},opacity=0.5,fill opacity=0.5](7.65,-1.35) rectangle (8.35, -2.05);
\draw[draw={rgb,255:red,0; green,0; blue,0},fill={rgb,255:red,200; green,113; blue,55},opacity=0.5,fill opacity=0.5](8.65,-1.35) rectangle (9.35, -2.05);
\draw[draw={rgb,255:red,0; green,0; blue,0},fill={rgb,255:red,200; green,113; blue,55},opacity=0.5,fill opacity=0.5](9.65,-1.35) rectangle (10.35, -2.05);
\draw[draw={rgb,255:red,0; green,0; blue,0},fill={rgb,255:red,200; green,113; blue,55},opacity=0.5,fill opacity=0.5](9.65,-2.35) rectangle (10.35, -3.05);
\draw[draw={rgb,255:red,0; green,0; blue,0},fill={rgb,255:red,200; green,113; blue,55},opacity=0.5,fill opacity=0.5](10.65,-2.35) rectangle (11.35, -3.05);
\draw[draw={rgb,255:red,0; green,0; blue,0},fill={rgb,255:red,200; green,113; blue,55},opacity=0.5,fill opacity=0.5](11.65,-2.35) rectangle (12.35, -3.05);
\draw[draw={rgb,255:red,0; green,0; blue,0},fill={rgb,255:red,200; green,113; blue,55},opacity=0.5,fill opacity=0.5](12.65,-2.35) rectangle (13.35, -3.05);
\draw[draw={rgb,255:red,0; green,0; blue,0},fill={rgb,255:red,200; green,113; blue,55},opacity=0.5,fill opacity=0.5](12.65,-3.35) rectangle (13.35, -4.05);
\draw[draw={rgb,255:red,0; green,0; blue,0},fill={rgb,255:red,200; green,113; blue,55},opacity=0.5,fill opacity=0.5](12.65,-4.35) rectangle (13.35, -5.05);
\draw[draw={rgb,255:red,0; green,0; blue,0},fill={rgb,255:red,200; green,113; blue,55},opacity=0.5,fill opacity=0.5](12.65,-5.35) rectangle (13.35, -6.05);
\draw[draw={rgb,255:red,0; green,0; blue,0},fill={rgb,255:red,200; green,113; blue,55},opacity=0.5,fill opacity=0.5](12.65,-6.35) rectangle (13.35, -7.05);
\draw[draw={rgb,255:red,0; green,0; blue,0},fill={rgb,255:red,200; green,113; blue,55},opacity=0.5,fill opacity=0.5](13.65,-6.35) rectangle (14.35, -7.05);
\draw[draw={rgb,255:red,0; green,0; blue,0},fill={rgb,255:red,200; green,113; blue,55},opacity=0.5,fill opacity=0.5](14.65,-6.35) rectangle (15.35, -7.05);
\draw[draw={rgb,255:red,0; green,0; blue,0},fill={rgb,255:red,200; green,113; blue,55},opacity=0.5,fill opacity=0.5](14.65,-5.35) rectangle (15.35, -6.05);
\draw[draw={rgb,255:red,0; green,0; blue,0},fill={rgb,255:red,200; green,113; blue,55},opacity=0.5,fill opacity=0.5](15.65,-5.35) rectangle (16.35, -6.05);
\draw[draw={rgb,255:red,0; green,0; blue,0},fill={rgb,255:red,200; green,113; blue,55},opacity=0.5,fill opacity=0.5](16.65,-5.35) rectangle (17.35, -6.05);
\draw[draw={rgb,255:red,0; green,0; blue,0},fill={rgb,255:red,200; green,113; blue,55},opacity=0.5,fill opacity=0.5](16.65,-4.35) rectangle (17.35, -5.05);
\draw[draw={rgb,255:red,0; green,0; blue,0},fill={rgb,255:red,200; green,113; blue,55},opacity=0.5,fill opacity=0.5](16.65,-3.35) rectangle (17.35, -4.05);
\draw[draw={rgb,255:red,0; green,0; blue,0},fill={rgb,255:red,200; green,113; blue,55},opacity=0.5,fill opacity=0.5](16.65,-2.35) rectangle (17.35, -3.05);
\draw[draw={rgb,255:red,0; green,0; blue,0},fill={rgb,255:red,200; green,113; blue,55},opacity=0.5,fill opacity=0.5](16.65,-1.35) rectangle (17.35, -2.05);
\draw[draw={rgb,255:red,0; green,0; blue,0},fill={rgb,255:red,200; green,113; blue,55},opacity=0.5,fill opacity=0.5](16.65,-0.35) rectangle (17.35, -1.05);
\draw[draw={rgb,255:red,0; green,0; blue,0},fill={rgb,255:red,200; green,113; blue,55},opacity=0.5,fill opacity=0.5](17.65,-0.35) rectangle (18.35, -1.05);
\draw[draw={rgb,255:red,0; green,0; blue,0},fill={rgb,255:red,200; green,113; blue,55},opacity=0.5,fill opacity=0.5](18.65,-0.35) rectangle (19.35, -1.05);
\draw[draw={rgb,255:red,0; green,0; blue,0},fill={rgb,255:red,200; green,113; blue,55},opacity=0.5,fill opacity=0.5](18.65,0.65) rectangle (19.35, -0.05);
\draw[draw={rgb,255:red,0; green,0; blue,0},fill={rgb,255:red,200; green,113; blue,55},opacity=0.5,fill opacity=0.5](19.65,0.65) rectangle (20.35, -0.05);
\draw[draw={rgb,255:red,0; green,0; blue,0},fill={rgb,255:red,200; green,113; blue,55},opacity=0.5,fill opacity=0.5](20.65,0.65) rectangle (21.35, -0.05);
\draw[draw={rgb,255:red,0; green,0; blue,0},fill={rgb,255:red,200; green,113; blue,55},opacity=0.5,fill opacity=0.5](20.65,-0.35) rectangle (21.35, -1.05);
\draw[draw={rgb,255:red,0; green,0; blue,0},fill={rgb,255:red,200; green,113; blue,55},opacity=0.5,fill opacity=0.5](20.65,-1.35) rectangle (21.35, -2.05);
\draw[draw={rgb,255:red,0; green,0; blue,0},fill={rgb,255:red,200; green,113; blue,55},opacity=0.5,fill opacity=0.5](21.65,-1.35) rectangle (22.35, -2.05);
\draw[draw={rgb,255:red,0; green,0; blue,0},fill={rgb,255:red,200; green,113; blue,55},opacity=0.5,fill opacity=0.5](22.65,-1.35) rectangle (23.35, -2.05);
\draw[draw={rgb,255:red,0; green,0; blue,0},fill={rgb,255:red,200; green,113; blue,55},opacity=0.5,fill opacity=0.5](23.65,-1.35) rectangle (24.35, -2.05);
\draw[draw={rgb,255:red,0; green,0; blue,0},fill={rgb,255:red,200; green,113; blue,55},opacity=0.5,fill opacity=0.5](23.65,-2.35) rectangle (24.35, -3.05);
\draw[draw={rgb,255:red,0; green,0; blue,0},fill={rgb,255:red,200; green,113; blue,55},opacity=0.5,fill opacity=0.5](23.65,-3.35) rectangle (24.35, -4.05);
\draw[draw={rgb,255:red,0; green,0; blue,0},fill={rgb,255:red,200; green,113; blue,55},opacity=0.5,fill opacity=0.5](24.65,-3.35) rectangle (25.35, -4.05);
\draw[draw={rgb,255:red,0; green,0; blue,0},fill={rgb,255:red,200; green,113; blue,55},opacity=0.5,fill opacity=0.5](25.65,-3.35) rectangle (26.35, -4.05);
\draw[draw={rgb,255:red,0; green,0; blue,0},fill={rgb,255:red,200; green,113; blue,55},opacity=0.5,fill opacity=0.5](26.65,-3.35) rectangle (27.35, -4.05);
\draw[draw={rgb,255:red,0; green,0; blue,0},fill={rgb,255:red,200; green,113; blue,55},opacity=0.5,fill opacity=0.5](26.65,-2.35) rectangle (27.35, -3.05);
\draw[draw={rgb,255:red,0; green,0; blue,0},fill={rgb,255:red,200; green,113; blue,55},opacity=0.5,fill opacity=0.5](26.65,-1.35) rectangle (27.35, -2.05);
\draw[draw={rgb,255:red,0; green,0; blue,0},fill={rgb,255:red,200; green,113; blue,55},opacity=0.5,fill opacity=0.5](26.65,-0.35) rectangle (27.35, -1.05);
\draw[draw={rgb,255:red,0; green,0; blue,0},fill={rgb,255:red,200; green,113; blue,55},opacity=0.5,fill opacity=0.5](26.65,0.65) rectangle (27.35, -0.05);
\draw[draw={rgb,255:red,0; green,0; blue,0},fill={rgb,255:red,200; green,113; blue,55},opacity=0.5,fill opacity=0.5](27.65,0.65) rectangle (28.35, -0.05);
\draw[draw={rgb,255:red,0; green,0; blue,0},fill={rgb,255:red,200; green,113; blue,55},opacity=0.5,fill opacity=0.5](28.65,0.65) rectangle (29.35, -0.05);
\draw[draw={rgb,255:red,0; green,0; blue,0},fill={rgb,255:red,200; green,113; blue,55},opacity=0.5,fill opacity=0.5](29.65,0.65) rectangle (30.35, -0.05);
\draw[draw={rgb,255:red,0; green,0; blue,0},fill={rgb,255:red,200; green,113; blue,55},opacity=0.5,fill opacity=0.5](30.65,0.65) rectangle (31.35, -0.05);
\draw[draw={rgb,255:red,0; green,0; blue,0},fill={rgb,255:red,200; green,113; blue,55},opacity=0.5,fill opacity=0.5](31.65,0.65) rectangle (32.35, -0.05);
\draw[draw={rgb,255:red,0; green,0; blue,0},fill={rgb,255:red,200; green,113; blue,55},opacity=0.5,fill opacity=0.5](32.65,0.65) rectangle (33.35, -0.05);
\draw[draw={rgb,255:red,200; green,113; blue,55},opacity=0.5](4,-5.7)--(5,-5.7)--(5,-3.7)--(6,-3.7)--(6,-4.7)--(8,-4.7)--(8,-1.7)--(10,-1.7)--(10,-2.7)--(13,-2.7)--(13,-6.7)--(15,-6.7)--(15,-5.7)--(17,-5.7)--(17,-0.7)--(19,-0.7)--(19,0.3)--(21,0.3)--(21,-1.7)--(24,-1.7)--(24,-3.7)--(27,-3.7)--(27,0.3)--(33,0.3);
\draw[draw={rgb,255:red,0; green,0; blue,0},fill={rgb,255:red,55; green,200; blue,55},opacity=0.5,fill opacity=0.5](16.65,-2.35) rectangle (17.35, -3.05);
\draw[draw={rgb,255:red,0; green,0; blue,0},fill={rgb,255:red,55; green,200; blue,55},opacity=0.5,fill opacity=0.5](17.65,-2.35) rectangle (18.35, -3.05);
\draw[draw={rgb,255:red,0; green,0; blue,0},fill={rgb,255:red,55; green,200; blue,55},opacity=0.5,fill opacity=0.5](18.65,-2.35) rectangle (19.35, -3.05);
\draw[draw={rgb,255:red,0; green,0; blue,0},fill={rgb,255:red,55; green,200; blue,55},opacity=0.5,fill opacity=0.5](18.65,-3.35) rectangle (19.35, -4.05);
\draw[draw={rgb,255:red,0; green,0; blue,0},fill={rgb,255:red,55; green,200; blue,55},opacity=0.5,fill opacity=0.5](18.65,-4.35) rectangle (19.35, -5.05);
\draw[draw={rgb,255:red,0; green,0; blue,0},fill={rgb,255:red,55; green,200; blue,55},opacity=0.5,fill opacity=0.5](18.65,-5.35) rectangle (19.35, -6.05);
\draw[draw={rgb,255:red,0; green,0; blue,0},fill={rgb,255:red,55; green,200; blue,55},opacity=0.5,fill opacity=0.5](18.65,-6.35) rectangle (19.35, -7.05);
\draw[draw={rgb,255:red,0; green,0; blue,0},fill={rgb,255:red,55; green,200; blue,55},opacity=0.5,fill opacity=0.5](19.65,-6.35) rectangle (20.35, -7.05);
\draw[draw={rgb,255:red,55; green,200; blue,55},opacity=0.5](17,-2.7)--(19,-2.7)--(19,-6.7)--(20,-6.7);
\draw[draw={rgb,255:red,0; green,0; blue,0},opacity=0.05](27,0.3)--(27,-3.7)--(24,-3.7)--(24,-1.7)--(21,-1.7)--(21,0.3)--(19,0.3)--(19,-0.7)--(17,-0.7)--(17,-5.7)--(15,-5.7)--(15,-6.7)--(13,-6.7)--(13,-2.7)--(10,-2.7)--(10,-1.7)--(8,-1.7)--(8,-4.7)--(6,-4.7)--(6,-3.7)--(5,-3.7)--(5,-5.7)--(3.5,-5.7)--(3.5,-17.2)--(38.5,-17.2)--(38.5,0.8)--(32.5,0.8)--(32.5,0.3)-- cycle;
\draw[draw=none,fill={rgb,255:red,0; green,0; blue,0},opacity=0.05](26,-18.2)--(39.5,-18.2)--(39.5,1.8)--(32.5,1.8)--(32.5,0.3)--(27,0.3)--(27,-3.7)--(24,-3.7)--(24,-1.7)--(21,-1.7)--(21,0.3)--(19,0.3)--(19,-0.7)--(17,-0.7)--(17,-2.7)--(19,-2.7)--(19,-6.7)--(20,-6.7)--(20,-8.7)--(22,-8.7)--(22,-10.7)--(21,-10.7)--(21,-11.7)--(24,-11.7)--(24,-9.7)--(26,-9.7);
\draw(37, 0.12) node[anchor=south west] {$\mathcal H_n$};
\draw[draw={rgb,255:red,0; green,0; blue,0},fill={rgb,255:red,55; green,200; blue,55},opacity=0.5,fill opacity=0.5](16.65,-2.35) rectangle (17.35, -3.05);
\draw[draw={rgb,255:red,0; green,0; blue,0},fill={rgb,255:red,55; green,200; blue,55},opacity=0.5,fill opacity=0.5](17.65,-2.35) rectangle (18.35, -3.05);
\draw[draw={rgb,255:red,0; green,0; blue,0},fill={rgb,255:red,55; green,200; blue,55},opacity=0.5,fill opacity=0.5](18.65,-2.35) rectangle (19.35, -3.05);
\draw[draw={rgb,255:red,0; green,0; blue,0},fill={rgb,255:red,55; green,200; blue,55},opacity=0.5,fill opacity=0.5](18.65,-3.35) rectangle (19.35, -4.05);
\draw[draw={rgb,255:red,0; green,0; blue,0},fill={rgb,255:red,55; green,200; blue,55},opacity=0.5,fill opacity=0.5](18.65,-4.35) rectangle (19.35, -5.05);
\draw[draw={rgb,255:red,0; green,0; blue,0},fill={rgb,255:red,55; green,200; blue,55},opacity=0.5,fill opacity=0.5](18.65,-5.35) rectangle (19.35, -6.05);
\draw[draw={rgb,255:red,0; green,0; blue,0},fill={rgb,255:red,55; green,200; blue,55},opacity=0.5,fill opacity=0.5](18.65,-6.35) rectangle (19.35, -7.05);
\draw[draw={rgb,255:red,0; green,0; blue,0},fill={rgb,255:red,55; green,200; blue,55},opacity=0.5,fill opacity=0.5](19.65,-6.35) rectangle (20.35, -7.05);
\draw[draw={rgb,255:red,0; green,0; blue,0},fill={rgb,255:red,55; green,200; blue,55},opacity=0.5,fill opacity=0.5](20.65,-6.35) rectangle (21.35, -7.05);
\draw[draw={rgb,255:red,0; green,0; blue,0},fill={rgb,255:red,55; green,200; blue,55},opacity=0.5,fill opacity=0.5](20.65,-5.35) rectangle (21.35, -6.05);
\draw[draw={rgb,255:red,0; green,0; blue,0},fill={rgb,255:red,55; green,200; blue,55},opacity=0.5,fill opacity=0.5](21.65,-5.35) rectangle (22.35, -6.05);
\draw[draw={rgb,255:red,0; green,0; blue,0},fill={rgb,255:red,55; green,200; blue,55},opacity=0.5,fill opacity=0.5](22.65,-5.35) rectangle (23.35, -6.05);
\draw[draw={rgb,255:red,0; green,0; blue,0},fill={rgb,255:red,55; green,200; blue,55},opacity=0.5,fill opacity=0.5](22.65,-4.35) rectangle (23.35, -5.05);
\draw[draw={rgb,255:red,0; green,0; blue,0},fill={rgb,255:red,55; green,200; blue,55},opacity=0.5,fill opacity=0.5](22.65,-3.35) rectangle (23.35, -4.05);
\draw[draw={rgb,255:red,0; green,0; blue,0},fill={rgb,255:red,55; green,200; blue,55},opacity=0.5,fill opacity=0.5](22.65,-2.35) rectangle (23.35, -3.05);
\draw[draw={rgb,255:red,55; green,200; blue,55},opacity=0.5](17,-2.7)--(19,-2.7)--(19,-6.7)--(21,-6.7)--(21,-5.7)--(23,-5.7)--(23,-2.7);
\draw[draw={rgb,255:red,0; green,0; blue,0},fill={rgb,255:red,55; green,200; blue,55},opacity=0.5,fill opacity=0.5](22.65,-2.35) rectangle (23.35, -3.05);
\draw[draw={rgb,255:red,0; green,0; blue,0},fill={rgb,255:red,55; green,200; blue,55},opacity=0.5,fill opacity=0.5](23.65,-2.35) rectangle (24.35, -3.05);
\draw[draw={rgb,255:red,0; green,0; blue,0},fill={rgb,255:red,55; green,200; blue,55},opacity=0.5,fill opacity=0.5](24.65,-2.35) rectangle (25.35, -3.05);
\draw[draw={rgb,255:red,0; green,0; blue,0},fill={rgb,255:red,55; green,200; blue,55},opacity=0.5,fill opacity=0.5](24.65,-3.35) rectangle (25.35, -4.05);
\draw[draw={rgb,255:red,0; green,0; blue,0},fill={rgb,255:red,55; green,200; blue,55},opacity=0.5,fill opacity=0.5](24.65,-4.35) rectangle (25.35, -5.05);
\draw[draw={rgb,255:red,0; green,0; blue,0},fill={rgb,255:red,55; green,200; blue,55},opacity=0.5,fill opacity=0.5](24.65,-5.35) rectangle (25.35, -6.05);
\draw[draw={rgb,255:red,0; green,0; blue,0},fill={rgb,255:red,55; green,200; blue,55},opacity=0.5,fill opacity=0.5](24.65,-6.35) rectangle (25.35, -7.05);
\draw[draw={rgb,255:red,0; green,0; blue,0},fill={rgb,255:red,55; green,200; blue,55},opacity=0.5,fill opacity=0.5](25.65,-6.35) rectangle (26.35, -7.05);
\draw[draw={rgb,255:red,0; green,0; blue,0},fill={rgb,255:red,55; green,200; blue,55},opacity=0.5,fill opacity=0.5](26.65,-6.35) rectangle (27.35, -7.05);
\draw[draw={rgb,255:red,0; green,0; blue,0},fill={rgb,255:red,55; green,200; blue,55},opacity=0.5,fill opacity=0.5](26.65,-5.35) rectangle (27.35, -6.05);
\draw[draw={rgb,255:red,0; green,0; blue,0},fill={rgb,255:red,55; green,200; blue,55},opacity=0.5,fill opacity=0.5](27.65,-5.35) rectangle (28.35, -6.05);
\draw[draw={rgb,255:red,0; green,0; blue,0},fill={rgb,255:red,55; green,200; blue,55},opacity=0.5,fill opacity=0.5](28.65,-5.35) rectangle (29.35, -6.05);
\draw[draw={rgb,255:red,0; green,0; blue,0},fill={rgb,255:red,55; green,200; blue,55},opacity=0.5,fill opacity=0.5](28.65,-4.35) rectangle (29.35, -5.05);
\draw[draw={rgb,255:red,0; green,0; blue,0},fill={rgb,255:red,55; green,200; blue,55},opacity=0.5,fill opacity=0.5](28.65,-3.35) rectangle (29.35, -4.05);
\draw[draw={rgb,255:red,0; green,0; blue,0},fill={rgb,255:red,55; green,200; blue,55},opacity=0.5,fill opacity=0.5](28.65,-2.35) rectangle (29.35, -3.05);
\draw[draw={rgb,255:red,55; green,200; blue,55},opacity=0.5](23,-2.7)--(25,-2.7)--(25,-6.7)--(27,-6.7)--(27,-5.7)--(29,-5.7)--(29,-2.7);
\draw[draw=none,fill={rgb,255:red,0; green,0; blue,0},thin](23, -2.7) ellipse (0.1cm and 0.1cm);\draw[draw=none,fill={rgb,255:red,0; green,0; blue,0},thin](11, -2.7) ellipse (0.1cm and 0.1cm);\draw[draw=none,fill={rgb,255:red,0; green,0; blue,0},thin](17, -2.7) ellipse (0.1cm and 0.1cm);\draw[draw=none,fill={rgb,255:red,0; green,0; blue,0},thin](14, -6.7) ellipse (0.1cm and 0.1cm);\draw(10.24, -2.46) node[anchor=south west] {$u_n$};
\draw(15.26, -3.26) node[anchor=south west] {$v_n$};
\draw(13.25, -8.23) node[anchor=south west] {$m_n$};
\draw[draw=none,fill={rgb,255:red,0; green,0; blue,0},thin](25, -3.7) ellipse (0.1cm and 0.1cm);\draw(25.2, -5.3) node[anchor=south west] {$m_{n+1}$};
\draw[draw=none,fill={rgb,255:red,0; green,0; blue,0},thin](13, -3.7) ellipse (0.1cm and 0.1cm);\draw(13.31, -4.35) node[anchor=south west] {$a$};
\draw[draw={rgb,255:red,0; green,0; blue,0}](14,-6.7)--(14,-8.7)--(16,-8.7)--(16,-10.7)--(15,-10.7)--(15,-11.7)--(18,-11.7)--(18,-9.7)--(20,-9.7)--(20,-16.7);
\draw[draw={rgb,255:red,0; green,0; blue,0}](20,-6.7)--(20,-8.7)--(22,-8.7)--(22,-10.7)--(21,-10.7)--(21,-11.7)--(24,-11.7)--(24,-9.7)--(26,-9.7)--(26,-16.7);
\draw[draw=none,fill={rgb,255:red,0; green,0; blue,0},thin](20, -6.7) ellipse (0.1cm and 0.1cm);\draw[draw={rgb,255:red,0; green,0; blue,0}](32,-16.7)--(32,-9.7)--(30,-9.7)--(30,-11.7)--(27,-11.7)--(27,-10.7)--(28,-10.7)--(28,-8.7)--(26,-8.7)--(26,-6.7)--(25,-6.7)--(25,-3.7);
\draw[draw=none,fill={rgb,255:red,0; green,0; blue,0},thin](32.5, 0.3) ellipse (0.1cm and 0.1cm);\draw(32.63, 4.3) node[anchor=south west] {$\ell^k$};
\draw[draw={rgb,255:red,0; green,0; blue,0}](32.5,0.3)--(32.5,5.3);
\draw[draw={rgb,255:red,0; green,0; blue,0},dashed,thin](20,-18.7)--(20,-16.7);
\draw(18.65, -17.51) node[anchor=south west] {$f_n$};
\draw[draw={rgb,255:red,0; green,0; blue,0},dashed,thin](26,-18.7)--(26,-16.7);
\draw(24.65, -17.51) node[anchor=south west] {$f_n+\vect{P_iP_j}$};
\draw(32.19, -17.51) node[anchor=south west] {$f_{n+1}$};
\draw[draw={rgb,255:red,0; green,0; blue,0},dashed,thin](32,-18.7)--(32,-16.7);
\end{tikzpicture}
    \end{center}
    \caption{Finding $m_{n+1}$ and $f_{n+1}$: we let $(a, t)$ be the largest pair of integers such that  $\pos{P_{a}}+t\vpij = \pos{P_{m_{n+1}}}$ for some integer $m_{n+1} \in \{ i+1,i+2, \ldots,k\}$, and let $f_{n+1}$ be equal to $f_n+t\vpij$, plus $P_{\rng{a}{m_n}}+t\vpij$. Here, $t=2$ and $a = u_n+3$.}
    \label{fig:induction-mn1}
  \end{figure}

  \argument{We claim that for all $s>1$, $P_{\rng a {m_n}}+s\vpij$ is entirely in $\mathcal H_n$}
  Suppose for the sake of contradiction that this is not the case, and let $s>1$ be the smallest integer such that $P_{\rng a {m_n}}+s\vpij$ is not entirely in $\mathcal H_n$.

  We first claim that $f_n+s\vpij$ is entirely in $\mathcal H_n$.
  Indeed, by~\ref{umv:s}, for all $u\geq 1$, $f_n+u\vpij$ intersects neither $f_n$ nor $c$.
  Therefore, since $s\geq 2$, $f_n+s\vpij$ intersects neither $f_n+\vpij$ nor $c+\vpij$ nor $c$, and therefore doesn't intersect $h_n$ (since $h_n$ is entirely composed of parts of $f_n+\vpij$, $c+\vpij$ and $c$).
  By~\ref{umv:f}, the beginning of $f_n$ is $\lmz+s_n\vpij$ for $s_n\geq 0$, and hence the beginning of $f_n+s\vpij$ is $\lmz+(s_n+s)\vpij$,
  which starts on the right-hand side of $h_n$ (because $\xcoord{\vpij}>0$), and doesn't intersect $h_n$.
  Hence, $f_n+s\vpij$ is entirely in $\mathcal H_n$, as claimed.

  Then, we claim that $P_{\rng a {m_n}}+s\vpij$ intersects $h_n$. Indeed, since $f_n+s\vpij$ is in $\mathcal H_n$, then
  $f_n(0)+s\vpij = \pos{P_{m_n}}+s\vpij$
  is also in $\mathcal H_n$.
  Therefore, since we assumed $P_{\rng a {m_n}}+s\vpij$ is not completely in $\mathcal H_n$, this means that
  $\embed{P_{\rng a m_n}+s\vpij}$ intersects $h_n$. That intersection can only happen on one of the five parts of $h_n$:
  \begin{enumerate}
  \item\label{case:fn:1} if $P_{\rng a {m_n}}+s\vpij$ intersects $l^k$ at some $l^k(z)$ with $z\geq 0$ then we claim that $P_{\rng a {m_n}}+(s-1)\vpij$ is in $\mathcal H_n$: if $s = 2$, this is because $P_{\rng a {m_n}}+\vpij$ is actually part of the border $h_n$ of $\mathcal H_n$, else, $s>2$, and this is because $s$ is minimal.

    Moreover, by maximality of $a$, $P_{a+1}+(s-1)\vpij$ is not on $P_{\rng {v_n} k}$. Therefore, $l^k(z)+\vpji$ is inside $\mathcal H_n$ and is not on $\embed{P_{\rng {v_n} k}}$.
    Since $\mathcal H_n \subset \mathcal C$ and since the only part of $P$ that is in $\mathcal H_n$ is $\embed{P_{\rng {v_n} k}}$, this means that $l^k(z)+\vpji$ is inside $\mathcal C$ and is not on $\embed{\Pik}$.
    This contradicts \subl{lem:c-lk}.

\item \label{case:fn:2} If $\reverse{\embed{P_{\rng a {m_n}}+s\vpij}}$ intersects $P_{\rng {v_n} k}$, then
    $P_{\range {a+1}{a+2} {m_n}}+s\vpij$ intersects $P_{\rng {v_n} k}$ (because $\pos{P_{m_n}}+s\vpij=f_n(0)+s\vpij$ is in $\mathcal H_n$).
    However, this contradicts the definition of $a$ as the largest integer such that $\pos{P_{a}}+t\vpij = \pos{P_{m_{n+1}}}$ for some $m_{n+1}$.

  \item \label{case:fn:3} If $P_{\rng a {m_n}}+s\vpij$ intersects $f_n+\vpij$ then $P_{\rng a {m_n}}+(s-1)\vpij$ intersects $f_n$.
    We have argued in Point~\ref{case:fn:1} that $P_{\rng a {m_n}}+(s-1)\vpij$ is entirely in $\mathcal H_n$.
    However, the only point of $f_n$ that can be in $\mathcal H_n$ is $f_n(0) = \pos{P_{m_n}}$, which happens only when $m_n = v_n$.

    If $s=2$, this means that $P_{\rng a {m_n}}+\vpij$ intersects $f_n$, but the only place where that can happen is at $\pos{P_a} = f_n(0)$, which means that $a = u_n$. However, by maximality of $a$, this means that for all $t>1$, $P_{\range {a+1}{a+2} {m_n}}+t\vpij$ does not intersect $h_n$, contradicting our assumption.

    If $s = 3$, $P_{\rng a {m_n}}+(s-2)\vpij$ intersects $P_{u_n}$, which is not in $\mathcal H_n$, contradicting that $P_{\rng a {m_n}}+\vpij$ is on the border $h_n$ of $\mathcal H_n$.

    Else, $s>3$, and $P_{\rng a {m_n}}+(s-2)\vpij$ intersects $P_{u_n}$, which is not in $\mathcal H_n$, contradicting the minimality of $s$.

  \item If $P_{\rng a {m_n}}+s\vpij$ intersects $P_{\rng {u_n} {v_n}}+\vpij$ then
    $P_{\rng a {m_n}}+(s-1)\vpij$ intersects $P_{\rng {u_n} {m_n}}$.

    However, by minimality of $s$, $P_{\rng a {m_n}}+(s-1)\vpij$ is in $\mathcal H_n$, and the only position of $P_{\rng {u_n} {m_n}}$ that may be in $\mathcal H_n$ is $\pos{P_{m_n}} = f_n(0)$. Therefore $P_{m_n}+s\vpij$ intersects $f_n(0)+\vpij = \pos{P_{m_n}}+\vpij$, and this case is already covered in Point~\ref{case:fn:3} above.

  \end{enumerate}
  This shows that for all $s>1$, $P_{\rng a {m_n}}+s\vpij$ is entirely in $\mathcal H_n$ as claimed.

  Now, remember that $\pos{P_{m_{n+1}}}=\pos{P_a}+t\vpij$. This means that $P_{m_{n+1}}$ is in $\mathcal H_n$. However, the only part of $P$ that is inside $\mathcal H_n$ is $P_{\rng {v_n} k}$, hence $m_{n+1}\geq v_n\geq m_n$.

  \argument{Satisfying the induction hypotheses for step $n+1$}
  Our final step is to prove that either $P$ is pumpable, or we can find two indices $u_{n+1}$ and $v_{n+1}$ (we have already specified the index $m_{n+1}$ and the curve $f_{n+1}$) to move on to the next step of the induction. There are two cases:
  \begin{enumerate}[label=\Alph*]
  \item\label{shield:induction:pumpable} If $m_{n+1} = m_n$ then we claim that $P$ is pumpable.

    Since $P_{\rng a {m_n}}+t\vpij$ is in $\mathcal H_n$, $\pos{P_a+t\vpij} = \pos{P_{m_{n+1}}} = \pos{P_{m_n}}$ is in $\mathcal H_n$.
    However, by definition of $\mathcal H_n$, the only position of $P_{\rng {m_n}{v_n}}$ in $\mathcal H_n$ is $P_{v_n}$, hence $m_n = v_n$.

    In this case, $P_{\rng {u_n} {v_n}}=P_{\rng {u_n} {m_n}}$ and the only position of $P_{\rng {u_n} {v_n}}$ that is in $\mathcal H_n$ is $\pos{P_{v_n}}$. Moreover, $\embed{P_{\rng {u_n} {v_n}}}+\vpij$ is on the border of $\mathcal H_n$, which implies that the only intersection between $P_{\rng {u_n} {v_n}}$ and $P_{\rng {u_n} {v_n}}+\vect{P_iP_j}$ is $\pos{P_{v_n}}=\pos{P_{u_n}}+\vpij$. By Lemma~\ref{lem:precious}, this means that for all $s>2$, $P_{\rng {u_n} {v_n}}+s\vect{P_iP_j}$ does not intersect  $P_{\rng {u_n} {v_n}}$. Therefore, $t=1$ and $a=u_n$.

    By definition of $a$, we have for all $s \geq1$, $P_{\rng{u_n}{v_n}}+s\vpij$ is in $\mathcal H_n\subset\mathcal C$, this means that the pumping of $P$ between $u_n$ and $v_n$ is simple (since $\sigma\cup\asm{P_{\range 0 1 i}}$ is not in $\mathcal C$ and by Lemma \ref{lem:precious}), and hence that $P$ is pumpable, and we are done: indeed, that is one of the conclusions of this lemma.

  \item Else $m_{n+1} > m_n$. We already defined $f_{n+1}= \concat{f_n, \reverse {\embed{P_{\rng a {m_n}}}}} + t\vpij$.
    We let $s_{n+1} = s_n + t$.
    We will define $u_{n+1}$ and $v_{n+1}$, and prove that the induction hypothesis holds for $(u_{n+1}, m_{n+1}, v_{n+1})$ and $f_{n+1}$:
    \begin{enumerate}[label=H\arabic*]
    \item $f_{n+1} = \concat{f_n, \reverse {\embed{P_{\rng a {m_n}}}}} + t\vpij$ is indeed the concatenation of $f_n+t\vpij$ and a finite curve made of the concatenation of embeddings of translations of segments of $P_{\range {i+1} {i+2} k}$ (we have merely added the reverse of one such segment). Since $\pos{P_a}+t\vpij = \pos{P_{m_{n+1}}}$ we get that $f_{n+1}$ connects
     $\lmz(0)+s_{n+1}\vpij$ to $\pos{P_{m_{n+1}}}$. Moreover, by \ref{umv:c} the only intersection between $f_n$ and $\embed{P_{\rng a {m_n}}}$ is $\pos{P_m}$ and thus $f_{n+1}$ is simple.

   \item $f_{n+1}$ is entirely in $\mathcal{H}_n \subset \mathcal C$ and intersects $c$ exactly once at $\pos{P_{m_{n+1}}}$: indeed, $f_n+t\vpij$ is in $\mathcal C$ and does not intersect $c$ at all by~\ref{umv:s}, and
      $\embed{P_{\rng a {m_{n+1}}}} + t\vpij$ starts at $f_{n+1}(0)\in\mathcal C$, and intersects $c$ exactly once by definition of $a$.
    \item We claim that for all $u>0$, $f_{n+1}+u\vpij$ does not intersect $c$: indeed, by~\ref{umv:s}, $f_n+(t+u)\vpij$ does not intersect $c$. Moreover, since $t$ is maximal, then $P_{\rng {a} {m_{n}}}+(t+u)\vpij$ does not intersect $c$.
      Moreover, we show that $f_{n+1}+u\vpij$ does not intersect $f_{n+1}$, by considering the two parts of $f_{n+1}+u\vpij$, namely $f_n+(t+u)\vpij$ and $P_{\rng a {m_n}}+(t+u)\vpij$:
      \begin{itemize}
      \item By~\ref{umv:s}, $f_n+(t+u)\vpij$ intersects neither $f_n+t\vpij$ nor $\embed{P_{\rng a {m_n}}}+\vpij$ (which is a segment of $c+\vpij$), nor $c$.
      \item
        We claim now that $\embed{P_{\rng a {m_n}}}+(t+u)\vpij$ does not intersect $\embed{P_{\rng a {m_n}}}+t\vpij$, or equivalently, that
        $\embed{P_{\rng a {m_n}}}+u\vpij$ does not intersect $\embed{P_{\rng a {m_n}}}$.

        We showed above that $\embed{P_{\rng a {m_n}}}+u\vpij$ is in $\mathcal H_n$. However, the only point of $\embed{P_{\rng a {m_n}}}$ that may be in $\mathcal H_n$ if $P_{m_n}$, which happens if and only if $v_n = m_n$. Moreover, if $u>1$, then $\embed{P_{\rng a {m_n}}}+(u-1)\vpij$ intersects $P_{u_n}$, which is not in $\mathcal H_n$, and this is a contradiction. Therefore, if $u = 1$, we also have $a = u_n$ and $m_n = m_{n+1}$, which was already handled in Case~\ref{shield:induction:pumpable}.

        Finally, if $\embed{P_{\rng a {m_n}}}+(t+u)\vpij$ intersects $f_n+t\vpij$, then $\embed{P_{\rng a {m_n}}}+u\vpij$ intersects $f_n$. Therefore, by~\ref{umv:unvn}, $m_n = v_n$ and $a = u_n$, which means that $m_{n+1} = m_n$, and this was already handled in Case~\ref{shield:induction:pumpable}.
      \end{itemize}
      Therefore, $f_{n+1}+u\vpij$ does not intersect $f_{n+1}$.
    \item
      We first define the curve $g_{n+1} = \concat{f_{n+1}, \embed{P_{\rng {m_{n+1}} k}}, \gs{\pos{P_k}}{ l^k(0)}, l^k}$.

      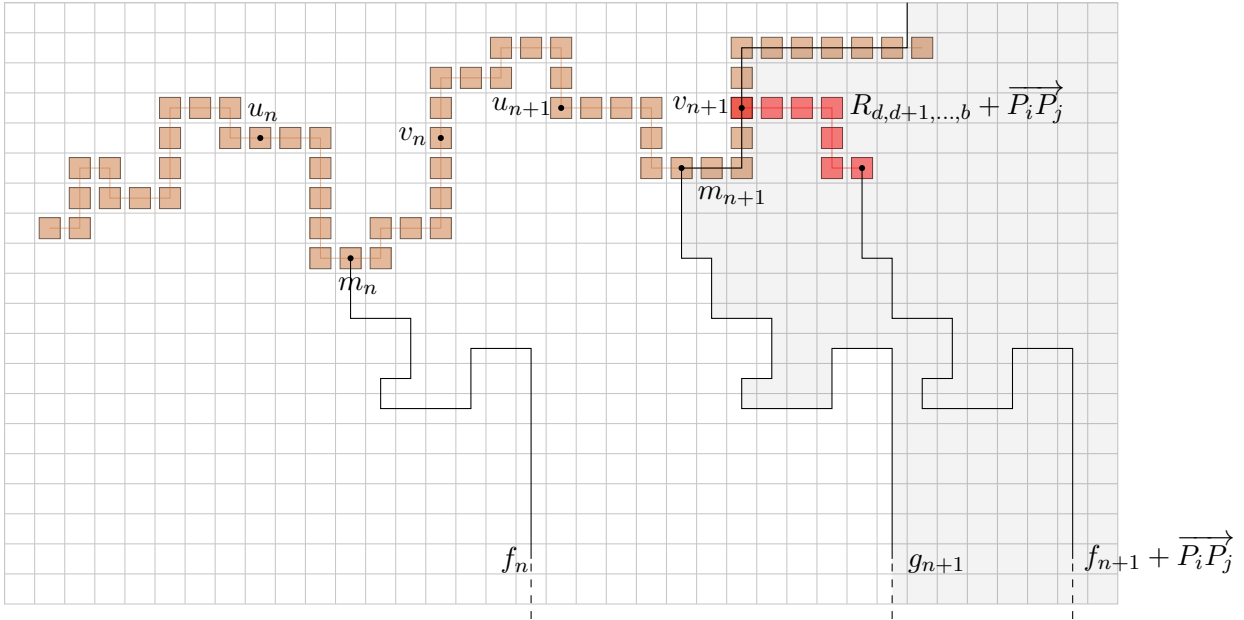
\begin{figure}[ht]
        \begin{center}
          \begin{tikzpicture}[scale=\scale]\draw[draw={rgb,255:red,200; green,200; blue,200}](2.5,1.8) rectangle (39.5, -18.2);
\draw[draw={rgb,255:red,200; green,200; blue,200}](2.5,-18.2)--(2.5,1.8);
\draw[draw={rgb,255:red,200; green,200; blue,200}](3.5,-18.2)--(3.5,1.8);
\draw[draw={rgb,255:red,200; green,200; blue,200}](4.5,-18.2)--(4.5,1.8);
\draw[draw={rgb,255:red,200; green,200; blue,200}](5.5,-18.2)--(5.5,1.8);
\draw[draw={rgb,255:red,200; green,200; blue,200}](6.5,-18.2)--(6.5,1.8);
\draw[draw={rgb,255:red,200; green,200; blue,200}](7.5,-18.2)--(7.5,1.8);
\draw[draw={rgb,255:red,200; green,200; blue,200}](8.5,-18.2)--(8.5,1.8);
\draw[draw={rgb,255:red,200; green,200; blue,200}](9.5,-18.2)--(9.5,1.8);
\draw[draw={rgb,255:red,200; green,200; blue,200}](10.5,-18.2)--(10.5,1.8);
\draw[draw={rgb,255:red,200; green,200; blue,200}](11.5,-18.2)--(11.5,1.8);
\draw[draw={rgb,255:red,200; green,200; blue,200}](12.5,-18.2)--(12.5,1.8);
\draw[draw={rgb,255:red,200; green,200; blue,200}](13.5,-18.2)--(13.5,1.8);
\draw[draw={rgb,255:red,200; green,200; blue,200}](14.5,-18.2)--(14.5,1.8);
\draw[draw={rgb,255:red,200; green,200; blue,200}](15.5,-18.2)--(15.5,1.8);
\draw[draw={rgb,255:red,200; green,200; blue,200}](16.5,-18.2)--(16.5,1.8);
\draw[draw={rgb,255:red,200; green,200; blue,200}](17.5,-18.2)--(17.5,1.8);
\draw[draw={rgb,255:red,200; green,200; blue,200}](18.5,-18.2)--(18.5,1.8);
\draw[draw={rgb,255:red,200; green,200; blue,200}](19.5,-18.2)--(19.5,1.8);
\draw[draw={rgb,255:red,200; green,200; blue,200}](20.5,-18.2)--(20.5,1.8);
\draw[draw={rgb,255:red,200; green,200; blue,200}](21.5,-18.2)--(21.5,1.8);
\draw[draw={rgb,255:red,200; green,200; blue,200}](22.5,-18.2)--(22.5,1.8);
\draw[draw={rgb,255:red,200; green,200; blue,200}](23.5,-18.2)--(23.5,1.8);
\draw[draw={rgb,255:red,200; green,200; blue,200}](24.5,-18.2)--(24.5,1.8);
\draw[draw={rgb,255:red,200; green,200; blue,200}](25.5,-18.2)--(25.5,1.8);
\draw[draw={rgb,255:red,200; green,200; blue,200}](26.5,-18.2)--(26.5,1.8);
\draw[draw={rgb,255:red,200; green,200; blue,200}](27.5,-18.2)--(27.5,1.8);
\draw[draw={rgb,255:red,200; green,200; blue,200}](28.5,-18.2)--(28.5,1.8);
\draw[draw={rgb,255:red,200; green,200; blue,200}](29.5,-18.2)--(29.5,1.8);
\draw[draw={rgb,255:red,200; green,200; blue,200}](30.5,-18.2)--(30.5,1.8);
\draw[draw={rgb,255:red,200; green,200; blue,200}](31.5,-18.2)--(31.5,1.8);
\draw[draw={rgb,255:red,200; green,200; blue,200}](32.5,-18.2)--(32.5,1.8);
\draw[draw={rgb,255:red,200; green,200; blue,200}](33.5,-18.2)--(33.5,1.8);
\draw[draw={rgb,255:red,200; green,200; blue,200}](34.5,-18.2)--(34.5,1.8);
\draw[draw={rgb,255:red,200; green,200; blue,200}](35.5,-18.2)--(35.5,1.8);
\draw[draw={rgb,255:red,200; green,200; blue,200}](36.5,-18.2)--(36.5,1.8);
\draw[draw={rgb,255:red,200; green,200; blue,200}](37.5,-18.2)--(37.5,1.8);
\draw[draw={rgb,255:red,200; green,200; blue,200}](38.5,-18.2)--(38.5,1.8);
\draw[draw={rgb,255:red,200; green,200; blue,200}](2.5,1.8)--(39.5,1.8);
\draw[draw={rgb,255:red,200; green,200; blue,200}](2.5,0.8)--(39.5,0.8);
\draw[draw={rgb,255:red,200; green,200; blue,200}](2.5,-0.2)--(39.5,-0.2);
\draw[draw={rgb,255:red,200; green,200; blue,200}](2.5,-1.2)--(39.5,-1.2);
\draw[draw={rgb,255:red,200; green,200; blue,200}](2.5,-2.2)--(39.5,-2.2);
\draw[draw={rgb,255:red,200; green,200; blue,200}](2.5,-3.2)--(39.5,-3.2);
\draw[draw={rgb,255:red,200; green,200; blue,200}](2.5,-4.2)--(39.5,-4.2);
\draw[draw={rgb,255:red,200; green,200; blue,200}](2.5,-5.2)--(39.5,-5.2);
\draw[draw={rgb,255:red,200; green,200; blue,200}](2.5,-6.2)--(39.5,-6.2);
\draw[draw={rgb,255:red,200; green,200; blue,200}](2.5,-7.2)--(39.5,-7.2);
\draw[draw={rgb,255:red,200; green,200; blue,200}](2.5,-8.2)--(39.5,-8.2);
\draw[draw={rgb,255:red,200; green,200; blue,200}](2.5,-9.2)--(39.5,-9.2);
\draw[draw={rgb,255:red,200; green,200; blue,200}](2.5,-10.2)--(39.5,-10.2);
\draw[draw={rgb,255:red,200; green,200; blue,200}](2.5,-11.2)--(39.5,-11.2);
\draw[draw={rgb,255:red,200; green,200; blue,200}](2.5,-12.2)--(39.5,-12.2);
\draw[draw={rgb,255:red,200; green,200; blue,200}](2.5,-13.2)--(39.5,-13.2);
\draw[draw={rgb,255:red,200; green,200; blue,200}](2.5,-14.2)--(39.5,-14.2);
\draw[draw={rgb,255:red,200; green,200; blue,200}](2.5,-15.2)--(39.5,-15.2);
\draw[draw={rgb,255:red,200; green,200; blue,200}](2.5,-16.2)--(39.5,-16.2);
\draw[draw={rgb,255:red,200; green,200; blue,200}](2.5,-17.2)--(39.5,-17.2);
\draw[draw={rgb,255:red,0; green,0; blue,0},fill={rgb,255:red,200; green,113; blue,55},opacity=0.5,fill opacity=0.5](3.65,-5.35) rectangle (4.35, -6.05);
\draw[draw={rgb,255:red,0; green,0; blue,0},fill={rgb,255:red,200; green,113; blue,55},opacity=0.5,fill opacity=0.5](4.65,-5.35) rectangle (5.35, -6.05);
\draw[draw={rgb,255:red,0; green,0; blue,0},fill={rgb,255:red,200; green,113; blue,55},opacity=0.5,fill opacity=0.5](4.65,-4.35) rectangle (5.35, -5.05);
\draw[draw={rgb,255:red,0; green,0; blue,0},fill={rgb,255:red,200; green,113; blue,55},opacity=0.5,fill opacity=0.5](4.65,-3.35) rectangle (5.35, -4.05);
\draw[draw={rgb,255:red,0; green,0; blue,0},fill={rgb,255:red,200; green,113; blue,55},opacity=0.5,fill opacity=0.5](5.65,-3.35) rectangle (6.35, -4.05);
\draw[draw={rgb,255:red,0; green,0; blue,0},fill={rgb,255:red,200; green,113; blue,55},opacity=0.5,fill opacity=0.5](5.65,-4.35) rectangle (6.35, -5.05);
\draw[draw={rgb,255:red,0; green,0; blue,0},fill={rgb,255:red,200; green,113; blue,55},opacity=0.5,fill opacity=0.5](6.65,-4.35) rectangle (7.35, -5.05);
\draw[draw={rgb,255:red,0; green,0; blue,0},fill={rgb,255:red,200; green,113; blue,55},opacity=0.5,fill opacity=0.5](7.65,-4.35) rectangle (8.35, -5.05);
\draw[draw={rgb,255:red,0; green,0; blue,0},fill={rgb,255:red,200; green,113; blue,55},opacity=0.5,fill opacity=0.5](7.65,-3.35) rectangle (8.35, -4.05);
\draw[draw={rgb,255:red,0; green,0; blue,0},fill={rgb,255:red,200; green,113; blue,55},opacity=0.5,fill opacity=0.5](7.65,-2.35) rectangle (8.35, -3.05);
\draw[draw={rgb,255:red,0; green,0; blue,0},fill={rgb,255:red,200; green,113; blue,55},opacity=0.5,fill opacity=0.5](7.65,-1.35) rectangle (8.35, -2.05);
\draw[draw={rgb,255:red,0; green,0; blue,0},fill={rgb,255:red,200; green,113; blue,55},opacity=0.5,fill opacity=0.5](8.65,-1.35) rectangle (9.35, -2.05);
\draw[draw={rgb,255:red,0; green,0; blue,0},fill={rgb,255:red,200; green,113; blue,55},opacity=0.5,fill opacity=0.5](9.65,-1.35) rectangle (10.35, -2.05);
\draw[draw={rgb,255:red,0; green,0; blue,0},fill={rgb,255:red,200; green,113; blue,55},opacity=0.5,fill opacity=0.5](9.65,-2.35) rectangle (10.35, -3.05);
\draw[draw={rgb,255:red,0; green,0; blue,0},fill={rgb,255:red,200; green,113; blue,55},opacity=0.5,fill opacity=0.5](10.65,-2.35) rectangle (11.35, -3.05);
\draw[draw={rgb,255:red,0; green,0; blue,0},fill={rgb,255:red,200; green,113; blue,55},opacity=0.5,fill opacity=0.5](11.65,-2.35) rectangle (12.35, -3.05);
\draw[draw={rgb,255:red,0; green,0; blue,0},fill={rgb,255:red,200; green,113; blue,55},opacity=0.5,fill opacity=0.5](12.65,-2.35) rectangle (13.35, -3.05);
\draw[draw={rgb,255:red,0; green,0; blue,0},fill={rgb,255:red,200; green,113; blue,55},opacity=0.5,fill opacity=0.5](12.65,-3.35) rectangle (13.35, -4.05);
\draw[draw={rgb,255:red,0; green,0; blue,0},fill={rgb,255:red,200; green,113; blue,55},opacity=0.5,fill opacity=0.5](12.65,-4.35) rectangle (13.35, -5.05);
\draw[draw={rgb,255:red,0; green,0; blue,0},fill={rgb,255:red,200; green,113; blue,55},opacity=0.5,fill opacity=0.5](12.65,-5.35) rectangle (13.35, -6.05);
\draw[draw={rgb,255:red,0; green,0; blue,0},fill={rgb,255:red,200; green,113; blue,55},opacity=0.5,fill opacity=0.5](12.65,-6.35) rectangle (13.35, -7.05);
\draw[draw={rgb,255:red,0; green,0; blue,0},fill={rgb,255:red,200; green,113; blue,55},opacity=0.5,fill opacity=0.5](13.65,-6.35) rectangle (14.35, -7.05);
\draw[draw={rgb,255:red,0; green,0; blue,0},fill={rgb,255:red,200; green,113; blue,55},opacity=0.5,fill opacity=0.5](14.65,-6.35) rectangle (15.35, -7.05);
\draw[draw={rgb,255:red,0; green,0; blue,0},fill={rgb,255:red,200; green,113; blue,55},opacity=0.5,fill opacity=0.5](14.65,-5.35) rectangle (15.35, -6.05);
\draw[draw={rgb,255:red,0; green,0; blue,0},fill={rgb,255:red,200; green,113; blue,55},opacity=0.5,fill opacity=0.5](15.65,-5.35) rectangle (16.35, -6.05);
\draw[draw={rgb,255:red,0; green,0; blue,0},fill={rgb,255:red,200; green,113; blue,55},opacity=0.5,fill opacity=0.5](16.65,-5.35) rectangle (17.35, -6.05);
\draw[draw={rgb,255:red,0; green,0; blue,0},fill={rgb,255:red,200; green,113; blue,55},opacity=0.5,fill opacity=0.5](16.65,-4.35) rectangle (17.35, -5.05);
\draw[draw={rgb,255:red,0; green,0; blue,0},fill={rgb,255:red,200; green,113; blue,55},opacity=0.5,fill opacity=0.5](16.65,-3.35) rectangle (17.35, -4.05);
\draw[draw={rgb,255:red,0; green,0; blue,0},fill={rgb,255:red,200; green,113; blue,55},opacity=0.5,fill opacity=0.5](16.65,-2.35) rectangle (17.35, -3.05);
\draw[draw={rgb,255:red,0; green,0; blue,0},fill={rgb,255:red,200; green,113; blue,55},opacity=0.5,fill opacity=0.5](16.65,-1.35) rectangle (17.35, -2.05);
\draw[draw={rgb,255:red,0; green,0; blue,0},fill={rgb,255:red,200; green,113; blue,55},opacity=0.5,fill opacity=0.5](16.65,-0.35) rectangle (17.35, -1.05);
\draw[draw={rgb,255:red,0; green,0; blue,0},fill={rgb,255:red,200; green,113; blue,55},opacity=0.5,fill opacity=0.5](17.65,-0.35) rectangle (18.35, -1.05);
\draw[draw={rgb,255:red,0; green,0; blue,0},fill={rgb,255:red,200; green,113; blue,55},opacity=0.5,fill opacity=0.5](18.65,-0.35) rectangle (19.35, -1.05);
\draw[draw={rgb,255:red,0; green,0; blue,0},fill={rgb,255:red,200; green,113; blue,55},opacity=0.5,fill opacity=0.5](18.65,0.65) rectangle (19.35, -0.05);
\draw[draw={rgb,255:red,0; green,0; blue,0},fill={rgb,255:red,200; green,113; blue,55},opacity=0.5,fill opacity=0.5](19.65,0.65) rectangle (20.35, -0.05);
\draw[draw={rgb,255:red,0; green,0; blue,0},fill={rgb,255:red,200; green,113; blue,55},opacity=0.5,fill opacity=0.5](20.65,0.65) rectangle (21.35, -0.05);
\draw[draw={rgb,255:red,0; green,0; blue,0},fill={rgb,255:red,200; green,113; blue,55},opacity=0.5,fill opacity=0.5](20.65,-0.35) rectangle (21.35, -1.05);
\draw[draw={rgb,255:red,0; green,0; blue,0},fill={rgb,255:red,200; green,113; blue,55},opacity=0.5,fill opacity=0.5](20.65,-1.35) rectangle (21.35, -2.05);
\draw[draw={rgb,255:red,0; green,0; blue,0},fill={rgb,255:red,200; green,113; blue,55},opacity=0.5,fill opacity=0.5](21.65,-1.35) rectangle (22.35, -2.05);
\draw[draw={rgb,255:red,0; green,0; blue,0},fill={rgb,255:red,200; green,113; blue,55},opacity=0.5,fill opacity=0.5](22.65,-1.35) rectangle (23.35, -2.05);
\draw[draw={rgb,255:red,0; green,0; blue,0},fill={rgb,255:red,200; green,113; blue,55},opacity=0.5,fill opacity=0.5](23.65,-1.35) rectangle (24.35, -2.05);
\draw[draw={rgb,255:red,0; green,0; blue,0},fill={rgb,255:red,200; green,113; blue,55},opacity=0.5,fill opacity=0.5](23.65,-2.35) rectangle (24.35, -3.05);
\draw[draw={rgb,255:red,0; green,0; blue,0},fill={rgb,255:red,200; green,113; blue,55},opacity=0.5,fill opacity=0.5](23.65,-3.35) rectangle (24.35, -4.05);
\draw[draw={rgb,255:red,0; green,0; blue,0},fill={rgb,255:red,200; green,113; blue,55},opacity=0.5,fill opacity=0.5](24.65,-3.35) rectangle (25.35, -4.05);
\draw[draw={rgb,255:red,0; green,0; blue,0},fill={rgb,255:red,200; green,113; blue,55},opacity=0.5,fill opacity=0.5](25.65,-3.35) rectangle (26.35, -4.05);
\draw[draw={rgb,255:red,0; green,0; blue,0},fill={rgb,255:red,200; green,113; blue,55},opacity=0.5,fill opacity=0.5](26.65,-3.35) rectangle (27.35, -4.05);
\draw[draw={rgb,255:red,0; green,0; blue,0},fill={rgb,255:red,200; green,113; blue,55},opacity=0.5,fill opacity=0.5](26.65,-2.35) rectangle (27.35, -3.05);
\draw[draw={rgb,255:red,0; green,0; blue,0},fill={rgb,255:red,200; green,113; blue,55},opacity=0.5,fill opacity=0.5](26.65,-1.35) rectangle (27.35, -2.05);
\draw[draw={rgb,255:red,0; green,0; blue,0},fill={rgb,255:red,200; green,113; blue,55},opacity=0.5,fill opacity=0.5](26.65,-0.35) rectangle (27.35, -1.05);
\draw[draw={rgb,255:red,0; green,0; blue,0},fill={rgb,255:red,200; green,113; blue,55},opacity=0.5,fill opacity=0.5](26.65,0.65) rectangle (27.35, -0.05);
\draw[draw={rgb,255:red,0; green,0; blue,0},fill={rgb,255:red,200; green,113; blue,55},opacity=0.5,fill opacity=0.5](27.65,0.65) rectangle (28.35, -0.05);
\draw[draw={rgb,255:red,0; green,0; blue,0},fill={rgb,255:red,200; green,113; blue,55},opacity=0.5,fill opacity=0.5](28.65,0.65) rectangle (29.35, -0.05);
\draw[draw={rgb,255:red,0; green,0; blue,0},fill={rgb,255:red,200; green,113; blue,55},opacity=0.5,fill opacity=0.5](29.65,0.65) rectangle (30.35, -0.05);
\draw[draw={rgb,255:red,0; green,0; blue,0},fill={rgb,255:red,200; green,113; blue,55},opacity=0.5,fill opacity=0.5](30.65,0.65) rectangle (31.35, -0.05);
\draw[draw={rgb,255:red,0; green,0; blue,0},fill={rgb,255:red,200; green,113; blue,55},opacity=0.5,fill opacity=0.5](31.65,0.65) rectangle (32.35, -0.05);
\draw[draw={rgb,255:red,0; green,0; blue,0},fill={rgb,255:red,200; green,113; blue,55},opacity=0.5,fill opacity=0.5](32.65,0.65) rectangle (33.35, -0.05);
\draw[draw={rgb,255:red,200; green,113; blue,55},opacity=0.5](4,-5.7)--(5,-5.7)--(5,-3.7)--(6,-3.7)--(6,-4.7)--(8,-4.7)--(8,-1.7)--(10,-1.7)--(10,-2.7)--(13,-2.7)--(13,-6.7)--(15,-6.7)--(15,-5.7)--(17,-5.7)--(17,-0.7)--(19,-0.7)--(19,0.3)--(21,0.3)--(21,-1.7)--(24,-1.7)--(24,-3.7)--(27,-3.7)--(27,0.3)--(33,0.3);
\draw[draw={rgb,255:red,0; green,0; blue,0},fill={rgb,255:red,255; green,0; blue,0},opacity=0.5,fill opacity=0.5](26.65,-1.35) rectangle (27.35, -2.05);
\draw[draw={rgb,255:red,0; green,0; blue,0},fill={rgb,255:red,255; green,0; blue,0},opacity=0.5,fill opacity=0.5](27.65,-1.35) rectangle (28.35, -2.05);
\draw[draw={rgb,255:red,0; green,0; blue,0},fill={rgb,255:red,255; green,0; blue,0},opacity=0.5,fill opacity=0.5](28.65,-1.35) rectangle (29.35, -2.05);
\draw[draw={rgb,255:red,0; green,0; blue,0},fill={rgb,255:red,255; green,0; blue,0},opacity=0.5,fill opacity=0.5](29.65,-1.35) rectangle (30.35, -2.05);
\draw[draw={rgb,255:red,0; green,0; blue,0},fill={rgb,255:red,255; green,0; blue,0},opacity=0.5,fill opacity=0.5](29.65,-2.35) rectangle (30.35, -3.05);
\draw[draw={rgb,255:red,0; green,0; blue,0},fill={rgb,255:red,255; green,0; blue,0},opacity=0.5,fill opacity=0.5](29.65,-3.35) rectangle (30.35, -4.05);
\draw[draw={rgb,255:red,0; green,0; blue,0},fill={rgb,255:red,255; green,0; blue,0},opacity=0.5,fill opacity=0.5](30.65,-3.35) rectangle (31.35, -4.05);
\draw[draw={rgb,255:red,255; green,0; blue,0},opacity=0.5](27,-1.7)--(30,-1.7)--(30,-3.7)--(31,-3.7);
\draw(30.24, -2.55) node[anchor=south west] {$R_{\rng d b}+\vpij$};
\draw[draw=none,fill={rgb,255:red,0; green,0; blue,0},thin](27, -1.7) ellipse (0.1cm and 0.1cm);\draw(24.36, -2.28) node[anchor=south west] {$v_{n+1}$};
\draw[draw=none,fill={rgb,255:red,0; green,0; blue,0},thin](21, -1.7) ellipse (0.1cm and 0.1cm);\draw(18.36, -2.34) node[anchor=south west] {$u_{n+1}$};
\draw[draw=none,fill={rgb,255:red,0; green,0; blue,0},thin](11, -2.7) ellipse (0.1cm and 0.1cm);\draw[draw=none,fill={rgb,255:red,0; green,0; blue,0},thin](17, -2.7) ellipse (0.1cm and 0.1cm);\draw[draw=none,fill={rgb,255:red,0; green,0; blue,0},thin](14, -6.7) ellipse (0.1cm and 0.1cm);\draw(10.24, -2.46) node[anchor=south west] {$u_n$};
\draw(15.26, -3.26) node[anchor=south west] {$v_n$};
\draw(13.25, -8.23) node[anchor=south west] {$m_n$};
\draw[draw=none,fill={rgb,255:red,0; green,0; blue,0},thin](25, -3.7) ellipse (0.1cm and 0.1cm);\draw(25.2, -5.3) node[anchor=south west] {$m_{n+1}$};
\draw[draw={rgb,255:red,0; green,0; blue,0}](14,-6.7)--(14,-8.7)--(16,-8.7)--(16,-10.7)--(15,-10.7)--(15,-11.7)--(18,-11.7)--(18,-9.7)--(20,-9.7)--(20,-16.7);
\draw[draw={rgb,255:red,0; green,0; blue,0}](32,-16.7)--(32,-9.7)--(30,-9.7)--(30,-11.7)--(27,-11.7)--(27,-10.7)--(28,-10.7)--(28,-8.7)--(26,-8.7)--(26,-6.7)--(25,-6.7)--(25,-3.7)--(27,-3.7)--(27,0.3)--(32.5,0.3)--(32.5,1.8);
\draw[draw=none,fill={rgb,255:red,0; green,0; blue,0},opacity=0.05](32,-9.7)--(30,-9.7)--(30,-11.7)--(27,-11.7)--(27,-10.7)--(28,-10.7)--(28,-8.7)--(26,-8.7)--(26,-6.7)--(25,-6.7)--(25,-3.7)--(27,-3.7)--(27,0.3)--(32.5,0.3)--(32.5,1.8)--(39.5,1.8)--(39.5,-18.2)--(32,-18.2)-- cycle;
\draw[draw={rgb,255:red,0; green,0; blue,0}](38,-16.7)--(38,-9.7)--(36,-9.7)--(36,-11.7)--(33,-11.7)--(33,-10.7)--(34,-10.7)--(34,-8.7)--(32,-8.7)--(32,-6.7)--(31,-6.7)--(31,-3.7);
\draw[draw=none,fill={rgb,255:red,0; green,0; blue,0},thin](31, -3.7) ellipse (0.1cm and 0.1cm);\draw[draw={rgb,255:red,0; green,0; blue,0},dashed,thin](20,-18.7)--(20,-16.7);
\draw(18.65, -17.51) node[anchor=south west] {$f_n$};
\draw(32.19, -17.51) node[anchor=south west] {$g_{n+1}$};
\draw[draw={rgb,255:red,0; green,0; blue,0},dashed,thin](32,-18.7)--(32,-16.7);
\draw[draw={rgb,255:red,0; green,0; blue,0},dashed,thin](38,-18.7)--(38,-16.7);
\draw(38.02, -17.51) node[anchor=south west] {$f_{n+1}+\vpij$};
\end{tikzpicture}
        \end{center}
        \caption{Finally, we find $u_{n+1}$ and $v_{n+1}$ by translating $f_{n+1}$ by $\vpij$, and looking at the first intersection between $P_{\rng i {m_{n+1}}}+\vpij$ and $P$. In order to find such an intersection, we define the infinite curve $g_{n+1}$, which cuts $\mathbb{R}^2$ into two connected components by Theorem~\ref{thm:infinite-jordan}.}
        \label{fig:induction-gn1}
      \end{figure}

      To conclude, we proceed as in Section \ref{subsec:u0v0} by using $R$ to find $u_{n+1}$ and $v_{n+1}$. The path $R$ starts on the left-hand side of $g_{n+1}$ (at $\pos{P_{i+1}}$, which is not on $f_{n+1}$ since by \subl{lem:m0i1} $m_0>i+1$) and ends on the right-hand side of $g_{n+1}$ or on $g_{n+1}$ (at a position within horizontal distance 0.5 of $l^k$, and possibly at $\pos{P_k}$). Therefore, $R$ must intersect $g_{n+1}$.

      Now, $f_{n+1}$ only intersects $\Pik$ at $\pos{P_{m_{n+1}}}$, and $f_{n+1}+\vpij$ does not intersect $c$ (by~\ref{umv:s}), and hence does not intersect $\Pik$. Therefore, $f_{n+1}$ does not intersect $\Pjk$.
      Therefore, $R$ intersects $P_{m_{n+1}}$.
      Let $b$ be the integer such that $R_{b} = P_{m_{n+1}}$, and let $d\leq b$ be the largest integer such that $R_d+\vpij$ is on $\Pik$ (at least $i+1\leq b$ has that property, hence $d$ is well-defined).

      Since $d$ is maximal, $R_{\rng d b}$ is a segment of $P$. Let $u_{n+1}$ be such that $R_d = P_{u_{n+1}}$. Since $d\leq b$, by \subl{lem:order}, $u_{n+1}\leq m_{n+1}$. Let $v_{n+1}$ be such that $R_d+\vpij = P_{v_{n+1}}$.
      These tiles have the same type by \subl{lem:r}.

      We claim that $v_{n+1}\geq m_{n+1}$.
      By Hypothesis~\ref{lem:hp:shield backup} of Definition~\ref{def:shield},
      $\embed{P_{\rng {u_{n+1}} {m_{n+1}}}+\vpij}$ can only intersect $l^k$ at $\pos{\glu P k}$, and  can only intersect $\gs{P_k}{l^k(0)}$ at $\pos{P_k}$, in which case $k = v_{n+1}$.
      Moreover, we have already argued that $\embed{P_{\rng {u_{n+1}} {m_{n+1}}}+\vpij}$ can only intersect $f_{n+1}$ at $\pos{P_{m_{n+1}}}$ (only in the case where $m_{n+1} = v_{n+1}$). Therefore, $\embed{P_{\rng {u_{n+1}} {m_{n+1}}}+\vpij}$ intersects $g_n$ exactly once, at $\pos{P_{v_{n+1}}}$ (or we are in the special case where $P$ is pumpable).

      Finally, since $\embed{P_{\rng {u_{n+1}} {m_{n+1}}}+\vpij}$ starts on the right-hand side of $g_{n+1}$, and intersects $g_{n+1}$ only at $\pos{P_{u_{n+1}}}+\vpij$, $\embed{P_{\rng {u_{n+1}} {m_{n+1}}}+\vpij}$ is entirely on the right-hand side of $g_{n+1}$, which means that $v_{n+1}\geq m_{n+1}$.
    \end{enumerate}
  \end{enumerate}

  \argument{Conclusion of the inductive argument} Since for all $n\geq 0$, $m_{n+1}>m_n$, and since $P$ is of finite length ($|P|=k+2$) we eventually run out of indices along $P$.
  This implies that at some point end up in  Case~\ref{shield:induction:pumpable} above, where $m_{n+1}=m_n$, which implies that $P$ is pumpable with pumping vector $\vpij$.
  This completes the proof of Lemma~\ref{lem:shield}.
\end{proof}

\section{Proof of main theorem}\label{sec:main thm proof}

\subsection{Intuition and roadmap for the proof of Theorem~\ref{thm:intro main thm}}\label{sec:main thm intuition}

To prove Theorem~\ref{thm:intro main thm}, we need to find three indices $i,j,k \in \{ 0,1,\ldots, |P|-1\}$ of $P$ that satisfy the hypotheses of the Shield Lemma (Lemma~\ref{lem:shield}, Definition~\ref{def:shield}) on $P$, the conclusion of which is that $P$ is pumpable or fragile. Throughout the proof we apply the Shield Lemma in several different ways.
We proceed in three steps:
\begin{enumerate}
\item First, in the proof of Theorem~\ref{thm:intro main thm}, 
we make some trivial modifications to $P$ (a rotation and translation of our frame of reference, and a truncation) so that $P$  is in a canonical form where, intuitively, it reaches far to the east ($P$'s final tile is to the east of the seed $\sigma$). We then invoke Theorem~\ref{thm:main}, the combinatorial-based proof of which goes through the following steps.
\item In Lemma~\ref{lem:couple:almost}, we use
Lemma~\ref{lem:glue:almost}
to show that either $P$ is pumpable or fragile (in turn, by applying Lemma~\ref{lem:shield}), or else that at most $|T|+1$ glues of $P$ that are visible from the south or from the north can be pointing west (the remaining $\geq (4|T|)^{(4|T|+1)}(4|\sigma|+6) - |T|-|\sigma|$ glues visible from the south are pointing east).
\item Then, using the notion of \emph{spans} (Definitions~\ref{def:span} and~\ref{def:span properties}, intuition in Figure~\ref{fig:spans})
  we show (Lemma~\ref{lem:couple:use}) that if we find two spans $S = (s, n)$ and $S' = (s',n')$ of the same orientation and type, and such that $s<s'$ and the height of $S'$ is at least the height of $S$, then we can apply Lemma~\ref{lem:shield} to $P$, proving that $P$ is pumpable or fragile.
\begin{figure}[t]
  \begin{center}
    \begin{tikzpicture}[scale=\scale]\draw[draw={rgb,255:red,200; green,200; blue,200}](1.5,3.8) rectangle (30.5, -9.2);
\draw[draw={rgb,255:red,200; green,200; blue,200}](1.5,-9.2)--(1.5,3.8);
\draw[draw={rgb,255:red,200; green,200; blue,200}](2.5,-9.2)--(2.5,3.8);
\draw[draw={rgb,255:red,200; green,200; blue,200}](3.5,-9.2)--(3.5,3.8);
\draw[draw={rgb,255:red,200; green,200; blue,200}](4.5,-9.2)--(4.5,3.8);
\draw[draw={rgb,255:red,200; green,200; blue,200}](5.5,-9.2)--(5.5,3.8);
\draw[draw={rgb,255:red,200; green,200; blue,200}](6.5,-9.2)--(6.5,3.8);
\draw[draw={rgb,255:red,200; green,200; blue,200}](7.5,-9.2)--(7.5,3.8);
\draw[draw={rgb,255:red,200; green,200; blue,200}](8.5,-9.2)--(8.5,3.8);
\draw[draw={rgb,255:red,200; green,200; blue,200}](9.5,-9.2)--(9.5,3.8);
\draw[draw={rgb,255:red,200; green,200; blue,200}](10.5,-9.2)--(10.5,3.8);
\draw[draw={rgb,255:red,200; green,200; blue,200}](11.5,-9.2)--(11.5,3.8);
\draw[draw={rgb,255:red,200; green,200; blue,200}](12.5,-9.2)--(12.5,3.8);
\draw[draw={rgb,255:red,200; green,200; blue,200}](13.5,-9.2)--(13.5,3.8);
\draw[draw={rgb,255:red,200; green,200; blue,200}](14.5,-9.2)--(14.5,3.8);
\draw[draw={rgb,255:red,200; green,200; blue,200}](15.5,-9.2)--(15.5,3.8);
\draw[draw={rgb,255:red,200; green,200; blue,200}](16.5,-9.2)--(16.5,3.8);
\draw[draw={rgb,255:red,200; green,200; blue,200}](17.5,-9.2)--(17.5,3.8);
\draw[draw={rgb,255:red,200; green,200; blue,200}](18.5,-9.2)--(18.5,3.8);
\draw[draw={rgb,255:red,200; green,200; blue,200}](19.5,-9.2)--(19.5,3.8);
\draw[draw={rgb,255:red,200; green,200; blue,200}](20.5,-9.2)--(20.5,3.8);
\draw[draw={rgb,255:red,200; green,200; blue,200}](21.5,-9.2)--(21.5,3.8);
\draw[draw={rgb,255:red,200; green,200; blue,200}](22.5,-9.2)--(22.5,3.8);
\draw[draw={rgb,255:red,200; green,200; blue,200}](23.5,-9.2)--(23.5,3.8);
\draw[draw={rgb,255:red,200; green,200; blue,200}](24.5,-9.2)--(24.5,3.8);
\draw[draw={rgb,255:red,200; green,200; blue,200}](25.5,-9.2)--(25.5,3.8);
\draw[draw={rgb,255:red,200; green,200; blue,200}](26.5,-9.2)--(26.5,3.8);
\draw[draw={rgb,255:red,200; green,200; blue,200}](27.5,-9.2)--(27.5,3.8);
\draw[draw={rgb,255:red,200; green,200; blue,200}](28.5,-9.2)--(28.5,3.8);
\draw[draw={rgb,255:red,200; green,200; blue,200}](29.5,-9.2)--(29.5,3.8);
\draw[draw={rgb,255:red,200; green,200; blue,200}](1.5,3.8)--(30.5,3.8);
\draw[draw={rgb,255:red,200; green,200; blue,200}](1.5,2.8)--(30.5,2.8);
\draw[draw={rgb,255:red,200; green,200; blue,200}](1.5,1.8)--(30.5,1.8);
\draw[draw={rgb,255:red,200; green,200; blue,200}](1.5,0.8)--(30.5,0.8);
\draw[draw={rgb,255:red,200; green,200; blue,200}](1.5,-0.2)--(30.5,-0.2);
\draw[draw={rgb,255:red,200; green,200; blue,200}](1.5,-1.2)--(30.5,-1.2);
\draw[draw={rgb,255:red,200; green,200; blue,200}](1.5,-2.2)--(30.5,-2.2);
\draw[draw={rgb,255:red,200; green,200; blue,200}](1.5,-3.2)--(30.5,-3.2);
\draw[draw={rgb,255:red,200; green,200; blue,200}](1.5,-4.2)--(30.5,-4.2);
\draw[draw={rgb,255:red,200; green,200; blue,200}](1.5,-5.2)--(30.5,-5.2);
\draw[draw={rgb,255:red,200; green,200; blue,200}](1.5,-6.2)--(30.5,-6.2);
\draw[draw={rgb,255:red,200; green,200; blue,200}](1.5,-7.2)--(30.5,-7.2);
\draw[draw={rgb,255:red,200; green,200; blue,200}](1.5,-8.2)--(30.5,-8.2);
\draw[draw={rgb,255:red,0; green,0; blue,0},fill={rgb,255:red,0; green,204; blue,255},opacity=0.5,fill opacity=0.5](3.65,-1.35) rectangle (4.35, -2.05);
\draw[draw={rgb,255:red,0; green,0; blue,0},fill={rgb,255:red,0; green,204; blue,255},opacity=0.5,fill opacity=0.5](4.65,-1.35) rectangle (5.35, -2.05);
\draw[draw={rgb,255:red,0; green,0; blue,0},fill={rgb,255:red,0; green,204; blue,255},opacity=0.5,fill opacity=0.5](5.65,-1.35) rectangle (6.35, -2.05);
\draw[draw={rgb,255:red,0; green,0; blue,0},fill={rgb,255:red,0; green,204; blue,255},opacity=0.5,fill opacity=0.5](5.65,-2.35) rectangle (6.35, -3.05);
\draw[draw={rgb,255:red,0; green,0; blue,0},fill={rgb,255:red,0; green,204; blue,255},opacity=0.5,fill opacity=0.5](4.65,-2.35) rectangle (5.35, -3.05);
\draw[draw={rgb,255:red,0; green,0; blue,0},fill={rgb,255:red,0; green,204; blue,255},opacity=0.5,fill opacity=0.5](3.65,-2.35) rectangle (4.35, -3.05);
\draw[draw={rgb,255:red,0; green,0; blue,0},fill={rgb,255:red,0; green,204; blue,255},opacity=0.5,fill opacity=0.5](3.65,-3.35) rectangle (4.35, -4.05);
\draw[draw={rgb,255:red,0; green,0; blue,0},fill={rgb,255:red,0; green,204; blue,255},opacity=0.5,fill opacity=0.5](3.65,-4.35) rectangle (4.35, -5.05);
\draw[draw={rgb,255:red,0; green,0; blue,0},fill={rgb,255:red,0; green,204; blue,255},opacity=0.5,fill opacity=0.5](3.65,-5.35) rectangle (4.35, -6.05);
\draw[draw={rgb,255:red,0; green,0; blue,0},fill={rgb,255:red,0; green,204; blue,255},opacity=0.5,fill opacity=0.5](3.65,-6.35) rectangle (4.35, -7.05);
\draw[draw={rgb,255:red,0; green,0; blue,0},fill={rgb,255:red,0; green,204; blue,255},opacity=0.5,fill opacity=0.5](4.65,-6.35) rectangle (5.35, -7.05);
\draw[draw={rgb,255:red,0; green,0; blue,0},fill={rgb,255:red,0; green,204; blue,255},opacity=0.5,fill opacity=0.5](5.65,-6.35) rectangle (6.35, -7.05);
\draw[draw={rgb,255:red,0; green,0; blue,0},fill={rgb,255:red,0; green,204; blue,255},opacity=0.5,fill opacity=0.5](5.65,-7.35) rectangle (6.35, -8.05);
\draw[draw={rgb,255:red,0; green,0; blue,0},fill={rgb,255:red,0; green,204; blue,255},opacity=0.5,fill opacity=0.5](4.65,-7.35) rectangle (5.35, -8.05);
\draw[draw={rgb,255:red,0; green,0; blue,0},fill={rgb,255:red,0; green,204; blue,255},opacity=0.5,fill opacity=0.5](3.65,-7.35) rectangle (4.35, -8.05);
\draw[draw={rgb,255:red,0; green,0; blue,0},fill={rgb,255:red,0; green,204; blue,255},opacity=0.5,fill opacity=0.5](2.65,-7.35) rectangle (3.35, -8.05);
\draw[draw={rgb,255:red,0; green,0; blue,0},fill={rgb,255:red,0; green,204; blue,255},opacity=0.5,fill opacity=0.5](2.65,-6.35) rectangle (3.35, -7.05);
\draw[draw={rgb,255:red,0; green,0; blue,0},fill={rgb,255:red,0; green,204; blue,255},opacity=0.5,fill opacity=0.5](2.65,-5.35) rectangle (3.35, -6.05);
\draw[draw={rgb,255:red,0; green,204; blue,255},opacity=0.5,thick](3.5,-1.7)--(6,-1.7)--(6,-2.7)--(4,-2.7)--(4,-6.7)--(6,-6.7)--(6,-7.7)--(3,-7.7)--(3,-5.7);
\draw(2.28, -5.33) node[anchor=south west] {$\sigma$};
\draw[draw={rgb,255:red,0; green,0; blue,0},fill={rgb,255:red,200; green,113; blue,55},opacity=0.5,fill opacity=0.5](4.65,-4.35) rectangle (5.35, -5.05);
\draw[draw={rgb,255:red,0; green,0; blue,0},fill={rgb,255:red,200; green,113; blue,55},opacity=0.5,fill opacity=0.5](5.65,-4.35) rectangle (6.35, -5.05);
\draw[draw={rgb,255:red,0; green,0; blue,0},fill={rgb,255:red,200; green,113; blue,55},opacity=0.5,fill opacity=0.5](6.65,-4.35) rectangle (7.35, -5.05);
\draw[draw={rgb,255:red,0; green,0; blue,0},fill={rgb,255:red,200; green,113; blue,55},opacity=0.5,fill opacity=0.5](7.65,-4.35) rectangle (8.35, -5.05);
\draw[draw={rgb,255:red,0; green,0; blue,0},fill={rgb,255:red,200; green,113; blue,55},opacity=0.5,fill opacity=0.5](8.65,-4.35) rectangle (9.35, -5.05);
\draw[draw={rgb,255:red,0; green,0; blue,0},fill={rgb,255:red,200; green,113; blue,55},opacity=0.5,fill opacity=0.5](8.65,-5.35) rectangle (9.35, -6.05);
\draw[draw={rgb,255:red,0; green,0; blue,0},fill={rgb,255:red,200; green,113; blue,55},opacity=0.5,fill opacity=0.5](8.65,-6.35) rectangle (9.35, -7.05);
\draw[draw={rgb,255:red,0; green,0; blue,0},fill={rgb,255:red,200; green,113; blue,55},opacity=0.5,fill opacity=0.5](9.65,-6.35) rectangle (10.35, -7.05);
\draw[draw={rgb,255:red,0; green,0; blue,0},fill={rgb,255:red,200; green,113; blue,55},opacity=0.5,fill opacity=0.5](10.65,-6.35) rectangle (11.35, -7.05);
\draw[draw={rgb,255:red,0; green,0; blue,0},fill={rgb,255:red,200; green,113; blue,55},opacity=0.5,fill opacity=0.5](11.65,-6.35) rectangle (12.35, -7.05);
\draw[draw={rgb,255:red,0; green,0; blue,0},fill={rgb,255:red,200; green,113; blue,55},opacity=0.5,fill opacity=0.5](11.65,-5.35) rectangle (12.35, -6.05);
\draw[draw={rgb,255:red,0; green,0; blue,0},fill={rgb,255:red,200; green,113; blue,55},opacity=0.5,fill opacity=0.5](11.65,-4.35) rectangle (12.35, -5.05);
\draw[draw={rgb,255:red,0; green,0; blue,0},fill={rgb,255:red,200; green,113; blue,55},opacity=0.5,fill opacity=0.5](11.65,-3.35) rectangle (12.35, -4.05);
\draw[draw={rgb,255:red,0; green,0; blue,0},fill={rgb,255:red,200; green,113; blue,55},opacity=0.5,fill opacity=0.5](11.65,-2.35) rectangle (12.35, -3.05);
\draw[draw={rgb,255:red,0; green,0; blue,0},fill={rgb,255:red,200; green,113; blue,55},opacity=0.5,fill opacity=0.5](10.65,-2.35) rectangle (11.35, -3.05);
\draw[draw={rgb,255:red,0; green,0; blue,0},fill={rgb,255:red,200; green,113; blue,55},opacity=0.5,fill opacity=0.5](9.65,-2.35) rectangle (10.35, -3.05);
\draw[draw={rgb,255:red,0; green,0; blue,0},fill={rgb,255:red,200; green,113; blue,55},opacity=0.5,fill opacity=0.5](9.65,-1.35) rectangle (10.35, -2.05);
\draw[draw={rgb,255:red,0; green,0; blue,0},fill={rgb,255:red,200; green,113; blue,55},opacity=0.5,fill opacity=0.5](9.65,-0.35) rectangle (10.35, -1.05);
\draw[draw={rgb,255:red,0; green,0; blue,0},fill={rgb,255:red,200; green,113; blue,55},opacity=0.5,fill opacity=0.5](10.65,-0.35) rectangle (11.35, -1.05);
\draw[draw={rgb,255:red,0; green,0; blue,0},fill={rgb,255:red,200; green,113; blue,55},opacity=0.5,fill opacity=0.5](11.65,-0.35) rectangle (12.35, -1.05);
\draw[draw={rgb,255:red,0; green,0; blue,0},fill={rgb,255:red,200; green,113; blue,55},opacity=0.5,fill opacity=0.5](12.65,-0.35) rectangle (13.35, -1.05);
\draw[draw={rgb,255:red,0; green,0; blue,0},fill={rgb,255:red,200; green,113; blue,55},opacity=0.5,fill opacity=0.5](13.65,-0.35) rectangle (14.35, -1.05);
\draw[draw={rgb,255:red,0; green,0; blue,0},fill={rgb,255:red,200; green,113; blue,55},opacity=0.5,fill opacity=0.5](13.65,-1.35) rectangle (14.35, -2.05);
\draw[draw={rgb,255:red,0; green,0; blue,0},fill={rgb,255:red,200; green,113; blue,55},opacity=0.5,fill opacity=0.5](13.65,-2.35) rectangle (14.35, -3.05);
\draw[draw={rgb,255:red,0; green,0; blue,0},fill={rgb,255:red,200; green,113; blue,55},opacity=0.5,fill opacity=0.5](13.65,-3.35) rectangle (14.35, -4.05);
\draw[draw={rgb,255:red,0; green,0; blue,0},fill={rgb,255:red,200; green,113; blue,55},opacity=0.5,fill opacity=0.5](13.65,-4.35) rectangle (14.35, -5.05);
\draw[draw={rgb,255:red,0; green,0; blue,0},fill={rgb,255:red,200; green,113; blue,55},opacity=0.5,fill opacity=0.5](13.65,-5.35) rectangle (14.35, -6.05);
\draw[draw={rgb,255:red,0; green,0; blue,0},fill={rgb,255:red,200; green,113; blue,55},opacity=0.5,fill opacity=0.5](14.65,-5.35) rectangle (15.35, -6.05);
\draw[draw={rgb,255:red,0; green,0; blue,0},fill={rgb,255:red,200; green,113; blue,55},opacity=0.5,fill opacity=0.5](15.65,-5.35) rectangle (16.35, -6.05);
\draw[draw={rgb,255:red,0; green,0; blue,0},fill={rgb,255:red,200; green,113; blue,55},opacity=0.5,fill opacity=0.5](16.65,-5.35) rectangle (17.35, -6.05);
\draw[draw={rgb,255:red,0; green,0; blue,0},fill={rgb,255:red,200; green,113; blue,55},opacity=0.5,fill opacity=0.5](17.65,-5.35) rectangle (18.35, -6.05);
\draw[draw={rgb,255:red,0; green,0; blue,0},fill={rgb,255:red,200; green,113; blue,55},opacity=0.5,fill opacity=0.5](17.65,-4.35) rectangle (18.35, -5.05);
\draw[draw={rgb,255:red,0; green,0; blue,0},fill={rgb,255:red,200; green,113; blue,55},opacity=0.5,fill opacity=0.5](17.65,-3.35) rectangle (18.35, -4.05);
\draw[draw={rgb,255:red,0; green,0; blue,0},fill={rgb,255:red,200; green,113; blue,55},opacity=0.5,fill opacity=0.5](16.65,-3.35) rectangle (17.35, -4.05);
\draw[draw={rgb,255:red,0; green,0; blue,0},fill={rgb,255:red,200; green,113; blue,55},opacity=0.5,fill opacity=0.5](15.65,-3.35) rectangle (16.35, -4.05);
\draw[draw={rgb,255:red,0; green,0; blue,0},fill={rgb,255:red,200; green,113; blue,55},opacity=0.5,fill opacity=0.5](15.65,-2.35) rectangle (16.35, -3.05);
\draw[draw={rgb,255:red,0; green,0; blue,0},fill={rgb,255:red,200; green,113; blue,55},opacity=0.5,fill opacity=0.5](15.65,-1.35) rectangle (16.35, -2.05);
\draw[draw={rgb,255:red,0; green,0; blue,0},fill={rgb,255:red,200; green,113; blue,55},opacity=0.5,fill opacity=0.5](16.65,-1.35) rectangle (17.35, -2.05);
\draw[draw={rgb,255:red,0; green,0; blue,0},fill={rgb,255:red,200; green,113; blue,55},opacity=0.5,fill opacity=0.5](17.65,-1.35) rectangle (18.35, -2.05);
\draw[draw={rgb,255:red,0; green,0; blue,0},fill={rgb,255:red,200; green,113; blue,55},opacity=0.5,fill opacity=0.5](17.65,-0.35) rectangle (18.35, -1.05);
\draw[draw={rgb,255:red,0; green,0; blue,0},fill={rgb,255:red,200; green,113; blue,55},opacity=0.5,fill opacity=0.5](17.65,0.65) rectangle (18.35, -0.05);
\draw[draw={rgb,255:red,0; green,0; blue,0},fill={rgb,255:red,200; green,113; blue,55},opacity=0.5,fill opacity=0.5](16.65,0.65) rectangle (17.35, -0.05);
\draw[draw={rgb,255:red,0; green,0; blue,0},fill={rgb,255:red,200; green,113; blue,55},opacity=0.5,fill opacity=0.5](15.65,0.65) rectangle (16.35, -0.05);
\draw[draw={rgb,255:red,0; green,0; blue,0},fill={rgb,255:red,200; green,113; blue,55},opacity=0.5,fill opacity=0.5](15.65,1.65) rectangle (16.35, 0.95);
\draw[draw={rgb,255:red,0; green,0; blue,0},fill={rgb,255:red,200; green,113; blue,55},opacity=0.5,fill opacity=0.5](15.65,2.65) rectangle (16.35, 1.95);
\draw[draw={rgb,255:red,0; green,0; blue,0},fill={rgb,255:red,200; green,113; blue,55},opacity=0.5,fill opacity=0.5](16.65,2.65) rectangle (17.35, 1.95);
\draw[draw={rgb,255:red,0; green,0; blue,0},fill={rgb,255:red,200; green,113; blue,55},opacity=0.5,fill opacity=0.5](17.65,2.65) rectangle (18.35, 1.95);
\draw[draw={rgb,255:red,0; green,0; blue,0},fill={rgb,255:red,200; green,113; blue,55},opacity=0.5,fill opacity=0.5](18.65,2.65) rectangle (19.35, 1.95);
\draw[draw={rgb,255:red,0; green,0; blue,0},fill={rgb,255:red,200; green,113; blue,55},opacity=0.5,fill opacity=0.5](19.65,2.65) rectangle (20.35, 1.95);
\draw[draw={rgb,255:red,0; green,0; blue,0},fill={rgb,255:red,200; green,113; blue,55},opacity=0.5,fill opacity=0.5](20.65,2.65) rectangle (21.35, 1.95);
\draw[draw={rgb,255:red,0; green,0; blue,0},fill={rgb,255:red,200; green,113; blue,55},opacity=0.5,fill opacity=0.5](20.65,1.65) rectangle (21.35, 0.95);
\draw[draw={rgb,255:red,0; green,0; blue,0},fill={rgb,255:red,200; green,113; blue,55},opacity=0.5,fill opacity=0.5](20.65,0.65) rectangle (21.35, -0.05);
\draw[draw={rgb,255:red,0; green,0; blue,0},fill={rgb,255:red,200; green,113; blue,55},opacity=0.5,fill opacity=0.5](20.65,-0.35) rectangle (21.35, -1.05);
\draw[draw={rgb,255:red,0; green,0; blue,0},fill={rgb,255:red,200; green,113; blue,55},opacity=0.5,fill opacity=0.5](20.65,-1.35) rectangle (21.35, -2.05);
\draw[draw={rgb,255:red,0; green,0; blue,0},fill={rgb,255:red,200; green,113; blue,55},opacity=0.5,fill opacity=0.5](20.65,-2.35) rectangle (21.35, -3.05);
\draw[draw={rgb,255:red,0; green,0; blue,0},fill={rgb,255:red,200; green,113; blue,55},opacity=0.5,fill opacity=0.5](20.65,-3.35) rectangle (21.35, -4.05);
\draw[draw={rgb,255:red,0; green,0; blue,0},fill={rgb,255:red,200; green,113; blue,55},opacity=0.5,fill opacity=0.5](19.65,-3.35) rectangle (20.35, -4.05);
\draw[draw={rgb,255:red,0; green,0; blue,0},fill={rgb,255:red,200; green,113; blue,55},opacity=0.5,fill opacity=0.5](19.65,-4.35) rectangle (20.35, -5.05);
\draw[draw={rgb,255:red,0; green,0; blue,0},fill={rgb,255:red,200; green,113; blue,55},opacity=0.5,fill opacity=0.5](19.65,-5.35) rectangle (20.35, -6.05);
\draw[draw={rgb,255:red,0; green,0; blue,0},fill={rgb,255:red,200; green,113; blue,55},opacity=0.5,fill opacity=0.5](19.65,-6.35) rectangle (20.35, -7.05);
\draw[draw={rgb,255:red,0; green,0; blue,0},fill={rgb,255:red,200; green,113; blue,55},opacity=0.5,fill opacity=0.5](20.65,-6.35) rectangle (21.35, -7.05);
\draw[draw={rgb,255:red,0; green,0; blue,0},fill={rgb,255:red,200; green,113; blue,55},opacity=0.5,fill opacity=0.5](21.65,-6.35) rectangle (22.35, -7.05);
\draw[draw={rgb,255:red,0; green,0; blue,0},fill={rgb,255:red,200; green,113; blue,55},opacity=0.5,fill opacity=0.5](22.65,-6.35) rectangle (23.35, -7.05);
\draw[draw={rgb,255:red,0; green,0; blue,0},fill={rgb,255:red,200; green,113; blue,55},opacity=0.5,fill opacity=0.5](22.65,-5.35) rectangle (23.35, -6.05);
\draw[draw={rgb,255:red,0; green,0; blue,0},fill={rgb,255:red,200; green,113; blue,55},opacity=0.5,fill opacity=0.5](22.65,-4.35) rectangle (23.35, -5.05);
\draw[draw={rgb,255:red,0; green,0; blue,0},fill={rgb,255:red,200; green,113; blue,55},opacity=0.5,fill opacity=0.5](22.65,-3.35) rectangle (23.35, -4.05);
\draw[draw={rgb,255:red,0; green,0; blue,0},fill={rgb,255:red,200; green,113; blue,55},opacity=0.5,fill opacity=0.5](22.65,-2.35) rectangle (23.35, -3.05);
\draw[draw={rgb,255:red,0; green,0; blue,0},fill={rgb,255:red,200; green,113; blue,55},opacity=0.5,fill opacity=0.5](22.65,-1.35) rectangle (23.35, -2.05);
\draw[draw={rgb,255:red,0; green,0; blue,0},fill={rgb,255:red,200; green,113; blue,55},opacity=0.5,fill opacity=0.5](22.65,-0.35) rectangle (23.35, -1.05);
\draw[draw={rgb,255:red,0; green,0; blue,0},fill={rgb,255:red,200; green,113; blue,55},opacity=0.5,fill opacity=0.5](22.65,0.65) rectangle (23.35, -0.05);
\draw[draw={rgb,255:red,0; green,0; blue,0},fill={rgb,255:red,200; green,113; blue,55},opacity=0.5,fill opacity=0.5](22.65,1.65) rectangle (23.35, 0.95);
\draw[draw={rgb,255:red,0; green,0; blue,0},fill={rgb,255:red,200; green,113; blue,55},opacity=0.5,fill opacity=0.5](22.65,2.65) rectangle (23.35, 1.95);
\draw[draw={rgb,255:red,0; green,0; blue,0},fill={rgb,255:red,200; green,113; blue,55},opacity=0.5,fill opacity=0.5](23.65,2.65) rectangle (24.35, 1.95);
\draw[draw={rgb,255:red,0; green,0; blue,0},fill={rgb,255:red,200; green,113; blue,55},opacity=0.5,fill opacity=0.5](24.65,2.65) rectangle (25.35, 1.95);
\draw[draw={rgb,255:red,0; green,0; blue,0},fill={rgb,255:red,200; green,113; blue,55},opacity=0.5,fill opacity=0.5](25.65,2.65) rectangle (26.35, 1.95);
\draw[draw={rgb,255:red,0; green,0; blue,0},fill={rgb,255:red,200; green,113; blue,55},opacity=0.5,fill opacity=0.5](26.65,2.65) rectangle (27.35, 1.95);
\draw[draw={rgb,255:red,0; green,0; blue,0},fill={rgb,255:red,200; green,113; blue,55},opacity=0.5,fill opacity=0.5](27.65,2.65) rectangle (28.35, 1.95);
\draw[draw={rgb,255:red,0; green,0; blue,0},fill={rgb,255:red,200; green,113; blue,55},opacity=0.5,fill opacity=0.5](27.65,1.65) rectangle (28.35, 0.95);
\draw[draw={rgb,255:red,0; green,0; blue,0},fill={rgb,255:red,200; green,113; blue,55},opacity=0.5,fill opacity=0.5](27.65,0.65) rectangle (28.35, -0.05);
\draw[draw={rgb,255:red,0; green,0; blue,0},fill={rgb,255:red,200; green,113; blue,55},opacity=0.5,fill opacity=0.5](27.65,-0.35) rectangle (28.35, -1.05);
\draw[draw={rgb,255:red,0; green,0; blue,0},fill={rgb,255:red,200; green,113; blue,55},opacity=0.5,fill opacity=0.5](28.65,-0.35) rectangle (29.35, -1.05);
\draw[draw={rgb,255:red,0; green,0; blue,0},fill={rgb,255:red,200; green,113; blue,55},opacity=0.5,fill opacity=0.5](28.65,-1.35) rectangle (29.35, -2.05);
\draw[draw={rgb,255:red,0; green,0; blue,0},fill={rgb,255:red,200; green,113; blue,55},opacity=0.5,fill opacity=0.5](28.65,-2.35) rectangle (29.35, -3.05);
\draw[draw={rgb,255:red,0; green,0; blue,0},fill={rgb,255:red,200; green,113; blue,55},opacity=0.5,fill opacity=0.5](28.65,-3.35) rectangle (29.35, -4.05);
\draw[draw={rgb,255:red,0; green,0; blue,0},fill={rgb,255:red,200; green,113; blue,55},opacity=0.5,fill opacity=0.5](28.65,-4.35) rectangle (29.35, -5.05);
\draw[draw={rgb,255:red,0; green,0; blue,0},fill={rgb,255:red,200; green,113; blue,55},opacity=0.5,fill opacity=0.5](28.65,-5.35) rectangle (29.35, -6.05);
\draw[draw={rgb,255:red,0; green,0; blue,0},fill={rgb,255:red,200; green,113; blue,55},opacity=0.5,fill opacity=0.5](27.65,-5.35) rectangle (28.35, -6.05);
\draw[draw={rgb,255:red,0; green,0; blue,0},fill={rgb,255:red,200; green,113; blue,55},opacity=0.5,fill opacity=0.5](26.65,-5.35) rectangle (27.35, -6.05);
\draw[draw={rgb,255:red,0; green,0; blue,0},fill={rgb,255:red,200; green,113; blue,55},opacity=0.5,fill opacity=0.5](25.65,-5.35) rectangle (26.35, -6.05);
\draw[draw={rgb,255:red,200; green,113; blue,55},opacity=0.5](4.5,-4.7)--(9,-4.7)--(9,-6.7)--(12,-6.7)--(12,-2.7)--(10,-2.7)--(10,-0.7)--(14,-0.7)--(14,-5.7)--(18,-5.7)--(18,-3.7)--(16,-3.7)--(16,-1.7)--(18,-1.7)--(18,0.3)--(16,0.3)--(16,2.3)--(21,2.3)--(21,-3.7)--(20,-3.7)--(20,-6.7)--(23,-6.7)--(23,2.3)--(28,2.3)--(28,-0.7)--(29,-0.7)--(29,-5.7)--(26,-5.7);
\draw[draw={rgb,255:red,0; green,0; blue,0}](10.5,3.8)--(10.5,-9.2);
\draw[draw={rgb,255:red,0; green,0; blue,0}](16.5,3.8)--(16.5,-9.2);
\draw[draw=none,fill={rgb,255:red,28; green,36; blue,31},thin](10, -6.7) ellipse (0.1cm and 0.1cm);\draw[draw=none,fill={rgb,255:red,28; green,36; blue,31},thin](10, -0.7) ellipse (0.1cm and 0.1cm);\draw[draw=none,fill={rgb,255:red,28; green,36; blue,31},thin](16, 2.3) ellipse (0.1cm and 0.1cm);\draw[draw=none,fill={rgb,255:red,28; green,36; blue,31},thin](16, -5.7) ellipse (0.1cm and 0.1cm);\draw(8.91, -8.73) node[anchor=south west] {$P_s$};
\draw(8.91, -0.2) node[anchor=south west] {$P_n$};
\draw(14.68, 2.6) node[anchor=south west] {$P_{n'}$};
\draw(14.89, -7.66) node[anchor=south west] {$P_{s'}$};
\draw[draw={rgb,255:red,0; green,0; blue,0}](27.5,3.8)--(27.5,-9.2);
\draw[draw=none,fill={rgb,255:red,28; green,36; blue,31},thin](27, 2.3) ellipse (0.1cm and 0.1cm);\draw[draw=none,fill={rgb,255:red,28; green,36; blue,31},thin](28, -5.7) ellipse (0.1cm and 0.1cm);\draw(25.57, 2.6) node[anchor=south west] {$P_{n''}$};
\draw(27.15, -7.66) node[anchor=south west] {$P_{s''}$};
\end{tikzpicture}
  \end{center}
  \caption{The seed is in blue, and $P$ is in brown. Three example spans are shown in this figure, $S = (s,n)$, $S'=(s',n')$ and $S''=(s'',n'')$.     A span $(s,n)$ is a pair of indices such that $\glu P s$ is visible from the south, $\glu P n$ is visible from the north, and $\glu P s$ and $\glu P n$ are on the same glue column. Here, both $S$ and $S'$ have both of their glues pointing east, hence we say that the span is pointing east. Moreover, since $s<n$ and $s'<n'$, $S$ and $S'$ are ``up spans''. On the other hand, $S''$ has its south glue pointing east, and its north glue pointing west, hence $S''$ is not pointing in any particular direction. Moreover, $n''<s''$, hence $S''$ is a ``down span''. Span $S$ has height 6, and spans $S'$ and $S''$ have height 8. 
    If the spans $(s,n)$ and $(s',n')$ happen to be of the same type (meaning $\type{\glueP{s}{s+1}}=\type{\glueP{s'}{s'+1}}$), since the height of $S'$ is at least the height of $S$, we could apply Lemma~\ref{lem:shield} directly to prove that $P$ is pumpable or fragile.}
  \label{fig:spans}
\end{figure}
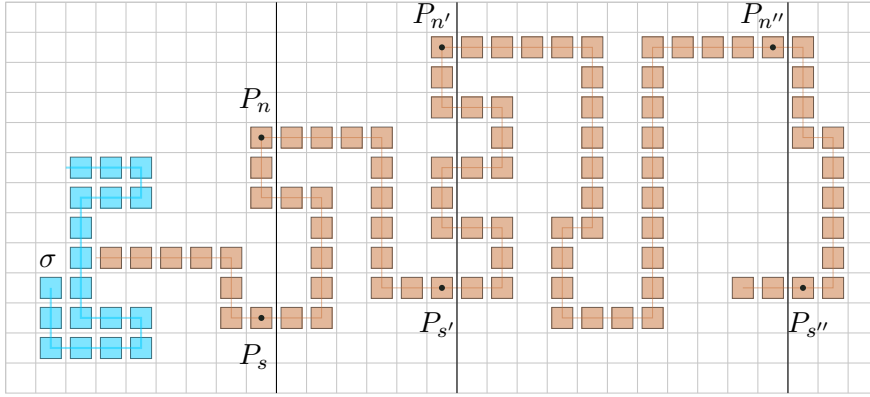

\item
Finally,  the proof of Theorem~\ref{thm:main} uses a combinatorial argument, showing that if the path is long enough, there are enough spans that we can always find two spans of the same type and orientation, and of increasing height.
\end{enumerate}

\subsection{First step: Putting $P$ into a canonical form}
We first restate Theorem~\ref{thm:intro main thm}, the short proof of which merely contains some initial setup to put $P$ in a canonical form and finishes by invoking Theorem~\ref{thm:main}, where the main heavy lifting happens. 

\begin{reptheorem}{thm:intro main thm}
  \introtheorem
\end{reptheorem}

\begin{proof}
Since the vertical height or horizontal width of $P$ is strictly greater than $|\sigma|$ (the number of tiles in $\sigma$) for all $|T| \geq 1$, we know that $P$ extends beyond the bounding box of $\sigma$.

Let $B$ be a square in $\mathbb{Z}^2$ (its four sides are parallel to the x and y axes) whose center is at the origin, and where each side is distance  $(4|T|)^{(4|T|+1)}(4|\sigma|+6)+|\sigma|$ from the origin.
Without loss of generality, we translate $P \cup \sigma$ so that $\pos{P_0}=(0,0) \in\mathbb{Z}^2$. Then we let $b\in \{0,1,\ldots, |P|-1 \}$ be the smallest index  so that tile $P_b$ is positioned on $B$. Note that $b$ is correctly defined since either the width or height of $P$ (which is at least $(8|T|)^{(4|T|+1)}(5|\sigma|+6)$) is greater than the length of the side of $B$ (which is less than $(8|T|)^{(4|T|+1)}(4|\sigma|+6)+2|\sigma|+1$)).
From now on we redefine $P$ to be its own prefix up to $P_b$, i.e.,  $P \defeq P_{0,1,\ldots,b}$ (since proving the theorem conclusion on any prefix of $P$ implies that conclusion also holds for $P$).

Without loss of generality, we assume that the last tile $P_b = P_{|P|-1}$ of $P$ is the unique {\em easternmost} tile of $P\cup \sigma$ (if it is not, we simply rotate all tiles of $T$, and the seed assembly $\sigma$, to make it so -- thus redefining $P$ again). Note that, $\pos{P_0}$ is still $(0,0)$ after this rotation.  
Uniqueness follows from the fact that $\sigma$ does not reach to $B$ for all $|T|\geq 1$, and that we assumed that $P$ places exactly one tile on $B$.

Finally, we also assume, without loss of generality, that the {\em westernmost} tile of $\sigma\cup \asm P$ is at x-coordinate $0$ and that the {\em southernmost} tile of $\sigma\cup \asm P$ is at y-coordinate $0$ (we translate $P \cup \sigma$ if this is not the case -- thus redefining $P$ again). Since we have $\pos{P_0}=(0,0)$ before the translation, then $\sigma \cup P$ is translated to the east. Moreover, since $P_0$ and the seed interacted then the last tile of $P$ is at least $(4|T|)^{(4|T|+1)}(4|\sigma|+6)$ to the east of the easternmost tile of $\sigma$.

$P$ now satisfies the hypotheses of  Theorem~\ref{thm:main}, thus $P$ is pumpable or fragile.
\end{proof}

\subsection{Second step: $P$ has many visible glues pointing east, or else is pumpable or fragile}

For the remainder of this section,
we define $X\in\Z$ to be the x-coordinate of the easternmost glue column, i.e. $X$ is the largest integer such that $P$ has a glue with x-coordinate $X+0.5$.
Recall (Section~\ref{sec:defs-paths}) that for a tile $P_i$, its x-coordinate is denoted $\xcoord {P_i}$.

We will need the following lemma, which uses a similar proof technique as Lemma~\ref{lem:glue:east}:

\begin{lemma}
\label{lem:glue:side}
Let $P$ be a path producible by some tile assembly system $\mathcal T=(T,\sigma,1)$ such that the last tile of $P$ is the unique easternmost tile of $\sigma \cup \asm P$. Either (a) all glues visible from the north relative to $P$ point to the east or (b) all glues visible from the south relative to $P$ point to the east.
\end{lemma}
\begin{proof}
  Suppose, for the sake of contradiction, that for some integer $i$, $\glu P i$ is visible from the south relative to $P$ and points west, and for some integer $j$, $\glu P j$ is visible from the north relative to $P$ and points west. 

  Assume without loss of generality that $i<j$, and let $l^i$ and $l^j$ be the respective visibility rays of $\glu P i$ and $\glu P j$.
  We define a curve $c$ as
  $$c = \concat{\reverse{l^i}, \gs{l^i(0)}{P_{i+1}}, \embed{P_{\range {i+1}{i+2}j}}, \gs{P_j}{l^j(0)}, l^j}$$
  By Theorem~\ref{thm:infinite-jordan}, $c$ cuts $\R^2$ into two connected components. Observe that the last tile of $P$ is on the right-hand side of $c$, yet $P_{j+1}$ is on the left-hand side of $c$.
  Therefore, $\embed{ P_{\range{j+1}{j+2}{|P|-1}} }$ must intersect $c$. However:
  \begin{itemize}
  \item That intersection is not be on $\reverse{l^i}$ or on $l^j$, by visibility of $\glu P i$ and $\glu P j$.
  \item That intersection is not be on $\embed{ P_{\range{i+1}{i+2}j} }$, because $P$ is simple.
  \item That intersection is not be on $\gs{l^i(0)}{P_{i+1}}$ nor $\gs{P_j}{l^j(0)}$,  since $P$ can only intersect these segments at the endpoints of these segments, and we have already handled this case.
  \end{itemize}
  Therefore, we get a contradiction, hence either all glues visible from the north relative to $P$ point to the east, or all glues visible from the south relative to $P$ point to the east.
\end{proof}

\begin{lemma}
\label{lem:glue:almost}
Let $P$ be a path producible by some tile assembly system $\mathcal T=(T,\sigma,1)$ such that the last tile of $P$ is the unique easternmost tile of $P \cup \sigma$. If more than $|T|+1$ glues of $P$ are visible from the south and are pointing to the west then $P$ is fragile or pumpable.
\end{lemma}
\begin{proof}
  We claim that the hypotheses of Definition~\ref{def:shield} are satisfied, and establishing that claim will  in turn allow us to apply Lemma~\ref{lem:shield}.
  Since there are $|T|+1$ glues visible from the south relative to $P$ and pointing west, then by the pigeonhole principle, there are two indices $i < j$ such that $\glu P i$ and $\glu P j$ are visible from the south, pointing west and of the same glue type.
  Since the last tile $P_{|P|-1}$ of $P$ is the unique easternmost tile of $P\cup \sigma$, then for $k=|P|-2$, $\glu P k$ is visible from the north relative to $P$. Hence, let $l^k$ be the visibility ray of $\glu P k$ (this means that $l^k$ is the vertical ray from $\glu P k$ to the north).
    By Lemma~\ref{lem:glue:east}, we have $x_{P_i} > x_{P_j}$. Therefore, $l^k+\vect{P_jP_i}$ is strictly to the east of $l_k$, and hence does not intersect $P$.

  We claim that the three hypotheses of Definition~\ref{def:shield} are satisfied by the horizontal mirror flip of $\mathcal{T}$ and of $P$:  
  In the mirror-flip of $P$ we (still) have $0 \leq i<j<k < |P|-1$, also 
   $\glu P i$ and $\glu P j$ are visible to the south (with respect to the mirror-flip of $P\cup \sigma$), point to the east, and are of the same type (Hypothesis~\ref{lem:hp:shield ij});  
  $\glu P k$ is visible to the north (Hypothesis~\ref{lem:hp:shield k});
  and 
  $\embed{P_{i,i+1,\ldots,k}} \cap ( l^k+\vect{P_jP_i}) = \emptyset$
(Hypothesis~\ref{lem:hp:shield backup}). 
Thus for the mirror-flip of $P$ and of $\mathcal{T}$, the mirror-flip of $P$ is pumpable or fragile, hence $P$ is  pumpable or fragile as claimed.
    \end{proof}

\subsection{Third step: repeated spans imply $P$ is  pumpable or fragile}
We need two key definitions that are illustrated in Figure~\ref{fig:spans}.
  \begin{definition}[span]\label{def:span}
    Let $P$ be a path producible by some tile assembly system $\mathcal T=(T,\sigma,1)$, whose last glue is the unique easternmost (\resp highest) glue of $P\cup\sigma$.
    A \emph{span} (\resp \emph{horizontal span}) of $P$ is a pair $(s,n) \in \{0,1,\ldots,|P|-2\}^2$ of indices of tiles of $P$ such that, at the same time:
    \begin{itemize}
    \item $\xcoord{\glu P n}=\xcoord{\glu P s}$ (\resp $\ycoord{\glu P n} = \ycoord{\glu P s}$), and
    \item $\glu{P}{n}$ is visible from the north (\resp from the east), and
    \item $\glu{P}{s}$ is visible from the south (\resp from the west).
    \end{itemize}
\end{definition}

Remark on Definition~\ref{def:span}:   Notice that $P$ has exactly one span on each glue column that does not have seed glues\footnote{Seed glues are glues that are located on tiles belonging to the seed $\sigma$. Thus, in particular, the $\geq 1$ glue(s) on $\sigma$ that are abutting to tile $P_0$ of $P$ are seed glues.} and on which $P$ has at least one glue. Hence in the rest of the proof we will use terms such as ``the span of $P$ on glue column $x$''.

The following properties of spans will be useful.

\begin{definition}[span properties: orientation,    pointing direction, type and height]\label{def:span properties}
    The \emph{orientation} of a span (\resp of a horizontal span) $(s, n)$ is  \emph{up} (\resp \emph{right}) if  $s\leq n$\footnote{In the case where $s=n$, $P$ has only one glue on the glue column of $\glu P s$, and $(s, n)$ is an up span.} or \emph{down} (\resp \emph{left}) if $s>n$. 
    
    A span (\resp horizontal span) $(s,n)$ \emph{points east} (\resp \emph{points north}) if both $\glu P s$ and $\glu P n$ point east (\resp point north), and that $(s,n)$ \emph{points west} (\resp \emph{points south}) if both $\glu P s$ and $\glu P n$ point west (\resp point south)\footnote{If $\glu P s$ and $\glu P n$ point to different direction, the span has no direction.}.

    The \emph{type} of a span or horizontal span $(s, n)$ is the type of its first glue in the growth order of $P$, i.e. the type of $\glu P s$ if $s\leq n$ and the type of $\glu P n$ otherwise. 
    
    The \emph{height} (\resp width) of a span (\resp horizontal span) $(s, n)$ is $\ycoord{P_n}-\ycoord{P_s}$ (\resp $\xcoord{P_n}-\xcoord{P_s}$).
  \end{definition}

\begin{lemma}
\label{lem:couple:almost}
Let $P$ be a path producible by some tile assembly system $\mathcal T=(T,\sigma,1)$ such that the last tile of $P$ is the unique easternmost tile of $\sigma  \cup \asm P$. Either $P$ is pumpable or fragile or else there is a glue column $x_0 \leq |T|+ |\sigma|$ such that for all $x\in\{x_0,x_0+1,\ldots,X\}$, the span of $P$ on glue column $x$ is pointing east.
\end{lemma}
\begin{proof}
  First, by Lemma~\ref{lem:glue:side}, we can assume without loss of generality (by vertical mirror-flip) that all glues visible from the north relative to $P$ are pointing east.

  Then, since $P_{|P|-1}$ is the unique easternmost tile of $P$, $\glue{P}{|P|-2}{|P|-1}$ is pointing east and visible both from the south and from the north. Therefore, the span of $P$ on glue column $X$ is pointing east. Let $x_0\in\Z$ be the westernmost glue column on which the glue visible from the south relative to $P$ is pointing east.
  By Lemma~\ref{lem:glue:east}, for all $x\in\{x_0,x_0+1,\ldots,X\}$, the glue visible from the south relative to $P$ on column $x$ is also pointing east.
  If $x_0>|T|+|\sigma|$ then there are at least $|T|+|\sigma|+1$ columns on which either $P$ has no visible glue (meaning that $\sigma$ has glues on these columns), or on which the glue visible from the south relative to $P$ is pointing west. Since the size of the seed is $|\sigma|$, this means that at least $|T|+1$ glues visible from the south relative to $P$ are pointing west, and we can conclude using Lemma~\ref{lem:glue:almost} that $P$ is pumpable or fragile.

  In all other cases, the span of $P$ on glue column $x$ is pointing east.
\end{proof}

In the final theorem, we need the following corollary of Lemma~\ref{lem:couple:almost}, whose statement is about horizontal spans instead of spans.
The corollary statement is simply a rotation of the statement of Lemma~\ref{lem:couple:almost} along with Y substituted for X in the latter, hence the proof is immediate.

\begin{corollary}
\label{cor:couple:almost}
Let $P$ be a path producible by some tile assembly system $\mathcal T=(T,\sigma,1)$ such that the last tile of $P$ is the unique northernmost (\resp southernmost) tile of $\sigma \cup \asm P$, let $Y$ be the largest integer such that $P$ or $\sigma$ has a glue on glue row $Y+0.5$ and let $y$ be the smallest integer such that $P$ or $\sigma$ has a glue on glue row $y+0.5$. Either $P$ is pumpable or fragile or else there is a glue row $y_0 \leq y+ |T|+ |\sigma|$ (\resp $y_0\geq Y-|T|-|\sigma|$) such that for all $z\in\{\rng {y_0} Y \}$
(respectively $z\in\{\rng y {y_0}\}$), the horizontal span of $P$ on glue row $z$ is pointing north (\resp south).
\end{corollary}

The following simple but powerful lemma shows how to exploit spans in order to apply Lemma~\ref{lem:shield}:

\begin{lemma}
\label{lem:couple:use}
Let $P$ be a path producible by some tile assembly system $\mathcal T=(T,\sigma,1)$ such that the last tile of $P$ is the unique easternmost tile of $P \cup \sigma$.
If there are two spans $S = (s, n)$ and $S' = (s',n')$ on $P$, both pointing east, of the same orientation and type, and such that $\xcoord{P_s}<\xcoord{P_{s'}}$ and the height of $S'$ is at least the height of $S$, then $P$ is pumpable or fragile.
\end{lemma}
\begin{proof}
  Without loss of generality, since $S$ and $S'$ have the same orientation, we assume that $S$ and $S'$ both have the ``up'' orientation (i.e. $s\leq n$ and $s'\leq n'$).
  We claim that we can apply Lemma~\ref{lem:shield} by setting $i = s$, $j = s'$ and $k = n'$.

  Indeed, both $\glu P s$ and $\glu P {s'}$ are visible from the south and pointing east. Since $S$ and $S'$ are of the same type, $\glu P s$ and $\glu P {s'}$ are also of the same type.
  Moreover, since $\xcoord{P_s}<\xcoord{P_{s'}}$, applying Lemma~\ref{lem:glue:east} shows that $i < j$.
  By definition of a span, $\glu P {n'}$ is visible from the north, and since $S'$ has the up orientation, $s'<n'$, hence $j\leq k$.

  Finally, let $l^{n'}$ be the visibility ray from $\glu P {n'}$ to the north. Note that $\pos{\glu P {s'}}+\vect{P_{s'}P_s} = \pos{\glu P s}$, and since the height of $S$ is smaller than or equal to the height of $S'$, $\embed{P}$ may only intersect $l^{n'}+\vect{P_{s'}P_{s}}$ at $l^{n'}(0)+\vect{P_{s'}P_s}$, which on the same column, and at least as high as $\glu P n$.\footnote{$S$ and $S'$ have the same height if and only if $l^{n'}\!(0)+\vect{P_{s'}P_s} = \pos{\glu P n}$}

  We can therefore conclude by applying Lemma~\ref{lem:shield} which shows that $P$ is pumpable or fragile.
\end{proof}

\subsection{Fourth step: proof of the main theorem}
\begin{theorem}
  \label{thm:main}
  Let $P$ be a path producible by some tile assembly system $\mathcal T=(T,\sigma,1)$ where the final tile $P_{|P|-1}$ is the easternmost one and it is  positioned at distance at least $(4|T|)^{(4|T|+1)}(4|\sigma|+6)$ to the east of the easternmost tile of $\sigma$. Then $P$ is pumpable or fragile.
\end{theorem}

\begin{proof}
  We first apply Lemma~\ref{lem:couple:almost}, yielding one of two possible conclusions:
  (1) either $P$ is pumpable or fragile, in which case we are done, or (2) there is an integer $x_0 \leq |T|+|\sigma|$ such that for all $x\in\{x_0,x_0+1,\ldots,X\}$, the span of $P$ on glue column $x$ is pointing east.

  Let $x_0<x_1<\ldots<x_n$ be the sequence of glue columns such that, both of the following hold:
  \begin{itemize}
    \item for all $a\neq b$, the spans of $P$ on glue columns $x_a$ and $x_b$ have different orientations or different types (or both); and 
  \item for all $x\in\{x_0,x_0+1,\ldots,X\}$, there is an integer $a \in \{0,1,\ldots,n\}$ such that $x_a\leq x$ and the spans of $P$ on glue columns $x_a$ and $x$ have the same orientation and the same type.
  \end{itemize}

  In other words, the sequence $x_0,x_1,\ldots,x_n$ is the sequence of the first occurrences, in increasing x-coordinate order, of each span type and orientation. That sequence has at least one element ($x_0$), and
  since that sequence contains at most one occurrence of each span type and orientation,  $n< 2|T|$.\footnote{There are 2 orientations, and $\leq |T|$ glue types that appear as the east glue of some tile in $T$, and the indexing up to $n$ begins at 0.}
  Moreover, as a consequence of our use of Lemma~\ref{lem:couple:almost} in the beginning of this proof, all the spans of $P$ on glue columns $x_0, x_1,\ldots,x_n$ point east.

  We extend the sequence by one point, by letting $x_{n+1} = X$ (i.e. we add the easternmost possible span), and for all $i \in\{0,1,\ldots,n\}$, we let $h_i$ be the height of the span of $P$ on glue column $x_i$.

  Intuitively, we consider a ``cone'' from $\sigma$ to the east (see Figure~\ref{fig:cone}), and consider two cases, depending on whether all spans are within the cone.
  The only difference with an actual geometric cone is that the y-coordinate of spans might not be aligned.

  \begin{figure}[ht]
    \begin{center}
      \begin{tikzpicture}[scale=\scale]\draw[draw={rgb,255:red,0; green,0; blue,0},fill={rgb,255:red,242; green,242; blue,242}](29,5.3)--(6,-4.7)--(29,-8.7);
\draw(4.5, -5.12) node[anchor=south west] {$\sigma$};
\draw[draw={rgb,255:red,0; green,0; blue,0},<->](15.5,-0.56)--(15.5,-6.35);
\draw[draw={rgb,255:red,0; green,0; blue,0},<->](13.62,-1.41)--(13.62,-6.03);
\draw[draw={rgb,255:red,0; green,0; blue,0},<->](19.2,1.02)--(19.2,-6.93);
\draw[draw={rgb,255:red,0; green,0; blue,0},<->](21.32,1.92)--(21.32,-7.3);
\draw(12.96, -7.42) node[anchor=south west] {$x_1$};
\draw(13.39, -3.77) node[anchor=south west] {$h_1$};
\draw[draw={rgb,255:red,212; green,85; blue,0},->](6,-4.7) .. controls (6,-4.7) and (8.76,-6.31) .. (11.43,-5.65) .. controls (12.95,-5.28) and (9.52,-1.95) .. (11.57,-2.29) .. controls (12.11,-2.37) and (12.96,-5.65) .. (13.53,-5.97) .. controls (14.11,-6.29) and (12.95,-1.03) .. (13.58,-1.41) .. controls (14.2,-1.8) and (14.82,-5.94) .. (15.46,-6.35) .. controls (16.09,-6.76) and (14.52,-1) .. (15.45,-0.57) .. controls (16.11,-0.27) and (17.59,-7.95) .. (19.18,-7) .. controls (21.36,-5.7) and (18.77,1.21) .. (19.22,0.99) .. controls (19.66,0.77) and (21.01,-7.23) .. (21.32,-7.3) .. controls (21.68,-7.29) and (19.49,2.24) .. (21.34,1.93) .. controls (22.26,1.78) and (22.3,-3.74) .. (23.97,-6.12) .. controls (25.64,-8.5) and (28.95,-7.75) .. (28.8,-6.98) .. controls (27.72,-5.6) and (23.76,-6.39) .. (24.55,-3.76) .. controls (25.54,-3.03) and (29.02,-2.08) .. (26.7,-0.74) .. controls (24.42,-0.6) and (23.43,-0.35) .. (24.5,1.8) .. controls (25.21,3.38) and (26.02,3.28) .. (28,3.3);
\draw[draw={rgb,255:red,0; green,0; blue,0},<->](11.5,-2.3)--(11.5,-5.66);
\draw(10.72, -6.92) node[anchor=south west] {$x_0$};
\draw(11.24, -4.27) node[anchor=south west] {$h_0$};
\end{tikzpicture}
    \end{center}
    \caption{$P$ is drawn schematically in brown. This proof has two cases: if the height of the spans are increasing no more than proportionally to their x-coordinate, which we represent by a cone here (even though the spans are not actually in a cone, since the spans might not be vertically aligned), we are in Case~\ref{main2}. Else, if we can find at least one $c$ such that $h_c$ gets out of the cone, i.e. such that $h_c > 4|T|^2(x_c+4|\sigma|+6)$, we are in Case~\ref{main1}.}
    \label{fig:cone}
  \end{figure}
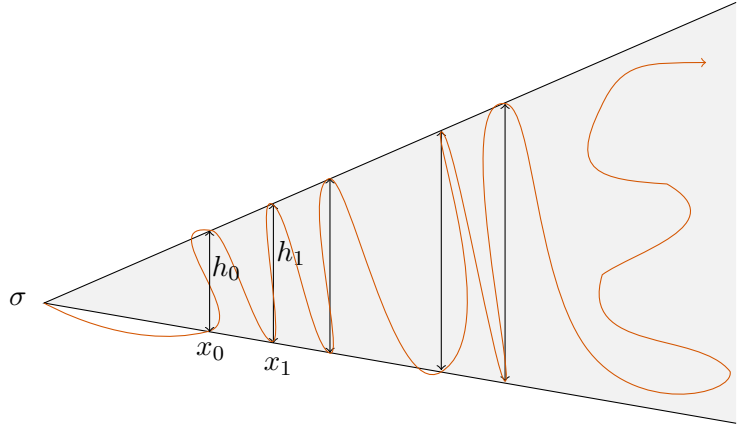

  There are two cases:
  \begin{enumerate}
  \item\label{main1} If there is an integer $c \in\{0,1,\ldots,n+1\}$ such that $h_c > 4|T|^2(x_c+4|\sigma|+6)$, then we claim that $P$ is pumpable or fragile.

    Indeed, without loss of generality, we assume that the span $(s_c,n_c)$ of $P$ on glue column $x_c$ has the ``up'' orientation (i.e. $s_c \leq n_c$). Let $X'$ be the easternmost glue column reached by $Q = P_{0,1,\ldots,n_c+1}$.
    See Figure~\ref{fig:main-cone1}.
    There are two subcases, depending on the relative sizes of $h_c$ and $X'$:

    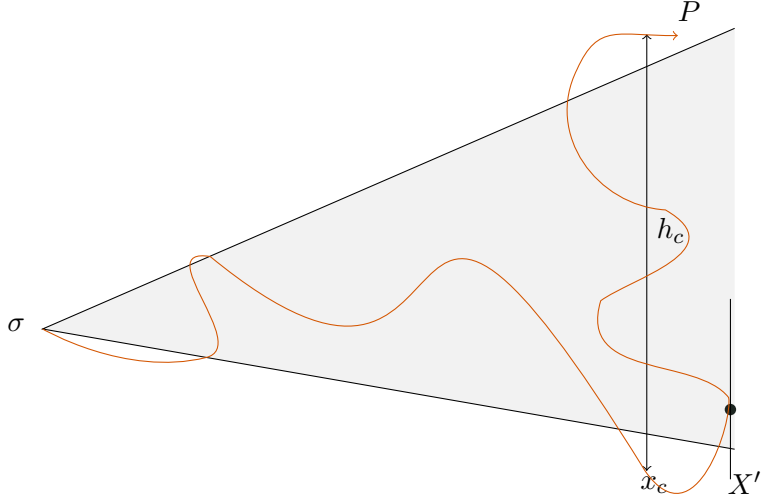
\begin{figure}[ht]
      \begin{center}
        \begin{tikzpicture}[scale=\scale]\draw[draw={rgb,255:red,0; green,0; blue,0},fill={rgb,255:red,242; green,242; blue,242}](29,5.3)--(6,-4.7)--(29,-8.7);
\draw(4.5, -5.12) node[anchor=south west] {$\sigma$};
\draw[draw={rgb,255:red,0; green,0; blue,0},<->](26.08,5.1)--(26.08,-9.42);
\draw(25.53, -10.49) node[anchor=south west] {$x_c$};
\draw(26.08, -2.16) node[anchor=south west] {$h_c$};
\draw(26.77, 5.19) node[anchor=south west] {$P$};
\draw[draw=none,fill={rgb,255:red,28; green,36; blue,31},thin](28.86, -7.38) ellipse (0.19cm and 0.19cm);\draw(28.44, -10.52) node[anchor=south west] {$X'$};
\draw[draw={rgb,255:red,0; green,0; blue,0}](28.86,-3.7)--(28.86,-9.7);
\draw[draw={rgb,255:red,212; green,85; blue,0},->](6,-4.7) .. controls (6,-4.7) and (8.76,-6.31) .. (11.43,-5.65) .. controls (12.95,-5.28) and (9.52,-1.95) .. (11.57,-2.29) .. controls (21.29,-10.28) and (16.31,6.27) .. (26.03,-9.43) .. controls (27.7,-11.81) and (28.95,-7.75) .. (28.8,-6.98) .. controls (27.72,-5.6) and (23.76,-6.39) .. (24.55,-3.76) .. controls (25.54,-3.03) and (29.02,-2.08) .. (26.7,-0.74) .. controls (24.42,-0.6) and (22.68,1.77) .. (23.76,3.93) .. controls (24.47,5.51) and (25.12,5.04) .. (27.11,5.06);
\end{tikzpicture}
      \end{center}
      \caption{$P$ is drawn schematically in brown. In Case~\ref{main1}, at least one span, at column $x_c$, gets out of the cone, i.e. is of height $h_c>4|T|^2(x_c+4|\sigma|+6)$. We define $X'$, the rightmost column reached by $P_{\range 0 1 {n_c}}$, and split this case into two Cases, ~\ref{main1.1} and~\ref{main1.2}, depending on the relative values of $h_c$ and $X'$.}
      \label{fig:main-cone1}
    \end{figure}

    \begin{enumerate}
    \item\label{main1.1} If $h_c< 4|T|(X'+2)+3|\sigma|+1$, we first claim that $X'>|T|(x_c+3)+|\sigma|+1$. See Figure~\ref{fig:main-cone1a}.

      \begin{figure}[ht]
        \begin{center}
          \begin{tikzpicture}[scale=\scale]\draw[draw={rgb,255:red,0; green,0; blue,0},fill={rgb,255:red,242; green,242; blue,242}](29,5.3)--(6,-4.7)--(29,-8.7);
\draw(4.5, -5.12) node[anchor=south west] {$\sigma$};
\draw[draw={rgb,255:red,0; green,0; blue,0},<->](26.08,5.1)--(26.08,-9.42);
\draw(25.53, -10.49) node[anchor=south west] {$x_c$};
\draw(26.08, -2.16) node[anchor=south west] {$h_c$};
\draw[draw={rgb,255:red,212; green,85; blue,0},->,thin](6,-4.7) .. controls (6.13,-4.34) and (10.17,-4.06) .. (11.61,-5.66) .. controls (11.79,-5.86) and (10.04,-2.23) .. (11.45,-2.27) .. controls (22.47,-8.97) and (14.96,4.37) .. (26.03,-9.41) .. controls (28.54,-8.62) and (37.24,-11.02) .. (38.38,-7.38) .. controls (38.39,-6.93) and (36.6,6.9) .. (35.24,7.12) .. controls (33.88,7.34) and (33.38,-8.38) .. (31.6,-8.21) .. controls (29.82,-8.03) and (27.64,-5.5) .. (26.38,-5.18) .. controls (25.12,-4.86) and (24.35,-4.42) .. (24.55,-3.76) .. controls (25.54,-3.03) and (29.02,-2.08) .. (26.7,-0.74) .. controls (24.42,-0.6) and (22.68,1.77) .. (23.76,3.93) .. controls (24.47,5.51) and (25.12,5.04) .. (27.11,5.06);
\draw(25.6, 5.19) node[anchor=south west] {$P_{n_c}$};
\draw[draw=none,fill={rgb,255:red,28; green,36; blue,31},thin](38.38, -7.38) ellipse (0.19cm and 0.19cm);\draw(37.97, -10.52) node[anchor=south west] {$X'$};
\draw[draw={rgb,255:red,0; green,0; blue,0}](38.39,-3.7)--(38.39,-11.29);
\draw[draw={rgb,255:red,0; green,0; blue,0},<->](11.5,-11.29)--(38.39,-11.29);
\draw[draw={rgb,255:red,0; green,0; blue,0},<->](11.5,-2.3)--(11.5,-5.66);
\draw(10.72, -6.92) node[anchor=south west] {$x_0$};
\draw(11.24, -4.27) node[anchor=south west] {$h_0$};
\end{tikzpicture}
        \end{center}
        \caption{$P$ is drawn schematically in brown. In Case~\ref{main1.1}, $X'$ is large enough that we can find two glue spans between $x_0$ and $X'$ that have the same type and orientation, and are at least $x_c$ columns apart. We can therefore apply Lemma~\ref{lem:shield} by letting $i$ and $j$ be the indices of the south glues of these spans on $P$, and letting $k$ be the index of $P_{n_c}$.}
        \label{fig:main-cone1a}
      \end{figure}

      Indeed, we have $h_c>4|T|^2(x_c+4|\sigma|+6)$, and hence:
      \begin{eqnarray*}
        4|T|(X'+2)+3|\sigma|+1&>&4|T|^2(x_c+4|\sigma|+6)\\
        X'&>&\frac{4|T|^2(x_c+4|\sigma|+6) - 3|\sigma|-1}{4|T|}-2\\
        X'&>&\frac{4|T|^2(x_c+3)}{4|T|} + \frac{4|T|^2(4|\sigma|+3) - 3|\sigma|-1 - 8|T|}{4|T|}\\
        X'&>&|T|(x_c+3) + |\sigma|\frac{16|T|^2-3}{4|T|} + \frac{12|T|^2 - 1 - 8|T|}{4|T|}\\
      \end{eqnarray*}
      Finally, since $|T|\geq 1$, $\frac{16|T|^2-3}{4|T|}\geq 3$, and $\frac{12|T|^2 - 1 - 8|T|}{4|T|}\geq \frac{3|T|-2}{2} - \frac{1}{4|T|}\geq  \frac{1}{2} - \frac{1}{4}\geq 0$, and therefore (ignoring the last term, and since $|\sigma|\geq 1$): $$X'>|T|(x_c+3) + 3|\sigma| \geq |T|(x_c+3) + |\sigma| + 2$$
      thus proving the claim that $X'>|T|(x_c+3)+|\sigma|+1$.

            Let $k=n_c$ and $l^k$ be the ray from $\glu P k$ to the north.
      By Lemma~\ref{lem:glue:east}, we have $s_0\leq s_c$, and since span $(s_c,n_c)$ has the ``up'' orientation then $s_0\leq n_c$.
      Then $\glu Q {s_0}$ is visible from the south relative to $Q$, and pointing east, hence by Lemma~\ref{lem:glue:east}, all the glues visible from the south relative to $Q$ and to the east of $\glu P {s_0}$ also point east. And since $X' \geq |T|(x_c+3)+|\sigma|+1$, there are at least $X' - x_0 \geq X' - (|T|+|\sigma|) \geq |T|(x_c+2)+1$ of them.

      By the pigeonhole principle, at least $x_c+2$ of these glues share the same type. Let $i\in\N$ be the index on $Q$ of the westernmost of these glues, and let $j$ be the index on $Q$ of the easternmost one. By Lemma~\ref{lem:glue:east}, this means that $i \leq j$.

      Moreover, $\vect{P_jP_i}$ is at least $x_c+1$ columns wide (and $\vect{P_jP_i}$ is to the west), since there are at least $x_c+1$ glue columns between $\glu P i$ and $\glu P j$.
      Therefore, $l^k + \vect{P_jP_i}$ is to the west of $x_c+0.5 - (x_c+1)\leq -0.5$, and hence cannot intersect $P$ (since we assumed in the beginning of this section that the westernmost tile of $P \cup \sigma$ has x-coordinate $0$).

      We can therefore apply Lemma~\ref{lem:shield}, using $( i, j, k)$, as a shield, to $P$ thus proving that $P$ is pumpable or fragile.

    \item\label{main1.2} Else, $h_c \geq 4|T|(X'+2)+3|\sigma|+1$. See Figure~\ref{fig:main-cone1b}.

      \begin{figure}[ht]
        \begin{center}
          \begin{tikzpicture}[scale=\scale]\draw[draw={rgb,255:red,0; green,0; blue,0},fill={rgb,255:red,242; green,242; blue,242}](29,5.3)--(6,-4.7)--(29,-8.7);
\draw(4.5, -5.12) node[anchor=south west] {$\sigma$};
\draw[draw={rgb,255:red,0; green,0; blue,0},<->](26.08,5.1)--(26.08,-9.42);
\draw(25.53, -10.49) node[anchor=south west] {$x_c$};
\draw(26.08, -2.16) node[anchor=south west] {$h_c$};
\draw[draw={rgb,255:red,212; green,85; blue,0},->,thin](6,-4.7) .. controls (6.13,-4.34) and (10.17,-4.06) .. (11.61,-5.66) .. controls (11.79,-5.86) and (10.04,-2.23) .. (11.45,-2.27) .. controls (22.47,-8.97) and (14.96,4.37) .. (26.03,-9.41) .. controls (28.54,-8.62) and (29.83,-11.02) .. (30.98,-7.38) .. controls (30.98,-6.93) and (30.67,-2.97) .. (29.31,-2.75) .. controls (28.28,-9.33) and (28.03,-4) .. (26.38,-5.18) .. controls (25.12,-4.86) and (24.35,-4.42) .. (24.55,-3.76) .. controls (25.54,-3.03) and (29.02,-2.08) .. (26.7,-0.74) .. controls (24.42,-0.6) and (22.68,1.77) .. (23.76,3.93) .. controls (24.47,5.51) and (25.12,5.04) .. (27.11,5.06);
\draw(25.6, 5.19) node[anchor=south west] {$P_{n_c}$};
\draw[draw=none,fill={rgb,255:red,28; green,36; blue,31},thin](30.98, -7.38) ellipse (0.19cm and 0.19cm);\draw(30.56, -10.52) node[anchor=south west] {$X'$};
\draw[draw={rgb,255:red,0; green,0; blue,0}](30.98,-3.7)--(30.98,-11.29);
\draw[draw={rgb,255:red,0; green,0; blue,0},<->](22,4.3)--(31,4.3);
\draw[draw={rgb,255:red,0; green,0; blue,0},<->](22,1.3)--(31,1.3);
\draw[draw={rgb,255:red,0; green,0; blue,0},<->](14,-0.7)--(31,-0.7);
\draw[draw={rgb,255:red,0; green,0; blue,0},<->](10,-2.7)--(31,-2.7);
\draw[draw={rgb,255:red,0; green,0; blue,0},<->](10,-6.7)--(31,-6.7);
\draw[draw={rgb,255:red,0; green,0; blue,0},<->](25,-8.7)--(31,-8.7);
\draw(31.3, -8.75) node[anchor=south west] {$y_0$};
\draw(31.3, -6.79) node[anchor=south west] {$y_1$};
\draw(31.3, -2.87) node[anchor=south west] {$y_2$};
\end{tikzpicture}
        \end{center}
        \caption{$P$ is drawn schematically in brown. In Case~\ref{main1.2}, $h_c$ is large enough, and $X'$ is small enough that if we consider all horizontal spans of $P$ (there are at least $h_c$ of them), two of them will have the same width, and we can apply Lemma~\ref{lem:couple:use}.}
        \label{fig:main-cone1b}
      \end{figure}
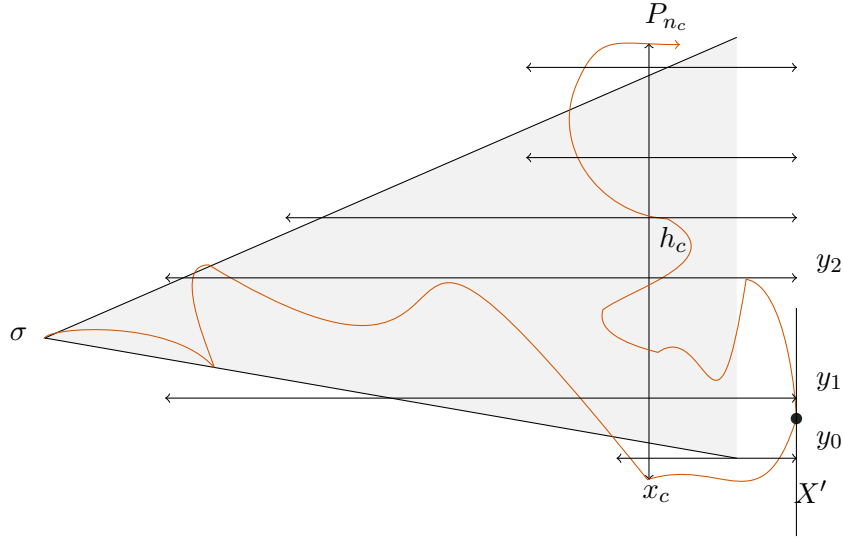
      Therefore, there is a tile of $Q$ at least $\left\lceil \frac{h_c-|\sigma|}{2} \right\rceil \geq 2|T|(X'+2)+|\sigma|+1$ to the north of all the tiles of $\sigma$, or to the south of all the tiles of $\sigma$.

      If the highest tile of $Q$ is at least $2|T|(X'+2)+|\sigma|+1$ rows to the north (\resp to the south) of all the tiles of $\sigma$, then let $Q'$ be the shortest prefix of $Q$ whose highest tile is $2|T|(X'+2)+|\sigma|+1$ rows to the north (\resp to the south) of all the tiles of $\sigma$.

                                                      By Corollary~\ref{cor:couple:almost}, either $Q$ is pumpable or fragile or at least $2|T|(X'+1)$ horizontal spans of $Q$ point north (\resp south), and by the pigeonhole principle, at least $X'+1$ of these spans have the same orientation and type.
      Since the width of all these horizontal spans is bounded by $X'$ (because $X'$ is by definition the easternmost column reached by $Q$), at least two of them have the same width, and hence by Lemma~\ref{lem:couple:use}, $Q$ is fragile or pumpable, and so is $P$ (because $Q$ is a prefix of $P$).

    \end{enumerate}
  \item \label{main2}Otherwise, $h_c\leq 4|T|^2(x_i+4|\sigma|+6)$ for all $c\in\{0,1,\ldots,n\}$, i.e. intuitively, just like in Figure~\ref{fig:cone}, all the spans are inside the ``cone''.

    We first claim that at least one of the $x_c$ (for $c\in\{0,1,\ldots,n+1\}$) must be ``large enough'', specifically we claim there is an $x_c > x_0+\sum_{d=0}^{c-1} h_d$.
    Indeed, suppose, for the sake of contradiction, that for all $c\in\{0,1,\ldots,n+1\}$, $x_c \leq x_0+\sum_{d=0}^{c-1} h_d$.
    We prove that the following inequality holds by induction on $c\in\{0,1,\ldots,n+1\}$:
    $$x_0+\sum_{d=0}^{c-1} h_d \leq (4|T|)^{2c+1}(4|\sigma|+6)-2$$
    First, the induction hypothesis holds for $c = 0$, since $x_0+\sum_{d=0}^{c-1}h_d = x_0\leq 4|T|(4|\sigma|+6)-2$.
    Then, if the induction hypothesis holds for any $c\in\{0,1,\ldots,n\}$, i.e. if $x_c \leq x_0+\sum_{d=0}^{c-1} h_d \leq (4|T|)^{2c+1}(4|\sigma|+6)-2$, then:
    \begin{eqnarray*}
      h_c&\leq&4|T|^2(x_c+4|\sigma|+6)\\
      x_0+\sum_{d=0}^{c-1} h_d + h_c&\leq&(4|T|)^{2c+1}(4|\sigma|+6) -2 + 4|T|^2(x_c+4|\sigma|+6)\\
         &\leq&(4|T|)^{2c+1}(4|\sigma|+6) + 4|T|^2x_c+ 4|T|^2(4|\sigma|+6) -2 \\
         &\leq&(4|T|)^{2c+1}(4|\sigma|+6) + 4|T|^2(4|T|)^{2c+1}(4|\sigma|+6)+ 4|T|^2(4|\sigma|+6) -2\\
         &\leq&((4|T|)^{2c+1} + 4|T|^2(4|T|)^{2c+1} + 4|T|^2)(4|\sigma|+6)-2\\
         &\leq&(4|T|)^{2c+1}\left(1 + 4|T|^2 + \frac{4|T|^2}{(4|T|)^{2c+1}}\right)(4|\sigma|+6)-2\\
         &\leq&(4|T|)^{2c+1}(4|T|^2+|T|+1)(4|\sigma|+6)-2\\
         &\leq&(4|T|)^{2c+1}(16|T|^2)(4|\sigma|+6) -2\\
         &\leq&(4|T|)^{2c+3}(4|\sigma|+6) - 2
    \end{eqnarray*}
    This proves the induction hypothesis for $c+1$, and hence by induction, the induction hypothesis holds for all $c\in\{0,1,\ldots,n+1\}$.

    However, $x_{n+1}=X\geq (4|T|)^{4|T|+1}(4|\sigma|+6)-1 > x_0+\sum_{d=0}^{n} h_d$, which is a contradiction. Therefore,
    there is at least one $c\in\{0,1,\ldots,n\}$ such that $x_c > x_0+\sum_{d=0}^{c-1} h_d$. 
    
    Let $c_0$ be the smallest such $c$.
    There is at least one $d<c_0$ such that the type and orientation of the span of $P$ on $x_d$ is repeated strictly more than $h_d$ times between $x_d$ (included) and $x_{c_0}$ (excluded).

    There are two cases:
    \begin{itemize}
    \item There is a glue column $x$, such that $x_d < x < c_0$ and the spans of $P$ on glue columns $x_d$ and  $x$ have the same type and orientation, and furthermore the height of the span of $P$ on $x$ is at least $h_d$. Since by definition, the height of the span on glue column $x_d$ is $h_d$,  we conclude immediately setting $S$ and $S'$ to be the spans on $x_d$ and $x$, respectively, in Lemma~\ref{lem:couple:use}. Thus  $P$ is pumpable or fragile.
    \item Otherwise, at least $h_d+1$ spans of the same type and orientation as the span of $P$ on $x_d$, have height at most $h_d$. By the pigeonhole principle, two of them have the same height and we can also apply Lemma~\ref{lem:couple:use} setting  $S$ to be the more western of the two spans, and $S'$ being the more eastern. Thus  $P$ pumpable or fragile.
    \end{itemize}
  \end{enumerate}
\end{proof}

\appendix

\section{The 2HAM is pumpable or blockable: proof sketch}\label{sec:2HAM}
This section contains a proof sketch of Corollary~\ref{cor:2ham}: an application of our main result to the two-handed, or hierarchical, tile assembly model (2HAM). 
A full proof requires using formal 2HAM definitions, and making some notions below more precise. 
We restate the Corollary here, see~\cite{AGKS05g,Versus,SFTSAFT} for 2HAM definitions.  

\begin{repcorollary}{cor:2ham}
  \introtheoremTwoHAM
\end{repcorollary}

\noindent {\em Proof sketch of Corollary~\ref{cor:2ham}}:  Let $\mathcal{H} = (T,1)$ be a  temperature 1 2HAM system with tile set~$T$ that grows a (tile) path $P$ of horizontal or vertical width $\boundHam$. 
Then redefine $P$ to be a shortest segment of $P$ that has horizontal or vertical width exactly $\boundHam$ (it is always possible to remove a prefix of $P$ and a suffix of $P$ so that this is the case). 
If necessary, redefine $\mathcal{H}$ by rotating all tiles $T$ by $90^\circ$ so that $P$ has horizontal width $\boundHam$.
We define an aTAM system $\mathcal{T} = (T,\sigma,1)$ by setting $T=H$ and by choosing the westernmost (smallest x-coordinate) tile $t$ of $P$, setting $\sigma = t$. 
Since~$P$ is a producible in $\mathcal{H}$, and since $\sigma$ is a tile of~$P$, $P$ is also producible by $\mathcal{T}$ (note that $\asm{P} =  \sigma \cup \asm{P}$, and $P$ is simple, and all adjacent glues along~$P$ match). 
Since $\bound \leq \boundHam$ for $|\sigma|=1$, we apply Theorem~\ref{thm:intro main thm} to $\mathcal{T}$, and $P$,
 and thus $P$ is pumpable or fragile in $\mathcal{T}$.
Specifically, this means that there is a producible assembly $\alpha \in \prodasm{\mathcal{T}}$ such that 
either 
(i)
$\alpha$ contains exactly  
$\sigma=t$, a prefix $P_{0,1,\ldots,i}$ of $P$, and then infinitely many repetitions of a segment $P_{i+1,i+2,\ldots,j}$, each translated by a unique vector $d \cdot \vect{P_i P_j}$ for all $d\in\mathbb{N}$, or 
(ii) 
$\alpha$ is an assembly that places a tile that conflicts with $P$ (prevents $P$ from growing). 
Since assemblies  that are producible in the aTAM (i.e.\ in $\mathcal{T}$) are producible in the 2HAM (i.e.\ in $\mathcal{H}$), we get that $\alpha$ is producible in~$\mathcal{H}$. Hence $P$ is pumpable or fragile in~$\mathcal{H}$.

\section{Wee lemmas and left/right turns for curves}\label{sec:app:wee lemmas}

\begin{lemma}[Lemma 6.3 of~\cite{OneTile}]
  \label{lem:precious}
  Consider a two-dimensional, bounded, connected, regular closed set $S$, i.e. $S$ is equal to the topological closure of its interior points. Suppose $S$ is translated by a vector $v$ to obtain shape $S_v$, such that $S$ and $S_v$ do not overlap. Then the shape $S_{c*v}$ obtained by translating $S$ by $c*v$ for any integer $c\neq 0$ also does not overlap $S$.
\end{lemma}

The following lemma~\cite{STOC2017} formalises the intuition behind Definition~\ref{def:pumpingPbetweeniandj}:
\begin{lemma}[Lemma 2.5 of \cite{STOC2017}]
  \label{lem:torture}
  Let $P$ be a path with tiles from some tileset $T$, $i<j$ be two integers, and $\olq$ be the pumping of $P$ between $i$ and $j$.
  Then for all integers $k\geq i$, $\olq_{k+(j-i)} = \olq_k + \vect{P_iP_j}$.
\end{lemma}
\begin{proof}
  By the definition of $\olq$ (and using the fact that $(j-i) \mod (j-i)=0$):
  \begin{eqnarray*}
    \olq_{k + (j-i)} &=& \torture P {k+(j-i)} i j\\
        &=& \torture P k i j + \vect{P_iP_j}\\
        &=& \olq_{k}+ \vect{P_iP_j}
      \end{eqnarray*}
\end{proof}

\subsection{Jordan curve theorem for infinite polygonal curves and left/right turns}\label{app:sec:inf JCT and turns}
The line $\ell:(-\infty,+\infty)\rightarrow \mathbb{R}^2$ of vector $\vect{w} = (u,v)$ passing through a point $(a,b) \in \mathbb{R}^2$ is defined as $\{(x,y) \mid -v(x-a)+u(y-b)=0\}$. This line cuts the 2D plane $\mathbb{R}^2$ into two connected components: the right-hand side $R=\{(x,y) \mid  -v(x-a)+u(y-b) \leq 0\}$ and the left-hand side $L=\{(x,y) \mid  -v(x-a)+u(y-b) \geq 0 \}$. We say that a curve $c$ turns right (\resp left) of $\ell$ if there exist $\epsilon>0$ and $t \in \mathbb{R}$ such that $c(t)$ is a point of $\ell$ and for all $t< z \leq t+\epsilon$, $c(z)$ belongs to $R$ (\resp $L$) and is not on $\ell$. Moreover, we say that {\em $c$ crosses $\ell$ from left to right} (\resp {\em from right to left}) if there exist $t$ and $\epsilon>0$ such that $c$ turns right (\resp left) of $\ell$ at $c(t)$, for all $t-\epsilon \leq z\leq t$, $c(t)\in L$ (\resp $c(t)\in R$) and $c(t-\epsilon)$ is not on~$\ell$.\footnote{\label{fn:cross}The definition of ``$c$ crosses $\ell$ from left to right'' ensures that certain kinds of intersections between $c$ and $\ell$ are not counted as crossing, those include: coincident intersection between two straight segments of length $>0$, and ``glancing'' of a line ``tangentially'' to a corner.} 
We say that {\em  $c$ crosses $\ell$} if $c$ crosses $\ell$ either from left to right or from right to left.  
Now, we generalise these notions to a specific class of polygonal curves:

\begin{definition}\label{def:simple infinite almost-vertical polygonal curve}
A {\em simple infinite almost-vertical polygonal curve} is a simple curve that is a concatenation of an (infinite) vertical ray from the south, a finite curve made of horizontal and vertical segments of lengths 1 or 0.5, and an (infinite) vertical ray to the north.
\end{definition}

Examples of simple infinite almost-vertical polygonal curves appear throughout this paper, including: the curve 
 $c$ in Definition~\ref{def:c} and shown as the border of $\mathcal C$ in Figure~\ref{fig:shield-setup-c}, 
 $e$ in Subsection~\ref{ref:subsubsec: visibility}, 
$c^{m_0}$ in Subsection~\ref{shield:split},
and $d$ in Definition~\ref{def:d}. 

Let $w_0,w_1\in\mathbb{R}^2$ and let $c$ be a curve. We say that {\em $w_0$ is   connected to $w_1$ while avoiding} $c$ if there is a curve $d: [0,1] \rightarrow \mathbb{R}^2$ with $d(0)=w_0$ and $d(1)=w_1$ and $d$ does not intersect $c$. 

\begin{theorem}
  \label{thm:infinite-jordan}
  Let $c : \mathbb{R} \rightarrow \mathbb{R}^2$ be a simple infinite almost-vertical polygonal curve, 
  and let   $x_\textrm{min}$ and $x_\textrm{max}$ be the respective minimum and maximum x-coordinates of $c(\mathbb{R})$. 
  
  Then $c$ cuts $\R^2$ into two connected components:
  \begin{enumerate}
  \item\label{thm:infinite-jordan:conc:1} the left-hand side of $c$:  
  $c(\mathbb{R}) \cup  \{ w \mid w \in\mathbb{R}^2 \textrm{ is connected to } (x_\textrm{min},0) \textrm{ while avoiding } c \}$, 
  \item\label{thm:infinite-jordan:conc:2} the right-hand side of $c$:  
  $c(\mathbb{R}) \cup  \{ w \mid w \in\mathbb{R}^2 \textrm{ is connected to } (x_\textrm{max},0) \textrm{ while avoiding } c \}$. 
  \end{enumerate}
\end{theorem}
\begin{proof}
  The proof is similar to the proof of the Jordan Curve Theorem for polygonal curves.

  For each point $(x, y)\in\R^2$ we define $\ell(x,y)$ to be the line of vector $(1,1)$ through $(x, y)$, $\ell^+(x,y)$ to be the ray of vector $(1,1)$ from $(x, y)$, and $\ell^-(x,y)$ to be the ray of vector $(-1,-1)$ from $(x,y)$.

  Since $c$ is made only of two vertical rays and a {\em finite} number of horizontal and vertical segments,
  for all $(x, y)\in\R^2$,
   $\ell^-(x,y)$ and $\ell^+(x, y)$ intersects $c$ at a finite number of points.
  Moreover, because $\ell(x,y)$ cuts the plane into two connected components, $\ell(x,y)$ crosses $c$
  an odd number of times,

  Let $L$ be the subset of $\R^2$ such that for all $(x, y)\in L$, $c$ crosses\footnoteref{fn:cross} $\ell^+(x, y)$ an odd number of times,\footnote{The term ``crosses'' was defined with respect to a line $\ell(x,y)$. This definition is easily generalised to the (coincident) ray $\ell^+(x,y)$, by considering only those crossings at locations that are simultaneously on both $\ell$ and the ray. Likewise for  the ray $\ell^-(x,y)$.} and let $R$ be the subset of $\R^2$ such that for all $(x, y)\in R$, $c$ crosses $\ell^-(x, y)$ an odd number of times. 
  Note that $L\cap R = c(\R)$  (where $c(\R)$ is the range of~$c$), or in other words the intersection of $L$ and $R$ is the set of all points of $c$.

  We claim that $L$ (\resp $R$) is a connected component: indeed, since $c$ is connected, if $(x_0,y_0)$ and $(x_1,y_1)$ are both in $L$ (\resp both in $R$), let $t_0,t_1\in\R$ be the smallest real numbers such that $\ell^+(x_0,y_0)(t_0)$ (\resp $\ell^-(x_0,y_0)(t_0)$) is on $c$, and $\ell^+(x_1,y_1)(t_1)$ (\resp $\ell^-(x_1,y_1)(t_1)$) is on $c$.
  We know there is at least one such intersection because $\ell^+(x_0,y_0)$ and $\ell^+(x_1,y_1)$ (\resp $\ell^-(x_0,y_0)$ and  $\ell^-(x_1,y_1)$) cross $c$ an odd number of times, hence at least once. Without loss of generality we suppose that $\ell^+(x_0,y_0)(t_0)$ (\resp $\ell^-(x_0,y_0)(t_0)$) is before $\ell^+(x_1,y_1)(t_1)$  (\resp $\ell^-(x_1,y_1)(t_1)$) according to the order of positions along $c$.

  Let $d$ be the curve defined as the concatenation of:
  \begin{itemize}
  \item $\ell^+(x_0,y_0)$ (\resp, $\ell^-(x_0,y_0)$) from $(x_0,y_0)$ up to $\ell^+(x_0,y_0)(t_0)$ (\resp, $\ell^-(x_0,y_0)(t_0)$),
  \item $c$ from $\ell^+(x_0,y_0)(t_0)$ to  $\ell^+(x_1,y_1)(t_1)$ (\resp, from $\ell^-(x_0,y_0)(t_0)$ to  $\ell^-(x_1,y_1)(t_1)$),
  \item $\reverse{\ell^+(x_1,y_1)}$  (\resp, $\reverse{\ell^-(x_1,y_1)}$) from $\ell^+(x_1,y_1)(t_1)$ (\resp, $\ell^-(x_1,y_1)(t_1)$) to $(x_1,y_1)$.
  \end{itemize}
  The curve $d$ is entirely in $L$ (\resp in $R$), starts at $(x_0,y_0)$ and ends at $(x_1,y_1)$, which proves that $L$ (\resp $R$) is indeed a connected component.

  We claim that $L\cup R=\R^2$: indeed, for all $(x,y)\in\R^2$, since $\ell(x,y)$ crosses $c$ an odd number of times, $\ell^+(x,y)$ crosses $c$ an even number of times if and only if $\ell^-(x,y)$ crosses $c$ an odd number of times (and vice-versa),
  hence $(x,y)$ is in at least one of $L$ and $R$, and only the points of $c$ are in both.

  As a conclusion, $c$ cuts the plane into two disjoint connected components:  
   $L\setminus c(\R)$ which we call the \emph{strict left-hand side} of $c$ and
   $R\setminus c(\R) $ which we call the \emph{strict right-hand side} of $c$.
  
  Let $(x,0)$ be a point so far to the east of $c$ that $\ell^+(x,0)$ intersects $c$ exactly once, along $c$'s horizontal ray to the north. 
  $(x,0)$ is in $L$, and since $L$ is connected, we get Conclusion~\ref{thm:infinite-jordan:conc:1}. 
  Let $(x,0)$ be a point so far to the west of $c$ that $\ell^-(x,0)$ intersects $c$ exactly once, along $c$'s horizontal ray to the north. 
  $(x,0)$ is in $R$, and since $R$ is connected, we get Conclusion~\ref{thm:infinite-jordan:conc:2}. 
\end{proof}

Using the technique in the proof of Theorem~\ref{thm:infinite-jordan}, the definitions of turning right and left can be extended from a line to an infinite polygonal curve. 

\begin{definition}[left-hand and right-hand side of a curve]\label{def:curve LHS RHS}
The conclusion of Theorem~\ref{thm:infinite-jordan} defines the {\em left-hand side} $\mathcal{L}\subsetneq \mathbb{R}^2$, and 
{\em right-hand side} $\mathcal{R}\subsetneq \mathbb{R}^2$ of a simple infinite almost-vertical polygonal curve $c$. 
The {\em strict left-hand side} of an infinite almost-vertical polygonal curve $c:\mathbb{R}\rightarrow \mathbb{R}^2$ is the set $\mathcal{L} \setminus c(\mathbb{R})$. 
Likewise the {\em strict right-hand side} of an infinite almost-vertical polygonal curve $c$ is the set $\mathcal{R} \setminus c(\mathbb{R})$, where $\mathcal{R}$ is the right-hand side of $c$.
\end{definition}

We have already defined what ti means for a curve to cross a line. 
Theorem~\ref{thm:infinite-jordan} enables us to generalise that definition to one curve turning from another
simple infinite almost-vertical polygonal curve.   
\begin{definition}[One curve turning left or right from another] \label{def:curve turn}
Let $d$ be a curve and let $c : \mathbb{R}\rightarrow \mathbb{R}^2$ be a simple infinite almost-vertical polygonal curve.
We say that $d$ turns left (\resp, right) from curve $c$ at the point $d(z) = c(w)$, for some $z,w \in \mathbb R$, if  there is an $\epsilon > 0$ such that $d(z+\epsilon)$ is in the strict left-hand (\resp, right-hand) side of $d$ and 
  $d(z')$ is not on $c$, for all $z'$ where $z < z' < z+\epsilon$.
\end{definition}

Definition~\ref{def:curve turn} is consistent with the definition of one path turning left/right from another (Section~\ref{sec:defs-paths}) in the following sense.  
Consider two paths $P$ and $Q$ such that $Q$ turns right (\resp, left) of $P$ and consider a curve $c$ which contains $\embed{P}$ (with the same orientation) then $\embed{Q}$ turns right (\resp, left) of $c$. 
However, not all curves are the embedding of some path, hence the reverse implication does not hold. 
Also, the curve turning definition has no requirement analogous to the orientation requirement for path turns (implied by $i>0$ and the definition of $\Delta$ in the path turning definition).

Throughout the article, we will use the following fact several times. First, let $c$ be a simple infinite almost-vertical polygonal curve (Definition~\ref{def:simple infinite almost-vertical polygonal curve}), 
and let $\ell^1$ and $\ell^2$ be two vertical rays from the south such that $\ell^1$ is the vertical ray from the south used to define $c$, and $\ell^2$ is strictly to the east (\resp, west) of $\ell^1$ and does not intersect $c$, then $\ell^2$ belongs to the right-hand side (\resp, left-hand side) of $c$.
The same conclusion holds if $\ell^1$ and $\ell^2$ go to the north, with $\ell^1$ being the vertical ray from the north used to define $c$.

\bibliographystyle{plain}
\bibliography{pumpable}

\end{document}